\documentclass[12pt,letterpaper]{article}
\usepackage{amsmath,amssymb,mathrsfs,amsthm}
\usepackage{microtype}
\usepackage{booktabs}
\usepackage{graphicx}
\usepackage{enumerate}
\usepackage{bbm}
\usepackage{blkarray}

\usepackage{bm}
\usepackage{listings}\usepackage{caption}
\usepackage{subcaption}
\usepackage{econometrics}
\usepackage{pdflscape}

\usepackage{setspace}
\usepackage{multirow}

\usepackage{xr}

\usepackage[authoryear]{natbib}

\usepackage[colorlinks=true,linkcolor=blue,citecolor=blue,backref=page]{hyperref}
\let \backreforig \backref
\renewcommand*{\backref}[1]{[\backreforig{#1}]}

\usepackage[toc,title,titletoc,header]{appendix}

\allowdisplaybreaks

\usepackage[left=1in,right=1in,top=0.8in,bottom=1in]{geometry}

\theoremstyle{definition}

\usepackage{lscape}

\newtheorem{example}{Example}[section]
\newtheorem{definition}{Definition} 
\newtheorem{theorem}{Theorem}[section]
\newtheorem{corollary}{Corollary}[section]
\newtheorem{lemma}{Lemma}[section]

\newtheorem{remark}{Remark}[section]

\newtheorem{assumption}{Assumption}

\newcommand{\indep}{\perp \!\!\! \perp}

\newcommand{\V}{\operatorname{Var}}
\newcommand{\clt}{\stackrel{d}{\rightarrow}}

\usepackage{lineno}
\setpagewiselinenumbers

\usepackage{rotating}

\usepackage{tikz}

\usetikzlibrary{shapes,decorations,arrows,calc,arrows.meta,fit,positioning}
\tikzset{
-Latex,auto,node distance =1 cm and 1 cm,semithick,
state/.style ={ellipse, draw, minimum width = 0.7 cm},
point/.style = {circle, draw, inner sep=0.04cm,fill,node contents={}},
bidirected/.style={Latex-Latex,dashed},
el/.style = {inner sep=2pt, align=left, sloped}
}

\title{Endogenous Interference in Randomized Experiments }
\author{\large Mengsi Gao\thanks{%
*Department of Economics, University of California Berkeley.
E-mail: mengsi.gao@berkeley.edu.
I am very grateful to my advisors Bryan Graham, Peng Ding, Michael Jansson and Demian Pouzo for their generous support and advice.
I also thank Eric Auerbach, Kirill Borusyak, Yong Cai, Kevin Chen, Patrick Kline, Shuangning Li, Yassine Sbai Sassi and Christopher Walters for valuable comments and discussions. All errors are my own.} \\
\date{
\today 
\medskip  \\
\href{https://mengsigao.github.io/files/JMP_MengsiGao.pdf}{\normalsize \textcolor{blue}{(Click here for the most recent version)}}
}
}

\begin{document}

\maketitle

\onehalfspacing

\thispagestyle{plain}

\pagenumbering{arabic}

\begin{spacing}{1.3}
\begin{abstract}

This paper investigates the identification and inference of treatment effects in randomized controlled trials with social interactions. 
Two key network features characterize the setting and introduce endogeneity: 
(1) latent variables may affect both network formation and outcomes,
and (2) the intervention may alter network structure, mediating treatment effects.
I make three contributions. 
First, I define parameters within a post-treatment network framework, distinguishing direct effects of treatment from indirect effects mediated through changes in network structure. 
I provide a causal interpretation of the coefficients in a linear outcome model.
For estimation and inference, I focus on a specific form of peer effects, represented by the fraction of treated friends.
Second, in the absence of endogeneity, I establish the consistency and asymptotic normality of ordinary least squares estimators. 
Third, if endogeneity is present, I propose addressing it through shift-share instrumental variables, demonstrating the consistency and asymptotic normality of instrumental variable estimators in relatively sparse networks.
For denser networks, I propose a denoised estimator based on eigendecomposition to restore consistency.
Finally, I revisit \cite{Prina2015} as an empirical illustration, demonstrating that treatment can influence outcomes both directly and through network structure changes.

\end{abstract}
\end{spacing}

\medskip
\noindent KEYWORDS: Causal inference, endogeneity, interference, mediation analysis, peer effect, random graph, shift-share instrument variable.

\newpage

\section{Introduction}

Peer effects have been extensively studied in the economics literature. However, identifying these effects can be challenging without randomized experiments. 
Recent research has integrated peer effects with randomized controlled trials (RCTs) to improve their identification across various fields, including education, microfinance, public health, agriculture, and social psychology.\footnote{See, for example, \cite{Sacerdote2001a}, \cite{MiguelKremer2004}, \cite{Sobel2006a}, \cite{BanerjeeChandrasekharDuflo2013}, \cite{CaiJanvrySadoulet2015} and \cite{PaluckShepherdAronow2016}}
These studies go beyond direct effects, exploring how interventions spread through social interaction, potentially amplifying or dampening their impact. This introduces a phenomenon known as ``interference.''\footnote{The term ``interference'' originates from the assumption of no interference between individuals in the potential outcomes framework for causal inference \citep{Cox1958, Rubin1980}. It initially referred to how spillover effects disrupted or ``interfered with'' the standard comparison between treated and untreated groups, complicating the estimation of direct treatment effects. However, the concept of interference has evolved, and in many contexts, it now encompasses situations where spillover effects are not just nuisances but are themselves of primary interest.}
The literature on spillover effects often assumes exogenous networks, where no latent variables influence both network formation and outcomes, and the network remains unchanged after the intervention. 
However, empirical evidence suggests that networks can be endogenous even in randomized experiments, as latent variables may influence both network formation and outcomes, and treatment can alter the network structure.
Motivated by these concerns, this paper investigates the identification and inference of causal effects in RCTs with interference, accounting for both sources of endogeneity.

Many empirical studies suggest that network structures are influenced by latent variables and evolve in response to interventions.
For example, \cite{Prina2015} conducted an RCT offering savings accounts to villagers in Nepal and found significant network changes before and after the intervention.
Similarly, \cite{BanerjeeBrezaChandrasekhar2023} found that introducing formal financial institutions in Indian villages reduced social connections, as access to formal institutions diminished the need for social ties, demonstrated through both an observational study and an RCT.
\cite{BarnhardtFieldPande2017} analyzed a housing lottery program designed to improve housing situations but found that it increased isolation from family and exacerbated financial insecurity. 
Likewise, \cite{CarrellSacerdoteWest2013} conducted a group formation experiment to improve low-ability students' performance, but the intervention backfired as low-ability students formed stronger bonds with similar peers, worsening their outcomes.

In empirical studies, researchers often rely on pre-intervention network data to fit regressions or control for unobserved confounders \citep{CarterLaajajYang2021a}. 
However, this approach can be problematic. 
When post-intervention networks ultimately influence outcomes, relying on pre-intervention networks introduces measurement error. 
Conversely, using post-intervention networks, which are shaped by both treatment and latent variables, introduces endogeneity issues, even in randomized experiments.
Specifically, treatment-induced changes in the network make it challenging to disentangle the direct effects of treatment from the indirect effects mediated by changes in network structure.
Separating these effects is crucial for understanding the mechanisms of the intervention and designing effective policies. 
Furthermore, latent variables complicate the identification and estimation of causal effects, potentially leading to biased and inconsistent estimates of causal effects.
This paper demonstrates a novel approach to utilizing panel network data for separating and consistently estimating these effects by fitting the regression with post-intervention network data while constructing instrumental variables from pre-intervention network data.

To account for network evolution induced by the intervention, I apply the mediation analysis framework \citep{Pearl2001, HeckmanPinto2015} to distinguish the direct effects of treatment from the indirect effects mediated by changes in network structure.
Researchers often use {effective treatment}, also known as {exposure mapping}, for dimensionality reduction \citep{Manski2013, AronowSamii2017}. 
This low-dimensional statistic captures interference patterns by mapping units, treatment assignments, and network structures to the exposures each unit receives.
In my context, since the network is a post-treatment variable, I treat the exposure mapping as a mediator transmitting the indirect effect of the treatment on the outcome.
I define the causal parameters within this framework, distinguishing between the direct effect of the treatment on the outcome and the indirect effect mediated by changes in the network. 
I assume a linear outcome model that incorporates both the treatment and the network mediator while allowing for additive and flexible forms of unobserved confounding.
In this model, the coefficients have clear causal interpretations, and the framework can be applied to any mediator. 

For estimation and inference, I focus on a specific mediator defined as the fraction of treated friends.
This approach aligns with the anonymous interference assumption from \cite{HudgensHalloran2008a} and \cite{Manski2013}.
Consequently, the linear model considered in this paper represents a special case of the linear-in-means (LIM) model \citep{Manski1993a, BlumeBrockDurlauf2015a}, focusing on contextual peer effects, which captures the mean impact of friends' treatments.
Within this LIM framework, the endogeneity issue arises only when both conditions are met: the network depends on the treatment but is not mean-independent, and an unobserved confounder is present.
Although this paper focuses on a specific form of peer effects, the discussion of the SSIV relevance condition offers valuable insights into other types of peer effects.
 
When no endogeneity is present, I propose using OLS estimation with post-intervention network data.
I demonstrate the consistency and asymptotic normality of the OLS estimators across different levels of network sparsity.
With the fraction mediator, I handle the diminishing variation of the regressor in denser networks and find that treatment-induced network changes can increase sample variation in the fraction and enhance the convergence rate.
Additionally, I demonstrate that the standard heteroskedasticity-consistent variance estimator ensures valid inference, even with dependent regressors and nonstandard convergence rates in estimation.

When unobserved confounders are present, I address endogeneity using shift-share (or ``Bartik'') instrumental variables (SSIV) \citep{Bartik1991}. 
The SSIV mitigates endogeneity by combining a set of shocks using exposure share weights.
This paper extends the shift-share design to the network setting by constructing the IV as a combination of the random shock (others' treatment assignments) and non-exogenous exposure (pre-intervention network structure), following \citet{BorusyakHull2023}. 
I provide a thorough examination of the conditions required for a valid IV, with particular emphasis on the relevance condition across varying levels of sparsity.
I demonstrate that the relevance condition for SSIV fails as networks become denser, as increased overlap in friendships reduces variation in SSIV across units, which limits its effectiveness as an instrument for the endogenous network mediator.
Specifically, I show that IV estimators using SSIV are consistent in relatively sparse networks. 
In denser networks, where IV estimators using SSIV may be inconsistent, I adapt the approach of \cite{LiWager2022} to propose an eigendecomposition-based denoised SSIV estimator for consistent estimation.
Furthermore, I establish the asymptotic normality of the IV estimators using both SSIV and modified SSIV and provide consistent variance estimators that account for unit dependencies introduced by the SSIV, ensuring valid inference.

I present simulation evidence supporting my theoretical results for both OLS and IV estimators using SSIV and modified SSIV. 
These results are derived in an asymptotic framework where networks are modeled as random graphs with varying sparsity rates. 
The estimation results across various sample sizes and sparsity levels confirm the asymptotic results of my theoretical analysis.

For empirical illustration, I revisit \cite{Prina2015}, which offered access to formal savings accounts to a random sample of female household heads in 19 villages in Nepal. My results demonstrate that IV estimation effectively addresses the endogeneity issue arising from unobserved confounders. 
Results from IV estimation and the OLS estimates in \citet{Prina2015} can differ in both sign and significance when the indirect effect is significant. 
Additionally, I find that the intervention influences various outcomes through multiple channels: some are directly impacted by the treatment, while others operate through changes in the network mediator, which captures patterns of interference.

\subsection*{Related Literature}

First, it builds on the peer effects literature, which explores how individuals' outcomes are shaped by their peers' behaviors, actions, or characteristics. The standard empirical framework for peer effects is the LIM model, which assumes that agents are influenced by the average actions of their peers. The seminal work of \citet{Manski1993a} formalizes the reflection problem, highlighting the difficulty of disentangling endogenous peer effects from correlated and contextual influences within the LIM framework.
\citet{BramoulleDjebbariFortin2009b} extend this line of inquiry by characterizing the identification conditions for peer effects in network settings, providing crucial insights for empirical applications. In a critical review of the literature, \citet{Angrist2014} scrutinizes various econometric methods and empirical studies on peer effects, offering sharp critiques of existing approaches and underscoring key challenges in this field.

The literature outlines four broad strategies for identifying peer effects while accounting for correlated effects: random peers, random shocks, structural endogeneity, and panel data
\citep{BramoulleDjebbariFortin2020}.
Researchers estimate causal peer effects by studying contexts with randomly assigned peers, through natural or artificial experiments \citep{Sacerdote2001a, CarrellSacerdoteWest2013, DeGiorgiPellizzariRedaelli2010a}. 
Identification relies on the fact that random assignment ensures an agent's characteristics, both observed and unobserved, are uncorrelated with those of her peers.
When peers are not random, researchers seek to identify peer effects using other sources of exogenous variation, such as randomized interventions or quasi-random experiments \citep{DieyeDjebbariBarrera-Osorio2014, NicolettiSalvanesTominey2018, DeGiorgiFrederiksenPistaferri2020, ArduiniPatacchiniRainone2020}.
In the absence of a clear source of exogenous variation, researchers have developed structural frameworks to address correlated effects and identify peer effects in networks \citep{Goldsmith-PinkhamImbens2013b, Graham2015, Graham2017, HsiehLee2016a, HsiehLeeBoucher2019, JohnssonMoon2019}, employing methods such as Bayesian approaches and control function techniques.
Combined with panel data, the inclusion of individual fixed effects enables researchers to control for agents' time-invariant unobserved characteristics, helping to address issues of correlated effects \citep{ArcidiaconoFosterGoodpaster2012, DeGiorgiFrederiksenPistaferri2020, ComolaPrina2021, dePaulaRasulSouza2023}.

This paper contributes to the peer effects literature by leveraging randomized treatment and panel network data in a novel way for identification and estimation, accounting for treatment-induced network changes. 
Specifically, I combine randomized treatments and pre-intervention networks to construct SSIV, instrumenting endogenous variables from the post-intervention network.
This approach is conceptually similar to the use of lagged variables as instruments in panel data models, as developed by \cite{ArellanoBond1991} and later expanded by \cite{BunSarafidis2015}. 
By incorporating pre-intervention networks into the SSIV framework, I aim to address potential endogeneity issues while capturing the dynamic evolution of network structures over time.
The most related work is \citet{ComolaPrina2021}, who study interventions that impact network structure.
They allow for endogenous peer effects but assume conditional exogeneity of pre- and post-intervention networks, addressing different endogeneity sources.
They adapt ``lagged'' partner characteristics as instruments, focusing on identification conditions rather than asymptotic properties.
Other recent studies explore dynamic networks, often focusing on edge-level variables. \citet{Auerbach2022} develops a test to assess how treatments like social programs or trade shocks impact network link formation, while \citet{AuerbachCai2023} examine social disruption by analyzing the formation and dissolution of network connections in response to policy using a random or quasi-random assignment framework.

Second, I contribute to the shift-share IV literature by demonstrating the relevance condition of SSIV across a broad range of network sparsity regimes.
Shift-share specifications are increasingly common in many contexts, including labor, public, development, macroeconomics, international trade, and finance; see \cite{Card2009}, \cite{AutorDornHanson2013}, \cite{NakamuraSteinsson2014}, \cite{BourveauSheZaldokas2020} and \cite{Breuer2022}.
There are two main approaches to identification in the SSIV literature.
\cite{Bartik1991} and \cite{Goldsmith-PinkhamSorkinSwift2020} suggest identification based on the exogeneity of the exposure shares.
In contrast, \cite{BorusyakHullJaravel2022} suggest identification through the quasi-random assignment of shocks, which allows for endogenous exposure shares, as do \cite{AdaoKolesarMorales2019} and \cite{BorusyakHull2023}.
\cite{AdaoKolesarMorales2019} investigate inference in shift-share regression designs and develop new results for the inference that remain valid even in the presence of arbitrary cross-regional correlation in the regression residuals.
Their findings suggest that cluster-robust standard errors, commonly reported in such settings, may lead to the overrejection of the null hypothesis.
\cite{BorusyakHull2023} extend the SSIV idea to the nonlinear case, where multiple sources of variation are combined according to a known formula. 
However, these papers impose a key condition for the relevance of SSIV to hold. 
Specifically, most observations are primarily exposed to a small number of shocks influencing treatment. 
In this paper, I characterize the regimes in which the relevance condition holds, ensuring that IV estimators with SSIV lead to consistent estimation.
I also establish a connection to the work of \cite{LiWager2022}, noting that their estimator for indirect effects fundamentally employs the SSIV approach to address the endogeneity issue of unobserved confounding. 

Third, I contribute to the interference literature by accounting for stochastic and treatment-induced network.
The concept of ``interference'' challenges the ``stable unit treatment value assumption'' (SUTVA), a foundational principle of classical causal inference that assumes no interference between units \citep{Cox1958, Rubin1980, ImbensRubin2015}.
The existing literature on estimating treatment effects under interference primarily follows a design-based approach \citep{HudgensHalloran2008a, AronowSamii2017, AbadieAtheyImbens2020, Leung2022, GaoDing2023}.
It makes no assumptions about outcome models and network formation, and inference is based on random treatment assignment. 
Given the stochastic nature of networks, it is natural to question how this affects inference.
Recent studies have started exploring the impact of network stochasticity by modeling network graphs as realizations from an (unknown) graphon.
For example, \cite{Leung2020} analyzes nonparametric and regression estimators for treatment and spillover effects in sparse networks, while \cite{LiWager2022} examines the asymptotics of treatment effect estimation under network interference, allowing for arbitrary dependencies on unobserved latent variables.
However, these studies do not consider how treatment-induced changes in network structure further influence causal effects.
Many studies on interference assume partial interference, with units divided into clusters where interference occurs only within each cluster \citep{Sobel2006a, HudgensHalloran2008a, TchetgenVanderWeele2012, LiuHudgensSaul2019}. 
In contrast, this paper examines a single large network, allowing for arbitrary interference patterns.

Our paper also relates to the literature on mediation analysis \citep{VanderweeleHongJones2013, Pearl2001, HeckmanPinto2015, ChengGuoLiu2022} by treating exposure mapping as a mediator, accounting for post-intervention network changes. 
Previous studies propose regression estimators for latent mediation but often lack rigorous asymptotic theory or causal interpretation \citep{CheJinZhang2021, LiuJinZhang2021, DiMariaAbbruzzoLovison2022}. \cite{HayesFredricksonLevin2023} consider latent network positions while excluding interference, while \cite{Sweet2018}, \cite{SweetAdhikari2020}, and \cite{GuhaRodriguez2021} treat entire networks as mediators.

\paragraph*{Organization of the paper}
In Section \ref{sec:setup}, I introduce the framework, define the parameters of interest, and provide the identification results.
In Section \ref{sec:OLS}, I discuss the endogeneity issue associated with the fraction mediator and the common practice of using pre-intervention network data to mitigate it. 
In Section \ref{sec:IV}, I explore estimation across various cases.
I first present the asymptotic properties of OLS estimators in the absence of endogeneity.
I then account for unobserved confounders and employ SSIV for estimation, analyzing its asymptotic properties. 
Additionally, I propose a modification to SSIV for cases where the network gets denser.
Section \ref{sec:MC} presents the results of the Monte Carlo simulations. 
Section \ref{sec:app} illustrates the results in an empirical application based on the RCT in \cite{Prina2015}.
Section \ref{sec:conclusion} offers the concluding remarks. 
The appendix of the paper collects all of the proofs, as well as some intermediate results.

\paragraph*{Notation} 
I use $O(), O_{\mathbb{P}}(), o_{\mathbb{P}}(), \asymp, \succ, \prec, \succcurlyeq, \preccurlyeq$ in the following sense: 
$a_n = O(b_n)$ if $|a_n| \leq C b_n$ for $n$ large enough; 
$X_n = O_{\mathbb{P}}(b_n)$, if for any $\delta > 0$, there exists $M, N > 0$, s.t. $\mathbb{P}[|X_n| \geq Mb_n] \leq \delta$ for any $n > N$; 
$X_n = o_{\mathbb{P}}(b_n)$, if $\lim \mathbb{P}[|X_n| \geq \varepsilon b_n] \rightarrow 0$ for any  $\varepsilon > 0$;
$a_n \asymp b_n$ if there exists $k_1, k_2>0$ and $n_0$, s.t. for all $n > n_0$, $k_1 a_n \leq b_n \leq k_2 a_n$;
$a_n \succ b_n$ if $\lim a_n/b_n = \infty$;
$a_n \prec b_n$ if $\lim a_n/b_n = 0$;
$a_n \succcurlyeq b_n$ if $b_n = O(a_n)$;
$a_n \preccurlyeq b_n$ if $a_n = O(b_n)$.
Let $1_{n-1}$ represent a vector of ones with $n - 1$ elements, and $0_{n-1}$ represent a vector of zeros with $n - 1$ elements.
Let $\| \cdot \|_{\textup{op}}$ denote the operator norm. Define $E_{w_i}(\cdot)$ as the expectation taken over the marginal distribution of $w_i$. Let $X_{-i}$ represent the set $\{X_j\}_{j \neq i}$. The abbreviation ``i.i.d.'' stands for ``independent and identically distributed.''

\section{Setup} \label{sec:setup}

\subsection{Framework and Notation}\label{sec:Framework}
I consider an RCT with $n$ participants. 
For each participant $i\in\{1,\dots,n\}$, let $Y_i \in \mathbb{R}$ denote the observed outcome of interest, $T_i \in \{0,1\}$ denote the treatment assignment where $T_i \overset{}{} \sim \text{Bernoulli}(\pi)$ for some $0<\pi<1$, and $w_i$ denote the unobserved covariates.
I consider a single large network, allowing for an arbitrary form of interference within it. 
The network structure is represented by an adjacency matrix $A = \{A_{ij}\}_{i,j=1}^n$, where the $(i,j)$th entry $A_{ij} \in \{0,1\}$ indicates whether units $i$ and $j$ are connected. 
The adjacency matrix is assumed to be undirected (symmetric), unweighted (binary values), and has no self-links ($A_{ii} = 0$).
I differentiate between two network structures: 
$A^\text{pre}$, which represents the network observed before the intervention, and $A^\text{post}$, which corresponds to the network observed after the intervention. I assume the researchers observe both $A^\text{pre}$ and $A^\text{post}$.

Figure \ref{fig:model_dag} illustrates the causal mechanism considered in this paper.\footnote{I use the graph for illustration without using formal graphical language.}
I assume that the post-intervention network $A^{\text{post}}$, influenced by the treatment vector, plays a critical role in determining the outcome. 
Consequently, the exposure mapping, also known as the effective treatment, which maps $\{T_i\}_{i=1}^n$ and $A^{\text{post}}$ to some low-dimensional statistics, serves as a mediator through which the treatment indirectly affects the outcome.
I denote this mediator as $M_i$.
Specifically, I assume that $T_i$ has both a direct effect on $Y_i$ and an indirect effect through changing network mediator $M_i$. 
The network $A^{\text{post}}$ is shaped by both the treatments and the latent variable $w_i$, which may also influence the outcome. 
As a result, $w_i$ acts as a confounder between the mediator and the outcome. 
The treatments of other units $T_{-i}$ affect $Y_i$ exclusively through the network mediator $M_i$.

\begin{figure}
\centering
\begin{tikzpicture}
\node[draw, circle, minimum size=1cm, text centered] (t1) {$T_i$};
\node[draw, circle, minimum size=1cm, above = 0.8 of t1, text centered] (t2) {$T_{-i}$};
\node[draw, circle,  minimum size=1cm, above = 2 of t1, text centered] (wj) {$w_{-i}$};
\node[draw, circle,  minimum size=1cm, right = 2.5 of t1, text centered] (m)
{$M_i$};
\node[draw, circle, minimum size=1cm, right = 2.5 of m, text centered] (y) {$Y_i$};
\node[draw, circle, minimum size=1cm, above = 2 of m, text centered] (x) {$w_i$};
\draw[->, line width=1] (t1) -- node[above, font=\footnotesize]{} (m);
\draw[-> , line width= 1] (t2) -- node[above, font=\footnotesize]{} (m);
\draw[-> , line width= 1 ] (wj) -- node[above, font=\footnotesize]{ } (m);
\draw [->, line width= 1] (m) -- node[above, font=\footnotesize]{ } (y);
\draw[->,line width= 1] (x) --node[right, font=\footnotesize]{}(m);
\draw[->,line width= 1] (x) -- node[right, font=\footnotesize]{}(y);
\draw[-> , line width=1 ] (t1) to [out=315,in=225, looseness=0.5] node[above, font=\footnotesize]{ } (y) ;
\end{tikzpicture}
\caption{Causal mechanism.}
\label{fig:model_dag}
\end{figure}
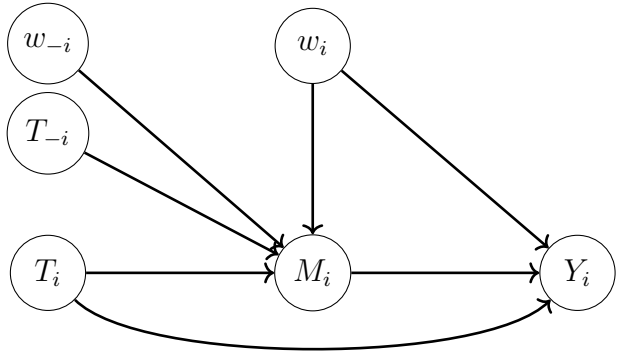

Mediators can take various forms. For example:
\begin{equation}
M_i = \frac{\sum_{j=1}^n A_{ij}^{\text{post}} T_j}{\sum_{j=1}^n A_{ij}^{\text{post}}}, 
\label{eq:Mi}
\end{equation}
which measures the fraction of treated friends after the intervention. This mediator depends solely on the number of (treated) friends, regardless of their identity. 
It is a specific form of the anonymous interference assumption proposed by \cite{HudgensHalloran2008a}, also referred to as the anonymous interactions assumption by \cite{Manski2013}.
Other examples of anonymous interference include:
\begin{enumerate}[(1)]
\item $M_i = \sum_{j=1}^n A_{ij}^{\text{post}} T_j$, which measures the total number of treated friends after the intervention;
\item $M_i = 1\{ \sum_{j=1}^n A_{ij}^{\text{post}} T_j > 0\}$, which measures whether there is at least one treated friend after the intervention.
\end{enumerate}

\subsubsection*{Potential Value Notation}
To define the causal parameters of interest, I introduce the potential value notations for the mediator and the outcome, following the mediation analysis literature \citep{RobinsGreenland1992, Pearl2001, HeckmanPinto2015}. 
First, I consider a hypothetical intervention on the treatment vector $t$ and define the potential values of the mediators for unit $i$ corresponding to this intervention as:
\begin{equation*}
\left\{ M_i(t): t \in \{0,1\}^n \right\},
\end{equation*}
Note that potential mediator $M_i(t)$ depends on the treatment vector $t$ in two ways: it is a direct function of $t$ and the network structure $A^{\text{post}}$, which implicitly depends on $t$ as well.
I then consider a hypothetical intervention on both $t_i$ and $m_i$. 
Define the potential outcomes corresponding to the interventions on $t_i$ and $m_i = M_i(t_i, t_{-i})$ for unit $i$ as:
\begin{equation*}
\left\{ Y_i(t_i, m_i): t_i \in \{0,1\}, m_i \in \mathcal{M} \right\},
\end{equation*}
where $\mathcal{M}$ contains all possible values of $m_i$.
This notation reflects that the potential outcome $Y_i(t_i, m_i)$ depends not only on $t_i$, the treatment assignment of unit $i$, but also on the network changes induced by the intervention, as captured by the function $M_i$. 
Additionally, it indicates that the treatment assignments of others influence the potential outcome of individual $i$ only indirectly through the mediator $M_i$.
I also define the nested potential outcomes corresponding to an intervention on $t_i$ and $m_i = M_i(t_i', t_{-i})$ as:
\begin{equation*}
\left\{ Y_i(t_i, M_i(t_i', t_{-i})): t_i \in \{0,1\}, t_i' \in \{0,1\}, t_{-i} \in \{0,1\}^{n-1} \right\}.
\end{equation*}
The notation $Y_i(t_i, M_i(t_i', t_{-i}))$ represents the hypothetical outcome when the treatment is set to level $t_i$ and the mediator is set to its potential value $M_i(t_i', t_{-i})$, corresponding to the treatments $t_i'$ for unit $i$ and $t_{-i}$ for the other units.
I allow $t_i$ and $t_i'$ to differ to define counterfactual outcomes where either the treatment level or the mediator level changes, while the other remains fixed. 
This approach enables the separation of the two channels through which the treatment affects the outcome, either directly or indirectly.
The observed values of the mediator $M_i$ and the outcome $Y_i$  are related to their potential values as follows:
\[
M_i = M_i(t_i, t_{-i}) \text{ if } T_i = t_i \text{ and } T_{-i} = t_{-i},
\]
and
\[
Y_i = Y_i(t_i, m_i) \text{ if } T_i = t_i \text{ and } M_i = m_i.
\]

\begin{remark}\label{remark:Y}
The interference literature typically defines potential outcomes as a function of the treatment vector \citep{HudgensHalloran2008a}:
\begin{equation*}
\left\{ Y_i(t_i, t_{-i}): t_i \in \{0,1\}; t_{-i} \in \{0,1\}^{n-1} \right\},
\end{equation*}
which does not account for how an individual's treatment affects their exposure to others in the network, thereby indirectly influencing the outcome.
Specifically, it cannot capture the scenario where the treatment or mediator levels change while the other remains fixed, i.e., when $t_i \ne t_i'$.
The expression $Y_i(t_i, M_i(t_i', t_{-i}))$ nests $Y_i(t_i, t_{-i})$ as a special case when $t_i = t_i'$ and $M_i = t_{-i}$.
In the causal graph shown in Figure \ref{fig:model_dag}, the arrow from $T_i$ to $M_i$ would be omitted if the treatment does not affect the network.
\end{remark}

\subsubsection*{Data-Generating Process for Networks}
I specify that the relationship between unit $i$ and $j$ before and after the intervention is determined by the following random graph models:
\begin{align}
A_{ij}^\text{pre} 
&= {1}\left\{\eta_{ij} \leq q_n^\text{pre} g^\text{pre}(w_i, w_j) \right\} 
{1}\{j \neq i\}, 
\label{eq:pre}  \\
A_{ij}^\text{post} 
&= {1}\left\{\eta_{ij} \leq q_n^\text{post} g^\text{post}(w_i, w_j, T_i, T_j)\right\} 
{1}\{j \neq i\},
\label{eq:post}
\end{align}
where $\{w_i\}_{i=1}^n$ are the i.i.d. latent variables and 
$\{\eta_{ij}\}_{i,j=1}^n$ is a symmetric matrix of unobserved scalar disturbances with upper diagonal entries that are mutually independent. 
I allow the sparsity parameters, $q_n^\text{pre}$ and $q_n^\text{post}$, to differ after the intervention for generality. 
The pre-intervention graphon, $g^\text{pre}$, depends on the latent variable of the pairs, $w_i$ and $w_j$, while the post-intervention graphon, $g^\text{post}$, also incorporates the treatment assignments $T_i$ and $T_j$ in forming links. 
Since the idiosyncratic term $\eta_{ij}$ and the latent variable $w_i$ are not time-varying, the framework to account for cases where the network remains unchanged after the intervention, with $g^\text{pre} = g^\text{post}$.\footnote{The assumption that $\eta_{ij}$ is not time-varying can be relaxed. Specifically, $\eta_{ij}$ in the pre- and post-intervention networks can be treated as independent draws from the same underlying distribution or as draws from correlated distributions. The conjecture is that the SSIV constructed based on $A^{\text{pre}}$ may be weaker, as additional exogenous shocks reduce its relevance.}
Changes in the network are primarily attributed to the intervention, while exogenous shocks unrelated to the intervention are reflected in the difference between $g^\text{pre}$ and $g^\text{post}$.
Thus, the network formation described in \eqref{eq:pre} and \eqref{eq:post} captures the features that network formation is driven by latent variables and may evolve over time, with changes induced by the intervention.

\begin{remark}
The latent variable $w_i$ may be either a scalar or a vector, with the key requirement being a common or correlated component that influences both the pre- and post-intervention networks. This relationship ensures that the pre-intervention network is predictive of the post-intervention network, allowing it to serve as an instrument.
\end{remark}

The distribution of \(\eta_{ij}\) is not separately identified from the graphons and sparsity parameters, so it is typically normalized to follow a standard uniform distribution. 
As a result, $q_n^{\diamond} g^{\diamond}$ represents the probability that a given pair is connected, for $\diamond \in \{\textup{pre}, \textup{post}\}$.
The graphon-based network model in \eqref{eq:pre} and \eqref{eq:post}, which takes the form of an inhomogeneous Erdős-Rényi graph, is motivated by Aldous-Hoover Theorem on exchangeable arrays \citep{Aldous1981, LovaszSzegedy2006,BickelChen2009} and has recently gained considerable attention in econometrics \citep{GaoLuZhou2015, ZhangLevinaZhu2017, Graham2020b, PariseOzdaglar2023}. 
Recent studies have also employed this model to account for randomness in network formation; 
see \cite{Auerbach2022a}, \cite{Cai2022}, and \cite{LiWager2022}. 
The sparsity parameters $q_n^{\diamond}$, for $\diamond \in \{\textup{pre}, \textup{post}\}$, serve as theoretical tools, (possibly) driving the probability of any pair being connected to zero as $n \to \infty $. 
The goal of this paper is to derive theoretical results across a broad range of sparsity rates, encompassing the following three cases:
\begin{enumerate}[(a)]
\item $q_n^{\diamond} \asymp n^{-1}$ (bounded degree graph): 
the total number of friends for each unit remains bounded and not vanishing in expectation;

\item $\lim_{n\to\infty} q_n^{\diamond} = 0$ and $\lim_{n\to\infty} n q_n^{\diamond} = \infty$ (sparse graph): the probability of any pair forming a connection decreases as the sample size increases, while the total number of friends for each unit would increase in expectation;

\item $q_n^{\diamond} \asymp 1$ (dense graph): 
the total number of friends for each unit is approximately of order $n$ in expectation, growing at the same rate as the sample size.
\end{enumerate}

I summarize the assumption on the network models below.

\begin{assumption}\label{asu:network}
The latent variables $\{w_i\}_{i=1}^n$ are i.i.d. and independent of $\{\eta_{ij}\}_{i,j=1}^n$, where $\eta_{ij} \overset{\text{i.i.d.}}{\sim } U[0,1]$ for $j>i$ and $\eta_{ij} = \eta_{ji}$.
The networks are randomly generated according to equations \eqref{eq:pre} and \eqref{eq:post}, where $g^\text{pre}$ is a symmetric measurable function of $w_i$ and $w_j$, and $g^\text{post}$ is a symmetric measurable function of $(T_i, w_i)$ and $(T_j, w_j)$, both mapping into $[0,1]$.
The sparsity parameters satisfy 
$n^{-1} \preccurlyeq q_n^{\diamond} \preccurlyeq 1$,
for $\diamond \in \{\textup{pre}, \textup{post}\}$.
\end{assumption}

\subsection{Parameters of Interest and Identification}

Now I define the causal parameters of interest in my setting.
\begin{definition}\label{def:parameter}
\begin{enumerate}[(1)]
For $t\in \{0,1\}$, define
\item the average \textit{total effect} $\textup{ToE}$ of unit $i$'s treatment as
\begin{equation}
\textup{ToE} = E\left[Y_i\left(1, M_i\left(1, T_{-i}\right) \right) - Y_i\left(0, M_i\left(0, T_{-i}\right) \right)\right];
\label{eq:TE}
\end{equation}
\item the average \textit{direct effect} $\textup{DE}(t)$ of unit $i$'s treatment as
\begin{equation}
\textup{DE}(t) 
= E\left[Y_i\left(1, M_i\left(t, T_{-i}\right) \right)-Y_i\left(0, M_i\left(t, T_{-i}\right) \right)\right];
\label{eq:DE}
\end{equation}
\item the average \textit{indirect effect} $\textup{IE}(t)$ of unit $i$'s treatment as
\begin{equation}
\textup{IE}(t) 
= E\left[Y_i\left(t, M_i\left(1, T_{-i}\right) \right) - Y_i\left(t, M_i\left(0, T_{-i}\right)  \right)\right];
\label{eq:IE}
\end{equation}
\item the average \textit{spillover effect} $\textup{SE}(t)$ as
\begin{equation}
\textup{SE}(t)
= E\left[Y_i\left(t, M_i\left(t_i = t, t_{-i} = 1_{n-1} \right)  \right) - Y_i\left(t, M_i\left(t_i = t, t_{-i} = 0_{n-1} \right) \right)\right]. 
\label{eq:SE}
\end{equation}
\end{enumerate}
\end{definition}

The average \textit{total effect} $\textup{ToE}$ measures the average overall effect of changing the treatment of unit $i$ on that unit's outcome.
I then decompose the $\textup{ToE}$ into two channels: one without changes in the mediator and the other through changes in the mediator.
The average \textit{direct effect} $\textup{DE}(t)$ captures the former, measuring the average effect of changing the treatment of unit $i$, while keeping the mediator at the level it would have taken if $T_i$ had been set to $t$.
The average \textit{indirect effect} $\textup{IE}(t)$ captures the latter, measuring the average effect of changing the treatment of unit $i$ solely through its impact on the mediator while keeping the intervention on unit $i$ at level $t$.
The average \textit{spillover effect} $\textup{SE}(t)$ measures the effect on unit $i$'s outcome of assigning everyone else to the treatment group versus the control group, while holding unit $i$'s treatment status fixed at level $t$.\footnote{The $\textup{SE}(t)$ in \eqref{eq:SE} is one way to define spillover effects. 
Alternatively, spillover effects can be defined as the change in unit $i$'s outcome resulting from differences between any two distinct intervention programs through variations in the treatment of others.}
In defining these causal effects, the expectation is taken over all sources of randomness, including potential outcomes and the treatment assignments of others.

\begin{remark}
The existing literature on interference typically defines the direct and indirect effects of a binary treatment, based on the potential outcome notation in Remark \ref{remark:Y}. For example, \cite{HuLiWager2022b} defined the average direct effect of a binary treatment as:
\[
\tau_{\textsc{ade}} 
= \frac{1}{n} \sum_{i=1}^n 
{E} \left[ Y_i(t_i = 1, T_{-i}) - Y_i(t_i = 0, T_{-i}) \right],
\]
and the average indirect effect as:
\begin{equation}
\tau_{\textsc{aie}} 
= \frac{1}{n} \sum_{i=1}^n \sum_{j \neq i} 
{E}\left[ Y_j(t_i = 1, T_{-i}) - Y_j(t_i = 0, T_{-i}) \right].  
\label{eq:indirect}
\end{equation}
The estimand $\tau_{\textsc{ade}}$ corresponds to the $\textup{ToE}$ in \eqref{eq:TE}. To better understand the intervention mechanism, I further decompose the effect of $T_i$ on $Y_i$ into the $\textup{DE}(t)$ in \eqref{eq:DE} and the $\textup{IE}(t)$ in \eqref{eq:IE}, which operates through changes in the mediator.
\end{remark}

The mediation formula relies on the following assumption.
\begin{assumption}\label{ass:uncon}
For all $i=1,\cdots,n$,
\begin{enumerate}[{(a)}]
\item $ (T_i, T_{-i}) \indep Y_{i}\left(t_i, m_i \right)$ for all $t_i \in \{0,1\}$ and $m_i \in \mathcal{M}$;
\item $ (T_i, T_{-i}) \indep M_i(t_i, t_{-i}) $ for all $t_i \in \{0,1\}$ and $t_{-i} \in \{0,1\}^{n-1}$;
\item $M_i \indep Y_i(t_i, m_i) \mid w_i$ for all $t_i \in \{0,1\}$ and $m_i \in \mathcal{M}$;
\item $M_i(t_i', t_{-i}) \indep   Y_{i}\left(t_i, m_i \right)  \mid w_i$ for all $t_i\in \{0,1\}$, $t_i'\in \{0,1\}$, $t_{-i} \in \{0,1\}^{n-1}$ and $m_i \in \mathcal{M}$.
\end{enumerate}
\end{assumption}
Assumption \ref{ass:uncon} is standard in the literature of mediation analysis \citep{Pearl2001}. 
Assumptions \ref{ass:uncon}(a)-(b) assume no treatment-outcome confounding and no treatment-mediator confounding, respectively.
Assumptions \ref{ass:uncon}(a)-(b) hold under experiments with randomized treatment.
Assumption \ref{ass:uncon}(c) assumes that the variable $w_i$ captures all the confounding factors between the mediator and the outcome. 
Assumption \ref{ass:uncon}(d) assumes the cross-world independence between the potential outcomes and potential mediators.
However, Assumption \ref{ass:uncon}(d) can never be validated, as it is impossible to observe both $M_i(t_i', t_{-i})$ and $Y_{i}\left(t_i, m_i \right)$ in any experiment if $t_i \ne t_i'$.

Theorem \ref{thm:identifiaction} expresses the parameters of interest in terms of the observed data up to the unknown distribution of $w_i$.
\begin{theorem}\label{thm:identifiaction}
Under Assumption \ref{ass:uncon},
\begin{enumerate}[(1)] 
\item
$\textup{ToE} 
= E(Y_i \mid T_i=1) - E(Y_i \mid T_i=0)$;
\item 
$\textup{DE}(t)
= E_{w_i}
\left\{ 
E\left[ E\left(Y_i \mid M_i, T_i = 1, w_i \right) - E\left(Y_i \mid M_i, T_i = 0, w_i \right) \mid T_i = t, w_i \right] 
\right\}$;
\item 
$\textup{IE}(t)
= E_{w_i}
\left\{ 
\begin{array}{c}
E\left[ E\left(Y_i \mid M_i, T_i = t, w_i \right) \mid T_i = 1, w_i \right] \\
- E\left[ E\left(Y_i \mid M_i, T_i = t, w_i \right) \mid T_i = 0, w_i \right]
\end{array}\right\} $;
\item 
$\textup{SE}(t) 
= E_{w_i} 
\left\{
\begin{array}{c}
E\left[ E\left(Y_i \mid M_i, T_i = t, w_i \right) \mid T_i = t, T_{-i} = 1_{n-1} , w_i \right] \\
- E\left[ E\left(Y_i \mid M_i, T_i = t, w_i \right) \mid T_i = t, T_{-i} = 0_{n-1}, w_i \right] 
\end{array}
\right\}$.
\end{enumerate}
\end{theorem}
If the parameter of interest is solely $\textup{ToE}$, it can be identified, and consistent estimators can be obtained by regressing the outcome on the treatment indicator with an intercept, or equivalently, by using the difference-in-means estimator, provided the treatment is randomly assigned.
However, if the focus is on distinguishing between the direct and indirect effects, or when the spillover effect is of interest, these causal effects are identified only up to the unknown distribution of $w_i$.

Following the classical Baron-Kenny method \citep{BaronKenny1986}, I assume that the potential outcome is linear in the treatment and mediator, and additive with respect to any form of $w_i$, as stated in Assumption \ref{asu:linearY}, which implies Assumptions \ref{ass:uncon}(c)-(d). However, I relax the assumption that the mediator $M_i$ is a linear function of $T_i$ and $w_i$.

\begin{assumption}\label{asu:linearY}
Assume the following (partially) linear model for the potential outcome:
\begin{equation}
Y_i(t_i, m_i) = \beta_0 + \beta_1 t_i + \beta_2 m_i + \lambda(w_i) + \varepsilon_i, 
\label{eq:linearY}
\end{equation}
where $\lambda(\cdot)$ is an unknown measurable function with $E(\lambda(w_i)) = 0$. 
Additionally, assume that $\{\varepsilon_i\}_{i=1}^n$ are i.i.d., $E(\varepsilon_i^2) < \infty$, and $E(\varepsilon_i \mid T_i, T_{-i}, A^{\text{post}}) = 0$.
\end{assumption}

Under Assumption \ref{asu:linearY}, the observed outcome also follows the linear model:
\[
Y_i = \beta_0 + \beta_1 T_i + \beta_2 M_i 
+ \lambda(w_i) + \varepsilon_i.
\]
Let $u_i$ denote the error term in the outcome model: $u_i = \lambda(w_i) + \varepsilon_i$.

Although the framework in Section \ref{sec:setup} applies to general forms of mediators, Sections \ref{sec:OLS} and \ref{sec:IV} focus on the estimation and inference with the mediator specified in \eqref{eq:Mi}.
I follow the convention of $0/0=0$.
I focus on $M_i$  in \eqref{eq:Mi} for several reasons. First, different forms of mediators exhibit varying levels of dependency and require different proofs. 
The magnitude of $M_i$  in \eqref{eq:Mi} is normalized and remains bounded as networks become denser and sample sizes increase, ensuring the outcome does not diverge to infinity under a constant treatment effect.
Second, with this fraction as the mediator, the linear outcome model introduced in Assumption \ref{asu:linearY} becomes a special case of the linear-in-means model \citep{Manski2013}. 
Incorporating endogenous peer effects is also of interest and is left for future work. 

\begin{remark}\label{re:Auerbach}
\cite{Auerbach2022} investigates the identification and estimation of the following partially linear model, with the parameters of interest being $\beta$ and $\lambda(w_i)$: 
\begin{align} 
Y_i 
&= \beta^\top X_i + \lambda(w_i) + \varepsilon_i, \label{eq:Auerbach2022}
\end{align}
where $w_i$ is some unobserved latent variable that also drives the network formation:
\[
A_{ij} 
= {1}\left\{\eta_{ij} \leq g(w_i, w_j)\right\} 
{1}\{j \neq i\}. 
\]
\cite{Auerbach2022a} focuses on the dense network, where the sparsity parameter equals $1$.
The regressor $X_i$ in \eqref{eq:Auerbach2022} corresponds to $(1, T_i, M_i)$ in my setting.
\cite{Auerbach2022a} uses a matching approach based on network data $A_{ij}$ to identify and estimate $\beta$ and $\lambda(w_i)$. However, this method relies on dense networks and sufficient variation in the regressors after controlling for the latent variable.
In my context, which involves exogenous shocks from the randomized treatment and does not prioritize identifying the latent variable part $\lambda(w_i)$, I tackle the endogeneity issue from the unobserved confounder using the IV method with SSIV.
\end{remark}

Corollary \ref{cor:identifiaction} simplifies the statement of Theorem \ref{thm:identifiaction} under Assumption \ref{asu:linearY}, and provides the causal interpretation of the coefficients, which is analogous to the Baron-Kenny formulas for mediation \citep{BaronKenny1986}.

\begin{corollary}\label{cor:identifiaction}
Under Assumption \ref{asu:linearY}, then
\begin{enumerate}[(1)]
\item 
$\textup{DE}(1) = \textup{DE}(0) = \textup{DE} = \beta_1$;
\item 
$\textup{IE}(1) = \textup{IE}(0) = \textup{IE} = \beta_2 \cdot \left\{ E(M_i \mid T_i=1) - E(M_i \mid T_i=0) \right\}$;
\item 
$\textup{ToE} = \beta_1 + \beta_2 \cdot \left\{ E(M_i \mid T_i=1) - E(M_i \mid T_i=0)\right\}$;
\item 
$\textup{SE}(t) 
= \beta_2 \cdot \left\{ E(M_i \mid T_i=t, T_{-i} = 1_{n-1}) - E(M_i \mid T_i=t, T_{-i} = 0_{n-1})\right\}$.
\end{enumerate}
\end{corollary}

Under Assumption \ref{asu:linearY}, where the treatment $T_i$ does not interact with the mediator $M_i$, $\textup{DE}(t)$ and $\textup{IE}(t)$ do not depend on the value of $t$.
Furthermore, Corollary \ref{cor:identifiaction} shows that $\beta_1$ captures the direct effect of the treatment, while the indirect and spillover effects depend on $\beta_2$, scaled by the magnitude of changes in the mediator in response to changes in $T_i$ and $T_{-i}$.
When the post-intervention network $A^{\text{post}}$ is uncorrelated with the treatment, the mediator does not respond to changes in $T_i$, resulting in a zero indirect effect and making the total effect equivalent to the direct effect.
However, the spillover effect would still be present.
Additionally, with $M_i$ in \eqref{eq:Mi}, $ E(M_i \mid T_i=t, T_{-i} = 1_{n-1}) - E(M_i \mid T_i=t, T_{-i} = 0_{n-1}) = 1$, indicating that $\beta_2$ captures the spillover effect $\textup{SE}(t)$ defined in \eqref{eq:SE}. 
Corollary \ref{cor:identifiaction} justifies the emphasis on the estimation and inference of $\beta_1$ and $\beta_2$ in this paper.

To summarize, I account for two key features of networks when estimating $\beta$'s:
\begin{enumerate}[(1)]
\item There exists some latent variable $w_i$ that influences both network formation, and consequently the mediator, as well as the outcome of interest.

\item Treatment assignments affect how units are exposed to others in the network, with the exposure mapping serving as a mediator for the treatment’s effect on the outcome. 

\end{enumerate}
As a result, even in a randomized setting, if the causal effects of interest involve the network, such as the $\textup{IE}(t)$ in \eqref{eq:IE} and the $\textup{SE}(t)$ in \eqref{eq:SE}, endogeneity becomes a concern, potentially biasing the estimation of $\beta_2$. 
Given the correlation between $T_i$ and $M_i$, this bias in estimating $\beta_2$ can further affect the estimation of $\beta_1$.

\subsection{Motivating Examples}

I present three empirical examples illustrating how social network structures can endogenously evolve in response to interventions, potentially undermining their intended effects.

\begin{example}\label{ex:Prina2015}
\cite{Prina2015} conducted an RCT to assess the impact of offering access to formal savings on households' financial situations in 19 villages in Nepal with 915 households.
Half of the female household heads were randomly offered the savings accounts, while the other half were not.
The financial links were measured by asking survey questions such as ``who did you exchange loans or gifts with?''
The network experienced considerable reshuffling, despite the total number of links remaining nearly constant (328 at baseline, 329 at endline). 
Of these, only 73 links persisted from baseline, while 255 links were broken after the intervention, and 256 new links were formed by the endline.
Importantly, the distribution of treat-treat, control-control, and treat-control pairs among these link types was imbalanced. 
Of the persisting links, 26 were treat-treat pairs, 17 were control-control pairs, and 30 were treat-control pairs. 
Among the broken links, 68 were treat-treat, 74 were control-control, and 113 were treat-control. 
In contrast, among the newly formed links, 78 were treat–treat, 53 were control-control, and 125 were treat-control.
I will revisit this study in Section \ref{sec:app} as an empirical illustration.
\end{example}

\begin{example}\label{ex:Banerjee2023}
\cite{BanerjeeBrezaChandrasekhar2023} investigated the effect of introducing formal financial institutions on informal lending practices and social networks through studies in two distinct settings. 
In their first study, they analyzed the non-random introduction of microfinance in 43 out of 75 villages in Karnataka, India. 
They found that social networks contracted more in villages where microfinance was introduced. 
To validate these findings, they conducted an RCT in Hyderabad, India, where 52 out of 104 neighborhoods were randomly selected for microfinance introduction. 
Similar to the first setting, these neighborhoods also experienced a reduction in social connections, even among households initially unlikely to borrow.
Together, these studies suggest that access to microfinance reduces the incentive to maintain and form social links, affecting both borrowers and non-borrowers alike.
\end{example}

\begin{example}\label{ex:Carrell2013}

\cite{CarrellSacerdoteWest2013} examined the impact of group randomization on academic performance at the United States Air Force Academy, with the hypothesis that low-ability students would benefit from exposure to high-ability peers.
Incoming freshmen were randomly assigned to treatment or control groups. The control group followed the usual distribution of abilities, while the treatment group paired low-ability students with high-ability peers and placed middle-ability students in separate squadrons.
Contrary to expectations, the intervention had a significantly negative effect on the academic performance of low-ability students. 
While several factors could explain this outcome, the study found that the endogenous sorting of roommates, study partners, and friends evolved differently in the treatment group compared to the control group. 
Specifically, low-ability students in the treatment group saw a notable increase in the number of low-ability study partners and friends. 
\end{example}

Examples \ref{ex:Prina2015}-\ref{ex:Carrell2013} also highlight concerns about latent variables that influence network formation.
Example \ref{ex:Prina2015} demonstrates that transfers were linked to treatment households having more assets and greater financial inclusion at baseline, suggesting that households with more resources and higher socioeconomic status may also increase transfers to others.
Example \ref{ex:Banerjee2023} suggests that access to microfinance reduces borrowers' incentives to maintain connections, leading even non-participants to scale back efforts to form links. This is partly because individuals connected to potential borrowers see diminished returns from these relationships.
Example \ref{ex:Carrell2013} indicates that the increase in low-ability study partners and friends was not merely due to the higher proportion of low-ability peers in the treatment group but also reflected a pattern of homophily, as students gravitated toward others with similar abilities.

\medskip

\section{OLS estimation} \label{sec:OLS}
In this section, I begin with discussing the properties of $M_i$ and the conditions under which endogeneity arises in Section \ref{sec:Mi}. 
In Section \ref{sec:ovb}, I  discuss the bias resulting from the common practice in empirical studies of collecting and using only pre-intervention network data.
When no endogeneity arises, I explore the asymptotic properties of OLS estimators in Section \ref{sec:ols}. 

\subsection{Property of $M_i$}\label{sec:Mi}
Combined with the form of mediator in \eqref{eq:Mi}, the outcome model of interest is
\[
Y_i = \beta_0 + \beta_1 T_i + \beta_2 \frac{\sum_{j=1}^n A_{ij}^\text{post} T_j }{\sum_{j=1}^n A_{ij}^\text{post} } + \lambda(w_i) + \varepsilon_i.
\]
In addition to assuming anonymous interference, this model assumes ``no second- or higher-order effects,'' meaning the treatment of units beyond direct connections has no impact.

The properties of the estimator depend critically on the characteristics of $M_i$.
The regressor $M_i$ is dependent across units, and its fractional form raises natural concerns about whether it provides sufficient variation, potentially leading to (near) multicollinearity issues.
To better understand the properties of $M_i$, I decompose it as $M_i = \xi_i + r^*_i$, where $\xi_i$ is defined as the ratio of the conditional expectations of the numerator and denominator of $M_i$ in \eqref{eq:Mi}:
\begin{equation}
\xi_i = \frac{\mathbb{E}(A_{ij}^\text{post} T_j \mid T_i, w_i)}{\mathbb{E}(A_{ij}^\text{post} \mid T_i, w_i)}
= P(T_j = 1 \mid A_{ij}^\text{post} = 1, T_i, w_i)
\label{eq:xi}
\end{equation}
and $r^*_i$ is some remainder term.
Note that $\xi_i$ represents the conditional probability that a neighbor of individual $i$ (i.e., a connected individual $j$) is treated, given $i$'s own characteristics $w_i$ and treatment status $T_i$.
Throughout the paper, I present the results based on the following two cases, depending on whether $\xi_i$ is degenerate or not:
\begin{enumerate}[(a)]
\item $\xi_i$ is an i.i.d. variable across $i$ with constant variance, implying $\V(\xi_i)>0$;\footnote{The term $\xi_i$ does not depend on $n$ because the sparsity rate $q_n^\text{post}$ cancels out between the numerator and denominator.}

\item $\xi_i = \pi$, implying $\V(\xi_i)=0$.
\end{enumerate}

Case (b) occurs  when $A_{ij}^\text{post}$ is mean independent of $T_j$ given $T_i$ and $ w_i$, i.e., $E(T_j \mid A_{ij}^\text{post}, T_i, w_i) = E(T_j \mid T_i, w_i) = \pi$.
This case encompasses scenarios where the network remains unchanged after the intervention (i.e., $A^{\text{pre}} = A^{\text{post}}$), undergoes exogenous changes unrelated to the treatment, or when $A_{ij}^{\text{post}}$ depends on $T_j$ but is mean-independent.
It has two important implications for discussion.
First, in Case (b), no endogeneity issue arises, even with the presence of an unobserved confounder:
\[
E\left[ M_i u_i \right] 
= E\left[ \frac{\sum_{j=1}^n A_{ij}^{\text{post}} T_j}{\sum_{j=1}^n A_{ij}^{\text{post}}} u_i \right]
= \pi E\left[ \frac{\sum_{j=1}^n A_{ij}^{\text{post}}}{\sum_{j=1}^n A_{ij}^{\text{post}}} u_i \right]
= 0.
\]
The equality holds because $A_{ij}^\text{post}$ is conditionally mean independent of $T_j$, and Assumption \ref{asu:network} implies that $A_{ik}^\text{post}$ is independent of $T_j$ for any $k \neq j$.
This property, resulting from the row normalization of the mediator, is also documented by \cite{DieyeDjebbariBarrera-Osorio2014} and applied in the construction of the normalized SSIV \citep{Card2009, AutorDornHanson2013, PeriShihSparber2015, AdaoKolesarMorales2019, Goldsmith-PinkhamSorkinSwift2020, BorusyakHullJaravel2022}.
Second, the consistency condition and convergence rate of OLS estimators depend on whether $\xi_i$ is degenerate. 
Treatment-induced network changes introduce greater variation in the fraction mediator, potentially improving the convergence rate.
By comparing the results of Cases (a) and (b), I evaluate how the network's response to the intervention influences inference.

\subsection{Bias from OLS estimation with pre-intervention network} \label{sec:ovb}
In empirical studies, researchers often assume the network remains unchanged and use only pre-intervention data for estimation. Alternatively, even when acknowledging possible network changes, they avoid post-intervention data due to endogeneity concerns and rely on pre-intervention data to control for unobserved confounders \citep{CarterLaajajYang2021a}.
However, using $A^\text{pre}$ instead of $A^\text{post}$ in the regression model introduces omitted variable bias (OVB) and impedes the identification and consistent estimation of direct, indirect, and spillover effects.

Define ${M}^\text{pre}_i$ as the mediator calculated using $A^\text{pre}$, i.e., 
${M}^\text{pre}_i = \sum_{j=1}^n A_{ij}^\text{pre} T_j / \sum_{j=1}^n A_{ij}^\text{pre}$. 
Researchers typically estimate a linear regression with the following regressors:
\[
{X}^\text{pre}_i = (1, T_i, {M}^\text{pre}_i).
\]
Let $\beta^\text{pre}$ denote the vector of population coefficients from the above OLS fit.
By the argument of OVB,
\begin{align*}
\beta_1^\text{pre} 
=& \beta_1 + \beta_2 \cdot \frac{ \textup{Cov}(T_i, M_i) }{ \textup{Var}(T_i)}
= \textup{ToE}
\text{ and }
\beta_2^\text{pre} 
= \beta_2 \cdot \frac{ \textup{Cov}({M}^\text{pre}_i, M_i) }{ \textup{Var}({M}^\text{pre}_i) }.
\end{align*}
This implies that if the mediator $M_i$ is uncorrelated with the treatment $T_i$, there is no OVB in $\beta_1^\text{pre}$. 
However, when a correlation exists between $T_i$ and $M_i$, $\beta_1^\text{pre}$ captures the total effect of $T_i$ on $Y_i$, encompassing both the direct effect and the indirect effect through changes in the mediator.
The coefficient $\beta_2^\text{pre}$ represents the attenuated indirect effect, which remains uncontaminated by $\beta_1$.
If researchers are only interested in the total effect $\textup{ToE}$, fitting an OLS regression with ${X}^\text{pre}_i$ using the pre-intervention network can still yield the desired result. 
However, this approach cannot distinguish between direct and indirect effect channels and fails to recover the spillover effect.

\subsection{Asymptotic properties of OLS estimators}\label{sec:ols}

As discussed in Section \ref{sec:Mi}, no endogeneity arises when either $\V\left( \xi_i \right) > 0$ with $\lambda(w_i) = 0$, or $\V\left( \xi_i \right) = 0$.
In such cases, I propose using OLS estimation with the regressor vector $X_i$ to estimate $\beta$:
\[
X_i = (1, T_i, M_i).
\]
Let $\hat{\beta}^{\textsc{ols}}$ denote the vector of coefficients obtained from the above OLS fit.

\begin{theorem}\label{thm:consistency_ratio}
Under Assumptions \ref{asu:network} and \ref{asu:linearY}, 
\begin{enumerate}[(a)]
\item if $\V\left( \xi_i \right) > 0$ and $\lambda(w_i)=0$, then 
$\hat{\beta}^{\textsc{ols}} - \beta = O_{\mathbb{P}}(n^{-1/2})$;
\item if $\V\left( \xi_i \right) = 0$, then 
$\hat{\beta}^{\textsc{ols}}_0 - \beta_0 = O_{\mathbb{P}}\left(\sqrt{q_n^\text{post}}\right)$, $\hat{\beta}^{\textsc{ols}}_1 - \beta_1 = O_{\mathbb{P}}(n^{-1/2})$ and $\hat{\beta}^{\textsc{ols}}_2 - \beta_2 = O_{\mathbb{P}}\left(\sqrt{q_n^\text{post}}\right)$.
\end{enumerate}
\end{theorem}

Theorem \ref{thm:consistency_ratio} establishes the consistency conditions of $\hat{\beta}^{\textsc{ols}}$ in the absence of endogeneity under Cases (a) and (b), respectively.
Under Case (a), the convergence rate is the standard $\sqrt{n}$ and does not depend on the sparsity rate $q_n^\text{post}$.
The intuition is that, by the decomposition, $\xi_i$ represents the key term in $M_i$. 
When $\xi_i$ is non-degenerate and thus i.i.d., the estimation reduces to the standard case with i.i.d. data.
Under Case (b), when $\xi_i$ is degenerate, the variation of $M_i$ comes from the remainder term $r_i^*$. 
As the network becomes denser, the dependence among units increases, and the variation of $M_i$ across units decreases, which slows the convergence rate.
This explains why $\hat{\beta}^{\textsc{ols}}_0$ and $\hat{\beta}^{\textsc{ols}}_2$ are consistent only when the network is not dense (i.e., $q_n^\text{post} \prec 1$), with a (possibly) slower convergence rate than the usual rate $\sqrt{n}$, scaled by $\sqrt{ n q_n^\text{post}}$. 
The convergence rate of $\hat{\beta}^{\textsc{ols}}_1$ remains the standard $\sqrt{n}$, since $T_i$ is uncorrelated with the other regressor $M_i$ and thus unaffected by the dependency of $M_i$ across $i$, as in the i.i.d. setting.

\medskip

Next, I study the asymptotic distribution of $\hat{\beta}^{\textsc{ols}}$, focusing on the regimes where it is consistent.
Define $\hat{u}_i^\textsc{ols}$ as the residual from the above OLS fit and stack $X_i$ to form the $n\times 3$ design matrix $X$. 
The variance estimator $\hat{V}^{\textsc{ols}}$ is the standard heteroskedasticity-consistent (HC) variance estimator, defined as:
\begin{equation}
\hat{V}^{\textsc{ols}} = (X^\top X)^{-1} \hat{V}^{\textsc{ols}}_{\text{num}} (X^\top X)^{-1},
\label{eq:HC_ratio}
\end{equation}
where the middle term, $\hat{V}^{\textsc{ols}}_{\text{num}}$, is given by
\[
\hat{V}^{\textsc{ols}}_{\text{num}} = X^\top \text{diag}\left\{(\hat{u}_i^{\textsc{ols}})^2, i = 1, \dots, n\right\} X.
\]
Theorem \ref{thm:asymnormal_ratio} focuses on the regimes of $q_n^{\text{post}}$ under which the OLS estimators are consistent. 

\begin{theorem}\label{thm:asymnormal_ratio}
Suppose either $\V\left( \xi_i \right) > 0$ with $\lambda(w_i)=0$ or $\V\left( \xi_i \right) = 0$ with $q_n^{\text{post}} \prec 1$.
Under Assumptions \ref{asu:network} and \ref{asu:linearY}, then
\[
\left( \hat{V}^{\textsc{ols}} \right)^{-1/2} 
\left(\hat{\beta}^{\textsc{ols}} - \beta \right) \clt \mathcal{N}(0, I_3).
\]
\end{theorem}

Theorem \ref{thm:asymnormal_ratio} establishes the asymptotic normality of $\hat{\beta}^{\textsc{ols}}$ and shows that $\hat{V}^{\textsc{ols}}$ closely approximates the asymptotic variance. 
There are two key implications. First, despite the dependence of the regressor $M_i$ across units, the usual HC variance estimator $\hat{V}^{\textsc{ols}}$ in \eqref{eq:HC_ratio} remains valid since the error term $u_i$ is assumed to be i.i.d. 
Second, the normal approximation exhibits a self-normalization property, ensuring that inference based on standard $t$-tests is reliable.
Although the convergence rate of $\hat{\beta}^{\textsc{ols}}$ (potentially) depends on the sparsity parameter $q_n^{\text{post}}$, its asymptotic normality is unaffected by the sparsity level, as the varying order of $\hat{\beta}^{\textsc{ols}}$ is ``canceled out'' by the order of the variance component. 
This result is crucial because the sparsity parameter $q_n^{\text{post}}$ is generally not identified \citep{BickelChenLevina2011}. 
A similar self-normalization property is observed by \cite{Cai2022} in the context of network centrality regression and by \cite{HansenLee2019} in the context of cluster-dependent data.

\section{IV estimation} \label{sec:IV}

In this section, I examine the IV estimation to address the potential endogeneity issue when $\lambda(w_i)\ne 0$.
I study the asymptotic properties of IV estimators using SSIV in Section \ref{sec:SSIV}, showing that they become inconsistent as networks grow denser.
To address this, in Section \ref{sec:IV_PC}, I propose a modification to SSIV to restore consistency in denser networks.

\subsection{Asymptotic properties of SSIV estimators}
\label{sec:SSIV}
Endogeneity issue (potentially) arises when $\lambda(w_i) \ne 0$.
To address this, one needs instruments that are uncorrelated with $u_i$ and shifts $M_i$ enough to identify $\beta_2$. 
I tackle this issue using the shift-share instrumental variable approach.
\cite{BorusyakHull2023} proposed a general approach for constructing SSIVs by combining exogenous shocks with non-exogenous exposure through a known formula, adjusting for expected treatment. 
Applying this approach to my setting, I derive the following linear SSIV for $M_i$: 
\[
Z_i^{\textsc{ssiv}} = \sum_{j=1}^n A_{ij}^\text{pre} (T_j - \pi),
\]
which leverages the treatment assignment of others, $T_j$, as the exogenous shock, weighted by pre-intervention network information as the share.

To estimate $\beta$, I propose using IV estimation with the IV vector ${Z}_i$:
\begin{equation*}
{Z}_i = \left( 1, T_i, \sum_{j=1}^n A_{ij}^\text{pre} (T_j - \pi) \right).
\end{equation*}
Stack ${Z}_i$ to obtain the $n\times 3$ matrix ${Z}$.
Let $\hat{\beta}^{\textsc{iv}}$ denote the vector of the coefficients obtained from the IV fits with ${Z}_i$.

I assess the validity of SSIV, and thus the identification of $\beta$, by verifying two conditions: exogeneity and relevance.
The exogeneity of SSIV holds by construction, due to the random assignment of treatments and the centering of treatment around the assignment probability:
\begin{align*}
E[Z_i^{\textsc{ssiv}} u_i ]
= E\left[ \left(\sum_{j=1}^n A_{ij}^\text{pre} (T_j - \pi) \right) u_i \right]
=& \sum_{j=1}^n E \left( A_{ij}^\text{pre} u_i \right)
E(T_j - \pi) 
= 0.
\end{align*}
The question now turns to the relevance of the SSIV.  
\citet[][Assumption 3]{BorusyakHull2023} assume weak mutual dependence of the SSIV to ensure the convergence of the sample first stage, which, in my model, roughly implies that the network is not too dense.
I thoroughly analyze the relevance condition in my context and characterize the sparsity regimes under which the SSIV provides consistent estimators, as stated in Theorem \ref{thm:IV_consistency}.
In line with \citet{BorusyakHull2023}, I demonstrate that consistency is achieved when both pre- and post-treatment networks are relatively sparse.

\begin{theorem}\label{thm:IV_consistency}
Under Assumptions \ref{asu:network} and \ref{asu:linearY}, 
\begin{enumerate}[(a)]
\item if $\V\left( \xi_i \right) > 0$ with
$\max\{q_n^\text{pre}, q_n^\text{post}\} \prec n^{-1/2}$, 
then
\[
\hat{\beta}^{\textsc{iv}} - \beta 
= O_{\mathbb{P}}\left( \sqrt{n} \max\{q_n^\text{pre}, q_n^\text{post}\} \right);
\]

\item if $\V\left( \xi_i \right) = 0$  
with $ \max\{q_n^\text{pre}, q_n^\text{post} \} \prec \sqrt{n} q_n^\text{post}$, then
\begin{align*}
\hat{\beta}^{\textsc{iv}}_1 - {\beta}_1 
~=~&  O_{\mathbb{P}}\left( \frac{1}{\sqrt{n}} \max\left\{ \frac{ \max\{q_n^\text{pre}, q_n^\text{post} \} }{  \sqrt{ q_n^\text{post}}  }, 1 \right\}  \right).
\end{align*}
Moreover, with $\max\{q_n^\text{pre}, q_n^\text{post}\} \prec n^{-1/2}$, then
\begin{align*}
\hat{\beta}^{\textsc{iv}}_0 - {\beta}_0 
~=~& O_{\mathbb{P}}\left( \sqrt{n} \max\{q_n^\text{pre}, q_n^\text{post}\} \right)
\text{ and }
\hat{\beta}^{\textsc{iv}}_2 - {\beta}_2 
~=~ O_{\mathbb{P}}\left( \sqrt{n} \max\{q_n^\text{pre}, q_n^\text{post}\} \right).
\end{align*}
\end{enumerate}
\end{theorem}

Theorem \ref{thm:IV_consistency} specifies the regimes under which $\hat{\beta}^{\textsc{iv}}$ is consistent for Cases (a) and (b). The consistency regime and convergence rates of $\hat{\beta}^{\textsc{iv}}_0$ and $\hat{\beta}^{\textsc{iv}}_2$ are invariant across these cases, both depending on $q_n^\text{pre}$ and $q_n^\text{post}$, and are (possibly) slower than the usual rate $\sqrt{n}$, scaled by the degree in the denser network, $n \max\{q_n^\text{pre}, q_n^\text{post}\}$.
In Case (b), $\hat{\beta}^{\textsc{iv}}_1$ converges at a (possibly) faster rate compared to Case (a) and is consistent under a less restrictive condition. 
This observation aligns with Theorem \ref{thm:consistency_ratio}, which shows that, in Case (b), the estimation of $\beta_1$ is less affected by the dependency of $M_i$ across $i$ and exhibits a faster convergence rate than the estimator of $\beta_2$.

Theorem \ref{thm:IV_consistency} demonstrates that the consistency of the SSIV estimators breaks down when networks are relatively dense, i.e., $\max\{q_n^\text{pre}, q_n^\text{post}\} \succcurlyeq n^{-1/2}$. 
Here is the intuition:
the SSIV method leverages variation in the number of treated friends for each individual, driven by random fluctuations in treatment assignment. 
Incorporating network information is essential, as the endogenous variable $M_i$ depends on the network structure. 
However, as the network becomes denser, the variation in the number of treated friends across individuals decreases. 
Consider an extreme case where all individuals are fully connected before the intervention. 
In this scenario, the SSIV reduces to $Z_i^{\textsc{ssiv}} = \sum_{j \neq i}(T_j - \pi)$, where the only distinction across units is whether they are treated or not.

\begin{remark}
A commonly used instrument in the peer effects literature is the ``peer-of-peer'' IV \citep{BramoulleDjebbariFortin2009b, DeGiorgiFrederiksenPistaferri2020}. 
When peers of peers are not direct peers, their characteristics influence individual outcomes only through their effect on peers' outcomes, providing valid instruments. 
Identification also requires a non-overlapping peer network, consistent with the findings on SSIV: the network must not be too dense and should exhibit sufficient structural variation.
\end{remark}

\begin{remark}
In Appendix \ref{app:IV_alt}, I discuss the performance of the IV estimators using the normalized SSIV suggested in \cite{BorusyakHullJaravel2022}: 
$\sum_{j=1}^n A_{ij}^\text{pre}T_j / \sum_{j=1}^n A_{ij}^\text{pre}$.
Instead of centering the treatment around the assignment probability, this approach normalizes the total number of treated friends to guarantee the exogeneity.
However, as the network becomes denser, this ratio converges to the assignment probability as the sample size increases.
This convergence can lead to a weak IV problem if the assignment probability is constant.
See the consistency condition of the IV estimators using this normalized SSIV in Theorem \ref{thm:IV_alt_consistency}.
\end{remark}

\begin{remark}
The BLP instrument is widely used in empirical industrial organization to address endogeneity in demand estimation by leveraging the characteristics of competing products as instruments \citep{BerryLevinsohnPakes1995, Armstrong2016}. 
Both the BLP instrument and the shift-share instrument share the fundamental idea of leveraging variation from external, exogenous sources to identify causal effects.
However, both the BLP instrument and the normalized SSIV face challenges in maintaining identifying power in limiting cases.
The BLP instrument has limited identifying power asymptotically in large markets, as market shares approach zero and the dependence of equilibrium markups on other products' characteristics diminishes.
Similarly, the normalized SSIV, $\sum_{j=1}^n A_{ij}^\text{pre} T_j / \sum_{j=1}^n A_{ij}^\text{pre}$, becomes degenerate in denser networks, weakening its relevance.
\end{remark}

\medskip
Corollary \ref{cor:IV_consistency} simplifies the results in Theorem \ref{thm:IV_consistency} to the special case where $q_n^\text{pre} \preccurlyeq q_n^\text{post}$.

\begin{corollary}\label{cor:IV_consistency}
Suppose $q_n^\text{pre} \preccurlyeq q_n^\text{post}$. 
Under Assumptions \ref{asu:network} and \ref{asu:linearY},
\begin{enumerate}[(a)]
\item if $\V\left( \xi_i \right) > 0$ with  $q_n^\text{post} \prec n^{-1/2}$, then
$\hat{\beta}^{\textsc{iv}} - \beta = O_{\mathbb{P}}(\sqrt{n} q_n^\text{post})$;

\item if $\V\left( \xi_i \right) = 0$ it holds that $\hat{\beta}^{\textsc{iv}}_1 - {\beta}_1 
= O_{\mathbb{P}}\left( \frac{1}{\sqrt{n}} \right)$.
Moreover, with $q_n^\text{post} \prec n^{-1/2}$, then
$\hat{\beta}^{\textsc{iv}}_0 - {\beta}_0 
= O_{\mathbb{P}}\left( \sqrt{n} q_n^\text{post} \right)$
and $\hat{\beta}^{\textsc{iv}}_2 - {\beta}_2 
= O_{\mathbb{P}}\left( \sqrt{n} q_n^\text{post} \right)$.

\end{enumerate}
\end{corollary}

Corollary \ref{cor:IV_consistency} shows that in both Cases (a) and (b), $\hat{\beta}_0^{\textsc{iv}}$ and $\hat{\beta}_2^{\textsc{iv}}$ exhibit the same convergence rate, which depends on $q_n^\text{post}$ and is (possibly) slower than the usual rate $\sqrt{n}$. 
In Case (a), the convergence rate of $\hat{\beta}_1^{\textsc{iv}}$ matches that of $\hat{\beta}_0^{\textsc{iv}}$ and $\hat{\beta}_2^{\textsc{iv}}$.
However, in Case (b), where $A^{\text{post}}$ is conditionally mean-independent of the treatment, $\hat{\beta}^{\textsc{iv}}_1$ is less affected by the dependence of $M_i$ across individuals and retains the usual rate $\sqrt{n}$.

\medskip

Now I show the asymptotic normality of the IV estimators $\hat{\beta}^{\textsc{iv}}$. 
I focus on the regimes where the IV estimators are consistent. 
Define $V^{\textsc{iv}}_{\text{num}} = \V\left(
\sum_{i=1}^n {Z}_i u_i \right)$, the variance of the numerator of the centered estimator $\hat{\beta}^{\textsc{iv}} - \beta$.
It can be shown that
\begin{align}
& V^{\textsc{iv}}_{\text{num}} 
= \begin{pmatrix}
\sum_{i=1}^n E(u_i^2) & \pi \sum_{i=1}^n E(u_i^2) & 0 \\
\pi \sum_{i=1}^n E(u_i^2) & \pi \sum_{i=1}^n E(u_i^2) & \pi(1-\pi)  \sum\limits_{i=1}^n \sum\limits_{j=1}^n E\left(  A_{ij}^\text{pre} u_i u_j \right) \\
0 & \pi(1-\pi) \sum\limits_{i=1}^n \sum\limits_{j=1}^n E\left( A_{ij}^\text{pre} u_i u_j \right) 
& \pi(1-\pi) \sum\limits_{i=1}^n E\left( \left( \sum\limits_{j=1}^n A_{ij}^\text{pre} u_j \right)^2 \right)
\end{pmatrix}.
\label{eq:V_num}
\end{align}
The $(2,3)$, $(3,2)$ and $(3,3)$ elements of $ V^{\textsc{iv}}_{\text{num}}$ in \eqref{eq:V_num} account for the dependence of ${Z}_i^{\textsc{ssiv}}$ across $i$, even though $u_i$ is i.i.d..
Let $\hat{u}_i^{\textsc{iv}}$ denote the residual from the IV fits with ${Z}_i$.
Define $\hat{V}^{\textsc{iv}}_{\text{num}}$ as the plug-in estimator of $V^{\textsc{iv}}_{\text{num}}$:
\begin{align}
\hat{V}^{\textsc{iv}}_{\text{num}}
=& \begin{pmatrix}
\sum\limits_{i=1}^n (\hat{u}_i^{\textsc{iv}})^2 & \pi \sum\limits_{i=1}^n (\hat{u}_i^{\textsc{iv}})^2 & 0 \\
\pi \sum\limits_{i=1}^n (\hat{u}_i^{\textsc{iv}})^2 & \pi \sum\limits_{i=1}^n (\hat{u}_i^{\textsc{iv}})^2 & \pi(1-\pi) \sum\limits_{i=1}^n \sum\limits_{j=1}^n A_{ij}^\text{pre} \hat{u}_j^{\textsc{iv}} \hat{u}_j^{\textsc{iv}} \\
0 & \pi(1-\pi) \sum\limits_{i=1}^n \sum\limits_{j=1}^n 
A_{ij}^\text{pre} \hat{u}_i^{\textsc{iv}} \hat{u}_j^{\textsc{iv}} & \pi(1-\pi)  \sum\limits_{i=1}^n \left( \sum\limits_{j=1}^n A_{ij}^\text{pre} \hat{u}_j^{\textsc{iv}} \right)^2
\end{pmatrix}.
\label{eq:V_hat_num}
\end{align}

Define $\hat{V}^{\textsc{iv}} = ({Z}^\top X)^{-1} \hat{V}_{\text{num}}^{\textsc{iv}} (X^\top {Z})^{-1}$ as the estimator of the asymptotic variance of $\hat{\beta}^{\textsc{iv}}$. 
I demonstrate the asymptotic normality of $\hat{\beta}^{\textsc{iv}}$, specifically focusing on the regimes where the IV estimators are consistent.

\begin{theorem}\label{thm:asymnormal_ratio_IV}
Under Assumptions \ref{asu:network} and \ref{asu:linearY}, and with $\max\{q_n^\text{pre}, q_n^\text{post}\} \prec n^{-1/2}$,
then
\[
\left(\hat{V}^{\textsc{iv}}\right)^{-1/2} 
\left( \hat{\beta}^{\textsc{iv}} - \beta \right) 
\clt \mathcal{N}(0, I_3).
\]
\end{theorem}

Theorem \ref{thm:asymnormal_ratio_IV} establishes the asymptotic normality of the IV estimators $\hat{\beta}^{\textsc{iv}}$ using SSIV, along with a consistent variance estimator, thereby validating inference results based on the usual $t$-test. 
Similar to Theorem \ref{thm:asymnormal_ratio}, it exhibits the property of self-normalization, which does not require to know the sparsity parameters.

\begin{remark}
The usual HC variance estimator of IV estimators is given by 
$\hat{V}^{{\textsc{iv}},\text{hc}} = ({Z}^\top X)^{-1} \hat{V}_{\text{num}}^{{\textsc{iv}},\text{hc}} (X^\top {Z})^{-1}$ where
\begin{align*}
\hat{V}^{{\textsc{iv}},\text{hc}}_{\text{num}}
= \begin{pmatrix}
\sum\limits_{i=1}^n (\hat{u}_i^{\textsc{iv}})^2 & \sum\limits_{i=1}^n T_i (\hat{u}_i^{\textsc{iv}})^2 & \sum\limits_{i=1}^n (\hat{u}_i^{\textsc{iv}})^2 \sum\limits_{j=1}^n A_{ij}^\text{pre} (T_j - \pi)  \\
\sum\limits_{i=1}^n T_i (\hat{u}_i^{\textsc{iv}})^2 & \sum\limits_{i=1}^n T_i (\hat{u}_i^{\textsc{iv}})^2 & \sum\limits_{i=1}^n T_i (\hat{u}_i^{\textsc{iv}})^2 \sum\limits_{j=1}^n A_{ij}^\text{pre} (T_j - \pi)  \\
\sum\limits_{i=1}^n (\hat{u}_i^{\textsc{iv}})^2 \sum\limits_{j=1}^n A_{ij}^\text{pre} (T_j - \pi)  & \sum\limits_{i=1}^n T_i (\hat{u}_i^{\textsc{iv}})^2 \sum\limits_{j=1}^n A_{ij}^\text{pre} (T_j - \pi)  & \sum\limits_{i=1}^n (\hat{u}_i^{\textsc{iv}})^2 \left(  \sum\limits_{j=1}^n A_{ij}^\text{pre} (T_j - \pi) \right)^2 
\end{pmatrix}.
\end{align*}
After appropriately scaling the $(3,1)$, $(1,3)$, $(3,2)$, and $(2,3)$ elements of $\hat{V}^{{\textsc{iv}},\text{hc}}_{\text{num}}$, these terms converge to zero. 
Thus, the primary differences between $\hat{V}^{\textsc{iv},\text{hc}}_{\text{num}}$ and $\hat{V}^{\textsc{iv}}_{\text{num}}$ lie in the $(2,3)$, $(3,2)$, and $(3,3)$ terms. 
The variance estimator $\hat{V}^{\textsc{iv},\text{hc}}_{\text{num}}$, designed for i.i.d. data, does not account for dependencies across units induced by the SSIV $Z_i^{\textsc{ssiv}}$. 
In contrast, $\hat{V}^{\textsc{iv}}_{\text{num}}$ captures these dependencies, leading to more accurate variance estimation.
This issue is similar to that documented by \cite{AdaoKolesarMorales2019}, who study inference in shift-share regression designs, where regional outcomes are regressed on a weighted average of sectoral shocks using regional sector shares as weights. 
They show that over-rejection arises when using cluster-robust standard errors because regression residuals are correlated across regions with similar sectoral shares, regardless of their geographic location.
\end{remark}

\subsection{Modification of SSIV} 
\label{sec:IV_PC}

I begin this subsection with a more formal explanation of the failure of relevance, offering insights into how to modify the SSIV.
To simplify the argument, consider the special case where $q_n^\text{pre} = q_n^\text{post} = q_n$. 
The relevance of SSIV is measured by $n^{-1}\sum_{i=1}^n M_i Z_i^{\textsc{ssiv}}$.\footnote{For the sake of a heuristic argument, I focus on $\frac{1}{n}\sum_{i=1}^n M_i Z_i^{\textsc{ssiv}}$ instead of the sample covariance $\frac{1}{n}\sum_{i=1}^n M_i Z_i^{\textsc{ssiv}} - \frac{1}{n}\sum_{i=1}^n M_i \frac{1}{n}\sum_{i=1}^n Z_i^{\textsc{ssiv}}$. Since the SSIV $Z_i^{\textsc{ssiv}}$ is mean zero, this simplification does not change the essence of the argument.}
Recall the decomposition that $M_i = \xi_i + r^*_i$, where $\xi_i$ is defined in \eqref{eq:xi} and $r^*_i$ is the remainder term, capturing information from all other units $j\ne i$.
The relevance of SSIV is then given by:
\[
\frac{1}{n}\sum_{i=1}^n M_i Z_i^{\textsc{ssiv}} = \frac{1}{n}\sum_{i=1}^n \xi_i Z_i^{\textsc{ssiv}} + \frac{1}{n}\sum_{i=1}^n r^*_i Z_i^{\textsc{ssiv}}.
\]
By definition, $\xi_i$ is a function of $T_i$ and $w_i$ and is uncorrelated with the treatments $T_j$ of all other units $j \neq i$, and thus uncorrelated with $Z_i^{\textsc{ssiv}}$. 
Consequently, the first term, $n^{-1} \sum_{i=1}^n \xi_i Z_i^{\textsc{ssiv}}$, has mean zero but a variance of order $O(n q_n^2)$, referred to as the noise term. 
The variance increases with the sparsity parameter as the dependence of $Z_i$ across units strengthens in denser networks. 
Thus, this term does not contribute to the signal in the relevance measure but instead adds noise.
In contrast, the second term, $n^{-1} \sum_{i=1}^n r_i^* Z_i^{\textsc{ssiv}}$, where $r_i^*$ contains information from all other units $j$, contributes to the signal in the relevance measure and is referred to as the signal term. 
It has a non-zero mean and remains of constant scale.
In other words, in $n^{-1} \sum_{i=1}^n M_i Z_i^{\textsc{ssiv}}$, the first term contributes only noise, while the second term carries the signal for the relevance measure. 
As networks become denser, the noise term increasingly dominates the signal term, preventing the first stage from converging to a nonzero constant and resulting in an inconsistent estimator. 
Therefore, restoring the consistency of the IV estimators is feasible if the noise term can be effectively reduced while preserving the signal term.

\begin{remark}
There is a remarkable connection between SSIV and the estimator of the indirect effect in \cite{LiWager2022}, who study asymptotics for treatment effect estimation under network interference, with the network randomly drawn from a graphon.
They consider anonymous interference, where potential outcomes depend on the fraction of treated friends, not their identities.
Unlike my setting, they assume that the network remains unchanged following the intervention.
The potential outcome for individual $i$ is 
$Y_i(t_i, t_{-i}) = f_i\left(t_i, \sum_{j=1}^n A_{ij} t_j/ \sum_{j =1}^n A_{ij} \right)$,
where $f_i \in \mathcal{F}$ allows arbitrary dependence on the latent variable $w_i$. 
They propose an estimator for the indirect effect defined in \eqref{eq:indirect}, 
\[
\hat{\tau}_{\text{IND}}^{\text{U}} = \frac{1}{n}\sum_{i=1}^n Y_i \left( \frac{\sum_{j\ne i} A_{ij} T_j }{\pi} - \frac{\sum_{j\ne i}A_{ij} (1-T_j)}{1-\pi} \right),
\]
which is derived as the difference between the Horvitz--Thompson estimators for direct and total effects, and thus is unbiased.
Rewriting this estimator gives
\[
\hat{\tau}_{\text{IND}}^{\text{U}} = \frac{1}{\pi(1-\pi)} \frac{1}{n}
\sum_{i=1}^n Y_i \left( \sum_{j\ne i}A_{ij} (T_j - \pi) \right),
\]
which implicitly uses SSIV to address endogeneity.
\end{remark}

Here, I explain how to reduce the noise term, drawing inspiration from the PC balancing idea in \cite{LiWager2022}. 
Ideally, the term $\frac{1}{n} \sum_{i=1}^n \xi_i Z_i^{\textsc{ssiv}}$ could be eliminated by adjusting the SSIV to be the residual from the projection of $Z_i^{\textsc{ssiv}}$ onto $\xi_i$, leveraging the orthogonality property.
However, since $\xi_i$ is unknown, the noise term $\frac{1}{n} \sum_{i=1}^n \xi_i Z_i$ can instead be reduced by projecting $w_i$ out of $Z_i$, thereby decreasing the noise while preserving the relevance signal from other units $j$. 
Although $w_i$ is unobserved, its information can be extracted through spectral analysis of the graphon.

To be more specific, let $G_n^\text{pre}$ denote the graphon matrix, where the $(i,j)$ element is given by $G^\text{pre}_{n,ij} = g^\text{pre}(w_i, w_j)$. 
By performing eigenvalue decomposition of $q_n^\text{pre} G_n^\text{pre}$, I express it as $q_n^\text{pre} G_n^\text{pre} = \sum_{k=1}^n \lambda_k^* \psi_k^* \psi_k^{*\top}$, where $\lambda_k^*$ represents the $k$-th largest eigenvalue of $q_n^\text{pre} G_n^\text{pre}$.
I assume that the graphon is approximately low-rank, as formalized in Assumption \ref{ass:low_rank}, meaning it can be well-approximated by the leading $r$ terms: $q_n^\text{pre} \tilde{G}^\text{pre}_n = \sum_{k=1}^r \lambda_k^* \psi_k^* \psi_k^{*\top}$.
I project $Z_i$ onto the first $r$ eigenvectors, $\{{\psi}_k^*(w_i)\}_{k=1}^r$, with the coefficients:
\begin{align*}
{\gamma}^*_k = \sum_{i=1}^n {\psi}_k^*(w_i) Z_i^{\textsc{ssiv}}, \text{ for all } k=1,\cdots,r,
\end{align*}
which follows from the orthogonality of eigenvectors. 
The residual from this projection yields the ``oracle version'' of the modified SSIV: ${Z}_i^{\textsc{de}*} = Z_i^{\textsc{ssiv}} - \sum_{k=1}^r {\gamma}^*_k {\psi}_k^*(w_i)$.
For a feasible estimatior, I approximate ${\psi}_k^*(w_i)$ using the eigenvector of pre-intervention adjacency matrix, $A^{\text{pre}} = \sum_{k=1}^n \hat{\lambda}_k \hat{\psi}_k \hat{\psi}_k^\top$. 
Therefore, the modified SSIV, denoted as ${Z}_i^{\textsc{de}}$ (for ``denoised''), is given by:
\[
{Z}_i^{\textsc{de}} = Z_i^{\textsc{ssiv}} - \sum_{k=1}^r \hat{\gamma}_k \hat{\psi}_k(w_i)
\]
where $\hat{\gamma}_k = \sum_{i=1}^n \hat{\psi}_k(w_i) Z_i^{\textsc{ssiv}}$.
By the properties of projection, the modified SSIV ${Z}_i^{\textsc{de}}$ is orthogonal to the space spanned by $\{\hat{\psi}_k(w_i)\}_{k=1}^r$, which encodes information about $w_i$. This effectively reduces noise without compromising the signal.

In the network literature, particularly in settings involving inference with eigenvectors \citep{Cai2022, LeLi2022, LiWager2022}, it is widely assumed that the graphon is low-rank in terms of its eigenfunctions, such that 
\begin{align}
g(w_i,w_j) = \sum_{k=1}^r \lambda_k \psi_k(w_i) \psi_k(w_j)
\label{eq:low_rank}
\end{align}
where $\mathbb{E}\left[\psi_k^2(w_i)\right]=1$ and $\mathbb{E}\left[\psi_k(w_i) \psi_l(w_i)\right]=0 \text { for } k \neq l$.
This assumption is satisfied by many well-known network models, including the stochastic block model \citep{HollandLaskeyLeinhardt1983a} and the random dot product graph \citep{YoungScheinerman2007, AthreyaFishkindTang2017}. 
Another popular model is the latent space model \citep{HoffRafteryHandcock2002}, which falls within the class of inhomogeneous Erdős--Rényi models, including the homophily model and the beta model.
The latent space model and the more general graphon models typically do not impose a low-rank structure as in \eqref{eq:low_rank}, but instead rely on certain smoothness conditions for the graphon function. 
This paper assumes that the graphon $g^\text{pre}$ can be well-approximated by a low-rank representation based on its eigenvalues, as outlined in Assumption \ref{ass:low_rank}.

\begin{assumption}\label{ass:low_rank}
There exists some constant $r$ $(r \prec n)$ such that
\begin{enumerate}[(a)]
\item $\min\limits_{k \in \{1,\cdots, r-1\}}(\lambda_k -\lambda_{k+1}) \asymp n q_n^\text{pre}$;
\item $\left\|\sum_{k=r+1}^n \lambda_k^* \psi_k^* \psi_k^{*\top} 
\right\|_{\textup{op}}=O_{\mathbb{P}}( q_n^\text{pre} )$.   
\end{enumerate}

\end{assumption}

Assumption \ref{ass:low_rank}(a) requires the minimum eigen-gap, i.e., the spacing between an eigenvalue and the rest of the spectrum, to be sufficiently large. 
Several papers propose eigenvector estimators that are robust to small eigen-gaps. 
For example, \citet{ChengWeiChen2021} tackle eigenvector estimation for low-rank matrices with small eigen-gaps and noisy observations, and \citet{LiCaiPoor2022} focus on estimating linear functionals of unknown eigenvectors under similarly tight eigen-gap conditions.
While these robust estimators could potentially extend my methods to such settings, doing so is beyond the scope of this paper.
Assumption \ref{ass:low_rank}(b) is analogous to the sparsity assumption in high-dimensional analysis, which assumes that only a small subset of features (variables) significantly contribute to the outcome. 
This assumption enables more efficient estimation and prediction in models with a large number of features, possibly exceeding the number of observations \citep{Tibshirani1996}.

To  estimate $\beta$'s, I propose the IV fits with the modified IV vector $\tilde{Z}_i^{\textsc{de}}$:
\[
\tilde{Z}_i^{\textsc{de}} 
= \left( 1, T_i, {Z}_i^{\textsc{de}} \right).
\]
Stack $\tilde{Z}_i^{\textsc{de}}$ to form the $n\times 3$ matrix $\tilde{Z}^{\textsc{de}}$.
Let $\hat{\beta}^{\textsc{de}}$ denote the vector of the coefficients obtained from the above IV fit.

\begin{theorem}\label{thm:consistency_ratio_PC_IV}
Suppose $q_n^\text{pre} \succ {\frac{\log(n)}{\log(\log(n))}} / n$.
Under Assumptions \ref{asu:network}, \ref{asu:linearY} and \ref{ass:low_rank}, 
\begin{enumerate}[(a)]
\item if $\V\left( \xi_i \right) > 0$ with $\max\{q_n^\text{pre}, q_n^\text{post}\} \prec \sqrt{q_n^\text{pre}}$, 
then
\begin{align*}
\hat{\beta}^{\textsc{de}} - {\beta} 
= O_{\mathbb{P}}\left( \frac{ \max\{q_n^\text{pre}, q_n^\text{post} \} }{ \sqrt{q_n^\text{pre} }}  \right);
\end{align*}

\item if $\V\left( \xi_i \right) = 0$, it holds that $\hat{\beta}^{\textsc{de}}_1 - {\beta}_1
= O_{\mathbb{P}}\left( \frac{ 1  }{  \sqrt{n} } \right)$.
Moreover, with $\max\{q_n^\text{pre}, q_n^\text{post}\} \prec \sqrt{q_n^\text{pre}}$, 
then
\begin{align*}
\hat{\beta}^{\textsc{de}}_0 - {\beta}_0
~=~ O_{\mathbb{P}}\left( \frac{ \max\{q_n^\text{pre}, q_n^\text{post} \} }{ \sqrt{q_n^\text{pre} }} \right)
\text{ and }
\hat{\beta}^{\textsc{de}}_2 - {\beta}_2
~=~ O_{\mathbb{P}}\left( \frac{ \max\{q_n^\text{pre}, q_n^\text{post} \} }{ \sqrt{q_n^\text{pre} }} \right).
\end{align*}

\end{enumerate}
\end{theorem}

To make the modified SSIV work, it requires precise estimation of the eigenvectors, which becomes problematic when the network is too sparse \citep{AltDucatezKnowles2021, Benaych-GeorgesBordenaveKnowles2019, Benaych-GeorgesBordenaveKnowles2020}. 
Despite this added requirement, ${\frac{\log(n)}{\log(\log(n))}} / n$ is much sparser than the threshold at which SSIV fails, $n^{-1/2}$.
As a result, the consistency regime of the modified SSIV overlaps with that of SSIV, while also extending to the regime where SSIV becomes inconsistent.

\begin{remark}
To clarify the condition $\max\{q_n^\text{pre}, q_n^\text{post}\} \prec \sqrt{q_n^\text{pre}}$,
I discuss it in two cases:
\begin{enumerate}[(1)]
\item if $q_n^\text{pre} \prec q_n^\text{post}$, then
$\max\{q_n^\text{pre}, q_n^\text{post}\} \prec \sqrt{q_n^\text{pre}}$ holds under $\sqrt{q_n^\text{pre}} \succ  q_n^\text{post}$, saying that the pre-intervention network can not be too sparse compared with the post-intervention network; 
\item if $q_n^\text{pre} \succcurlyeq q_n^\text{post}$, then 
$\max\{q_n^\text{pre}, q_n^\text{post}\} \prec \sqrt{q_n^\text{pre}}$ holds under $q_n^\text{pre} \prec 1$.
\end{enumerate}
The modified SSIV cannot be extended to the case where both pre- and post-intervention networks are dense. 
In such a case, one potential solution is to use the matching method from \cite{Auerbach2022a} if some regularity assumption holds, e.g., $M_i$ has sufficient variation when controlling for $w_i$.

\end{remark}

Corollary \ref{cor:IV_PC_consistency} simplifies the results in Theorem \ref{thm:consistency_ratio_PC_IV} under the special case $q_n^\text{post} \succcurlyeq q_n^\text{pre}$.

\begin{corollary}
\label{cor:IV_PC_consistency}
Suppose $q_n^\text{post} \succcurlyeq q_n^\text{pre} \succ {\frac{\log(n)}{\log(\log(n))}} / n$.
Under Assumptions \ref{asu:network}, \ref{asu:linearY} and \ref{ass:low_rank}, then
\begin{enumerate}[(a)]
\item if $\V\left( \xi_i \right) > 0$ with $\sqrt{q_n^\text{pre}} \succ  q_n^\text{post}$, then $\hat{\beta}^{\textsc{de}} - \beta = O_{\mathbb{P}}\left( \frac{q_n^\text{post}}{\sqrt{q_n^\text{pre}}} \right)$;
\item if $\V\left( \xi_i \right) = 0$, it holds that $\hat{\beta}^{\textsc{de}}_1 - {\beta}_1 = O_{\mathbb{P}}\left( \frac{ 1  }{  \sqrt{n} } \right)$.
Moreover, with $\sqrt{q_n^\text{pre}} \succ  q_n^\text{post}$,  
then
$\hat{\beta}^{\textsc{de}}_0 - {\beta}_0 = O_{\mathbb{P}}\left( \frac{q_n^\text{post}}{\sqrt{q_n^\text{pre}}} \right)$ and $\hat{\beta}^{\textsc{de}}_2 - \beta_2 = O_{\mathbb{P}}\left( \frac{q_n^\text{post}}{\sqrt{q_n^\text{pre}}} \right)$. 
\end{enumerate}
\end{corollary}

Aligning with Corollary \ref{cor:IV_consistency}, Corollary \ref{cor:IV_PC_consistency} demonstrates that under Case (a), the convergence rate of $\hat{\beta}$ depends on the sparse parameters and is (possibly) slower than the usual rate $\sqrt{n}$, but faster than the IV estimators using SSIV, which converge at the rate $1/(\sqrt{n} q_n^\text{post})$. 
In Case (b), I also find that when the post-intervention network $A^{\text{post}}$ is conditionally mean independent of the treatment, $\hat{\beta}^{\textsc{de}}_1$ is less affected by the dependency of $M_i$ across $i$ and retains the usual rate $\sqrt{n}$. 
Meanwhile, the convergence rates of $\hat{\beta}^{\textsc{de}}_0$ and $\hat{\beta}^{\textsc{de}}_2$ are (possibly) slower than $\sqrt{n}$.

\begin{remark}
One takeaway from this subsection is that the discussion on the relevance condition of SSIV sheds light on other forms of mediators. Specifically, if a mediator can be decomposed into two parts, $M_i = \xi_i + r_i^*$, where $\xi_i$ contains only unit $i$'s own information, then $\xi_i$ introduces noise that potentially weakens relevance. 
The proposed denoising procedure has the potential to restore consistency.
\end{remark}

\medskip

Now I show the asymptotic normality of $\hat{\beta}^{\textsc{de}}$.  
Define $\mu_k^{u} = \sum_{i=1}^n u_i \psi_k^*(w_i)$ and $\eta_i^{u} = u_i - \sum_{k=1}^r \mu_k^{u} \psi_k^*(w_i)$.
Define $V^{\textsc{de}}_{\text{num}} = \V\left( \sum_{i=1}^n \tilde{Z}^{\textsc{de}}_i u_i \right)$, the variance of the numerator of the estimation bias $(\hat{\beta}^{\textsc{de}} - \beta)$. 
We can show
\begin{align}
& {V}^{\textsc{de}}_{\text{num}}
= \begin{pmatrix}
\sum\limits_{i=1}^n E(u_i^2) & \sum\limits_{i=1}^n \pi E(u_i^2) & 0 \\
\sum\limits_{i=1}^n \pi E(u_i^2) & \sum\limits_{i=1}^n \pi E(u_i^2) & 0 \\
0 & 0 & \pi(1-\pi) \sum\limits_{i=1}^n \sum\limits_{j=1}^n E\left( A_{ij}^{\text{pre}} (\eta_i^{u})^2  \right)
\end{pmatrix}.
\label{eq:V_num_PC}
\end{align}
By comparing ${V}^{\textsc{iv}}_{\text{num}}$ in \eqref{eq:V_num} with ${V}^{\textsc{de}}_{\text{num}}$ in \eqref{eq:V_num_PC}, I observe that the $(3,2)$ and $(2,3)$ terms in \eqref{eq:V_num_PC} are zero, whereas they are non-zero in \eqref{eq:V_num}. Additionally, the $(3,3)$ term in \eqref{eq:V_num} accounts for the covariance component, while the $(3,3)$ term in \eqref{eq:V_num_PC} accounts only for the diagonal component. 
This modification effectively reduces noise by mitigating the dependency of the IV across units due to information from $w_i$.

Let $\hat{u}_i^{\textsc{de}}$ denote the residual from the IV fits with $\tilde{Z}_i^{\textsc{de}}$.
Define $\hat{\mu}_k^{u} = \sum_{i=1}^n \hat{u}_i^{\textsc{de}} \hat{\psi}_{ki}$ and $\hat{\eta}^{u}_i = \hat{u}_i^{\textsc{de}} - \sum_{k=1}^r \hat{\mu}_k^u \hat{\psi}_{ki} $.
Define the plug-in estimator of ${V}_{\text{num}}^{\textsc{de}}$ as
\begin{align*}
\hat{V}_{\text{num}}^{\textsc{de}}
=& \begin{pmatrix}
\sum\limits_{i=1}^n (\hat{u}_i^{\textsc{de}})^2 & \pi \sum\limits_{i=1}^n (\hat{u}_i^{\textsc{de}})^2 & 0 \\
\pi \sum\limits_{i=1}^n (\hat{u}_i^{\textsc{de}})^2 & \pi \sum\limits_{i=1}^n (\hat{u}_i^{\textsc{de}})^2 & 0 \\
0 & 0 & \pi(1-\pi) \sum\limits_{i=1}^n \sum\limits_{j=1}^n  A_{ij}^\text{pre} \left( \hat{\eta}^u_i \right)^2
\end{pmatrix}.
\end{align*}
Define
$\hat{V}^{\textsc{de}} = ((\tilde{Z}^{\textsc{de}})^\top X)^{-1} \hat{V}_{\text{num}}^{\textsc{de}} (X^\top \tilde{Z}^{\textsc{de}})^{-1}$ as the variance estimator of $\hat{\beta}^{\textsc{de}}$.

Theorem \ref{thm:asymnormal_ratio_PC_IV} focuses on the sparsity regimes under which the IV estimators $\hat{\beta}^{\textsc{de}}$ are consistent. 

\begin{theorem}\label{thm:asymnormal_ratio_PC_IV}
Suppose $q_n^\text{pre} \succ {\frac{\log(n)}{\log(\log(n))}} / n$.
Under Assumptions \ref{asu:network}, \ref{asu:linearY} and \ref{ass:low_rank}, and with $\max\{q_n^\text{pre}, q_n^\text{post}\} \prec \sqrt{q_n^\text{pre}}$,
then
\[
\left(\hat{V}^{\textsc{de}}\right)^{-1/2} 
\left( \hat{\beta}^{\textsc{de}} - \beta \right) 
\clt \mathcal{N}(0, I_3).
\]
\end{theorem}
Theorem \ref{thm:asymnormal_ratio_PC_IV} demonstrates the asymptotic normality of the IV estimators $\hat{\beta}^{\textsc{de}}$ using the modified SSIV, and the consistency of the variance estimator, which together validates the inference results based on the usual $t$ test. 
While comparing the asymptotic variances of $\hat{\beta}^{\textsc{iv}}$ and $\hat{\beta}^{\textsc{de}}$ is not straightforward, simulation results in Section \ref{sec:MC_IV} indicate that, in the regimes where both SSIV and modified SSIV yield consistent estimators, using the modified SSIV does not increase the variance.

\section{Monte Carlo Simulation}\label{sec:MC}
In this section, I provide simulation evidence to support my theory. I draw treatment indicator $T_i \overset{\text{i.i.d.}}{\sim} \text{Bern}(0.5)$.
I consider the following four designs for network generation: 
\paragraph{Design 1: Rank-3 Stochastic Block Model}
\begin{itemize}

\item $g^\text{pre}(w_i, w_j)
= \begin{cases}
3/5 & \text { if } \Phi(w_i) \leq \frac13 \text { and } \Phi(w_j) \le \frac13; \\
1/3 & \text { if } \frac13 < \Phi(w_i) \leq \frac23 \text { and } \frac13 < \Phi(w_j) \leq \frac23; \\
1/2 & \text { if } \frac23 < \Phi(w_i) \leq 1 \text { and }
\frac23 < \Phi(w_j) \leq 1; \\
1/5 & \text { otherwise}.
\end{cases}$

\item $g^\text{post}(w_i, w_j, T_i, T_j)
= \begin{cases}
3/5 & \text { if } \Phi(w_i(1-T_i)) \leq \frac13 \text { and } \Phi(w_j(1-T_j)) \leq \frac13; \\
1/3 & \text { if } \frac13 < \Phi(w_i(1-T_i)) \leq \frac23 \text { and } \frac13 < \Phi(w_j(1-T_j)) \leq \frac23; \\
1/2 & \text { if } \frac23 <\Phi(w_i(1-T_i)) \le 1 \text { and } \frac23 <\Phi(w_j(1-T_j)) \le 1; \\
1/5 & \text { otherwise}.    
\end{cases}$

\end{itemize}

\paragraph{Design 2: Homophily Model}
\begin{itemize}
\item $g^\text{pre}(w_i, w_j) = 1-(\Phi(w_i)-\Phi(w_j))^2$.
\item $g^\text{post}(w_i, w_j, T_j, T_j) = 1- \left(\Phi(w_i)(1-T_i) -\Phi(w_j)(1-T_j) \right)^2$.
\end{itemize}

\paragraph{Design 3: Beta Model}
\begin{itemize}
\item $g^\text{pre}(w_i, w_j) = \frac{\exp(\Phi(w_i) + \Phi(w_j))}{1+ \exp(\Phi(w_i) + \Phi(w_j))}$.
\item $g^\text{post}(w_i, w_j, T_j, T_j) =  \frac{\exp(\Phi(w_i) + \Phi(w_j) +  T_i + T_j + T_iT_j)}{1+ \exp(\Phi(w_i) + \Phi(w_j) + T_i + T_j + T_iT_j)}$.
\end{itemize}

\paragraph{Design 4: Homophily Model}
\begin{itemize}
\item $g^\text{pre}(w_i, w_j) = 1- \left(\Phi(w_i)-\Phi(w_j)\right)^2$.
\item $g^\text{post}(w_i, w_j, T_i, T_j) = 1- \left(\Phi(w_i(1-T_i))-\Phi(w_j (1-T_j))\right)^2 $.
\end{itemize}
For all four models, I set the diagonal elements of pre- and post-intervention networks to zero, i.e., $g^\text{pre}(w_i, w_i) = 0$ and $g^\text{post}(w_i, w_i, T_i, T_i) = 0$. 
The ranks $r$ for these four designs are $3$, $2$, $2$ and $2$, respectively.

I explore the sparsity regime where $q_n^\text{pre} = q_n^\text{post} = q_n$, considering different levels of sparsity.
The adjacency matrices for the pre- and post-intervention networks are generated as follows:
\begin{align*}
A_{ij}^\text{pre}
~=~& 1\left\{\eta_{i j} \leq q_n g^\text{pre}(w_i, w_j)\right\}, \\
A_{ij}^\text{post}
~=~& 1\left\{\eta_{i j} \leq q_n g^\text{post}(w_i, w_j, T_i, T_j)\right\},
\end{align*}
where $w_i \overset{\text{i.i.d.}}{\sim } \mathcal{N}(0,1)$ and $\eta_{ij} \overset{\text{i.i.d.}}{\sim } U[0,1]$ for $i<j$ and $\eta_{ij} = \eta_{ji}$.
Each simulated dataset has $n\in \{200, 800\}$ with $5,000$ replications.

Table \ref{tab:Mi} reports the within-sample mean and standard error of the mediator $M_i$ for varying levels of sparsity and sample sizes over four designs. 
The results for $q_n = n^{-1}$ show that Designs 1--4 maintain stable standard errors as the sample size increases. 
For $q_n = n^{-1/2}$, the standard errors for Designs 1--4 decrease with sample size, with Designs 3 and 4 shrinking at a faster rate. 
Finally, for $q_n = 1$, Designs 1 and 2 maintain stable standard errors as the sample size grows, while Designs 3 and 4 exhibit a reduction in standard errors by half when the sample size increases from 200 to 800.
Moreover, for Designs 3 and 4, when the network is denser, the standard errors of $M_i$ decrease, suggesting insufficient variation within the sample and a potential issue of near collinearity when fitting linear regression.
Consequently, I use Designs 1 and 2 to represent Case (a) with non-degenerate $\xi_i$, and Designs 3 and 4 to represent Case (a) with constant $\xi_i$.
Notably, both Designs 2 and 4 are homophily models, but how treatments enter the network model differs, consequently influencing the properties of $M_i$ and the performance of the estimators.

\begin{table}[ht]
\caption{$M_i$ in different designs}
\centering
\footnotesize
\label{tab:Mi}
\begin{tabular}{cl|ccc|ccc}
\hline \hline
Design & statistics & \multicolumn{6}{c}{$q_n$} \\
\hline
 & & \multicolumn{3}{c}{$n=200$} & \multicolumn{3}{c}{$n=800$}	\\						
\hline	
  &  & $-1$ & $-1/2$ & 1 & $-1$ & $-1/2$ & 1 \\
\hline 
1	&	mean	&	0.173	&	0.533	&	0.534	&	0.174	&	0.536	&	0.535	\\
	&	std	&	0.367	&	0.272	&	0.142	&	0.369	&	0.208	&	0.137	\\
\hline															
2	&	mean	&	0.282	&	0.507	&	0.507	&	0.284	&	0.508	&	0.508	\\
	&	std	&	0.420	&	0.199	&	0.134	&	0.422	&	0.170	&	0.135	\\
\hline															
3	&	mean	&	0.282	&	0.521	&	0.522	&	0.282	&	0.521	&	0.522	\\
	&	std	&	0.419	&	0.155	&	0.020	&	0.421	&	0.108	&	0.012	\\
\hline															
4	&	mean	&	0.312	&	0.523	&	0.523	&	0.313	&	0.523	&	0.523	\\
	&	std	&	0.425	&	0.141	&	0.012	&	0.427	&	0.099	&	0.006	\\
\hline
\hline
\end{tabular}
\caption*{\small Note: Within-sample mean and standard deviation of $M_i$ across four designs with $n \in \{200, 800\}$, $q_n \in \{n^{-1}, n^{-1/2}, 1\}$, and 5,000 replications.}
\end{table}

Section \ref{sec:MC_ols} presents the OLS estimation results without unobserved confounders, while Section \ref{sec:MC_IV} presents the IV estimation results accounting for unobserved confounders.
I omit the results of $\beta_0$ for brevity.

\subsection{OLS estimation }\label{sec:MC_ols}
I generate the outcomes as follows:
\[
Y_i = \beta_0+ \beta_1 T_i  + \beta_2 \frac{\sum_{j=1}^n A_{ij}^\text{post}T_j}{\sum_{j=1}^n A_{ij}^\text{post}}  +\varepsilon_i
\]
with $(\beta_0, \beta_1, \beta_2) = (1,1,0.5)$, $\varepsilon_i \overset{\text{i.i.d.}}{\sim} U[-1,1]$ and $\varepsilon_i \indep w_i$. 
Thus, there is no endogeneity issue.

Table \ref{tab:OLS_pn-1/2} presents the results for $q_n = n^{-1/2}$. 
The results for $q_n = n^{-2/3}$ (sparser networks) and $q_n = n^{-1/3}$ (denser networks) follow similar patterns and are therefore omitted.
The top panel of Table \ref{tab:OLS_pn-1/2} reports the results for $n = 200$, while the bottom panel shows the results for $n = 800$. 
I report the mean and the standard deviation of the OLS estimators across simulation draws under ``$\hat{\beta}^{\textsc{ols}}$'' and ``std($\hat{\beta}^{\textsc{ols}}$)'', using $A^\text{post}$ and $A^\text{pre}$ to construct the fraction of treated friends, respectively, in order to approximate the true value and asymptotic standard deviation of the estimators.
For the OLS fits using $A^\text{post}$, I also report the heteroskedastic consistent standard errors under ``s.e.'', and the corresponding coverage of 95\% Confidence Intervals (CIs) under ``coverage.''
For the OLS fits using $A^\text{pre}$, I report the coverage of 95\% CIs based on the corresponding heteroskedastic consistent standard error under ``coverage.''

In line with Theorem \ref{thm:consistency_ratio}, the estimators $\hat{\beta}^{\textsc{ols}}$ are consistent, with estimators for all four designs concentrated around the true values. 
As the sample size increases from 200 to 800, the standard deviations of $\hat{\beta}_1^{\textsc{ols}}$ across all models decrease by approximately half, consistent with the expected $\sqrt{n}$ convergence rate.
For $\hat{\beta}_2^{\textsc{ols}}$, the standard deviations in Designs 1 and 2 also decrease by half, whereas those in Designs 3 and 4 exhibit a slower convergence rate due to smaller variation in $M_i$ within the sample. 
In line with Theorem \ref{thm:asymnormal_ratio}, the coverage of the 95\% CIs using the heteroskedasticity-consistent standard errors is close to the target level.

However, if I mistakenly use $A^{\text{pre}}$ to calculate the mediator variable $M_i$, the coefficients may be inaccurate, and the coverage can become arbitrarily poor. It is worth noting that, for Designs 3 and 4, where the post-intervention networks have a weaker dependency on the treatment, the estimators are not as far off, and the coverage of the CIs is better than in Designs 1 and 2.

\begin{table}
\caption{Simulation results of OLS estimation with $q_n^\text{pre} = q_n^\text{post}=n^{-1/2}$}
\centering
\footnotesize
\label{tab:OLS_pn-1/2}
\begin{tabular}{cl|cccc|ccc}
\hline \hline
Design & $\beta$ & \multicolumn{4}{c}{OLS with $A^\text{post}$} & \multicolumn{3}{c}{OLS with $A^\text{pre}$} \\
\hline
& &  $\hat{\beta}^{\textsc{ols}}$ &  std($\hat{\beta}^{\textsc{ols}}$) & s.e. & coverage  & $\hat{\beta}^{\textsc{ols}}$ &  std($\hat{\beta}^{\textsc{ols}}$)  & coverage  \\
\hline 
\multicolumn{9}{c}{$n=200$}	\\						
\hline	
1	&	$\beta_1=1$	&	1.000	&	0.088	&	0.087	&	0.941	&	1.097	&	0.085	&	0.789	\\
	&	$\beta_2=0.5$	&	0.500	&	0.161	&	0.159	&	0.940	&	0.326	&	0.167	&	0.819	\\
\hline															
2	&	$\beta_1=1$	&	1.000	&	0.092	&	0.091	&	0.946	&	1.093	&	0.083	&	0.796	\\
	&	$\beta_2=0.5$	&	0.502	&	0.232	&	0.228	&	0.945	&	0.409	&	0.275	&	0.937	\\
\hline															
3	&	$\beta_1=1$	&	1.001	&	0.084	&	0.082	&	0.941	&	0.991	&	0.084	&	0.945	\\
	&	$\beta_2=0.5$	&	0.506	&	0.285	&	0.278	&	0.936	&	0.412	&	0.255	&	0.935	\\
\hline															
4	&	$\beta_1=1$	&	0.999	&	0.082	&	0.082	&	0.943	&	0.998	&	0.083	&	0.948	\\
	&	$\beta_2=0.5$	&	0.500	&	0.292	&	0.288	&	0.940	&	0.435	&	0.277	&	0.943	\\
\hline															\multicolumn{9}{c}{$n=800$}	\\							
\hline	
1	&	$\beta_1=1$	&	0.999	&	0.047	&	0.046	&	0.945	&	1.095	&	0.042	&	0.360	\\
	&	$\beta_2=0.5$	&	0.502	&	0.112	&	0.111	&	0.950	&	0.355	&	0.123	&	0.785	\\
\hline																
2	&	$\beta_1=1$	&	1.000	&	0.049	&	0.049	&	0.950	&	1.092	&	0.042	&	0.396	\\
	&	$\beta_2=0.5$	&	0.500	&	0.146	&	0.143	&	0.945	&	0.415	&	0.198	&	0.929	\\
\hline																
3	&	$\beta_1=1$	&	1.000	&	0.083	&	0.082	&	0.942	&	0.990	&	0.083	&	0.946	\\
	&	$\beta_2=0.5$	&	0.496	&	0.481	&	0.466	&	0.936	&	0.406	&	0.431	&	0.944	\\
\hline																
4	&	$\beta_1=1$	&	1.000	&	0.040	&	0.041	&	0.951	&	0.999	&	0.041	&	0.951	\\
	&	$\beta_2=0.5$	&	0.500	&	0.207	&	0.206	&	0.948	&	0.441	&	0.198	&	0.942	\\
\hline
\hline
\end{tabular}
\caption*{\small Note: Simulation results for the OLS estimators using $A^{\text{post}}$ and $A^{\text{pre}}$ with $n \in \{200, 800\}$, $q=n^{-1/2}$, and 5, 000 replications. }
\end{table}

\subsection{IV estimation with SSIV} \label{sec:MC_IV}
I generate the outcome as follows:
\[
Y_i = \beta_0+ \beta_1 T_i  + \beta_2 \frac{\sum_{j=1}^n A_{ij}^\text{post}T_j}{\sum_{j=1}^n A_{ij}^\text{post}}  
 +  (w_i + \varepsilon_i)/2
\]
with $(\beta_0, \beta_1, \beta_2) = (1,1,0.5)$, $\varepsilon_i \overset{\text{i.i.d.}}{\sim} U[-1,1]$ and $\varepsilon_i \indep w_i$. Hence, the error term is $u_i = (w_i + \varepsilon_i)/2$.

Tables \ref{tab:IV_pn-log} and \ref{tab:IV_pn-1/5} present the estimation results of IV fits using SSIV and modified SSIV with $q_n = {\frac{\log(n)}{\log(\log(n))}} / n$ and $q_n = n^{-1/5}$, respectively.
For each table, the top panel reports the results for $n = 200$, while the bottom panel shows the results for $n = 800$. 
The left panel presents the mean and standard deviation of the IV estimators using SSIV, along with our estimated standard errors and the coverage of 95\% CIs across simulations, reported as ``$\hat{\beta}^{\textsc{iv}}$'', ``std($\hat{\beta}^{\textsc{iv}})$'', ``s.e.'', and ``coverage''. The right panel provides the corresponding statistics for the IV estimators using modified SSIV, reported as ``$\hat{\beta}^{\textsc{de}}$'', ``std($\hat{\beta}^{\textsc{de}})$'', ``s.e.'', and ``coverage''.

I begin with Table \ref{tab:IV_pn-log}, where the network is relatively sparse.
In this regime, the network is relatively sparse, and Theorem \ref{thm:IV_consistency} asserts that the SSIV estimators $\hat{\beta}^{\textsc{iv}}$ are consistent. I observe that the estimators are concentrated around the true values. As the sample size increases from $200$ to $800$, the standard deviations of $\hat{\beta}_1^{\textsc{iv}}$ decrease by half, indicating a convergence rate of $\sqrt{n}$. However, the standard deviation of $\hat{\beta}_2^{\textsc{iv}}$ shows a slower convergence rate.
In line with Theorem \ref{thm:asymnormal_ratio_IV}, the coverage of the 95\% CIs using our variance estimate is close to the target level. 

Table \ref{tab:IV_pn-1/5} presents the results for the IV fits with $q_n = n^{-1/5}$, denser than the threshold where SSIV remains valid.
In Table \ref{tab:IV_pn-1/5}, I also report the results with $n = 1600$.
In line with Theorem \ref{thm:IV_consistency}, I observe that $\hat{\beta}^{\textsc{iv}}$ is no longer consistent for Designs 1 and 2; the estimators can be arbitrarily off, and the standard deviation does not shrink as the sample size increases. 
For Designs 3 and 4, $\hat{\beta}^{\textsc{iv}}_1$ remains consistent with convergence rate $\sqrt{n}$, while $\hat{\beta}^{\textsc{iv}}_2$ is not. 
By applying the modified SSIV, I effectively project out some noise and reduce the estimation error.
The IV estimator $\hat{\beta}^{\textsc{de}}_1$ is consistent with a convergence rate of $\sqrt{n}$, while $\hat{\beta}^{\textsc{de}}_2$ is consistent with a convergence rate of $1/\sqrt{q_n}$; the coverage of the 95\% confidence intervals for both estimators is approximately at the target level.

\begin{table}
\caption{Simulation results of IV estimation with $q_n = {\frac{\log(n)}{\log(\log(n))}} / n$}
\centering
\footnotesize
\label{tab:IV_pn-log}
\begin{tabular}{cl|cccc|ccccc}
\hline \hline
Design  &  $\beta$ & \multicolumn{4}{c}{SSIV} & \multicolumn{4}{c}{modified SSIV}  \\
\hline
&  & $\hat{\beta}^{\textsc{iv}}$ &  std($\hat{\beta}^{\textsc{iv}}$) & s.e. & coverage & $\hat{\beta}^{\textsc{de}}$ &  std($\hat{\beta}^{\textsc{de}}$) & s.e. & coverage \\
\hline 
\multicolumn{10}{c}{$n=200$}	\\					
\hline	
1	&	$\beta_1=1$	&	1.001	&	0.083	&	0.089	&	0.950	&	1.001	&	0.083	&	0.087	&	0.950	\\
	&	$\beta_2=0.5$	&	0.501	&	0.207	&	0.215	&	0.953	&	0.499	&	0.203	&	0.214	&	0.956	\\
\hline	
2	&	$\beta_1=1$	&	0.999	&	0.083	&	0.087	&	0.954	&	1.000	&	0.083	&	0.082	&	0.940	\\
	&	$\beta_2=0.5$	&	0.505	&	0.150	&	0.157	&	0.952	&	0.505	&	0.168	&	0.210	&	0.983	\\
\hline																
3	&	$\beta_1=1$	&	0.999	&	0.083	&	0.083	&	0.944	&	0.999	&	0.083	&	0.083	&	0.944	\\
	&	$\beta_2=0.5$	&	0.500	&	0.142	&	0.146	&	0.953	&	0.500	&	0.143	&	0.149	&	0.949	\\
\hline																
4	&	$\beta_1=1$	&	1.000	&	0.082	&	0.084	&	0.949	&	1.000	&	0.082	&	0.084	&	0.950	\\
	&	$\beta_2=0.5$	&	0.502	&	0.137	&	0.142	&	0.955	&	0.502	&	0.138	&	0.145	&	0.956	\\
\hline															\multicolumn{10}{c}{$n=800$}	\\							\hline	
1	&	$\beta_1=1$	&	1.000	&	0.040	&	0.044	&	0.955	&	1.000	&	0.040	&	0.042	&	0.954	\\
	&	$\beta_2=0.5$	&	0.501	&	0.099	&	0.101	&	0.951	&	0.501	&	0.098	&	0.097	&	0.946	\\
\hline																
2	&	$\beta_1=1$	&	1.000	&	0.040	&	0.042	&	0.957	&	1.000	&	0.040	&	0.042	&	0.956	\\
	&	$\beta_2=0.5$	&	0.500	&	0.076	&	0.076	&	0.948	&	0.501	&	0.078	&	0.081	&	0.958	\\
\hline																
3	&	$\beta_1=1$	&	1.000	&	0.041	&	0.041	&	0.950	&	1.000	&	0.041	&	0.041	&	0.950	\\
	&	$\beta_2=0.5$	&	0.500	&	0.072	&	0.071	&	0.946	&	0.500	&	0.072	&	0.072	&	0.946	\\
\hline																
4	&	$\beta_1=1$	&	1.000	&	0.041	&	0.041	&	0.950	&	1.000	&	0.041	&	0.041	&	0.950	\\
	&	$\beta_2=0.5$	&	0.501	&	0.068	&	0.069	&	0.957	&	0.501	&	0.068	&	0.069	&	0.953	\\
\hline
\hline
\end{tabular}
\caption*{\small Note: Simulation results for the IV estimators using SSIV and the modified SSIV with $n \in \{200, 800\}$, $q_n = {\frac{\log(n)}{\log(\log(n))}} / n$, and 5, 000 replications. }
\end{table}

\begin{table}[ht]
\caption{Simulation results of IV estimation with $q_n = n^{-1/5}$}
\centering
\footnotesize
\label{tab:IV_pn-1/5}
\begin{tabular}{cl|cccc|ccccc}
\hline \hline
Design  & $\beta$ & \multicolumn{4}{c}{SSIV} & \multicolumn{4}{c}{modified SSIV}  \\
\hline
&  & $\hat{\beta}^{\textsc{iv}}$ &  std($\hat{\beta}^{\textsc{iv}}$) & s.e. & coverage & $\hat{\beta}^{\textsc{de}}$ &  std($\hat{\beta}^{\textsc{de}}$) & s.e. & coverage \\
\hline 
\multicolumn{10}{c}{$n=200$}	\\		
\hline	
1	&	$\beta_1=1$	&	1.006	&	0.236	&	0.307	&	0.980	&	0.999	&	0.114	&	0.123	&	0.945	\\
	&	$\beta_2=0.5$	&	0.455	&	1.501	&	1.584	&	0.950	&	0.514	&	0.452	&	0.522	&	0.974	\\
\hline																
2	&	$\beta_1=1$	&	0.946	&	0.330	&	0.369	&	0.972	&	0.997	&	0.144	&	0.167	&	0.952	\\
	&	$\beta_2=0.5$	&	0.944	&	2.257	&	2.041	&	0.959	&	0.537	&	0.750	&	0.846	&	0.978	\\
\hline																
3	&	$\beta_1=1$	&	0.999	&	0.085	&	0.090	&	0.939	&	0.999	&	0.084	&	0.088	&	0.940	\\
	&	$\beta_2=0.5$	&	0.511	&	0.990	&	1.189	&	0.963	&	0.504	&	0.883	&	1.027	&	0.963	\\
\hline																
4	&	$\beta_1=1$	&	1.000	&	0.091	&	0.086	&	0.950	&	1.000	&	0.083	&	0.084	&	0.950	\\
	&	$\beta_2=0.5$	&	0.526	&	1.562	&	1.773	&	0.976	&	0.507	&	0.669	&	0.777	&	0.973	\\
\hline															\multicolumn{10}{c}{$n=800$}	\\								\hline	
1	&	$\beta_1=1$	&	1.023	&	0.379	&	0.401	&	0.989	&	1.001	&	0.077	&	0.082	&	0.947	\\
	&	$\beta_2=0.5$	&	0.376	&	2.167	&	2.106	&	0.947	&	0.498	&	0.348	&	0.380	&	0.966	\\
\hline																
2	&	$\beta_1=1$	&	0.879	&	1.380	&	0.453	&	0.973	&	0.998	&	0.092	&	0.103	&	0.953	\\
	&	$\beta_2=0.5$	&	1.278	&	8.703	&	2.507	&	0.957	&	0.509	&	0.463	&	0.521	&	0.971	\\
\hline
3	&	$\beta_1=1$	&	1.000	&	0.044	&	0.047	&	0.949	&	1.000	&	0.041	&	0.044	&	0.949	\\
	&	$\beta_2=0.5$	&	0.518	&	0.957	&	1.041	&	0.952	&	0.494	&	0.709	&	0.777	&	0.961	\\
\hline																
4	&	$\beta_1=1$	&	1.000	&	0.047	&	0.042	&	0.948	&	1.000	&	0.041	&	0.041	&	0.948	\\
	&	$\beta_2=0.5$	&	0.514	&	1.921	&	1.991	&	0.973	&	0.489	&	0.435	&	0.489	&	0.969	\\
 \hline															\multicolumn{10}{c}{$n=1600$}	\\								\hline	
1	&	$\beta_1=1$	&	1.018	&	0.907	&	0.908	&	0.990	&	1.000	&	0.068	&	0.071	&	0.944	\\
	&	$\beta_2=0.5$	&	0.405	&	5.250	&	5.023	&	0.968	&	0.499	&	0.320	&	0.340	&	0.960	\\
\hline	
2	&	$\beta_1=1$	&	0.800	&	2.976	&	3.023	&	0.979	&	0.999	&	0.077	&	0.085	&	0.953	\\
	&	$\beta_2=0.5$	&	1.685	&	17.444	&	17.494	&	0.916	&	0.506	&	0.395	&	0.438	&	0.967	\\
\hline	
3	&	$\beta_1=1$	&	1.001	&	0.031	&	0.030	&	0.943	&	1.001	&	0.031	&	0.030	&	0.946	\\
	&	$\beta_2=0.5$	&	0.507	&	0.876	&	0.903	&	0.958	&	0.499	&	0.579	&	0.615	&	0.958	\\
\hline	
4	&	$\beta_1=1$	&	1.000	&	0.034	&	0.030	&	0.953	&	1.000	&	0.029	&	0.029	&	0.952	\\
	&	$\beta_2=0.5$	&	0.477	&	2.201	&	2.219	&	0.970	&	0.501	&	0.370	&	0.414	&	0.970	\\
\hline
\hline
\end{tabular}
\caption*{\small Note: Simulation results for the IV estimators using SSIV and the modified SSIV with $n \in \{200, 800, 1600\}$, $q_n = n^{-1/5}$, and 5, 000 replications. }
\end{table}

\section{Empirical Application}\label{sec:app}
In this section, I revisit \cite{Prina2015}, which offered access to formal savings accounts to a random sample of poor households in 19 villages in Nepal.\footnote{\cite{Prina2015} conducted the RCT and collected network data but did not perform any network analysis.
\cite{ComolaPrina2021} studied the treatment effects while accounting for network changes and revisited \cite{Prina2015} as an empirical illustration. \cite{ComolaPrina2023} also revisited the RCT from \cite{Prina2015} and focused on dyadic-level regressions to assess the treatment's impact on financial flows. I integrate the panel network data released by \cite{ComolaPrina2023} with the outcome variables from \cite{Prina2015}. The treatment variable is available in both datasets.}
I apply the method in this paper to make estimations and inference on the causal effects of interest.

I now provide a more detailed description of the empirical setting of the RCT in \cite{Prina2015}, expanding on Example \ref{ex:Prina2015}.
Before the introduction of savings accounts, a baseline survey was conducted in May 2010 across 19 villages in Pokhara. All households with a female head aged 18-55 were surveyed.
Following the baseline survey, half of the female household heads were randomly assigned to the treatment group through a public lottery between late May and early June 2010, offering them savings accounts at the local bank.
The remaining half, in the control group, did not receive this offer. An endline survey was conducted in June 2011, after the intervention.
This study analyzes a sample of 915 households that participated in both survey waves.
Both surveys collected data on household socioeconomic characteristics and informal financial transactions.

Following \cite{ComolaPrina2023}, I define financial links as transfers--loans or gifts, either sent or received--between households, based on survey questions such as, ``With whom did you exchange loans or gifts?''
The resulting networks, both pre- and post-intervention, were sparse, with small groups and minimal clustering.
Despite the total number of links remaining stable (328 at baseline, 329 at endline), the network underwent significant reshuffling. As shown in Example \ref{ex:Prina2015}, 255 links were broken, and 256 new links formed by the endline. 
\cite{ComolaPrina2023} reported high uptake and usage of the savings accounts, with over 84\% of treated households opening accounts and depositing about 8\% of their baseline weekly income nearly every week during the first year.

Table \ref{tab:estimate_app} presents the results of OLS estimates and IV estimators using both SSIV and modified SSIV, across various outcomes, with corresponding standard errors shown in parentheses.
I also report the point estimates using the normalized SSIV, $\frac{\sum_{j=1}^n A_{ij}^\text{pre}T_j}{\sum_{j=1}^n A_{ij}^\text{pre}}$, without standard errors, as I do not conduct an asymptotic distribution analysis.
The top panel reports the results of household expenditures in the 30 days before the endline survey for different categories, measured by Nepalese rupee. 
The expenditure categories include health, education, festivals and ceremonies, meat and fish, fish, and other expenditures.\footnote{Other expenditures include clothes and footwear, personal care items, house cleaning articles, house maintenance, and bus and taxi fares.} 
The bottom panel focuses on education-related expenditures, including school fees, textbooks, uniforms, and school supplies (e.g., pens and pencils). 
For these expenditures, the sample is restricted to households with school-age children (ages 6--16).
For brevity, I omit the results for $\beta_0$ across the various estimation specifications.

The discrepancy between the OLS and the IV estimators using SSIV, in either the sign of the coefficients or their significance, highlights the bias caused by endogeneity due to unobserved confounders.
The IV results using SSIV show that the treatment affects various outcomes through different channels.
In the ``Fish'' expenditure column, the direct effect of access to savings accounts on fish consumption is not statistically significant, while the indirect effect, mediated through the fraction of treated friends, is significantly positive.
A $0.1$ increase in the fraction of friends with access to savings accounts leads to an increase in fish consumption by Rs. $25.29$. Since the fraction of friends with access to savings accounts is $0.062$ higher for the treatment group compared to the control group, this implies an additional Rs. $15.63$ in fish consumption for the treatment group.
The spillover effects, from assigning others to the treated group versus the control group, result in an increase in fish consumption by Rs. $252.91$.

For total education expenditures, as well as expenditures on textbooks and school supplies, the direct effects of access to savings accounts are positive and significant, while the indirect effects through fraction of treated friends are not significant.
Access to free savings accounts leads to an increase of Rs. 574.43, Rs. 223.26, and Rs. 104.47 in total education expenditure, school textbooks, and school supplies, respectively, compared to control units, holding the fraction of treated friends constant.
Our methods effectively disentangle the direct treatment effects from the indirect effects mediated through networks, providing a clearer understanding of the mechanisms by which the intervention influences the outcomes of interest.
Moreover, for outcomes where $\hat{\beta}_2^{\textsc{iv}}$ is not significant, the estimates $\hat{\beta}_1^{\textsc{ols}}$ and $\hat{\beta}_1^{\textsc{iv}}$ are closely aligned, as the treatment is randomly assigned and omitted variable bias is not a concern in these cases.

Furthermore, the point estimators from IV fits using SSIV and normalized SSIV are closely aligned, demonstrating the effectiveness of both versions of SSIV.
Despite the relative sparsity of the pre-intervention network, the results from the modified SSIV show that the point estimators are closely aligned with those from using the SSIV and normalized SSIV.
This coherence further suggests that the modification based on the estimated eigenvectors does not introduce excessive noise into the estimation.

For the sake of comparison, Table \ref{tab:estimate_app} also includes the corresponding results from \citet[Tables 5 and 6]{Prina2015}.\footnote{\citet{Prina2015} did not analyze the expenditure on ``Fish'', so I leave the result blank.} 
\citet{Prina2015} estimated the average effect of being assigned to the treatment group on each outcome variable, with explanatory variables including the treatment, the baseline value of the outcome variable, and a vector of baseline characteristics.\footnote{Baseline characteristics include age, years of education, marital status of the account holder; number of household members; baseline household income; and three dummies for the main source of household income.}
The treatment indicator's coefficient is referred to as the ITT due to the presence of non-compliance.
Similar to the OLS results in the top panel, for outcomes where $\hat{\beta}_2^{\textsc{iv}}$ in the IV fits is not significant, the ITT estimators reported by \citet{Prina2015} and the $\hat{\beta}_1^{\textsc{iv}}$ estimators from the IV fits are closely aligned. 
This alignment suggests that unobserved confounders do not bias the estimation of the (direct) treatment effect in these cases, as the treatment is randomly assigned.

\begin{table}
\caption{Estimation result based on data from \cite{Prina2015}}
\centering
\footnotesize
\label{tab:estimate_app}
\begin{tabular}{ccccccccc}
\hline \hline
Method	&	$\hat{\beta}$	& Total &	Health		&	Festivals 	 	&	Meat	&	 Fish &	Other	\\
	&			& Expend.	&	 &	\& Ceremonies		&	\& Fish & 	&	Expend. \\
\hline
OLS	&	$\hat{\beta}_1^{\textsc{ols}}$	&	141.072	&	-670.366	&	267.112*	&	99.824	&	48.505	&	206.433	\\
	&		&	(1030.252)	&	(644.920)	&	(141.412)	&	(96.476)	&	(39.683)	&	(687.006)	\\
	&	$\hat{\beta}_2^{\textsc{ols}}$	&	-39.811	&	-50.981	&	-157.612	&	151.700	&	-55.696	&	-277.696	\\
	&		&	(1125.085)	&	(773.712)	&	(119.425)	&	(115.674)	&	(35.796)	&	(648.408)	\\
\hline															
SSIV	&	$\hat{\beta}_1^{\textsc{iv}}$	&	-143.459	&	-729.080	&	232.207	&	107.500	&	33.276	&	106.640	\\
	&		&	(1122.582)	&	(681.558)	&	(153.465)	&	(106.339)	&	(42.139)	&	(741.432)	\\
	&	$\hat{\beta}_2^{\textsc{iv}}$	&	5726.414	&	1138.892	&	549.756	&	-3.858	&	252.913*	&	1744.679	\\
	&		&	(6997.446)	&	(4723.695)	&	(964.534)	&	(742.139)	&	(150.832)	&	(3573.983)	\\
\hline															
modified 	&	$\hat{\beta}_1^{\textsc{de}}$	&	-47.271	&	-745.607	&	272.332	&	135.879	&	27.030	&	129.909	\\
SSIV	&		&	(1130.105)	&	(704.263)	&	(155.618)	&	(110.431)	&	(42.708)	&	(738.796)	\\
	&	$\hat{\beta}_2^{\textsc{de}}$	&	3777.096	&	1473.816	&	-263.414	&	-578.965	&	379.510*	&	1273.114	\\
	&		&	(9034.384)	&	(6485.695)	&	(1314.334)	&	(962.715)	&	(206.481)	&	(4549.693)	\\
\hline															
normalized	&	$\hat{\beta}_1$	&	34.099	&	-760.932	&	229.050	&	97.796	&	42.299	&	300.023	\\
SSIV	&	$\hat{\beta}_2$	&	2128.082	&	675.746	&	613.734	&	192.805	&	70.062	&	-2174.355	\\
\hline															
\cite{Prina2015}	& ITT	&	340.780	&	-290.34	&	248.070*	&	153.600	&	--	&	49.120	\\
	&		&	(860.930)	&	(540.710)	&	(123.670)	&	(87.770)	&	--	&	(569.010)	\\
\hline
\hline
 & 	& Total Exp. 	& Exp. on	&	Exp. on & Exp. on &	Exp. on	&   \\
 & 	& 	in Education	& Sch. Fees &	Textbooks	& Sch. Uniforms	& Sch. Supplies	&	 \\
\hline 
OLS	&	$\hat{\beta}_1^{\textsc{ols}}$	&	562.276**	&	137.764	&	246.481**	&	74.264	&	103.767**	&		\\
	&		&	(281.922)	&	(162.391)	&	(105.114)	&	(63.945)	&	(45.880)	&		\\
	&	$\hat{\beta}_2^{\textsc{ols}}$	&	87.445	&	160.412	&	-91.316	&	-3.846	&	22.194	&		\\
	&		&	(336.824)	&	(212.844)	&	(111.511)	&	(78.980)	&	(56.184)	&		\\
\hline															
SSIV	&	$\hat{\beta}_1^{\textsc{iv}}$	&	574.427*	&	183.217	&	223.261*	&	63.479	&	104.469**	&		\\
	&		&	(305.988)	&	(180.058)	&	(117.887)	&	(67.727)	&	(50.387)	&		\\
	&	$\hat{\beta}_2^{\textsc{iv}}$	&	-109.182	&	-575.128	&	284.441	&	170.681	&	10.823	&		\\
	&		&	(1318.538)	&	(633.715)	&	(668.099)	&	(271.842)	&	(223.504)	&		\\
\hline															
modified	&	$\hat{\beta}_1^{\textsc{de}}$	&	576.798*	&	197.519	&	211.983*	&	64.979	&	100.048**	&		\\
SSIV	&		&	(295.529)	&	(175.501)	&	(114.828)	&	(65.675)	&	(48.519)	&		\\
	&	$\hat{\beta}_2^{\textsc{de}}$	&	-147.559	&	-806.571	&	466.935	&	146.417	&	82.376	&		\\
	&		&	(1395.153)	&	(718.228)	&	(644.052)	&	(314.866)	&	(219.500)	&		\\
\hline															
normalized	&	$\hat{\beta}_1$	&	560.511	&	191.865	&	214.801	&	57.291	&	96.554	&		\\
SSIV	&	$\hat{\beta}_2$	&	116.003	&	-715.069	&	421.338	&	270.828	&	138.906	&		\\
\hline															
\cite{Prina2015}	&	ITT	&	667.560**	&	224.720	&	213.740*	&	113.150*	&	115.950**	&		\\
	&		&	(320.630)	&	(193.600)	&	(120.220)	&	(65.400)	&	(49.980)	&		\\
\hline
\hline
\end{tabular}
\caption*{\small Note: Dependent variables are expressed in Nepalese rupees (the exchange rate was roughly Rs. 70 to USD 1 during the study period). The significance level of the estimators is indicated with stars in the usual manner: ``***'' means significant at $\alpha=1\%$, ``**'' means significant at $\alpha=5\%$, and``*'' means significant at $\alpha=10\%$. }
\end{table}

\section{Conclusion} \label{sec:conclusion}

This paper investigates the identification and inference of treatment effects in RCTs with network interference, focusing on two key aspects: (1) unobserved confounders affecting both network formation and outcomes, and (2) network changes induced by the intervention.
The framework demonstrates that treatment affects outcomes in two distinct ways: directly and indirectly through changes in the network mediator. 
Disentangling these channels offers deeper insights into the intervention's mechanisms and informs more effective policy design.
This paper presents methods for the estimation and inference of causal effects in the presence of network interference, addressing endogeneity concerns. 
In the absence of endogeneity, I recommend OLS estimation using post-intervention network data. 
When unobserved confounders are present, I recommend IV estimation with SSIV for relatively sparse networks and propose a modification to SSIV for denser networks. 
For all estimators, I provide consistent variance estimators for normal approximation, ensuring valid inference.
More generally, this paper highlights a concern regarding the relevance condition of SSIV as the network becomes denser. While I focus on the mean impact of peers, the discussion of the SSIV relevance condition and denoising modification also provides insights applicable to other forms of mediators.

This paper opens several avenues for future research. A natural extension would be to incorporate endogenous and correlated peer effects, along with covariates, into the linear model.
An expanding body of research highlights the importance of measuring and accounting for network changes when evaluating the welfare impacts of policies \citep{CarrellSacerdoteWest2013, Jackson2021, BanerjeeBrezaChandrasekhar2023}. While the literature extensively examines optimal treatment assignment under interference \citep{BairdBohrenMcIntosh2018b, CaiPouget-AbadieAiroldi2022, Viviano2023, Viviano2023a}, the challenge of designing effective policies that account for endogenous network evolution remains unresolved.
Furthermore, my methods rely on detailed panel network data, but relaxing this requirement to single-period or aggregated network data would be desirable \citep{AlidaeeAuerbachLeung2020, BrezaChandrasekharMcCormick2020}. 
I leave these questions for future research.

\newpage 

\bibliographystyle{ecta}
\bibliography{endogenous_network.bib}

\newpage

\begin{appendix} 
\appendix

\fontsize{11}{13.2}

\begin{center}
\Large\bfseries{Appendix for ``Endogenous Interference in Randomized Experiments''}
\end{center}

\addcontentsline{toc}{section}{Supplementary Material to ``Regression analysis with endogenous interference pattern"} %

The supplementary materials contain the following sections.

Section \ref{app:identification} gives the proofs for the identification. 

Section \ref{app:OLS} gives the proofs for the results of OLS estimation. 

Section \ref{app:IV} gives the proofs for the results of IV estimation with SSIV. 

Section \ref{app:PC_IV} gives the proofs for the results of IV estimation with the modified SSIV. 

Section \ref{app:IV_alt} gives the proofs for the results of IV estimation with the normalized SSIV. 

Section \ref{app:lemma} collects the proofs for some useful lemmas

\section*{Notation}

Let $N_i$ denote the degree of unit $i$ in network $A$, i.e., $N_i = \sum_{j=1}^n A_{ij}$. 
Define $g_0(i,j) = g^{\text{pre}}(w_i, w_j)$ and $g_1(i,j) = g^{\text{post}}(w_i, w_j, T_i, T_j)$. Define $g_0(k) = E(g_0(j,k) \mid w_k)$.
I use $\Omega(), \Omega_{\mathbb{P}}()$ in the following sense: 
$a_n = \Omega(b_n)$ if $a_n \geq C b_n$ for $n$ large enough, where $C$ is a positive constant; 
$X_n = \Omega_{\mathbb{P}}(b_n)$, if for any $\delta > 0$, there exists $M, N > 0$, s.t. $\mathbb{P}(|X_n| \leq Mb_n) \leq \delta$ for any $n > N$.
Let $\| \cdot \|_{\textup{op}}$ denote the operator norm and $\| \cdot \|_{\textup{F}}$ denote the Frobenious norm.
I use ``LLN'', ``CMT'' and ``CLT'' to denote ``law of large number'', ``continuous mapping theorem'' and ``central limit theorem,'' respectively.

\section{Proof of results in Section \ref{sec:setup}}\label{app:identification}

\begin{proof}[Proof of Theorem \ref{thm:identifiaction}]
The proof is adapted from the standard proof in the mediation analysis \citep{Pearl2001}.
We only prove for the discrete case, as the continuous and mixed variable cases follow by analogous arguments.
Given $T_{-i} = t_{-i}$ and $w_i=w$, we have 
\begin{align*}  
& E\left[Y_{i}\left(t_i, M_i(t_i', T_{-i}) \right) \mid T_{-i} = t_{-i}, w_i=w \right]   \\
=& \sum_{m_i \in \mathcal{M}} E\left(Y_i\left(t_i, M_i(t_i', t_{-i})\right) \mid M_i(t_i', t_{-i}) = m_i, T_{-i} = t_{-i}, w_i=w\right)  
\textup{Pr}\left( M_i(t_i', t_{-i}) = m_i \mid T_{-i} = t_{-i}, w_i=w \right) 
 \\
=& \sum_{m_i \in \mathcal{M}} 
E\left(Y_i(t_i, m_i) \mid M_i(t_i', t_{-i})=m_i, T_{-i} = t_{-i}, w_i=w \right)
\textup{Pr}\left( M_i(t_i', t_{-i}) = m_i \mid T_{-i} = t_{-i}, w_i=w \right)   \\
=& \sum_{m_i \in \mathcal{M}} E\left(Y_i(t_i, m_i) \mid T_{-i} = t_{-i}, w_i=w \right)
\textup{Pr}\left( M_i(t_i', t_{-i}) = m_i \mid T_{-i} = t_{-i}, w_i=w \right)   \\
=& \sum_{m_i \in \mathcal{M}} 
E\left(Y_i(t_i, m_i) \mid w_i=w \right)
\textup{Pr}\left( M_i = m_i \mid T_i = t_i',  T_{-i} = t_{-i}, w_i=w \right)  \\
=& \sum_{m_i \in \mathcal{M}} 
E\left(Y_i \mid T_i = t_i, M_i = m_i, w_i=w \right)
\textup{Pr}\left( M_i = m_i \mid T_i = t_i',  T_{-i} = t_{-i}, w_i=w \right). 
\end{align*}
Taking expectations over $T_{-i} \mid w_i=w$, we have 
\begin{align}  
& E\left[Y_{i}\left(t_i, M_i(t_i', T_{-i}) \right) \mid w_i=w \right] \notag \\
=& \sum_{m_i \in \mathcal{M}} 
\left( 
\begin{array}{c}
E\left(Y_i \mid T_i = t_i, M_i = m_i, w_i=w \right) \\
\cdot 
\sum\limits_{t_{-i}\in\{0,1\}^{n-1}} 
\textup{Pr}\left( M_i = m_i \mid T_i = t_i,  T_{-i} = t_{-i}, w_i=w \right) 
\textup{Pr}\left( T_{-i} = t_{-i} \mid w_i=w \right)   \\
\end{array}
\right)  \notag \\
=& \sum_{m_i \in \mathcal{M}} E\left(Y_i \mid T_i = t_i, M_i = m_i, w_i=w \right)
\textup{Pr}\left( M_i = m_i \mid T_i = t_i', w_i=w \right)  
\notag \\
=& E\left\{ E\left(Y_i \mid M_i, T_i = t_i, w_i=w \right) \mid T_i = t_i', w_i=w \right\}.
\label{eq:Y_t_t'}
\end{align}
If $t_i=t_i'$, we have 
$E\left[Y_{i}\left(t_i, M_i(t_i, T_{-i}) \right) \mid w_i=w \right] = E(Y_i \mid T_i=t_i, w_i=w )$.

By \eqref{eq:Y_t_t'}, with $t_i=t_i'$ and $T_{-i} = t_{n-1}$ and given $w_i=w$, we have 
\begin{align*}
& E\left[Y_{i}\left(t_i, M_i(t_i,t_{n-1} ) \right) \mid w_i=w \right]   \\
=& E\left\{ E\left(Y_i \mid M_i, T_i = t_i, w_i=w \right) \mid T_i = t_i, T_{-i} = t_{n-1}, w_i=w \right\}.
\end{align*}
Now we are ready to show the results. 
Define $\textup{DE}(t_i; w)$ as the conditional direct effect with $T_i=t_i$ and $w_i=w$,
$\textup{IE}(t_i; w)$ as the conditional indirect effect with $T_i=t_i$ and $w_i=w$,
$\textup{ToE}(w)$ as the conditional total effect with $w_i=w$,
and $\textup{SE}(t_i;w)$ as the conditional spillover effect with $T_i=t_i$ and $w_i=w$.
We only show $t_i=0$ since the results of $t_i=1$ follow from analogous arguments.
By \eqref{eq:Y_t_t'}, we can show that
\begin{align*}
\textup{DE}(0; w )
~=~& E\left[Y_i\left(1, M_i(0, T_{-i}) \right)-Y_i\left(0, M_i(0, T_{-i}) \right) \mid w_i=w \right] \\
~=~& E\left\{ E(Y_i \mid M_i, T_i=1, w_i=w ) \mid T_i=0, w_i=w \right\}
- E(Y_i \mid T_i=0, w_i=w ), \\
\textup{IE}(0;w)
~=~& E\left[Y_i\left(0, M_i(1, T_{-i}) \right) - Y_i\left(0, M_i(0, T_{-i}) \right) \mid w_i = w \right] \\
~=~& E\left\{ E(Y_i \mid M_i, T_i=0, w_i=w ) \mid T_i=1, w_i=w \right\}
- E(Y_i \mid T_i=0, w_i = w), \\
\textup{ToE}(w)
~=~& E\left[Y_i\left(1, M_i(1, T_{-i}) \right) - Y_i\left(0, M_i(0, T_{-i}) \right) \mid w_i = w \right] \\
~=~& E(Y_i \mid T_i=1, w_i = w)
- E(Y_i \mid T_i=0, w_i = w), \\
\textup{SE}(t_i; w) 
~=~& E\left[Y_i\left(t_i, M_i(t_i, T_{-i} = 1_{n-1}) \right) - Y_i\left(t_i, M_i(t_i, T_{-i} = 0_{n-1}) \right) \mid w_i = w \right] \\
~=~&   
E\left[ E\left(Y_i \mid M_i, T_i = t_i, w_i=w \right) \mid T_i = t_i, T_{-i} = 1_{n-1}, w_i=w \right] \\
&- E\left[ E\left(Y_i \mid M_i, T_i = t_i, w_i=w \right) \mid T_i = t_i, T_{-i} = 0_{n-1}, w_i=w \right].
\end{align*}
The desired result follows from marginalizing over the distribution of $w$.
\end{proof}

\begin{proof}[Proof of Corollary \ref{cor:identifiaction}]
We only prove for the discrete case, as the continuous and mixed variable cases follow by analogous arguments.
Under Assumption \ref{asu:linearY} and combing with Theorem \ref{thm:identifiaction}, the conditional DE equals
\[
\textup{DE}(w) = 
\sum_{m_i \in \mathcal{M}}
\beta_1 \textup{Pr}(M_i = m_i \mid T_i=0, w_i=w)     
= \beta_1,
\]
and thus $\textup{DE} = \beta_1$ by marginalizing over the distribution of $w$.
The conditional IE equals 
\begin{align*}
\textup{IE}(w)
=& \sum_{m_i \in \mathcal{M}} (\beta_0 + \beta_2 m_i + \lambda(w) )
\left\{ \textup{Pr}(M_i = m_i \mid T_i=1,w_i=w) - \textup{Pr}(M_i = m_i \mid T_i=0,w_i=w) \right\}  \\
=& \beta_2 \cdot
\left\{ E(M_i \mid T_i=1,w_i=w) - E(M_i \mid T_i=0,w_i=w) \right\}.
\end{align*}
Therefore, 
$\textup{IE} = \beta_2 \cdot \left\{ E(M_i \mid T_i=1) - E(M_i \mid T_i=0) \right\}$.
The conditional $\textup{ToE}$ equals 
\begin{align*}
\textup{ToE}(w)
=& E(Y_i \mid T_i=1, w_i = w )
- E(Y_i \mid T_i=0, w_i = w ) \\
=& \sum_{m_i \in \mathcal{M}} \left\{ (\beta_1 + \beta_2 m_i  )\cdot
\textup{Pr}(M_i = m_i \mid T_i=1, w_i=w) - \beta_2 m_i \cdot \textup{Pr}(M_i = m_i \mid T_i=0, w_i=w) 
\right\}  \\
=& \beta_1 + \beta_2 \cdot \left\{ E(M_i \mid T_i=1, w_i=w) - E(M_i \mid T_i=0, w_i=w) \right\}.
\end{align*}
Therefore, $\textup{ToE} = \beta_1 + \beta_2 \cdot \left\{ E(M_i \mid T_i=1) - E(M_i \mid T_i=0) \right\}$.
The conditional SE equals 
\begin{align*}
& \textup{SE}(t_i;w)
= 
\left[ 
\begin{array}{c}
E\left\{ E\left(Y_i \mid M_i, T_i = t_i, w_i=w \right) \mid T_i = t_i, T_{-i} = 1_{n-1}, w_i=w \right\} \\
- E\left\{ E\left(Y_i \mid M_i, T_i = t_i, w_i=w \right) \mid T_i = t_i, T_{-i} = 0_{n-1}, w_i=w \right\} 
\end{array}
\right] \\
&= 
\beta_2 
\sum_{m_i \in \mathcal{M}} m_i
\left[
\textup{Pr}(M_i = m_i \mid T_i=t_i, T_{-i} = 1_{n-1},w_i=w) - \textup{Pr}(M_i = m_i \mid T_i=t_i, T_{-i} = 0_{n-1},w_i=w) 
\right] \\
&= \beta_2 \cdot \left\{ E(M_i \mid T_i=t_i, T_{-i} = 1_{n-1},w_i=w) - E(M_i \mid T_i=t, T_{-i} = 0_{n-1},w_i=w) \right\}.
\end{align*}
Therefore, 
$\textup{SE}(t_i) = \beta_2 \cdot \left\{ E(M_i \mid T_i=t_i, T_{-i} = 1_{n-1}) - E(M_i \mid T_i=t_i, T_{-i} = 0_{n-1}) \right\}$.

\end{proof}

\section{Proof of results in Section \ref{sec:ols}}
\label{app:OLS}

\fontsize{11}{13.2}

Denote by $x_i = \frac{1}{n-1} \sum_{j\ne i} A_{ij}^{\text{post}} T_j$ and $y_i = \frac{1}{n-1} \sum_{j\ne i} A_{ij}^{\text{post}}$.
With the second-order Taylor expansion of $M_i$ at the conditional mean of $(x_i, y_i)$:
\[
\theta_i = 
\left(E(A_{ij}^{\text{post}} T_j \mid T_i, w_i), E(A_{ij}^{\text{post}} \mid T_i, w_i) \right), 
\]
we have $M_i = \xi_i + r_{0,i} + r_{1,i}$ 
where
\begin{align}
\xi_i 
= \frac{E(A_{ij}^{\text{post}} T_j \mid T_i, w_i)}{E(A_{ij}^{\text{post}} \mid T_i, w_i)} 
\label{eq:gi}
\end{align}
and
\begin{align*}
r_{0,i}
=& \frac{E(A_{ij}^{\text{post}} T_j \mid T_i, w_i)}{E(A_{ij}^{\text{post}} \mid T_i, w_i)} 
\frac{1}{n-1} \sum_{j\ne i} \left( \frac{ A_{ij}^{\text{post}}T_j }{E(A_{ij}^{\text{post}} T_j\mid T_i, w_i)} - \frac{A_{ij}^{\text{post}}}{E(A_{ij}^{\text{post}} \mid T_i, w_i)}  \right), \\
r_{1,i} 
=& - \frac{1}{\theta_i^*(y)^2} \frac{1}{(n-1)^2} 
\left( \sum_{j\ne i} \left(A_{ij}^{\text{post}} T_j - E(A_{ij}^{\text{post}} T_j \mid T_i, w_i) \right) 
\sum_{j\ne i} \left(A_{ij}^{\text{post}} - E(A_{ij}^{\text{post}} \mid T_i, w_i) \right)\right) \notag \\
&+ \frac{2\theta_i^*(x)}{\theta_i^*(y)^3} \frac{1}{(n-1)^2} \left(  \sum_{j\ne i} \left(A_{ij}^{\text{post}} - E(A_{ij}^{\text{post}} \mid T_i, w_i) \right) \right)^2 
\end{align*}
for some $\theta_i^* = (\theta_i^*(x), \theta_i^*(y))$ between $\theta_i$ and $(x_i, y_i)$. 
Note that $\xi_i$ is either a function of $T_i$ and $w_i$, making it i.i.d. across units and independent of $n$, or it equals $\pi$ when $A^{\text{post}}$ is conditionally independent of $T$.
To simplify notations, define 
\begin{align*}
R_{ij}
~=~& \frac{ A_{ij}^{\text{post}} T_j }{E(A_{ij}^{\text{post}} T_j\mid T_i, w_i)} - \frac{A_{ij}^{\text{post}}}{E(A_{ij}^{\text{post}} \mid T_i, w_i)}, \\
U_{ij} 
~=~& A_{ij}^{\text{post}} T_j - E(A_{ij}^{\text{post}} T_j \mid T_i, w_i), \\
W_{ij} 
~=~& A_{ij}^{\text{post}} - E(A_{ij}^{\text{post}} \mid T_i, w_i).
\end{align*}
By definition, we have 
\begin{align*}
E(R_{ij} \mid T_i, w_i) = 0, 
E(U_{ij} \mid T_i, w_i) = 0, 
\text{ and } E(W_{ij} \mid T_i, w_i) = 0. 
\end{align*}
Therefore, we can rewrite the remainder terms as 
\begin{align*}
r_{0,i}
~=~& \xi_i \frac{1}{n-1} \sum_{k\ne i} R_{ik}, \\
r_{1,i} 
~=~& \frac{1}{(n-1)^2} \sum_{k\ne i} \left( - \frac{U_{ik} W_{ik}}{\theta_i^*(y)^2} + \frac{2\theta_i^*(x)}{\theta_i^*(y)^3}  W_{ik}^2 \right) 
+ \frac{1}{(n-1)^2} \sum_{\substack{(k,l) \\ k\ne l}} \left( - \frac{U_{ik} W_{il}}{\theta_i^*(y)^2} + \frac{2\theta_i^*(x)}{\theta_i^*(y)^3}  W_{ik} W_{il} \right).
\end{align*}
Define $\mu_{r_1,i} = E(r_{1,i} \mid T_i, w_i)$. 
By direct algebra, we can show that
\begin{align}
\mu_{r_1,i} = \frac{1}{n-1} E\left( - \frac{U_{ik} W_{ik}}{\theta_i^*(y)^2} + \frac{2\theta_i^*(x)}{\theta_i^*(y)^3}  W_{ik}^2 \mid  T_i, w_i \right)
= O_{\mathbb{P}}\left( \frac{1}{n q_n^{\text{post}}} \right).
\label{eq:mu_r1i}
\end{align}
Then we can decompose $M_i$ as $M_i = \xi_i + \mu_{r_1,i} + r_{0,i} + (r_{1,i} - \mu_{r_1,i})$, where the last two terms are mean zero.

\subsection{Auxiliary Lemmas}

\begin{lemma} \label{lemma:consistencyM}
Under Assumptions \ref{asu:network}, we have 
\begin{align}
\frac{1}{n} \sum_{i=1}^n M_i 
=& \frac{1}{n} \sum_{i=1}^n (\xi_i + \mu_{r_1,i}) 
+ O_{\mathbb{P}}\left( \frac{1}{n \sqrt{q_n^{\text{post}}}} \right), 
\label{eq:M_i_con} \\
\frac{1}{n} \sum_{i=1}^n T_i M_i 
=& \frac{1}{n} \sum_{i=1}^n T_i (\xi_i + \mu_{r_1,i}) 
+ O_{\mathbb{P}}\left( \frac{1}{n \sqrt{q_n^{\text{post}}}} \right), 
\label{eq:T_iM_i_con} \\
\frac{1}{n} \sum_{i=1}^n M_i^2 
=& \frac{1}{n} \sum_{i=1}^n (\xi_i + \mu_{r_1,i})^2 
+ E\left((r_{0,i} + r_{1,i} - \mu_{r_1,i})^2 \right) 
+ O_{\mathbb{P}}\left( \frac{1}{n \sqrt{q_n^{\text{post}}}} \right).
\label{eq:M_i^2_con} 
\end{align}
\end{lemma}
\begin{proof}[Proof of Lemma \ref{lemma:consistencyM}]
See proof in Section \ref{proof:consistencyM}.
\end{proof}

\begin{lemma}\label{lemma:inverselimit}
If $\Delta = O_{\mathbb{P}}(\mu)$ with $\mu = o(1)$, and $\Lambda$ converges in probability to a finite and invertible matrix, then $(\Lambda + \Delta)^{-1} - \Lambda^{-1} = O_{\mathbb{P}}(\mu)$.
\end{lemma}

\begin{proof}[Proof of Lemma \ref{lemma:inverselimit}]
This is Lemma A5 in \cite{SuDing2021}, which is useful for deriving the probability limit of the inverse of a matrix.
With $\Delta = o_{\mathbb{P}}(1)$, we have $\Lambda + \Delta = \Lambda + o_{\mathbb{P}}(1)$. Because $\Lambda$ converges in probability to a finite and invertible matrix, $(\Lambda + \Delta)^{-1} = \Lambda^{-1} + o_{\mathbb{P}}(1)$. \citet[Lemma A10]{LiDing2020} stated that
\[
(\Lambda + \Delta)^{-1} - \Lambda^{-1} = \Lambda^{-1} \Delta (\Lambda + \Delta)^{-1} \Lambda^{-1} - \Lambda^{-1} \Delta \Lambda^{-1},
\]
which implies
\[
(\Lambda + \Delta)^{-1} - \Lambda^{-1} 
= O_{\mathbb{P}}(1) O_{\mathbb{P}}(\mu) O_{\mathbb{P}}(1) = O_{\mathbb{P}}(\mu).
\]        
\end{proof}

\begin{lemma}\label{lemma:airi_ri^2}
Define $a_i$ as an i.i.d. random variable with nonzero mean and constant variance.
Under Assumption \ref{asu:network}, then
\begin{align}
\frac{1}{n}\sum_{i=1}^n a_i (r_{0,i} + r_{1,i} - \mu_{r_1,i})
~=~& O_{\mathbb{P}}\left( \frac{1}{n \sqrt{q_n^{\text{post}}}} \right), 
\label{eq:airi} \\
\frac{1}{n} \sum_{i=1}^n (r_{0,i} + r_{1,i} - \mu_{r_1,i})^2 
~=~& E\left((r_{0,i} + r_{1,i} - \mu_{r_1,i})^2 \right) 
+ O_{\mathbb{P}}\left( \frac{1}{\sqrt{n}  n q_n^{\text{post}}} \right)
+ O_{\mathbb{P}}\left( \frac{1}{n} \right).
\label{eq:ri^2}
\end{align}    
\end{lemma}
\begin{proof}[Proof of Lemma \ref{lemma:airi_ri^2}]
See proof in Section \ref{proof:airi_ri^2}.
\end{proof}

\subsection{Consistency}
To establish the consistency of the OLS estimates $\hat{\beta}^{\textsc{ols}} = ( X^\top X )^{-1} ( X^\top Y)$, we first show the probability limit of $( X^\top X )^{-1}$.

\begin{theorem}\label{thm:XX-1}
Assume Assumptions \ref{asu:network} and \ref{asu:linearY}.
\paragraph{Case (a):}
Define 
\begin{align*}
\Lambda
= \frac{1}{n} \sum_{i=1}^n 
\begin{pmatrix}
1 & T_i & \xi_i + \mu_{r_1,i} \\
T_i & T_i & T_i(\xi_i + \mu_{r_1,i}) \\
\xi_i + \mu_{r_1,i} & T_i(\xi_i + \mu_{r_1,i}) & (\xi_i + \mu_{r_1,i})^2 + E((r_{0,i} + r_{1,i} - \mu_{r_1,i})^2)  
\end{pmatrix}.
\end{align*}
Then 
$X^\top  X = \Lambda + O_{\mathbb{P}}\left( \frac{1}{n \sqrt{q_n^{\text{post}}}} \right) $, and thus $(X^\top  X)^{-1} = E(\Lambda)^{-1} + O_{\mathbb{P}}\left( \frac{1}{ \sqrt{n}} \right)$.

\paragraph{Case (b):} We show that
\begin{align*}
& \left( X^\top  X \right)^{-1} 
= \frac{\begin{pmatrix}
a_{11}^* + O_{\mathbb{P}}\left( \frac{1}{\sqrt{n} } \right) 
& a_{12}^* + O_{\mathbb{P}}\left( \frac{1}{n \sqrt{q_n^{\text{post}}} } \right) 
& a_{13}^* + O_{\mathbb{P}}\left( \frac{1}{\sqrt{n}} \right) \\
a_{12}^* + O_{\mathbb{P}}\left( \frac{1}{n \sqrt{q_n^{\text{post}}} } \right) & a_{22}^* + O_{\mathbb{P}}\left( \frac{1}{ n \sqrt{n}  q_n^{\text{post}} } \right) & O_{\mathbb{P}}\left( \frac{1}{n\sqrt{q_n^{\text{post}}}} \right) \\
a_{13}^* + O_{\mathbb{P}}\left( \frac{1}{\sqrt{n}} \right) & O_{\mathbb{P}}\left( \frac{1}{n\sqrt{q_n^{\text{post}}}} \right) & a_{33}^* + O_{\mathbb{P}}\left( \frac{1}{\sqrt{n}} \right)
\end{pmatrix}}{ a_{33}^* a_{22}^* - (a_{23}^*)^2
+ O_{\mathbb{P}} \left( \frac{1}{n \sqrt{n} q_n^{\text{post}}} \right) }
\end{align*}
where 
\begin{align*}
a_{11}^*
~=~& \pi E\left((\xi + \mu_{r_1,i})^2 \right)
- \left( E\left(T_i (\xi + \mu_{r_1,i})\right) \right)^2
+ \pi E\left((r_{0,i} + r_{1,i} - \mu_{r_1,i})^2 \right) \\
a_{12}^*
~=~& 
- \pi \V(\mu_{r_1,i} ) 
- \pi E\left( (r_{0,i} + r_{1,i} - \mu_{r_1,i})^2 \right) \\
a_{13}^*
~=~& - \pi (1-\pi) E\left(\xi + \mu_{r_1,i}\right)  \\
a_{22}^*
~=~& \V(\mu_{r_1,i}^2)  
+ E\left((r_{0,i} + r_{1,i} - \mu_{r_1,i})^2\right) \\
a_{33}^* 
~=~& \pi(1-\pi).
\end{align*}

\end{theorem}

\begin{proof}[Proof of Theorem \ref{thm:XX-1}]
We prove the results under Cases (a) and (b), respectively.

\paragraph{Case (a):}
In this case, $\xi_i$ is a function of $T_i$ and $w_i$ with constant variance. It suffices to show that $X^\top X$ converges in probability to a finite and invertible matrix. 
As we show in Lemma \ref{lemma:consistencyM}, 
\[
X^\top X = \Lambda + O_{\mathbb{P}}\left( \frac{1}{n \sqrt{q_n^{\text{post}}}} \right), 
\]
where $\Lambda$ converges in probability to a finite and invertible matrix $E(\Lambda)$.
By Lemma \ref{lemma:inverselimit} and CMT, we have 
$(X^\top X)^{-1}
= (E(\Lambda))^{-1} + o_{\mathbb{P}}(1)$.

\paragraph{Case (b):}
In this case, we have constant $\xi_i = \xi$. The probability limit of $X^\top X$, denoted by $E(\Lambda)$, is not invertible, and therefore Lemma \ref{lemma:inverselimit} cannot be applied. 
Thus, we first derive the closed form of $(X^\top X)^{-1}$, then show the probability limit. 
By direct algebra, the closed form of $( X^\top X )^{-1}$ is 
\begin{align}
& \left( X^\top X \right)^{-1} 
= \frac{1}{ \det(X^\top X) } 
\begin{pmatrix}
a_{11} & a_{12} & a_{13} \\
a_{12} & a_{22} & a_{23} \\
a_{13} & a_{23} & a_{33}
\end{pmatrix}
\label{eq:XX-1}
\end{align}  
where 
\begin{align*}
a_{11}
~=~& \left(\frac{1}{n}\sum_{i=1}^n T_i \right) 
\left(\frac{1}{n}\sum_{i=1}^n M_i^2 \right) 
- \left(\frac{1}{n}\sum_{i=1}^n T_iM_i \right)^2 \\
a_{12}
~=~& \left( \frac{1}{n}\sum_{i=1}^n M_i \right) 
\left( \frac{1}{n}\sum_{i=1}^n T_i M_i \right) -
\left( \frac{1}{n}\sum_{i=1}^n T_i \right) 
\left( \frac{1}{n}\sum_{i=1}^n M_i^2 \right) \\
a_{13}
~=~& - \left( \frac{1}{n}\sum_{i=1}^n T_i \right) \left( \frac{1}{n}\sum_{i=1}^n (1-T_i)M_i \right) \\
a_{22}
~=~& \left( \frac{1}{n}\sum_{i=1}^n M_i^2 \right) - \left( \frac{1}{n}\sum_{i=1}^n M_i \right)^2 \\
a_{23}
~=~& \left( \frac{1}{n}\sum_{i=1}^n T_i \right) 
\left( \frac{1}{n}\sum_{i=1}^n M_i \right)
- \left( \frac{1}{n}\sum_{i=1}^n T_iM_i \right) \\
a_{33}
~=~& \left(\frac{1}{n}\sum_{i=1}^n T_i \right) - \left(\frac{1}{n}\sum_{i=1}^n T_i \right)^2
\end{align*}
and 
\begin{equation*}
\det(X^\top X) 
= a_{33} a_{22} - a_{23}^2.
\end{equation*}
We show the probability limits of the terms in \eqref{eq:XX-1} one by one.  \\
\underline{(1) $a_{11}$:}
With constant $\xi$ and by Lemma \ref{lemma:consistencyM}, we expand $a_{11}$ into:
\begin{align*}
a_{11}
=& \left(\frac{1}{n}\sum_{i=1}^n T_i \right) 
\left(\frac{1}{n}\sum_{i=1}^n M_i^2 \right) 
- \left( \frac{1}{n}\sum_{i=1}^n T_iM_i \right)^2 \\
=& \left( \pi + O_{\mathbb{P}}\left( \frac{1}{\sqrt{n} } \right) \right)
\left( E\left((\xi + \mu_{r_1,i})^2 \right)
+ E\left((r_{0,i} + r_{1,i} - \mu_{r_1,i})^2 \right) 
+ O_{\mathbb{P}}\left( \frac{1}{n \sqrt{q_n^{\text{post}}}} \right) \right) \\
&- \left( E(T_i (\xi + \mu_{r_1,i})) 
+ O_{\mathbb{P}}\left( \frac{1}{\sqrt{n}} \right)
+ O_{\mathbb{P}}\left( \frac{1}{n \sqrt{q_n^{\text{post}}}} \right) \right)^2 \\
=& \left( \pi + O_{\mathbb{P}}\left( \frac{1}{\sqrt{n} } \right) \right)
\left( E\left((\xi + \mu_{r_1,i})^2 \right)
+ E\left((r_{0,i} + r_{1,i} - \mu_{r_1,i})^2 \right) 
+ O_{\mathbb{P}}\left( \frac{1}{n \sqrt{q_n^{\text{post}}}} \right) \right) \\
&- \left( E\left(T_i (\xi + \mu_{r_1,i})\right)^2  
+ O_{\mathbb{P}}\left( \frac{1}{\sqrt{n}} \right) \right) \\
=& 
a_{11}^*
+ O_{\mathbb{P}}\left( \frac{1}{\sqrt{n}} \right).
\end{align*}
\underline{(2) $a_{12}$:}
With constant $\xi$ and by Lemma \ref{lemma:consistencyM}, we have
\begin{align}
& a_{12}
= \left( \frac{1}{n}\sum_{i=1}^n M_i \right) 
\left( \frac{1}{n}\sum_{i=1}^n T_i M_i \right) -
\left( \frac{1}{n}\sum_{i=1}^n T_i \right) 
\left( \frac{1}{n}\sum_{i=1}^n M_i^2 \right) \notag \\
=& \xi \left( \frac{1}{n}\sum_{i=1}^n T_i \mu_{r_1,i} - \bar{T} \frac{1}{n}\sum_{i=1}^n \mu_{r_1,i} \right)
+ \frac{1}{n}\sum_{i=1}^n \mu_{r_1,i} \frac{1}{n}\sum_{i=1}^n T_i \mu_{r_1,i} 
- \bar{T} \frac{1}{n}\sum_{i=1}^n \mu_{r_1,i}^2 
- \bar{T} \frac{1}{n}\sum_{i=1}^n (r_{0,i} + r_{1,i} - \mu_{r_1,i})^2 \notag \\
&+ \xi \left( \frac{1}{n}\sum_{i=1}^n T_i (r_{0,i} + r_{1,i} - \mu_{r_1,i}) - \frac{1}{n}\sum_{i=1}^n T_i \frac{1}{n}\sum_{i=1}^n (r_{0,i} + r_{1,i} - \mu_{r_1,i}) \right) \notag \\
&+ \frac{1}{n}\sum_{i=1}^n \mu_{r_1,i} \frac{1}{n}\sum_{i=1}^n T_i (r_{0,i} + r_{1,i} - \mu_{r_1,i})
+ \frac{1}{n}\sum_{i=1}^n (r_{0,i} + r_{1,i} - \mu_{r_1,i})
\frac{1}{n}\sum_{i=1}^n T_i \mu_{r_1,i} \notag \\
&+ \frac{1}{n}\sum_{i=1}^n (r_{0,i} + r_{1,i} - \mu_{r_1,i})
\frac{1}{n}\sum_{i=1}^n T_i (r_{0,i} + r_{1,i} - \mu_{r_1,i})
- 2 \bar{T}\frac{1}{n}\sum_{i=1}^n \mu_{r_1,i} (r_{0,i} + r_{1,i} - \mu_{r_1,i}) \notag \\
=& a_{12}^*
+ O_{\mathbb{P}}\left( \frac{1}{n \sqrt{q_n^{\text{post}}} } \right).
\label{eq:a_12_limit}
\end{align}
\underline{(3) $a_{13}$.}
By Lemma \ref{lemma:consistencyM}, we have 
\begin{align}
a_{13}
=& - \left( \pi + O_{\mathbb{P}}\left( \frac{1}{\sqrt{n} } \right) \right) \left( E\left((1-T_i) (\xi + \mu_{r_1,i})\right) 
+ O_{\mathbb{P}}\left( \frac{1}{\sqrt{n}} \right) \right) 
= a_{13}^* + O_{\mathbb{P}}\left( \frac{1}{\sqrt{n} } \right).
\label{eq:a_{13}_limit}
\end{align}
\underline{(4) $a_{22}$.}
With constant $\xi$ and by Lemma \ref{lemma:consistencyM}, we have
\begin{align}
& a_{22}
= \left( \frac{1}{n}\sum_{i=1}^n M_i^2 \right) - \left( \frac{1}{n}\sum_{i=1}^n M_i \right)^2 \notag \\
&= \frac{1}{n}\sum_{i=1}^n \mu_{r_1,i}^2 
- \left( \frac{1}{n}\sum_{i=1}^n \mu_{r_1,i} \right)^2
+ \frac{1}{n}\sum_{i=1}^n (r_{0,i} + r_{1,i} - \mu_{r_1,i})^2 \notag \\
&+ \frac{2}{n}\sum_{i=1}^n \mu_{r_1,i}(r_{0,i} + r_{1,i} - \mu_{r_1,i}) 
- \frac{2}{n}\sum_{i=1}^n \mu_{r_1,i}
\frac{1}{n}\sum_{i=1}^n (r_{0,i} + r_{1,i} - \mu_{r_1,i})
- \left( \frac{1}{n}\sum_{i=1}^n (r_{0,i} + r_{1,i} - \mu_{r_1,i}) \right)^2 \notag \\
&= \frac{1}{n}\sum_{i=1}^n \mu_{r_1,i}^2 
- \left( \frac{1}{n}\sum_{i=1}^n \mu_{r_1,i} \right)^2
+ \frac{1}{n}\sum_{i=1}^n (r_{0,i} + r_{1,i} - \mu_{r_1,i})^2 
+ O_{\mathbb{P}}\left( \frac{1}{n^2 q_n^{\text{post}} \sqrt{q_n^{\text{post}}}} \right) \notag \\
&= a_{22}^* + O_{\mathbb{P}}\left( \frac{1}{\sqrt{n} n q_n^{\text{post}}} \right). 
\label{eq:a_22_limit}
\end{align}
\underline{(5) $a_{23}$.}
With constant $\xi$ and by Lemma \ref{lemma:consistencyM}, we have
\begin{align}
a_{23}
=& \left( \frac{1}{n}\sum_{i=1}^n T_i \right) 
\left( \frac{1}{n}\sum_{i=1}^n M_i \right)
- \left( \frac{1}{n}\sum_{i=1}^n T_iM_i \right) \notag \\
=& \left( \frac{1}{n}\sum_{i=1}^n T_i \right) 
\left( \frac{1}{n}\sum_{i=1}^n \mu_{r_1,i} \right)
- \left( \frac{1}{n}\sum_{i=1}^n T_i \mu_{r_1,i} \right) \notag \\
&+ \left( \frac{1}{n}\sum_{i=1}^n T_i \right) 
\left( \frac{1}{n}\sum_{i=1}^n (r_{0,i} + r_{1,i} - \mu_{r_1,i}) \right)
- \left( \frac{1}{n}\sum_{i=1}^n T_i (r_{0,i} + r_{1,i} - \mu_{r_1,i}) \right) \notag \\
=& O_{\mathbb{P}}\left( \frac{1}{n\sqrt{q_n^{\text{post}}}} \right).
\label{eq:a_23_limit}
\end{align}
\underline{(6) $a_{33}$.}
By the i.i.d. treatment and CLT, we have $a_{33} = a_{33}^* + O_{\mathbb{P}}\left( \frac{1}{\sqrt{n}} \right)$. \\
\underline{(7) $\det(X^\top X)$.}
By combining \eqref{eq:a_22_limit} and \eqref{eq:a_23_limit}, the denominator of \eqref{eq:XX-1} is 
\begin{align*}
& \det(X^\top X) 
= a_{33} a_{22} - a_{23}^2 
= a_{33}^* a_{22}^*
+ O_{\mathbb{P}} \left( \frac{1}{\sqrt{n} n q_n^{\text{post}}} \right) 
\text{ with }
a_{33}^* a_{22}^*
\asymp \frac{1}{n q_n^{\text{post}}}. 
\end{align*}

\end{proof}

\begin{proof}[Proof of Theorem \ref{thm:consistency_ratio}]
We prove the results under Cases (a) and (b), respectively.
\paragraph{Case (a):} 
By Theorem \ref{thm:XX-1}, $(X^\top  X)^{-1} = (E(\Lambda))^{-1} + o_{\mathbb{P}}(1)$, which is $O_{\mathbb{P}}(1)$. 
By Lemma \ref{lemma:airi_ri^2},
\begin{align*}
\frac{1}{n}\sum_{i=1}^n M_i u_i 
= \frac{1}{n}\sum_{i=1}^n (\xi_i + \mu_{r_1,i})u_i 
+ \frac{1}{n}\sum_{i=1}^n (r_{0,i} + r_{1,i} - \mu_{r_1,i})u_i 
= O_{\mathbb{P}}\left( \frac{1}{\sqrt{n}} \right).
\end{align*}
With i.i.d. data, we have 
\[
\frac{1}{n}\sum_{i=1}^n u_i
= O_{\mathbb{P}}\left( \frac{1}{\sqrt{n}} \right)
\text{ and } 
\frac{1}{n}\sum_{i=1}^n T_i u_i
= O_{\mathbb{P}}\left( \frac{1}{\sqrt{n}} \right).
\]
Together with Theorem \ref{thm:XX-1} and CMT, this completes the proof that $\hat{\beta}^{\textsc{ols}} - \beta = O_{\mathbb{P}}\left( \frac{1}{\sqrt{n}} \right)$.

\paragraph{Case (b):}
We prove by the closed form of $\hat{\beta}^{\textsc{ols}} - \beta$. By direct algebra, we have
\begin{align*}
\hat{\beta}^{\textsc{ols}}_1 - {\beta}_1
~=~& \frac{ a_{22} \frac{1}{n}\sum_{i=1}^n (T_i - \bar{T}) u_i 
+ a_{23} \frac{1}{n}\sum_{i=1}^n (M_i - \bar{M}) u_i}{\det(X^\top X)}, \\
\hat{\beta}^{\textsc{ols}}_2 - {\beta}_2
~=~& \frac{ a_{32} \frac{1}{n}\sum_{i=1}^n (T_i - \bar{T}) u_i 
+ a_{33} \frac{1}{n}\sum_{i=1}^n (M_i - \bar{M}) u_i}{\det(X^\top X)}, \\
\hat{\beta}^{\textsc{ols}}_0 - {\beta}_0
~=~& \bar{u} - (\hat{\beta}^{\textsc{ols}}_1 - \beta_1) \bar{T} - (\hat{\beta}^{\textsc{ols}}_2 - \beta_2) \bar{M}.
\end{align*}
With i.i.d. data, we have 
\[
\frac{1}{n}\sum_{i=1}^n (T_i - \bar{T}) u_i
= O_{\mathbb{P}}\left( \frac{1}{\sqrt{n}} \right).
\]
It remains to show $\frac{1}{n}\sum_{i=1}^n (M_i - \bar{M}) u_i$. 
By \eqref{eq:mu_r1i} and Lemma \ref{lemma:airi_ri^2}, we can show that
\begin{align*}
\frac{1}{n}\sum_{i=1}^n (M_i - \bar{M}) u_i 
=& \frac{1}{n}\sum_{i=1}^n \mu_{r_1,i} u_i
- \frac{1}{n}\sum_{i=1}^n  \mu_{r_1,i} \frac{1}{n}\sum_{i=1}^n u_i \\
&+ \frac{1}{n}\sum_{i=1}^n (r_{0,i} + r_{1,i} - \mu_{r_1,i}) u_i
- \frac{1}{n}\sum_{i=1}^n (r_{0,i} + r_{1,i} - \mu_{r_1,i}) \frac{1}{n}\sum_{i=1}^n u_i \\
=& O_{\mathbb{P}}\left( \frac{1}{n \sqrt{q_n^{\text{post}}}} \right).
\end{align*}
By combining these results and Theorem \ref{thm:XX-1}, we have
\begin{align*}
\hat{\beta}^{\textsc{ols}}_1 - {\beta}_1
=& \frac{ a_{22} \frac{1}{n}\sum_{i=1}^n (T_i - \bar{T}) u_i 
+ a_{23} \frac{1}{n}\sum_{i=1}^n (M_i - \bar{M}) u_i}{\det(X^\top X)} \\
=& \frac{ \left( a_{22}^* + O_{\mathbb{P}}\left( \frac{1}{\sqrt{n} n q_n^{\text{post}}} \right) \right) O_{\mathbb{P}}\left(\frac{1}{\sqrt{n}}\right) +  O_{\mathbb{P}}\left( \frac{1}{n\sqrt{q_n^{\text{post}}}} \right) O_{\mathbb{P}}\left(\frac{1}{n\sqrt{q_n^{\text{post}}}} \right) }{ \pi(1-\pi) a_{22}^*
+ O_{\mathbb{P}} \left( \frac{1}{\sqrt{n} n q_n^{\text{post}}} \right) } \\
=& O_{\mathbb{P}}\left( \frac{1}{\sqrt{n}} \right)
\end{align*}
and 
\begin{align*}
\hat{\beta}^{\textsc{ols}}_2 - {\beta}_2
=& \frac{ a_{23} \frac{1}{n}\sum_{i=1}^n (T_i - \bar{T}) u_i 
+ a_{33} \frac{1}{n}\sum_{i=1}^n (M_i - \bar{M}) u_i}{\det(X^\top X)} \\
=& \frac{ O_{\mathbb{P}}\left( \frac{1}{n\sqrt{q_n^{\text{post}}}} \right) O_{\mathbb{P}}\left(\frac{1}{\sqrt{n}}\right) + \left( \pi(1-\pi) + O_{\mathbb{P}}\left(\frac{1}{\sqrt{n}} \right) \right) O_{\mathbb{P}}\left(\frac{1}{n\sqrt{q_n^{\text{post}}}}\right)  }{ \pi(1-\pi) a_{22}^* 
+ O_{\mathbb{P}} \left( \frac{1}{\sqrt{n} n q_n^{\text{post}}} \right) } \\
=& O_{\mathbb{P}}\left(\sqrt{q_n^{\text{post}}}\right).
\end{align*}
Therefore, it follows that
\begin{align*}
\hat{\beta}^{\textsc{ols}}_0 - {\beta}_0
=& \bar{u} - (\hat{\beta}^{\textsc{ols}}_1 - \beta_1) \bar{T} - (\hat{\beta}^{\textsc{ols}}_2 - \beta_2) \bar{M} 
= O_{\mathbb{P}}\left(\sqrt{q_n^{\text{post}}}\right).
\end{align*}
Thus, we complete the proof.

\end{proof}

\subsection{Asymptotic normality}
To prove the asymptotic normality of the OLS estimators in Theorem \ref{thm:asymnormal_ratio}, we divide the proof into two cases: when $n q_n^{\text{post}} \succ 1$ and when the network is with bounded degree $n q_n^{\text{post}} \asymp 1$.
For the former case, we show that $\frac{1}{n}\sum_{i=1}^n M_i u_i$ can be approximated by the average of i.i.d. variables and then apply the Lindeberg--Lévy CLT.
For the latter case, we establish the CLT with weak dependence in Theorem \ref{thm:clt}, which makes uses of the proof of Theorem 3.5 in \cite{Ross2011} and is the unconditional version of Theorem 4 in \cite{Leung2020}.

Let $H_1, \ldots, H_n$ be real--valued random variables. 
Define the indicator of dependency as
\begin{align}
D_{ij}
= 1\left\{ A_{ij}^{\diamond} + \max_k A_{ik}^{\diamond} A_{kj}^{\diamond} + 1\{i=j\} > 0 \right\}
\label{eq:G}
\end{align}
where ${\diamond} \in \{\text{pre}, \text{post}\}$. 
In words, two units are dependent if they are directly connected, share at least one common friend, or refer to the same index.
An $n \times n$ binary symmetric matrix $D$ with $(i,j)$th element being $D_{ij}$ is a {dependency graph} on $H$ if for any two disjoint subsets $I_1, I_2 \subseteq \{1, \ldots, n\}$, we have $\{H_i : i \in I_1\} \perp\!\!\!\perp \{H_i : i \in I_2\}$ conditional on the event that $D_{ij} = 0$ for all $i \in I_1$ and $j \in I_2$. By construction, $D_{ii} = 1$ for all $i$.

Define $\mathcal{S}_i = \{j: D_{ij} = 1\}$ and $S_i = |\mathcal{S}_i| = \sum_{j=1}^n D_{ij}$.
Define $\sigma_n^2 = \V\left(\sum_{i=1}^n H_i\right)$.
Let $(D^3)_{ij}$ denote the $(i,j)$th entry of the third matrix power of $D$.
\begin{theorem}[CLT for local dependence]\label{thm:clt}
Under the following assumptions, $\sigma_n^{-1} \sum_{i=1}^{n} H_i \overset{d}{\to} \mathcal{N}(0,1)$ as $n \to \infty$:
\begin{itemize}
\item[(a)] $\sigma_n^2 = O(n)$;
\item[(b)] $\max_i E \left[|H_i|^4 \right] = O(1)$;
\item[(c)] $ n^{-1} \sum_{i=1}^{n} S_i^2 = o_{\mathbb{P}}(\sqrt{n})$, $ n^{-1} \sum_{i=1}^{n} S_i^3 = o_{\mathbb{P}}(n)$, and $ n^{-1} \sum_{i=1}^{n} \sum_{j \neq i} (D^3)_{ij} = o_{\mathbb{P}}(n)$.
\end{itemize}
\end{theorem}
\begin{proof}[Proof of Theorem \ref{thm:clt}]
For any two random variables $U, V$ with respective probability laws $\mu(\cdot)$ and $\nu(\cdot )$,
define their Wasserstein distance as
\[
\Delta(U, V) = \sup_{f \in \mathcal{F}} \left| \int f(x) d\mu(x) - \int f(x) d\nu(x ) \right|,
\]
where $\mathcal{F} = \{f: \mathbb{R} \to \mathbb{R} : |f(x) - f(y)| \leq |x - y|\}$. 
The convergence of $\Delta(U, V )$ implies convergence of the Kolmogorov metric distance between $U$ and $V$ (see e.g. \citet[Proposition 1.2]{Ross2011}), which, in turn, implies weak convergence. 
Thus, the theorem holds if
\[
\Delta\left( \frac{1}{\sigma_n} \sum_{i=1}^{n} H_i, \mathcal{N}(0, 1)  \right) \overset{p}{\to} 0.
\]
The distance is bounded above by 
\begin{align*}
\Delta\left( \frac{1}{\sigma_n} \sum_{i=1}^{n} H_i, \mathcal{N}(0, 1) \right)
& \le \frac{1}{\sigma_n^3} \sum_{i=1}^{n} E \left[ \left|H_i  \left( \sum_{j\in \mathcal{S}_i} H_j \right)^2 \right| \right] 
+ \frac{\sqrt{2}}{\sqrt{\pi} \sigma_n^2} \sqrt{\text{Var} \left( \sum_{i=1}^{n} \sum_{j \in \mathcal{S}_i} H_i H_j \right)} \\
& \le \frac{1}{\sigma_n^3} \sum_{i=1}^{n} \sum_{j, k \in \mathcal{S}_i} E \left[ |H_i H_j H_k| \right] 
+ \frac{\sqrt{2}}{\sqrt{\pi} \sigma_n^2} \sqrt{\text{Var} \left( \sum_{i=1}^{n} \sum_{j \in \mathcal{S}_i} H_i H_j \right)} \\
& \le \max_i E \left[ |H_i|^3 \right] \frac{n}{\sigma_n^3} \frac{1}{n} \sum_{i=1}^{n} S_i^2 \\
& + \frac{1}{\sqrt{n}} \left( \frac{\sigma_n^2}{n} \right)^{-1} \left( \frac{4}{\pi} \max_i E \left[ H_i^4 \right] \right)^{1/2} \left( \frac{4}{n} \sum_{i=1}^{n} S_i^3 + \frac{3}{n} \sum_{i=1}^{n} \sum_{j=1} (G^3)_{ij} \right)^{1/2},
\end{align*}
which is $o_{\mathbb{P}}(1)$ under Assumptions (a)-(c).

\end{proof}

We verify Assumption (c) in Theorem \ref{thm:clt} when $n q_n^{\diamond} \asymp 1$ for ${\diamond} \in \{\text{pre}, \text{post} \}$. 

\begin{lemma}\label{lemma:Assu_B2}
Suppose $n q_n^{\diamond} \asymp 1$ for ${\diamond} \in \{\text{pre}, \text{post} \}$. Then
\begin{itemize}
\item[] (i) $\frac{1}{n} \sum_{i=1}^{n} S_i^2 = o_{\mathbb{P}}(\sqrt{n})$, (ii) $\frac{1}{n} \sum_{i=1}^{n} S_i^3 = o_{\mathbb{P}}(n)$, and (iii) $\frac{1}{n} \sum_{i=1}^{n} \sum_{j \neq i} (D^3)_{ij} = o_{\mathbb{P}}(n)$.
\end{itemize}

\begin{proof}[Proof of Lemma \ref{lemma:Assu_B2}]
\underline{For (i)}, we bound it in $L_1$ norm:
\begin{align*}
E\left( \frac{1}{n}\sum_{i=1}^n S_i^2 \right)
\le & E\left[ \frac{1}{n}\sum_{i=1}^n \left( \sum_{j=1}^n A_{ij}^{\diamond} + \sum_{j=1}^n \sum_{k \ne j} A_{ik}^{\diamond} A_{kj}^{\diamond} \right)^2 \right] \\
\le & 2 E\left[ \frac{1}{n}\sum_{i=1}^n \left( \sum_{j=1}^n A_{ij}^{\diamond} \right)^2 \right]
+ 2 E\left[ \frac{1}{n}\sum_{i=1}^n \left( \sum_{j=1}^n \sum_{k \ne j} A_{ik}^{\diamond} A_{kj}^{\diamond} \right)^2 \right]
\end{align*}
where 
\begin{align*}
E\left[ \frac{1}{n}\sum_{i=1}^n \left( \sum_{j=1}^n A_{ij}^{\diamond} \right)^2 \right]
= \frac{1}{n}\sum_{i=1}^n 
E \left( \sum_{j=1}^n A_{ij}^{\diamond} + \sum_{j=1}^n \sum_{k\ne j} A_{ij}^{\diamond} A_{kj}^{\diamond} \right)
\le C_1 n q_n^{\diamond} + C_2 n^2 (q_n^{\diamond})^2 \le C
\end{align*}
and by the AM-GM inequality, the second term is bounded above by
\begin{align*}
& E\left[ \frac{1}{n}\sum_{i=1}^n \left( \sum_{j=1}^n \sum_{k\ne j} A_{ik}^{\diamond} A_{kj}^{\diamond} \right)^2 \right]
= \frac{1}{n}\sum_{i=1}^n
E\left[ 
\sum_{(j_1,k_1)} A_{ik_1}^{\diamond} A_{k_1j_1}^{\diamond} 
\sum_{(j_2,k_2)} A_{ik_2}^{\diamond} A_{k_2j_2}^{\diamond}  \right]
\le C n^4 (q_n^{\diamond})^4 \le C.
\end{align*}
\underline{For (ii)}, we bound it in $L_1$ norm:
\begin{align*}
E\left( \frac{1}{n}\sum_{i=1}^n S_i^3 \right)
\le & E\left[ \frac{1}{n}\sum_{i=1}^n \left( \sum_{j=1}^n A_{ij}^{\diamond}  + \sum_{\substack{(j,k)}} A_{ik}^{\diamond} A_{kj}^{\diamond} \right)^3  \right] \\
\le & 4 E\left[ \frac{1}{n}\sum_{i=1}^n \left( \sum_{j=1}^n A_{ij}^{\diamond}  \right)^3 \right]
+ 4 E\left[ \frac{1}{n}\sum_{i=1}^n \left( \sum_{\substack{(j,k)}} A_{ik}^{\diamond} A_{kj}^{\diamond} \right)^3 \right].
\end{align*}
By analogous argument to showing (i), we have for the first term
\begin{align*}
E\left[ \left( \sum_{j=1}^n A_{ij}^{\diamond} \right)^3 \right]
= \frac{1}{n}\sum_{i=1}^n 
E\left[ \sum_{j=1}^n A_{ij}^{\diamond}  
+ \sum_{\substack{(j_1,j_2)\\ j_1\ne j_2} } 
A_{ij_1}^{\diamond} A_{ij_2}^{\diamond}  
+ \sum_{\substack{(j_1,j_2,j_3)\\ \text{all distinct}}} A_{ij_1}^{\diamond} A_{ij_2}^{\diamond} A_{ij_3}^{\diamond} \right]
\le C
\end{align*}
and for the second term, 
\begin{align*}
& \frac{1}{n}\sum_{i=1}^n E\left[ \left( \sum_{\substack{(j,k) \\ j \ne k}} A_{ik}^{\diamond} A_{kj}^{\diamond} \right)^3 \right] 
= \frac{1}{n}\sum_{i=1}^n E\left[ \sum_{\substack{(j_1,k_1)\\ j_1\ne k_1}} A_{ik_1}^{\diamond} A_{k_1j_1} \sum_{\substack{(j_2,k_2)\\ j_2\ne k_2}} A_{ik_2}^{\diamond} A_{k_2j_2}^{\diamond} \sum_{\substack{(j_3,k_3)\\ j_3\ne k_3}} A_{ik_3}^{\diamond} A_{k_3j_3}^{\diamond} \right]
\le C.
\end{align*}
\underline{For (iii)}, by definition, we have
\begin{align*}
& (D^3)_{ij}
= \sum_{k=1}^n (D^2)_{ik} D_{kj}
= \sum_{(k,l)} D_{il} D_{lk} D_{kj} 
\end{align*}
where for $i\ne l$, we have $D_{il} \le A_{il}^{\diamond} +  \max_{h} A_{ih}^{\diamond} A_{hl}^{\diamond} $.
Then 
\begin{align*}
& \frac{1}{n}\sum_{i=1}^n \sum_{j\ne i} (D^3)_{ij}
= \frac{1}{n}\sum_{i=1}^n \sum_{j\ne i} 
\sum_{(k,l)} D_{il} D_{lk}  D_{kj} \\
\le & \frac{1}{n}\sum_{i=1}^n \sum_{j\ne i} \sum_{(k,l)}
\left( A_{il}^{\diamond} + \max_{h} A_{ih}^{\diamond} A_{hl}^{\diamond} \right)
\left( A_{lk}^{\diamond} + \max_{h} A_{lh}^{\diamond} A_{hk}^{\diamond} \right)
\left( A_{kj}^{\diamond} + \max_{h} A_{kh}^{\diamond} A_{hj}^{\diamond} \right).
\end{align*}
Expanding the brackets involves four cases, depending on whether the individuals are directly connected or connected through a common friend. We only show one case as other cases follow from analogous arguments.
For the case where they are all connected through a common friend, we have 
\begin{align*}
& \frac{1}{n}\sum_{i=1}^n \sum_{j\ne i} \sum_{(k,l)}
E\left[  
\left( \max_{h} A_{ih}^{\diamond} A_{hl}^{\diamond} \right)
\left( \max_{h} A_{lh}^{\diamond} A_{hk}^{\diamond} \right)
\left( \max_{h} A_{kh}^{\diamond} A_{hj}^{\diamond} \right)
\right] \\
\le & \frac{1}{n}\sum_{i=1}^n \sum_{j\ne i} \sum_{(k,l)}
\sum_{(h_1,h_2,h_3)}
E\left[  
\left( A_{ih_1}^{\diamond} A_{h_1l}^{\diamond} \right)
\left( A_{lh_2}^{\diamond} A_{h_2k}^{\diamond} \right)
\left( A_{kh_3}^{\diamond} A_{h_3j}^{\diamond} \right)
\right]
\le C n^6 (q_n^{\diamond})^6 \le C. 
\end{align*}
By combining these four cases, we can show that $\frac{1}{n}\sum_{i=1}^n \sum_{j\ne i} (D^3)_{ij} = o_{\mathbb{P}}(\sqrt{n})$.

\end{proof}

\end{lemma}

\begin{proof}[Proof of Theorem \ref{thm:asymnormal_ratio}]
    
We complete the proof in three steps.
\paragraph*{Step 1: asymptotic normality.}
To obtain our desired result, we first show that $ \sum_{i=1}^n X_i u_i$, the numerator of $\hat{\beta}^{\textsc{ols}} - \beta$, converges to a normal distribution with asymptotic variance 
\begin{align}
V_{\text{num}}^{\textsc{ols}} 
\equiv
\V\left( \sum_{i=1}^n X_i u_i \right)
= 
\begin{pmatrix}
\sum_{i=1}^nE(u_i^2) & \pi \sum_{i=1}^nE(u_i^2) & \sum_{i=1}^n E(M_i u_i^2)  \\
\pi \sum_{i=1}^n E(u_i^2) & \pi \sum_{i=1}^n E(u_i^2) & \sum_{i=1}^n E(T_i M_i u_i^2)  \\
\sum_{i=1}^n E(M_i u_i^2)  & \sum_{i=1}^n E(T_i M_i u_i^2)  & \sum_{i=1}^n E(M_i^2 u_i^2) 
\end{pmatrix}. 
\label{eq:asyV_OLS}
\end{align}
We divide the proof into two cases. \\
\underline{Case 1: $n q_n^{\text{post}} \succ 1$.} 
By Taylor expansion and Lemma \ref{lemma:airi_ri^2}, 
\begin{align}
\frac{1}{n}\sum_{i=1}^n M_i u_i
&= \frac{1}{n}\sum_{i=1}^n (\xi_i + \mu_{r_1,i}) u_i
+ \frac{1}{n}\sum_{i=1}^n (r_{0,i} + r_{1,i} - \mu_{r_1,i}) u_i \notag \\
&= \frac{1}{n}\sum_{i=1}^n \xi_i u_i
+ O_{\mathbb{P}}\left(\frac{1}{\sqrt{n}n q_n^{\text{post}}} \right)
+ O_{\mathbb{P}}\left(\frac{1}{\sqrt{n} \sqrt{n q_n^{\text{post}}} } \right) \notag \\
&= \frac{1}{n} \sum_{i=1}^n \xi_i u_i + o_{\mathbb{P}}\left(\frac{1}{\sqrt{n}} \right). 
\label{eq:tilde_b}
\end{align}
It implies that $\frac{1}{n}\sum_{i=1}^n M_i u_i$ can be approximated by $\frac{1}{n}\sum_{i=1}^n \xi_i u_i$, which is the average of i.i.d. random variables, with a small noise $o_{\mathbb{P}}\left(\frac{1}{\sqrt{n}} \right)$. 
Define $\tilde{b}
= \frac{1}{\sqrt{n}} \sum_{i=1}^n (u_i,T_i u_i,\xi_i u_i)^\top$.
By the Cramer--Wold theorem and Lindeberg–Lévy CLT, we can show that
\begin{equation*}
\V(\tilde{b})^{-1/2}
\tilde{b}
~\stackrel{\textup{d}}{\rightarrow}~ 
\mathcal{N}(0, I_3).
\end{equation*}
With $n q_n^{\text{post}} \succ 1$, we have 
$E[ (r_{0,i} + r_{1,i} - \mu_{r_1,i})^2 ] = o(1)$, and thus 
\begin{align*}
\V( \tilde{b})
=& \begin{pmatrix}
E(u_i^2) & \pi E(u_i^2) & E(\xi_i) E(u_i^2) \\
\pi E(u_i^2) & \pi E(u_i^2) & E(T_i \xi_i) E(u_i^2) \\
E(\xi_i) E(u_i^2)  & E(T_i \xi_i) E(u_i^2) & E( \xi_i^2) E(u_i^2)
\end{pmatrix} 
= \frac{1}{n} V_{\text{num}}^{\textsc{ols}}  + o(1).
\end{align*}
Together with \eqref{eq:tilde_b}, the CLT follows: 
\begin{align*}
\left( V_{\text{num}}^{\textsc{ols}} \right)^{-1/2}
\left( \sum_{i=1}^n X_i u_i \right)
~\stackrel{\textup{d}}{\rightarrow} ~
\mathcal{N}(0, I_3).
\end{align*}
\underline{Case 2: $n q_n^{\text{post}} \asymp 1$.}
It suffices to verify the assumptions in Theorem \ref{thm:clt}.
Assumption (a) follows from $V_{\text{num}}^{\textsc{ols}}=O(1)$ and (b) follows from Assumption \ref{asu:linearY} and boundedness of $T_i$ and $M_i$.
Assumption (c) follows from applying Lemma \ref{lemma:Assu_B2} to $A^{\text{post}}$. 
Then by Theorem \ref{thm:clt} and Cramer--Wold device, we can show 
\begin{align*}
\left( V_{\text{num}}^{\textsc{ols}} \right)^{-1/2} 
\left( \sum_{i=1}^n X_i u_i \right)
~\stackrel{\textup{d}}{\rightarrow} ~
\mathcal{N}\left(0, I_3\right).
\end{align*}

\paragraph*{Step 2: consistent variance estimator.} 
We show that $\frac{1}{n}(\hat{V}_{\text{num}}^{\textsc{ols}} - V_{\text{num}}^{\textsc{ols}}) = o_{\mathbb{P}}(1)$.
We only show the convergence of $\frac{1}{n} \sum_{i=1}^n M_i^2 (\hat{u}_i^\textsc{ols})^2$ to $E(M_i^2 u_i^2)$ since the remaining terms follow from similar arguments and we omit them for brevity.
Notice that $\hat{u}_i^\textsc{ols} = u_i - (\hat{\beta}^{\textsc{ols}} - \beta )^\top X_i$ and $\hat{\beta}^{\textsc{ols}} - \beta = o_{\mathbb{P}}(1)$. 
By expanding the square, we have 
\begin{align}
\frac{1}{n} \sum_{i=1}^n M_i^2 (\hat{u}_i^\textsc{ols})^2
=& \frac{1}{n} \sum_{i=1}^n M_i^2 u_i^2 
- \frac{2}{n} \sum_{i=1}^n M_i^2 u_i (\hat{\beta}^{\textsc{ols}} - \beta )^\top X_i
+ \frac{1}{n} \sum_{i=1}^n M_i^2 \left((\hat{\beta}^{\textsc{ols}} - \beta )^\top X_i\right)^2. 
\label{eq:M_i^2hat{e}_i^2}
\end{align}
For the first term of \eqref{eq:M_i^2hat{e}_i^2}, we have 
\begin{align*}
\frac{1}{n} \sum_{i=1}^n M_i^2 u_i^2
=& \frac{1}{n} \sum_{i=1}^n (\xi_i + \mu_{r_1,i})^2  u_i^2
+ \frac{2}{n} \sum_{i=1}^n (\xi_i + \mu_{r_1,i})(r_{0,i} + r_{1,i} - \mu_{r_1,i}) u_i^2
+ \frac{1}{n} \sum_{i=1}^n (r_{0,i} + r_{1,i} - \mu_{r_1,i})^2 u_i^2 \\
=& E\left( (\xi_i + \mu_{r_1,i})^2 u_i^2 \right)
+ E\left((r_{0,i} + r_{1,i} - \mu_{r_1,i})^2 u_i^2 \right) 
+ O_{\mathbb{P}}\left( \frac{1}{\sqrt{n}} \right)
+ O_{\mathbb{P}}\left( \frac{1}{n \sqrt{q_n^{\text{post}}}} \right) \\
=& E(M_i^2 u_i^2) + o_{\mathbb{P}}(1)
\end{align*}
by Lemma \ref{lemma:airi_ri^2}.
For the second term of \eqref{eq:M_i^2hat{e}_i^2}, we have 
\begin{align*}
& \frac{1}{n} \sum_{i=1}^n M_i^2 u_i (\hat{\beta}^{\textsc{ols}} - \beta )^\top X_i 
= (\hat{\beta}^{\textsc{ols}}_0 - \beta_0) \frac{1}{n} \sum_{i=1}^n M_i^2 u_i  
+ (\hat{\beta}^{\textsc{ols}}_1 - \beta_1) \frac{1}{n} \sum_{i=1}^n M_i^2 T_i u_i 
+ (\hat{\beta}^{\textsc{ols}}_2 - \beta_2) \frac{1}{n} \sum_{i=1}^n M_i^3 u_i.
\end{align*}
By the bounded support of $M_i$, we can show that 
\begin{align*}
E\left[ \left( \frac{1}{n} \sum_{i=1}^n M_i^2 u_i \right)^2 \right]
= E\left[ \frac{1}{n^2} \sum_{i=1}^n M_i^4 u_i^2 + \frac{1}{n^2} \sum_{i=1}^n \sum_{j\ne i} M_i^2 M_j^2 u_i u_j \right]
\le \frac{C}{n},
\end{align*}
and it implies that $\frac{1}{n} \sum_{i=1}^n M_i^2 u_i = O_{\mathbb{P}}\left( \frac{1}{\sqrt{n}} \right)$.
By applying analogous arguments to the other terms and with $\hat{\beta}^{\textsc{ols}} - \beta  = o_{\mathbb{P}}(1)$, we can show that 
\begin{align*}
\frac{1}{n} \sum_{i=1}^n M_i^2 u_i (\hat{\beta}^{\textsc{ols}} - \beta )^\top X_i = o_{\mathbb{P}}(1).
\end{align*}
For the last term of \eqref{eq:M_i^2hat{e}_i^2}, we have
\begin{align*}
& \frac{1}{n}\sum_{i=1}^n M_i^2 \left((\hat{\beta}^{\textsc{ols}} - \beta )^\top X_i\right)^2
\le 3 \left[  (\hat{\beta}^{\textsc{ols}}_0 - \beta_0)^2 \frac{1}{n}\sum_{i=1}^n M_i^2 + (\hat{\beta}^{\textsc{ols}}_1 - \beta_1)^2 \frac{1}{n}\sum_{i=1}^n T_i M_i^2 + (\hat{\beta}^{\textsc{ols}}_2 - \beta_2)^2 \frac{1}{n}\sum_{i=1}^n M_i^4 \right], 
\end{align*}
which is $o_{\mathbb{P}}(1)$ by the bounded support of $M_i$ and $\hat{\beta}^{\textsc{ols}} - \beta  = o_{\mathbb{P}}(1)$.
By combining these results, we can show that
\begin{align*}
\frac{1}{n} \sum_{i=1}^n M_i^2 (\hat{u}_i^\textsc{ols})^2
=& E( M_i^2 u_i^2 ) + o_{\mathbb{P}}(1).
\end{align*}

\paragraph*{Step 3:}
Steps 1--2 and CMT together imply the desired result:
\begin{align*}
& \left(  \hat{V}^{\textsc{ols}}_{\text{num}} \right)^{-1/2} 
\left(  \sum_{i=1}^n X_i u_i \right)
\stackrel{\textup{d}}{\rightarrow}  \mathcal{N}(0, I_3).
\end{align*}
By Theorem \ref{thm:XX-1}, we have shown the probability limit of $(X^\top X)^{-1}$. 
Then together with CMT, we can show that 
\[
\left( \hat{V}^{\textsc{ols}} \right)^{-1/2} (\hat{\beta}^{\textsc{ols}}-\beta) \stackrel{\textup{d}}{\rightarrow} \mathcal{N}(0, I_3).
\]

\end{proof}

\section{Proof of results in Section \ref{sec:SSIV}}
\label{app:IV}

\subsection{Auxiliary Lemmas}

Define $\phi_i$ as an i.i.d. random variable with constant variance.
Define the conditional expectation 
$Q_j^{\phi} = E\left[ \frac{ A_{ij}^{\text{pre}} \phi_i }{ q_n^{\text{pre}}  } \mid w_j \right]$. 
\begin{lemma}\label{lemma:Ziphii}
Under Assumptions \ref{asu:network} and \ref{asu:linearY},
\begin{align}
\frac{1}{n^2 q_n^{\text{pre}}} \sum_{i=1}^n Z_i^{\textsc{ssiv}} \phi_i
= \frac{1}{n}\sum_{i=1}^n (T_i - \pi) Q_i^{\phi}
+ O_{\mathbb{P}}\left( \frac{1}{\sqrt{n} \sqrt{n q_n^{\text{pre}}}} \right). 
\label{eq:Zi_alt_phi}
\end{align}
\end{lemma}
\begin{proof}[Proof of Lemma \ref{lemma:Ziphii}]
By reordering the index, we have
\begin{align*}
\frac{1}{n^2 q_n^{\text{pre}}}\sum_{i=1}^n Z_i^{\textsc{ssiv}} \phi_i
=& \frac{1}{n^2 q_n^{\text{pre}}}\sum_{i=1}^n \sum_{j\ne i} A_{ij}^{\text{pre}} (T_j - \pi)  \phi_i 
= \frac{1}{n}\sum_{j=1}^n (T_j - \pi) 
\frac{1}{n q_n^{\text{pre}}} \sum_{i\ne j} A_{ij}^{\text{pre}} \phi_i. 
\end{align*}   
Note that, for fixed $j$, given $w_j$, $\frac{ A_{ij}^{\text{pre}} \phi_i }{ q_n^{\text{pre}} } $ are i.i.d.. 
Then we have 
\begin{align*}
E\left[ \left( \frac{1}{n-1} \sum_{i\ne j}
\frac{ A_{ij}^{\text{pre}} \phi_i }{ q_n^{\text{pre}} } - Q_j^{\phi} \right)^2 \right]
&= \frac{1}{n-1} 
E\left[ \left( 
\frac{ A_{ij}^{\text{pre}} \phi_i }{ q_n^{\text{pre}} } - Q_j^{\phi} \right)^2 \right] \\
&\le \frac{1}{(n-1) q_n^{\text{pre}}} 
E\left[ \frac{A_{ij}^{\text{pre}} \phi_i^2}{q_n^{\text{pre}}} \right]
\le \frac{C}{(n-1) q_n^{\text{pre}}}. 
\end{align*}
We complete the proof by showing that
\begin{align*}
& E\left[ \left( \frac{1}{n}\sum_{j=1}^n (T_j - \pi) 
\left( \frac{1}{n q_n^{\text{pre}}} \sum_{i\ne j} A_{ij}^{\text{pre}} \phi_i - Q_j^ \phi \right) \right)^2 \right] \\
=& \frac{1}{n^2}\sum_{j=1}^n \pi(1-\pi) E\left[ \left( \frac{1}{n q_n^{\text{pre}}} \sum_{i\ne j} A_{ij}^{\text{pre}} \phi_i - Q_j^ \phi \right)^2 \right] 
\le \frac{C}{n^2 q_n^{\text{pre}}}.
\end{align*}

\end{proof}

\begin{lemma} \label{lemma:Zixx}
Under Assumption \ref{asu:network}, we have
\begin{align}
\frac{1}{n^2 q_n^{\text{pre}}}\sum_{i=1}^n M_i Z_i^{\textsc{ssiv}}
=& \frac{ E\left[ Z_i^{\textsc{ssiv}}(r_{0,i} + r_{1,i} - \mu_{r_1,i}) \right] }{n q_n^{\text{pre}} }
+ O_{\mathbb{P}}\left( \frac{1}{\sqrt{n}} \right), \label{eq:MiZi} 
\end{align}
\end{lemma}

\begin{proof}[Proof of Lemma \ref{lemma:Zixx}]
See proof in Section \ref{proof_lemma:Zixx}.
\end{proof}

Now we analyze $\frac{1}{n}\sum_{i=1}^n (Z_i^{\textsc{ssiv}} - \bar{Z}^{\textsc{ssiv}})(M_i - \bar{M})$, which measures the relevance of the SSIV.
\begin{lemma}\label{lem:ZM}
Under Assumptions \ref{asu:network} and \ref{asu:linearY}, then
\begin{align*}
\frac{1}{n}\sum_{i=1}^n (Z_i^{\textsc{ssiv}} - \bar{Z}^{\textsc{ssiv}})(M_i - \bar{M})
= \begin{cases}
O_{\mathbb{P}}\left( \frac{ \min\{q_n^{\text{pre}}, q_n^{\text{post}}\} }{ q_n^{\text{post}} } \right)
+ O_{\mathbb{P}}\left( \frac{ q_n^{\text{pre}} }{ \sqrt{n} q_n^{\text{post}} } \right)
& \text{ if } \V(\xi_i) = 0; \\
O_{\mathbb{P}}\left( \frac{ \min\{ q_n^{\text{pre}}, q_n^{\text{post}} \} }{ q_n^{\text{post}} } \right)
+ O_{\mathbb{P}}\left( \sqrt{n} q_n^{\text{pre}} \right)
& \text{ if } \V(\xi_i) > 0.
\end{cases}
\end{align*}
\end{lemma}
\begin{proof}[Proof of Lemma \ref{lem:ZM}]
We can rewrite it as 
\begin{align*}
\frac{1}{n}\sum_{i=1}^n Z_i^{\textsc{ssiv}} M_i - \bar{Z}^{\textsc{ssiv}} \bar{M} 
=& \frac{1}{n}\sum_{i=1}^n Z_i^{\textsc{ssiv}} (\xi_i + \mu_{r_1,i})
- \bar{Z}^{\textsc{ssiv}} \frac{1}{n}\sum_{i=1}^n (\xi_i + \mu_{r_1,i}) \\
&+ \frac{1}{n}\sum_{i=1}^n Z_i^{\textsc{ssiv}} (r_{0,i} + r_{1,i} - \mu_{r_1,i})
- \bar{Z}^{\textsc{ssiv}} \frac{1}{n}\sum_{i=1}^n (r_{0,i} + r_{1,i} - \mu_{r_1,i}). 
\end{align*}
By the argument in the proof of Lemma \ref{lemma:Zixx} (see \eqref{eq:Ziri} for details), we show that
\begin{equation*}
\frac{1}{n}\sum_{i=1}^n Z_i^{\textsc{ssiv}}(r_{0,i} + r_{1,i} - \mu_{r_1,i})
= E\left[ Z_i^{\textsc{ssiv}}(r_{0,i} + r_{1,i} - \mu_{r_1,i}) \right]
+ O_{\mathbb{P}}\left( \frac{ q_n^{\text{pre}} }{\sqrt{n} q_n^{\text{post}}} \right)
\end{equation*}
where 
\[
E\left[ Z_i^{\textsc{ssiv}}(r_{0,i} + r_{1,i} - \mu_{r_1,i}) \right]
\asymp \frac{ \min\{q_n^{\text{pre}}, q_n^{\text{post}}\} }{ q_n^{\text{post}} }.
\]
By Lemma \ref{lemma:Ziphii} and Lemma \ref{lemma:airi_ri^2}, we can show that 
\begin{equation*}
\bar{Z}^{\textsc{ssiv}} \frac{1}{n}\sum_{i=1}^n (r_{0,i} + r_{1,i} - \mu_{r_1,i})
= O_{\mathbb{P}}\left( \frac{ n q_n^{\text{pre}} }{ \sqrt{n} } \right) 
O_{\mathbb{P}}\left( \frac{ 1 }{ n \sqrt{q_n^{\text{post}}} } \right)
= O_{\mathbb{P}}\left( \frac{ q_n^{\text{pre}} }{  \sqrt{ n q_n^{\text{post}}} } \right).
\end{equation*}
For the sum of the first two terms, we consider Cases (a) and (b), respectively.
\paragraph{Case (a)} 
By Lemma \ref{lemma:Ziphii}, we have 
\begin{align*}
\frac{1}{n}\sum_{i=1}^n Z_i^{\textsc{ssiv}} (\xi_i + \mu_{r_1,i})
- \bar{Z}^{\textsc{ssiv}} \frac{1}{n}\sum_{i=1}^n (\xi_i + \mu_{r_1,i})
= O_{\mathbb{P}}\left( \sqrt{n} q_n^{\text{pre}} \right).
\end{align*}

\paragraph{Case (b)} 
By Lemma \ref{lemma:Ziphii} and recall \eqref{eq:mu_r1i}, we have 
\begin{align*}
\frac{1}{n}\sum_{i=1}^n Z_i^{\textsc{ssiv}}(\xi_i + \mu_{r_1,i}) 
- \bar{Z}^{\textsc{ssiv}} \frac{1}{n}\sum_{i=1}^n (\xi_i + \mu_{r_1,i}) 
= \frac{1}{n}\sum_{i=1}^n Z_i^{\textsc{ssiv}} \mu_{r_1,i}
- \bar{Z}^{\textsc{ssiv}} \frac{1}{n}\sum_{i=1}^n \mu_{r_1,i} 
= O_{\mathbb{P}}\left( \frac{q_n^{\text{pre}}}{\sqrt{n} q_n^{\text{post}}} \right).
\end{align*}
To conclude,
\begin{align*}
\frac{1}{n}\sum_{i=1}^n ( Z_i^{\textsc{ssiv}} - \bar{Z}^{\textsc{ssiv}})( M_i - \bar{M})
= \begin{cases}
O_{\mathbb{P}}\left( \frac{ \min\{q_n^{\text{pre}}, q_n^{\text{post}}\} }{ q_n^{\text{post}} } \right)
+ O_{\mathbb{P}}\left( \frac{ q_n^{\text{pre}} }{ \sqrt{n} q_n^{\text{post}} } \right)
& \text{ if }  \V(\xi_i) = 0; \\
O_{\mathbb{P}}\left( \frac{ \min\{q_n^{\text{pre}}, q_n^{\text{post}}\} }{ q_n^{\text{post}} } \right)
+ O_{\mathbb{P}}\left( \sqrt{n} q_n^{\text{pre}} \right)
& \text{ if }  \V(\xi_i)>0.
\end{cases} 
\end{align*}

\end{proof}

\subsection{Consistency}
Define the matrix
\begin{align*}
D_n^{\textsc{iv}} = \begin{pmatrix}
1 & 0 & 0 \\
0 & 1 & 0 \\
0 & 0 & \frac{1}{n q_n^{\text{pre}} } 
\end{pmatrix}.
\end{align*}
The closed form of $( D_n^{\textsc{iv}} {Z}^\top X )^{-1}$ is 
\begin{align}
& \left( D_n^{\textsc{iv}} {Z}^\top X \right)^{-1} 
= \frac{ 1 }{ \det(D_n^{\textsc{iv}} {Z}^\top X) } \begin{pmatrix}
b_{11} & b_{12} & b_{13} \\
b_{21} & b_{22} & b_{23} \\
b_{31} & b_{32} & b_{33}
\end{pmatrix}
\label{eq:ZX-1}
\end{align}  
where 
\begin{align*}
b_{11}
~=~& \left( \frac{1}{n}\sum_{i=1}^n T_i \right) \left( \frac{1}{n^2 q_n^{\text{pre}}}\sum_{i=1}^n M_i Z_i^{\textsc{ssiv}} \right) - \left( \frac{1}{n }\sum_{i=1}^n T_i M_i \right) \left( \frac{1}{n^2 q_n^{\text{pre}}}\sum_{i=1}^n T_i Z_i^{\textsc{ssiv}} \right) \\
b_{12}
~=~& - \left( \frac{1}{n}\sum_{i=1}^n T_i \right) \left( \frac{1}{n^2 q_n^{\text{pre}}}\sum_{i=1}^n M_i Z_i^{\textsc{ssiv}} \right) + \left( \frac{1}{n^2 q_n^{\text{pre}}}\sum_{i=1}^n T_i Z_i^{\textsc{ssiv}} \right) \left( \frac{1}{n}\sum_{i=1}^n M_i \right)  \\
b_{13}
~=~& a_{13} \\
b_{21}
~=~& - \left( \frac{1}{n}\sum_{i=1}^n T_i \right) \left( \frac{1}{n^2 q_n^{\text{pre}}}\sum_{i=1}^n M_i Z_i^{\textsc{ssiv}} \right) 
+ \left( \frac{1}{n^2 q_n^{\text{pre}}}\sum_{i=1}^n Z_i^{\textsc{ssiv}} \right) 
\left( \frac{1}{n}\sum_{i=1}^n T_i M_i \right) \\
b_{22}
~=~& \left( \frac{1}{n^2 q_n^{\text{pre}}}\sum_{i=1}^n M_i Z_i^{\textsc{ssiv}} \right) - \left( \frac{1}{n}\sum_{i=1}^n M_i \right) \left( \frac{1}{n^2 q_n^{\text{pre}}}\sum_{i=1}^n Z_i^{\textsc{ssiv}} \right) \\
b_{23}
~=~& a_{23} \\
b_{31}
~=~& \left( \frac{1}{n}\sum_{i=1}^n T_i \right) \left( \frac{1}{n^2 q_n^{\text{pre}}}\sum_{i=1}^n T_i Z_i^{\textsc{ssiv}} \right) - \left( \frac{1}{n}\sum_{i=1}^n T_i \right) \left( \frac{1}{n^2 q_n^{\text{pre}}}\sum_{i=1}^n Z_i^{\textsc{ssiv}} \right) \\
b_{32}
~=~& - \left( \frac{1}{n^2 q_n^{\text{pre}}}\sum_{i=1}^n T_i Z_i^{\textsc{ssiv}} \right) + \left( \frac{1}{n^2 q_n^{\text{pre}}}\sum_{i=1}^n Z_i^{\textsc{ssiv}} \right) \left( \frac{1}{n}\sum_{i=1}^n T_i \right) \\
b_{33} 
~=~& a_{33}
\end{align*}
and 
\begin{align*}
\det({Z}^\top X)
= b_{33} b_{22} - b_{32} b_{23}.
\end{align*}
Define $b_{22}^* = \frac{E\left[ Z_i^{\textsc{ssiv}}(r_{0,i} + r_{1,i} ) \right]}{n q_n^{\text{pre}}} \asymp \frac{1}{n \max\{q_n^{\text{pre}},q_n^{\text{post}}\}} $.
Below, we show the probability limit of $\left( D_n^{\textsc{iv}} {Z}^\top X \right)^{-1}$.
\begin{theorem}\label{thm:ZX-1}
Under Assumption \ref{asu:network}, then
\paragraph{Case (a):}
\begin{align*}
& \left( D_n^{\textsc{iv}} {Z}^\top X \right)^{-1} 
= \frac{ \begin{pmatrix}
\pi b_{22}^* & - \pi b_{22}^* & a_{13}^* \\
- \pi b_{22}^* & b_{22}^*  & - \textup{Cov}(T_i, \xi_i + \mu_{r_1,i}) \\
0 & 0 & \pi (1-\pi) 
\end{pmatrix} + O_{\mathbb{P}}\left( \frac{1}{ \sqrt{n} } \right) }{ \pi (1-\pi) b_{22}^* + O_{\mathbb{P}}\left( \frac{1}{ \sqrt{n} } \right) };
\end{align*}

\paragraph{Case (b):}
\begin{align*}
& \left( D_n^{\textsc{iv}} {Z}^\top X \right)^{-1} 
= \frac{ \begin{pmatrix}
\pi b_{22}^* + O_{\mathbb{P}}\left( \frac{1}{ \sqrt{n} } \right) & - \pi b_{22}^*  + O_{\mathbb{P}}\left( \frac{1}{ \sqrt{n} } \right)  & a_{13}^* + O_{\mathbb{P}}\left( \frac{1}{ \sqrt{n} } \right) \\
- \pi b_{22}^* + O_{\mathbb{P}}\left( \frac{1}{ \sqrt{n} } \right) & b_{22}^* + O_{\mathbb{P}}\left( \frac{1}{\sqrt{n} nq_n^{\text{post}}} \right)  &  O_{\mathbb{P}}\left( \frac{1}{ n \sqrt{q_n^{\text{post}}} } \right)  \\
O_{\mathbb{P}}\left( \frac{1}{ \sqrt{n} } \right) & O_{\mathbb{P}}\left( \frac{1}{ \sqrt{n} \sqrt{n q_n^{\text{pre}}} } \right) & \pi (1-\pi) + O_{\mathbb{P}}\left( \frac{1}{ \sqrt{n} } \right) 
\end{pmatrix} }{ \pi (1-\pi) b_{22}^* + O_{\mathbb{P}}\left( \frac{ 1 }{ \sqrt{n} n q_n^{\text{post}}} \right) }. 
\end{align*}

\end{theorem}

\begin{proof}[Proof of Theorem \ref{thm:ZX-1}]
By Lemma \ref{lemma:Ziphii}, we can show that 
\begin{align}
\frac{1}{n^2 q_n^{\text{pre}} }\sum_{i=1}^n Z_i^{\textsc{ssiv}}
=& \frac{1}{n}\sum_{i=1}^n (T_i - \pi) E\left[ \frac{ A_{ij}^{\text{pre}} }{ q_n^{\text{pre}}  } \mid w_i \right]
+ O_{\mathbb{P}}\left( \frac{1}{\sqrt{n} \sqrt{nq_n^{\text{pre}} }} \right)
= O_{\mathbb{P}}\left( \frac{1}{\sqrt{n}} \right), 
\label{eq:Zi} \\
\frac{1}{n^2 q_n^{\text{pre}}}\sum_{i=1}^n T_i Z_i^{\textsc{ssiv}}
=& \frac{1}{n}\sum_{i=1}^n (T_i - \pi) E\left[ \frac{ A_{ij}^{\text{pre}} T_j }{ q_n^{\text{pre}}  } \mid w_i \right]
+ O_{\mathbb{P}}\left( \frac{1}{\sqrt{n} \sqrt{nq_n^{\text{pre}}}} \right)
= O_{\mathbb{P}}\left( \frac{1}{\sqrt{n}} \right).
\label{eq:TiZi}
\end{align}
Then we analyze terms in \eqref{eq:ZX-1} individually. \\
\underline{(1): $b_{11}$.}
By Lemma \ref{lemma:consistencyM} and Lemma \ref{lemma:Zixx}, we have
\begin{align*}
& b_{11}
= \left( \frac{1}{n}\sum_{i=1}^n T_i \right) \left( \frac{1}{n^2 q_n}\sum_{i=1}^n M_i Z_i^{\textsc{ssiv}} \right) - \left( \frac{1}{n}\sum_{i=1}^n T_i M_i \right) \left( \frac{1}{n^2 q_n}\sum_{i=1}^n T_i Z_i^{\textsc{ssiv}} \right) \\
&= \left( \pi + O_{\mathbb{P}}\left( \frac{1}{\sqrt{n}} \right) \right) \left( b_{22}^*  + O_{\mathbb{P}}\left( \frac{1}{\sqrt{n}} \right) \right) 
- \left( E[T_i (\xi_i + \mu_{r_1,i})] 
+ O_{\mathbb{P}}\left( \frac{1}{\sqrt{n}} \right) \right) 
O_{\mathbb{P}}\left( \frac{1}{\sqrt{n}} \right) \\
&= \pi b_{22}^* + O_{\mathbb{P}}\left( \frac{1}{\sqrt{n}} \right).
\end{align*}
\underline{(2): $b_{12}$.}
By Lemma \ref{lemma:consistencyM} and Lemma \ref{lemma:Zixx}, we have
\begin{align*}
& b_{12}
= \left( \frac{1}{n^2 q_n^{\text{pre}}}\sum_{i=1}^n T_i Z_i^{\textsc{ssiv}} \right) 
\left( \frac{1}{n}\sum_{i=1}^n M_i \right)
- \left( \frac{1}{n}\sum_{i=1}^n T_i \right) 
\left( \frac{1}{n^2 q_n^{\text{pre}}}\sum_{i=1}^n M_i Z_i^{\textsc{ssiv}} \right)  \\
&= \left( \frac{1}{n^2 q_n^{\text{pre}}}\sum_{i=1}^n T_i Z_i^{\textsc{ssiv}} \right) \left( \frac{1}{n}\sum_{i=1}^n (\xi_i + \mu_{r_1,i}) \right)
- \left( \frac{1}{n}\sum_{i=1}^n T_i \right) \left( \frac{1}{n^2 q_n^{\text{pre}}}\sum_{i=1}^n (\xi_i + \mu_{r_1,i}) Z_i^{\textsc{ssiv}} \right) \\
&+ \left( \frac{1}{n^2 q_n^{\text{pre}}}\sum_{i=1}^n T_i Z_i^{\textsc{ssiv}} \right) \left( \frac{1}{n}\sum_{i=1}^n (r_{0,i} + r_{1,i}- \mu_{r_1,i}) \right)
- \left( \frac{1}{n}\sum_{i=1}^n T_i \right) \left( \frac{1}{n^2 q_n^{\text{pre}}}\sum_{i=1}^n (r_{0,i} + r_{1,i} - \mu_{r_1,i}) Z_i^{\textsc{ssiv}} \right) \\
&= \left( E( \xi_i + \mu_{r_1,i} ) 
+ O_{\mathbb{P}}\left( \frac{1}{\sqrt{n}} \right) \right) 
O_{\mathbb{P}}\left( \frac{1}{\sqrt{n}} \right)
- \left( \pi + O_{\mathbb{P}}\left( \frac{1}{\sqrt{n}} \right)  \right) \left( b_{22}^* + O_{\mathbb{P}}\left( \frac{1}{\sqrt{n}} \right) \right) \\
&= - \pi b_{22}^* + O_{\mathbb{P}}\left( \frac{1}{\sqrt{n}} \right).
\end{align*}
\underline{(3): $b_{13}$.}
By \eqref{eq:a_{13}_limit}, we have shown 
$b_{13}
= a_{13}^* + O_{\mathbb{P}}\left( \frac{1}{\sqrt{n}} \right)$. \\
\underline{(4): $b_{21}$.}
By Lemma \ref{lemma:consistencyM} and Lemma \ref{lemma:Zixx}, we can show
\begin{align*}
b_{21}
=& - \left( \frac{1}{n}\sum_{i=1}^n T_i \right) 
\left( \frac{1}{n^2 q_n^{\text{pre}}}\sum_{i=1}^n M_i Z_i^{\textsc{ssiv}} \right) 
+ \left( \frac{1}{n^2 q_n^{\text{pre}}}\sum_{i=1}^n Z_i^{\textsc{ssiv}} \right) 
\left( \frac{1}{n}\sum_{i=1}^n T_i M_i \right) \\
=& - \left( \frac{1}{n}\sum_{i=1}^n T_i \right) 
\left( \frac{1}{n^2 q_n^{\text{pre}}}\sum_{i=1}^n \xi_i Z_i^{\textsc{ssiv}} \right) 
+ \left( \frac{1}{n^2 q_n^{\text{pre}}}\sum_{i=1}^n Z_i^{\textsc{ssiv}} \right) 
\left( \frac{1}{n}\sum_{i=1}^n T_i \xi_i \right) \\
&- \left( \frac{1}{n}\sum_{i=1}^n T_i \right) 
\left( \frac{1}{n^2 q_n^{\text{pre}}}\sum_{i=1}^n \mu_{r_1,i} Z_i^{\textsc{ssiv}} \right) 
+ \left( \frac{1}{n^2 q_n^{\text{pre}}}\sum_{i=1}^n Z_i^{\textsc{ssiv}} \right) 
\left( \frac{1}{n}\sum_{i=1}^n T_i \mu_{r_1,i} \right) \\
&- \left( \frac{1}{n}\sum_{i=1}^n T_i \right) \left( \frac{1}{n^2 q_n^{\text{pre}}}\sum_{i=1}^n (r_{0,i} + r_{1,i} - \mu_{r_1,i}) Z_i^{\textsc{ssiv}} \right) + \left( \frac{1}{n^2 q_n^{\text{pre}}}\sum_{i=1}^n Z_i^{\textsc{ssiv}} \right) \left( \frac{1}{n}\sum_{i=1}^n T_i (r_{0,i} + r_{1,i} - \mu_{r_1,i}) \right) \\
=& - \left( \pi + O_{\mathbb{P}}\left( \frac{1}{\sqrt{n}} \right) \right) \left( \frac{E\left[ Z_i^{\textsc{ssiv}}(r_{0,i} + r_{1,i} - \mu_{r_1,i}) \right]}{n q_n^{\text{pre}}} + O_{\mathbb{P}}\left( \frac{1}{\sqrt{n}} \right) \right) 
+ O_{\mathbb{P}}\left( \frac{1}{n \sqrt{q_n^{\text{post}}}} \right)
O_{\mathbb{P}}\left( \frac{1}{\sqrt{n}} \right) \\
=& - \pi b_{22}^* + O_{\mathbb{P}}\left( \frac{1}{\sqrt{n}} \right).
\end{align*}
\underline{(5): $b_{31}$.}
By analogous argument to \eqref{eq:TiZi}, we can show that 
\begin{align*}
b_{31}
=& \left( \frac{1}{n}\sum_{i=1}^n T_i \right)
\left( \frac{1}{n^2 q_n^{\text{pre}}}\sum_{i=1}^n T_i Z_i^{\textsc{ssiv}} \right) 
- \left( \frac{1}{n}\sum_{i=1}^n T_i \right) 
\left( \frac{1}{n^2 q_n^{\text{pre}}}\sum_{i=1}^n Z_i^{\textsc{ssiv}} \right) \\
=& - \left( \frac{1}{n}\sum_{i=1}^n T_i \right)
\left( \frac{1}{n^2 q_n^{\text{pre}}}\sum_{i=1}^n (1 - T_i) Z_i^{\textsc{ssiv}} \right)
= O_{\mathbb{P}}\left( \frac{1}{\sqrt{n}} \right).
\end{align*}
\underline{(6): $b_{32}$.}
By Lemma \ref{lemma:Ziphii}, we have
\begin{align}
b_{32} 
=& - \left( \frac{1}{n^2 q_n^{\text{pre}}}\sum_{i=1}^n T_i Z_i^{\textsc{ssiv}} \right) + \left( \frac{1}{n^2 q_n^{\text{pre}}}\sum_{i=1}^n Z_i^{\textsc{ssiv}} \right) \left( \frac{1}{n}\sum_{i=1}^n T_i \right) \notag \\
=& - \left( \frac{1}{n^2 q_n^{\text{pre}}}\sum_{i=1}^n (T_i - \pi) Z_i^{\textsc{ssiv}} \right) + \left( \frac{1}{n^2 q_n^{\text{pre}}}\sum_{i=1}^n Z_i^{\textsc{ssiv}} \right) \left( \frac{1}{n}\sum_{i=1}^n (T_i-\pi) \right) \notag \\
=& O_{\mathbb{P}}\left( \frac{1}{\sqrt{n} \sqrt{n q_n^{\text{pre}}}} \right)
+ O_{\mathbb{P}}\left( \frac{1}{\sqrt{n} } \right)
O_{\mathbb{P}}\left( \frac{1}{\sqrt{n} } \right)
= O_{\mathbb{P}}\left( \frac{1}{\sqrt{n} \sqrt{n q_n^{\text{pre}}}} \right).
\label{eq:b_33_limit}
\end{align}
\underline{(7): $b_{33}$.} $b_{33} = \pi (1-\pi) + O_{\mathbb{P}}\left( \frac{1}{\sqrt{n}} \right)$.

We consider Cases (a) and (b) for the remaining terms, respectively.

\paragraph{Case (a)}

\underline{(1): $b_{22}$.}
By Lemma \ref{lem:ZM}, we have 
\begin{align}
b_{22}
= b_{22}^* + O_{\mathbb{P}}\left( \frac{1}{\sqrt{n} } \right).
\label{eq:b_22_limit_a}
\end{align}
\underline{(2): $b_{23}$.}
By Lemma \ref{lemma:consistencyM}, we can show 
\begin{align}
b_{23}
= - \textup{Cov}(T_i, \xi_i + \mu_{r_1,i}) 
+ O_{\mathbb{P}}\left( \frac{1}{ \sqrt{n} } \right).
\label{eq:b_23_limit_a}
\end{align}
\underline{(3): $\det({Z}^\top X)$.}
By combining \eqref{eq:b_33_limit}, \eqref{eq:b_22_limit_a} and \eqref{eq:b_23_limit_a}, we have
\begin{align*}
& \det({Z}^\top X)
= b_{33} b_{22} - b_{32} b_{23} \\
&= \left( \pi (1-\pi) + O_{\mathbb{P}}\left( \frac{1}{\sqrt{n}} \right) \right)
\left( b_{22}^* + O_{\mathbb{P}}\left( \frac{ 1 }{ \sqrt{n} } \right) \right) 
- \left( - \textup{Cov}(T_i, \xi_i + \mu_{r_1,i}) 
+ O_{\mathbb{P}}\left( \frac{1}{ \sqrt{n} } \right) \right)
O_{\mathbb{P}}\left(\frac{1}{n\sqrt{q_n^{\text{pre}}}}\right) \\
&= \pi (1-\pi) b_{22}^* + O_{\mathbb{P}}\left( \frac{1}{ \sqrt{n} } \right).
\end{align*}

\paragraph{Case (b)}

\underline{(1): $b_{22}$.}
By Lemma \ref{lem:ZM}, we have
\begin{align}
b_{22}
= b_{22}^* + O_{\mathbb{P}}\left( \frac{ 1 }{ \sqrt{n} n q_n^{\text{post}}} \right).
\label{eq:b_22_limit_b}
\end{align}
\underline{(2): $b_{23}$.}
By \eqref{eq:a_23_limit}, we have shown 
$b_{23} = O_{\mathbb{P}}\left( \frac{1}{n\sqrt{q_n^{\text{post}}}} \right)$. \\
\underline{(3): $\det({Z}^\top X)$.}
By combining \eqref{eq:b_33_limit}, \eqref{eq:b_22_limit_b} and \eqref{eq:a_23_limit}, we have
\begin{align*}
& \det({Z}^\top X)
= b_{33} b_{22} - b_{32} b_{23} \\
&= \left( \pi (1-\pi) + O_{\mathbb{P}}\left( \frac{1}{\sqrt{n}} \right) \right)
\left( b_{22}^* + O_{\mathbb{P}}\left( \frac{ 1 }{ \sqrt{n} n q_n^{\text{post}}} \right) \right) 
- O_{\mathbb{P}}\left( \frac{1}{n\sqrt{q_n^{\text{post}}}} \right) 
O_{\mathbb{P}}\left(\frac{1}{n\sqrt{q_n^{\text{pre}}}}\right) \\
&= \pi (1-\pi) b_{22}^* + O_{\mathbb{P}}\left( \frac{ 1 }{ \sqrt{n} n q_n^{\text{post}}} \right).
\end{align*}
    
\end{proof}

\begin{proof}[Proof of Theorem \ref{thm:IV_consistency}]
By direct algebra, the closed form of the IV estimators $\hat{\beta}^{\textsc{iv}}$ is 
\begin{align*}
\hat{\beta}^{\textsc{iv}}_1 - {\beta}_1
~=~& \frac{b_{22} \frac{1}{n}\sum_{i=1}^n (T_i - \bar{T}) u_i + b_{23} \frac{1}{n^2 q_n^{\text{pre}}}\sum_{i=1}^n (Z_i^{\textsc{ssiv}} - \bar{Z}^{\textsc{ssiv}}) u_i }{ b_{33} b_{22} - b_{32} b_{23} }, \\
\hat{\beta}^{\textsc{iv}}_2 - {\beta}_2
~=~& \frac{ b_{33} \frac{1}{n^2 q_n^{\text{pre}}}\sum_{i=1}^n (Z_i^{\textsc{ssiv}} - \bar{Z}^{\textsc{ssiv}}) u_i + b_{32} \frac{1}{n}\sum_{i=1}^n (T_i - \bar{T}) u_i }{ b_{33} b_{22} - b_{32} b_{23} }, \\
\hat{\beta}^{\textsc{iv}}_0 - {\beta}_0
~=~& \frac{1}{n}\sum_{i=1}^n u_i - (\hat{\beta}^{\textsc{iv}}_1 - {\beta}_1) \bar{T} - (\hat{\beta}^{\textsc{iv}}_2 - {\beta}_2) \bar{M}.
\end{align*}
By Lemma \ref{lemma:Ziphii}, we can show that 
\begin{align}
\frac{1}{n^2 q_n^{\text{pre}}}\sum_{i=1}^n Z_i^{\textsc{ssiv}} u_i
= \frac{1}{n }\sum_{i=1}^n (T_i-\pi) E\left[ \frac{ A_{ij}^{\text{pre}} u_j }{ q_n^{\text{pre}}  } \mid w_i \right] 
+ O_{\mathbb{P}}\left( \frac{1}{\sqrt{n}\sqrt{n q_n^{\text{pre}}}} \right)
= O_{\mathbb{P}}\left( \frac{1}{\sqrt{n}} \right).
\label{eq:Ziui}
\end{align}
By combining \eqref{eq:Zi} and \eqref{eq:Ziui}, we can show that 
\begin{align*}
\frac{1}{n^2 q_n^{\text{pre}}}\sum_{i=1}^n (Z_i^{\textsc{ssiv}} -\bar{Z}^{\textsc{ssiv}})u_i
= O_{\mathbb{P}}\left( \frac{1}{\sqrt{n}} \right).
\end{align*}
With i.i.d. data, we also have 
\begin{equation*}
\frac{1}{n}\sum_{i=1}^n u_i = O_{\mathbb{P}}\left( \frac{1}{ \sqrt{n} } \right) 
\text{ and }
\frac{1}{n}\sum_{i=1}^n (T_i - \bar{T}) u_i = O_{\mathbb{P}}\left( \frac{1}{ \sqrt{n} } \right). \label{eq:ui}
\end{equation*}

\paragraph{Case (a)}

By Theorem \ref{thm:ZX-1}, we can conclude that
\begin{align*}
\hat{\beta}^{\textsc{iv}}_1 - {\beta}_1
~=~& \frac{b_{22} \frac{1}{n}\sum_{i=1}^n (T_i - \bar{T}) u_i + b_{23} \frac{1}{n^2 q_n^{\text{pre}}}\sum_{i=1}^n (Z_i^{\textsc{ssiv}} - \bar{Z}^{\textsc{ssiv}}) u_i }{ b_{33} b_{22} - b_{32} b_{23} }  \\
~=~& \frac{ \left( b_{22}^* + O_{\mathbb{P}}\left( \frac{1}{\sqrt{n} } \right) \right) O_{\mathbb{P}}\left( \frac{1}{\sqrt{n} } \right) + \left( - \textup{Cov}(T_i, \xi_i + \mu_{r_1,i}) 
+ O_{\mathbb{P}}\left( \frac{1}{ \sqrt{n} } \right) \right) O_{\mathbb{P}}\left( \frac{1}{\sqrt{n}} \right)  }{ \pi (1-\pi) b_{22}^* + O_{\mathbb{P}}\left( \frac{1}{ \sqrt{n} } \right) }  \\
~=~& \frac{ b_{22}^*  O_{\mathbb{P}}\left( \frac{1}{\sqrt{n} } \right) + \left( - \textup{Cov}(T_i, \xi_i + \mu_{r_1,i})  \right) O_{\mathbb{P}}\left( \frac{1}{\sqrt{n}} \right) +  O_{\mathbb{P}}\left( \frac{1}{n} \right) }{ \pi (1-\pi) b_{22}^* + O_{\mathbb{P}}\left( \frac{1}{ \sqrt{n} } \right) } 
\end{align*}
and 
\begin{align*}
\hat{\beta}^{\textsc{iv}}_2 - {\beta}_2
~=~& \frac{ b_{33} \frac{1}{n^2 q_n^{\text{pre}}}\sum_{i=1}^n (Z_i^{\textsc{ssiv}} - \bar{Z}^{\textsc{ssiv}}) u_i + b_{32} \frac{1}{n}\sum_{i=1}^n (T_i - \bar{T}) u_i }{ b_{33} b_{22} - b_{32} b_{23} } \\
~=~& \frac{ \left(\pi(1-\pi) + O_{\mathbb{P}}\left( \frac{1}{ \sqrt{n} } \right) \right)O_{\mathbb{P}}\left( \frac{1}{ \sqrt{n} } \right) +  O_{\mathbb{P}}\left( \frac{1}{ \sqrt{n} \sqrt{n q_n^{\text{pre}}} } \right) O_{\mathbb{P}}\left( \frac{1}{ \sqrt{n} } \right)}{ \pi (1-\pi) b_{22}^* + O_{\mathbb{P}}\left( \frac{1}{ \sqrt{n} } \right) } \\
~=~& \frac{ \pi(1-\pi) O_{\mathbb{P}}\left( \frac{1}{ \sqrt{n} } \right) +  O_{\mathbb{P}}\left( \frac{1}{ n } \right) }{ \pi (1-\pi) b_{22}^* + O_{\mathbb{P}}\left( \frac{1}{ \sqrt{n} } \right) }.
\end{align*}
Recall $b_{22}^* \asymp \frac{1}{n \max\{q_n^{\text{pre}},q_n^{\text{post}}\}} $.
Therefore, $\hat{\beta}^{\textsc{iv}}_1$ and $\hat{\beta}^{\textsc{iv}}_2$ are consistent when $\max\{q_n^{\text{pre}}, q_n^{\text{post}}\} \prec \frac{1}{\sqrt{n}}$
 with
\begin{align*}
\hat{\beta}^{\textsc{iv}}_1 - {\beta}_1 
~=~& O_{\mathbb{P}}\left( \sqrt{n} \max\{q_n^{\text{pre}}, q_n^{\text{post}}\} \right)
\text{ and }
\hat{\beta}^{\textsc{iv}}_2 - {\beta}_2 
~=~ O_{\mathbb{P}}\left( \sqrt{n} \max\{q_n^{\text{pre}}, q_n^{\text{post}}\} \right).
\end{align*}
The consistency of $\hat{\beta}^{\textsc{iv}}_0$ follows when $\max\{q_n^{\text{pre}}, q_n^{\text{post}}\} \prec \frac{1}{\sqrt{n}}$ with 
\begin{align*}
\hat{\beta}^{\textsc{iv}}_0 - {\beta}_0 
~=~& O_{\mathbb{P}}\left( \sqrt{n} \max\{q_n^{\text{pre}}, q_n^{\text{post}}\} \right).
\end{align*}

\paragraph{Case (b):}
By Theorem \ref{thm:ZX-1}, 
for $\hat{\beta}^{\textsc{iv}}_1 - {\beta}_1$, we have
\begin{align*}
\hat{\beta}^{\textsc{iv}}_1 - {\beta}_1
~=~& \frac{b_{22} \frac{1}{n}\sum_{i=1}^n (T_i - \bar{T}) u_i + b_{23} \frac{1}{n^2 q_n^{\text{pre}}}\sum_{i=1}^n (Z_i^{\textsc{ssiv}} - \bar{Z}^{\textsc{ssiv}}) u_i }{ b_{33} b_{22} - b_{32} b_{23} }  \\
~=~& \frac{ \left( b_{22}^* 
+ O_{\mathbb{P}}\left( \frac{ 1 }{ \sqrt{n} n q_n^{\text{post}}} \right) \right) O_{\mathbb{P}}\left( \frac{ 1 }{ \sqrt{n} } \right)  + O_{\mathbb{P}}\left( \frac{1}{n\sqrt{q_n^{\text{post}}}} \right) O_{\mathbb{P}}\left( \frac{ 1 }{ \sqrt{n} } \right) }{ \pi (1-\pi) b_{22}^* + O_{\mathbb{P}}\left( \frac{ 1 }{ \sqrt{n} n q_n^{\text{post}}} \right) } \\
~=~& \frac{ b_{22}^* O_{\mathbb{P}}\left( \frac{ 1 }{ \sqrt{n} } \right)  + O_{\mathbb{P}}\left( \frac{1}{n\sqrt{q_n^{\text{post}}}} \right) O_{\mathbb{P}}\left( \frac{ 1 }{ \sqrt{n} } \right) }{ \pi (1-\pi) b_{22}^* + O_{\mathbb{P}}\left( \frac{ 1 }{ \sqrt{n} n q_n^{\text{post}}} \right) }.
\end{align*}
Therefore, $\hat{\beta}^{\textsc{iv}}_1$ is consistent when $ \max\{q_n^{\text{pre}}, q_n^{\text{post}}\} \prec \sqrt{n} q_n^{\text{post}} $ with 
\begin{align*}
\hat{\beta}^{\textsc{iv}}_1 - {\beta}_1 
~=~& O_{\mathbb{P}}\left( \frac{1}{\sqrt{n}} \max\left\{ \frac{ \max\{q_n^\text{pre}, q_n^\text{post} \} }{  \sqrt{ q_n^\text{post}}  }, 1 \right\}  \right).
\end{align*}
For $\hat{\beta}^{\textsc{iv}}_2 - {\beta}_2$, we have
\begin{align*}
\hat{\beta}^{\textsc{iv}}_2 - {\beta}_2
~=~& \frac{ b_{33} \frac{1}{n^2 q_n^{\text{pre}}}\sum_{i=1}^n (Z_i^{\textsc{ssiv}} - \bar{Z}^{\textsc{ssiv}}) u_i + b_{32} \frac{1}{n}\sum_{i=1}^n (T_i - \bar{T}) u_i }{ b_{33} b_{22} - b_{32} b_{23} } \\
~=~& \frac{ \left( \pi (1-\pi) + O_{\mathbb{P}}\left( \frac{1}{\sqrt{n}} \right) \right) O_{\mathbb{P}}\left( \frac{1}{\sqrt{n}} \right) + O_{\mathbb{P}}\left( \frac{1}{\sqrt{n} \sqrt{n q_n^{\text{pre}}} } \right) O_{\mathbb{P}}\left( \frac{1}{\sqrt{n}} \right) }{ \pi (1-\pi) b_{22}^* + O_{\mathbb{P}}\left( \frac{ 1 }{ \sqrt{n} n q_n^{\text{post}}} \right)}  \\
~=~& \frac{ \pi (1-\pi) O_{\mathbb{P}}\left( \frac{1}{\sqrt{n}} \right) + O_{\mathbb{P}}\left( \frac{1}{n} \right) }{ \pi (1-\pi) b_{22}^* + O_{\mathbb{P}}\left( \frac{ 1 }{ \sqrt{n} n q_n^{\text{post}}} \right) }.
\end{align*}
Therefore, $\hat{\beta}^{\textsc{iv}}_2$ is consistent when $\max\{q_n^{\text{pre}}, q_n^{\text{post}}\} \prec n^{-1/2}$, with estimation error
\begin{align*}
\hat{\beta}^{\textsc{iv}}_2 - {\beta}_2 
~=~ O_{\mathbb{P}}\left( \sqrt{n} \max\{q_n^{\text{pre}}, q_n^{\text{post}}\} \right).
\end{align*}
The consistency of $\hat{\beta}^{\textsc{iv}}_0$ follows when $\max\{q_n^{\text{pre}}, q_n^{\text{post}}\} \prec n^{-1/2}$ with 
\begin{align*}
\hat{\beta}^{\textsc{iv}}_0 - {\beta}_0 
~=~& O_{\mathbb{P}}\left( \sqrt{n} \max\{q_n^{\text{pre}}, q_n^{\text{post}}\} \right).
\end{align*}
\end{proof}

\begin{proof}[Proof of Corollary \ref{cor:IV_consistency}]
Corollary \ref{cor:IV_consistency} is a direct result of Theorem \ref{thm:IV_consistency} under 
$q_n^{\text{pre}} \preccurlyeq q_n^{\text{post}}$.    
\end{proof}

\subsection{Asymptotic normality}
Below we prove the asymptotic distribution when the estimators are consistent.

\begin{proof}[Proof of Theorem \ref{thm:asymnormal_ratio_IV}]
We complete the proof in three steps.
\paragraph{Step 1: asymptotic normality.}
We prove the asymptotic normality of the numerator:
\begin{align*}
D_n^{\textsc{iv}} \frac{1}{n}\sum_{i=1}^n {Z}_i u_i 
= 
\frac{1}{n}\sum_{i=1}^n 
\begin{pmatrix}
u_i \\
T_i u_i \\
Z_i^{\textsc{ssiv}} u_i / (nq_n^{\text{pre}})
\end{pmatrix}.
\end{align*}
Analogous to Theorem \ref{thm:asymnormal_ratio}, we divide the proof into two cases.

\underline{Case 1: $n q_n^{\text{pre}} \succ 1$}.
By \eqref{eq:Ziui}, we have shown that
\begin{align*}
\frac{1}{n^2 q_n^{\text{pre}}}\sum_{i=1}^n Z_i^{\textsc{ssiv}} u_i
= \frac{1}{n }\sum_{i=1}^n (T_i-\pi) E\left[ \frac{ A_{ij}^{\text{pre}} u_j }{ q_n^{\text{pre}}  } \mid w_i \right] 
+ O_{\mathbb{P}}\left( \frac{1}{\sqrt{n}\sqrt{nq_n^{\text{pre}}}} \right),
\end{align*}
which can be approximated by the average of i.i.d. random variable with a small error $o_{\mathbb{P}}\left(\frac{1}{\sqrt{n}}\right)$. Then the asymptotic normality of $D_n^{\textsc{iv}} \frac{1}{n}\sum_{i=1}^n {Z}_i u_i $ follows from standard argument using Cramer--Wold theorem and Lindeberg--Lévy CLT.

\underline{Case 2: $n q_n^{\text{pre}} \asymp 1$}.
It suffices to verify the assumptions in Theorem \ref{thm:clt}.
Assumption (a) follows from  $V_{\text{num}}^{\textsc{iv}} = O(n)$.
Assumption (b) follows from Lemma \ref{lemma:consistencyM}.
Assumption (c) follows from applying Lemma \ref{lemma:Assu_B2} to $A^{\text{pre}}$ with $n q_n^{\text{pre}} \asymp 1$. 
Then by Theorem \ref{thm:clt} and Cramer--Wold device, we can show 
\begin{align*}
\left( V_{\text{num}}^{\textsc{iv}} \right)^{-1/2} 
\left(  \sum_{i=1}^n {Z}_i u_i \right)
~\overset{d}{\to}~
\mathcal{N}\left(0, I_3 \right).
\end{align*}

\paragraph{Step 2: consistent variance estimator.}
We show that 
\begin{align*}
D_n^{\textsc{iv}} 
\left( \hat{V}^{\textsc{iv}}_{\text{num}} - {V}^{\textsc{iv}}_{\text{num}} \right) 
D_n^{\textsc{iv}}
= o_{\mathbb{P}}(1).
\end{align*}
Recall that $\hat{u}_i^{\textsc{iv}} = u_i - ( \hat{\beta}^{\textsc{iv}} - \beta)^\top X_i$.
We show the consistency of the variance estimators one by one. \\
\underline{(1): The consistency of $(1,1)$th element of $\hat{V}^{\textsc{iv}}_{\text{num}}$.} 
By expanding the square,
\begin{align*}
& \frac{1}{n} \sum_{i=1}^n (\hat{u}_i^{\textsc{iv}})^2
= \frac{1}{n} \sum_{i=1}^n u_i^2
- \frac{2}{n} \sum_{i=1}^n u_i ( \hat{\beta}^{\textsc{iv}} - \beta)^\top X_i
+ \frac{1}{n} \sum_{i=1}^n \left(( \hat{\beta}^{\textsc{iv}} - \beta)^\top X_i\right)^2. 
\end{align*}
First, by LLN, we have $\frac{1}{n} \sum_{i=1}^n u_i^2 = E(u_i^2) + o_{\mathbb{P}}(1)$.
Second, with $\hat{\beta}^{\textsc{iv}} - \beta = o_{\mathbb{P}}(1)$, we have
\begin{align*}
\frac{1}{n} \sum_{i=1}^n u_i ( \hat{\beta}^{\textsc{iv}} - \beta)^\top X_i
=& ( \hat{\beta}^{\textsc{iv}}_0 - \beta_0) \frac{1}{n} \sum_{i=1}^n u_i + ( \hat{\beta}^{\textsc{iv}}_1 - \beta_1) \frac{1}{n} \sum_{i=1}^n T_i u_i + ( \hat{\beta}^{\textsc{iv}}_2 - \beta_2) \frac{1}{n} \sum_{i=1}^n M_i u_i 
= o_{\mathbb{P}}(1).
\end{align*}
Third, 
\begin{align*}
\frac{1}{n} \sum_{i=1}^n \left(( \hat{\beta}^{\textsc{iv}} - \beta)^\top X_i \right)^2
=& \frac{1}{n} \sum_{i=1}^n \left( ( \hat{\beta}^{\textsc{iv}}_0 - \beta_0) + ( \hat{\beta}^{\textsc{iv}}_1 - \beta_1) T_i + ( \hat{\beta}^{\textsc{iv}}_2 - \beta_2) M_i \right)^2 \\
\le & 3 \left[ ( \hat{\beta}^{\textsc{iv}}_0 - \beta_0)^2 
+ ( \hat{\beta}^{\textsc{iv}}_1 - \beta_1)^2 \frac{1}{n} \sum_{i=1}^n T_i 
+ ( \hat{\beta}^{\textsc{iv}}_2 - \beta_2)^2 \frac{1}{n} \sum_{i=1}^n M_i^2 \right]
= o_{\mathbb{P}}(1),
\end{align*}
where the last equality follows from the bounded support of $T_i$ and $M_i$, and $ \hat{\beta}^{\textsc{iv}} - \beta = o_{\mathbb{P}}(1)$.
Therefore, it implies that
\begin{align*}
\frac{1}{n} \sum_{i=1}^n (\hat{u}_i^{\textsc{iv}})^2 
= \frac{1}{n} \sum_{i=1}^n u_i^2 + o_{\mathbb{P}}(1)
= E(u_i^2) + o_{\mathbb{P}}(1).    
\end{align*}
\underline{(2): The consistency of $(3,3)$th element of  $\hat{V}^{\textsc{iv}}_{\text{num}}$.} 
By definition,
\begin{align}
& \frac{1}{n} \sum_{j=1}^n \left( \sum_{i\ne j} A_{ij}^\text{pre} \hat{u}_i^{\textsc{iv}} \right)^2
= \frac{1}{n} \sum_{j=1}^n \left( \sum_{i\ne j} A_{ij}^\text{pre} \left(u_i -  ( \hat{\beta}^{\textsc{iv}} - \beta)^\top X_i \right) \right)^2 \notag \\
=& \frac{1}{n} \sum_{j=1}^n \left( \sum_{i\ne j} A_{ij}^\text{pre} u_i \right)^2
- \frac{2}{n} \sum_{j=1}^n  
\left( \sum_{i\ne j} A_{ij}^\text{pre} u_i \right)
\left( \sum_{i\ne j} A_{ij}^\text{pre} ( \hat{\beta}^{\textsc{iv}} - \beta)^\top X_i \right)
+ \frac{1}{n} \sum_{j=1}^n \left( \sum_{i\ne j} A_{ij}^\text{pre} ( \hat{\beta}^{\textsc{iv}} - \beta)^\top X_i \right)^2. 
\label{eq:var_Aij_ui}
\end{align}
For the first term of \eqref{eq:var_Aij_ui}, the expectation is
\begin{align}
E\left[ \frac{1}{n} \sum_{j=1}^n \left( \sum_{i\ne j} A_{ij}^\text{pre} u_i \right)^2 \right]
= E\left[ \frac{1}{n } \sum_{j=1}^n \left( \sum_{i\ne j} A_{ij}^\text{pre} u_i^2 
+ \sum_{(i,k)} A_{ij}^\text{pre} A_{kj}^\text{pre} u_i u_k \right) \right]
\asymp (n q_n^\text{pre})^2.
\label{eq:Aijui}
\end{align}
The variance is 
\begin{align*}
& \V\left( \frac{1}{n} \sum_{j=1}^n \left( \sum_{i\ne j} A_{ij}^\text{pre} u_i \right)^2 \right) 
= \V\left( \frac{1}{n} \sum_{j=1}^n \left( \sum_{i\ne j} A_{ij}^\text{pre} u_i^2 + \sum_{(i,k)} A_{ij}^\text{pre}  A_{kj}^\text{pre} u_i u_k \right) \right) \\
=& \V\left( \frac{1}{n} \sum_{j=1}^n \left( \sum_{i\ne j} A_{ij}^\text{pre} u_i^2 \right) \right) 
+ \V\left( \frac{1}{n} \sum_{j=1}^n \left( \sum_{(i,k)} A_{ij}^\text{pre} A_{kj}^\text{pre} u_i u_k \right) \right) \\
&+ 2 \textup{Cov}\left( \frac{1}{n} \sum_{j=1}^n \left( \sum_{i\ne j} A_{ij}^\text{pre} u_i^2 \right), \frac{1}{n} \sum_{j=1}^n \left( \sum_{(i,k)} A_{ij}^\text{pre} A_{kj}^\text{pre} u_i u_k \right) \right) \\
=& \frac{1}{n^2} \sum_{j=1}^n \sum_{i\ne j} 
\V\left( A_{ij}^\text{pre} u_i^2 \right) 
+ \frac{1}{n^2} \sum_{j=1}^n \sum_{(i,k)}
\textup{Cov}\left( A_{ij}^\text{pre} u_i^2, A_{kj}^\text{pre} u_k^2 \right) 
+ \frac{1}{n^2} \sum_{(j,k)} 
\textup{Cov}\left( \sum_{i\ne j} A_{ij}^\text{pre} u_i^2, \sum_{l\ne k} A_{lk}^\text{pre} u_l^2 \right)  \\
&+ \frac{1}{n^2} \sum_{j=1}^n \sum_{(i,k)} \V\left( A_{ij}^\text{pre} A_{kj}^\text{pre} u_i u_k \right)
+ \frac{1}{n^2} \sum_{j=1}^n \sum_{(i,k) \ne (h,m)} \textup{Cov}\left( A_{ij}^\text{pre} A_{kj}^\text{pre} u_i u_k, A_{hj}^\text{pre} A_{mj}^\text{pre} u_h u_m \right) \\
&+ \frac{1}{n^2} \sum_{(j_1, j_2)} \textup{Cov}\left( \sum_{(i,k)} A_{ij_1}^\text{pre} A_{kj_1}^\text{pre} u_i u_k, \sum_{(i,k)} A_{ij_2}^\text{pre} A_{kj_2}^\text{pre} u_i u_k \right) 
+ 2 \frac{1}{n^2} \textup{Cov}\left( \sum_{j=1}^n \sum_{i\ne j} A_{ij}^\text{pre} u_i^2, \sum_{j=1}^n \sum_{(i,k)} A_{ij}^\text{pre} A_{kj}^\text{pre} u_i u_k \right) \\
=& \frac{1}{n^2} \sum_{j=1}^n \sum_{i\ne j} 
\V\left( A_{ij}^\text{pre} u_i^2 \right) 
+ \frac{1}{n^2} \sum_{j=1}^n \sum_{(i,k)}
\textup{Cov}\left( A_{ij}^\text{pre} u_i^2, A_{kj}^\text{pre} u_k^2 \right) 
+ \frac{1}{n^2} \sum_{(j,k)} \sum_{i\ne j,k}
\textup{Cov}\left( A_{ij}^\text{pre} u_i^2, A_{ik}^\text{pre} u_i^2 \right)  \\
&+ \frac{1}{n^2} \sum_{j=1}^n \sum_{(i,k)} \V\left( A_{ij}^\text{pre} A_{kj}^\text{pre} u_i u_k \right)
+ \frac{1}{n^2} \sum_{j=1}^n \sum_{(i,k) \ne (h,m)} \textup{Cov}\left( A_{ij}^\text{pre} A_{kj}^\text{pre} u_i u_k, A_{hj}^\text{pre} A_{mj}^\text{pre} u_h u_m \right) \\
&+ \frac{1}{n^2} \sum_{(j_1, j_2)} \sum_{(i_1,i_2,k)} \textup{Cov}\left( A_{i_1j_1}^\text{pre} A_{kj_1}^\text{pre} u_i u_k, A_{i_2j_2}^\text{pre} A_{kj_2}^\text{pre} u_i u_k \right) 
+ 2 \frac{1}{n^2} \sum_{(i,j,k,l)} \textup{Cov}\left( A_{ij}^\text{pre} u_i^2, A_{il}^\text{pre} A_{kl}^\text{pre} u_i u_k \right) \\
\asymp & 
n^3 \left(q_n^\text{pre} \right)^4.
\end{align*}
Therefore, 
\begin{align}
\frac{1}{n (n q_n^{\text{pre}})^2} \sum_{j=1}^n \left( \sum_{i\ne j} A_{ij}^\text{pre} u_i \right)^2
= E\left[ \frac{1}{n (n q_n^{\text{pre}})^2} \sum_{j=1}^n \left( \sum_{i\ne j} A_{ij}^\text{pre} u_i \right)^2 \right]
+ O_{\mathbb{P}}\left( \frac{ 1 }{ \sqrt{n} } \right).
\label{eq:Aijui^2}
\end{align}
For the second term of \eqref{eq:var_Aij_ui}, 
\begin{align*}
& \frac{1}{n} \sum_{j=1}^n  
\left( \sum_{i\ne j} A_{ij}^\text{pre} u_i \right)
\left( \sum_{i\ne j} A_{ij}^\text{pre} ( \hat{\beta}^{\textsc{iv}} - \beta)^\top X_i \right) \\
=& ( \hat{\beta}^{\textsc{iv}}_0 - \beta_0) \frac{1}{n} \sum_{j=1}^n  
\left( \sum_{i\ne j} A_{ij}^\text{pre} u_i 
 \sum_{i\ne j} A_{ij}^\text{pre} \right)
+ ( \hat{\beta}^{\textsc{iv}}_1 - \beta_1) \frac{1}{n} \sum_{j=1}^n  
\left( \sum_{i\ne j} A_{ij}^\text{pre} u_i 
 \sum_{i\ne j} A_{ij}^\text{pre} T_i \right) \\
&+ ( \hat{\beta}^{\textsc{iv}}_2 - \beta_2) \frac{1}{n} \sum_{j=1}^n  
\left( \sum_{i\ne j} A_{ij}^\text{pre} u_i  \sum_{i\ne j} A_{ij}^\text{pre} M_i \right) \\
\le & | \hat{\beta}^{\textsc{iv}}_0 - \beta_0| \frac{1}{n} \sum_{j=1}^n  
\left[ \left( \sum_{i\ne j} A_{ij}^\text{pre} u_i \right)^2
+ \left( \sum_{i\ne j} A_{ij}^\text{pre} \right)^2 \right] 
+ | \hat{\beta}^{\textsc{iv}}_1 - \beta_1| \frac{1}{n} \sum_{j=1}^n  
\left[ \left( \sum_{i\ne j} A_{ij}^\text{pre} u_i \right)^2
+ \left( \sum_{i\ne j} A_{ij}^\text{pre} T_i \right)^2 \right] \\
&+ | \hat{\beta}^{\textsc{iv}}_2 - \beta_2| \frac{1}{n} \sum_{j=1}^n  
\left[ \left( \sum_{i\ne j} A_{ij}^\text{pre} u_i \right)^2
+ \left( \sum_{i\ne j} A_{ij}^\text{pre} M_i \right)^2 \right].
\end{align*}
By analogous arguments to \eqref{eq:Aijui}, we can show that 
\begin{align}
& \frac{1}{n} \sum_{j=1}^n \left( \sum_{i\ne j} A_{ij}^\text{pre} \right)^2 
= O_{\mathbb{P}}\left( (n q_n^{\text{pre}})^2\right)
\text{ and }
\frac{1}{n} \sum_{j=1}^n \left( \sum_{i\ne j} A_{ij}^\text{pre} T_i \right)^2 
= O_{\mathbb{P}}\left( (n q_n^{\text{pre}})^2\right)
\label{eq:Aijai}
\end{align}
by bounding them in $L_1$ norm.
Also, by the bounded support of $M_i$, we have
\begin{align*}
E\left[ \frac{1}{n } \sum_{j=1}^n \left( \sum_{i\ne j} A_{ij}^\text{pre} M_i \right)^2 \right]
\le E\left[ \frac{1}{n } \sum_{j=1}^n \left( \sum_{i\ne j} A_{ij}^\text{pre} \right)^2 \right]
\le C (n q_n^{\text{pre}})^2,
\end{align*}
which implies that
\begin{align}
\frac{1}{n} \sum_{j=1}^n \left( \sum_{i\ne j} A_{ij}^\text{pre} M_i \right)^2
=& O_{\mathbb{P}}\left( (n q_n^{\text{pre}})^2 \right).
\label{eq:AijMi^2}
\end{align}
With consistent estimators that $ \hat{\beta}^{\textsc{iv}} - \beta = o_{\mathbb{P}}(1)$, we can show that 
\begin{equation}
\frac{1}{n (n q_n^{\text{pre}})^2} \sum_{j=1}^n  
\left( \sum_{i\ne j} A_{ij}^\text{pre} u_i \right)
\left( \sum_{i\ne j} A_{ij}^\text{pre} ( \hat{\beta}^{\textsc{iv}} - \beta)^\top X_i \right) = o_{\mathbb{P}}(1).
\label{eq:AijuiAijXi}
\end{equation}
For the third term of \eqref{eq:var_Aij_ui}, 
\begin{align*}
\frac{1}{3}\frac{1}{n} \sum_{j=1}^n 
\left( \sum_{i\ne j} A_{ij}^\text{pre}
(\hat{\beta}^{\textsc{iv}} - \beta)^\top X_i \right)^2 
&\le ( \hat{\beta}^{\textsc{iv}}_0 - \beta_0)^2 \frac{1}{n} \sum_{j=1}^n \left( \sum_{i\ne j} A_{ij}^\text{pre} \right)^2 \notag \\
+  ( \hat{\beta}^{\textsc{iv}}_1 - \beta_1)^2 & \frac{1}{n} \sum_{j=1}^n \left( \sum_{i\ne j} A_{ij}^\text{pre} T_i \right)^2 
+ ( \hat{\beta}^{\textsc{iv}}_2 - \beta_2)^2 \frac{1}{n} \sum_{j=1}^n \left( \sum_{i\ne j} A_{ij}^\text{pre} M_i \right)^2.
\end{align*}
Together with \eqref{eq:Aijai}, \eqref{eq:AijMi^2} and $ \hat{\beta}^{\textsc{iv}} - \beta = o_{\mathbb{P}}(1)$, we have 
\begin{equation}
\frac{1}{n (n q_n^{\text{pre}})^2} \sum_{j=1}^n \left( \sum_{i\ne j} A_{ij}^\text{pre} ( \hat{\beta}^{\textsc{iv}} - \beta)^\top X_i \right)^2 
= O_{\mathbb{P}}(1)o_{\mathbb{P}}(1) = o_{\mathbb{P}}(1).
\label{eq:AijXi^2}
\end{equation}
Thus, \eqref{eq:Aijui^2}, \eqref{eq:AijuiAijXi} and \eqref{eq:AijXi^2} together imply that
\begin{align*}
\frac{1}{n (n q_n^{\text{pre}})^2} \sum_{j=1}^n \left( \sum_{i\ne j} A_{ij}^\text{pre} \hat{u}_i^{\textsc{iv}} \right)^2
= \frac{1}{n (n q_n^{\text{pre}})^2} \sum_{j=1}^n E\left[ \left( \sum_{i\ne j} A_{ij}^\text{pre} u_i \right)^2 \right]
+ o_{\mathbb{P}}(1).
\end{align*}
\underline{(3): The consistency of $(2,3)$ element of $\hat{V}^{\textsc{iv}}_{\text{num}}$.}
By decomposition, 
\begin{align}
& \frac{1}{n} \sum_{i=1}^n \sum_{j\ne i} A_{ij}^\text{pre} \hat{u}_i^{\textsc{iv}} \hat{u}_j^{\textsc{iv}}
= \frac{1}{n} \sum_{i=1}^n \sum_{j\ne i} A_{ij}^\text{pre} \left( u_i - ( \hat{\beta}^{\textsc{iv}} - \beta)^\top X_i \right) \left( u_j - ( \hat{\beta}^{\textsc{iv}} - \beta)^\top X_j \right) \notag \\
&= \frac{1}{n} \sum_{i=1}^n \sum_{j\ne i} A_{ij}^\text{pre} u_i u_j 
- \frac{2}{n} \sum_{i=1}^n \sum_{j\ne i} A_{ij}^\text{pre} u_i ( \hat{\beta}^{\textsc{iv}} - \beta)^\top X_j  
+ \frac{1}{n} \sum_{i=1}^n \sum_{j\ne i} A_{ij}^\text{pre} ( \hat{\beta}^{\textsc{iv}} - \beta)^\top X_i ( \hat{\beta}^{\textsc{iv}} - \beta)^\top X_j. 
\label{eq:var_Aijuiuj}
\end{align}
For the first term of \eqref{eq:var_Aijuiuj}, the expectation is 
\begin{align*}
E\left[\frac{1}{n} \sum_{i=1}^n \sum_{j\ne i} A_{ij}^\text{pre} u_i u_j \right] 
= O\left(nq_n^\text{pre} \right), 
\end{align*}
and the variance is 
\begin{align*}
& \V\left(\frac{1}{n} \sum_{i=1}^n \sum_{j\ne i} A_{ij}^\text{pre} u_i u_j \right)
= \frac{1}{n^2} \sum_{i=1}^n  \V\left(\sum_{j\ne i} A_{ij}^\text{pre} u_i u_j \right)
+ \frac{1}{n^2} \sum_{(i,k)} \textup{Cov}\left(\sum_{j\ne i} A_{ij}^\text{pre} u_i u_j, \sum_{j\ne k} A_{kj}^\text{pre} u_k u_j \right) \\
&= \frac{1}{n^2} \sum_{i=1}^n  \sum_{j\ne i} \V\left( A_{ij}^\text{pre} u_i u_j \right)
+ \frac{1}{n^2} \sum_{i=1}^n \sum_{(j,k)} \textup{Cov}\left( A_{ij}^\text{pre} u_i u_j, A_{ik}^\text{pre} u_i u_k \right) \\
&+ \frac{1}{n^2} \sum_{(i,k)} \sum_{j\ne i,k} \textup{Cov}\left(A_{ij}^\text{pre} u_i u_j,  A_{kj}^\text{pre} u_k u_j \right)
+ \frac{1}{n^2} \sum_{(i,j) \ne (k,l)} 
\textup{Cov}\left(A_{ij}^\text{pre} u_i u_j,  A_{kl}^\text{pre} u_k u_l \right) \\
&= \frac{1}{n^2} \sum_{i=1}^n  \sum_{j\ne i} \V\left( A_{ij}^\text{pre} u_i u_j \right)
+ \frac{1}{n^2} \sum_{i=1}^n \sum_{(j,k)} \textup{Cov}\left( A_{ij}^\text{pre} u_i u_j, A_{ik}^\text{pre} u_i u_k \right) 
+ \frac{1}{n^2} \sum_{(i,k)} \sum_{j\ne i,k} \textup{Cov}\left(A_{ij}^\text{pre} u_i u_j,  A_{kj}^\text{pre} u_k u_j \right) \\
&\le C_1 q_n^\text{pre} + C_2 (n q_n^\text{pre}) q_n^\text{pre}
= O\left( n (q_n^\text{pre})^2 \right).
\end{align*}
Therefore, 
\begin{align}
\frac{1}{n^2 q_n^\text{pre}} \sum_{i=1}^n \sum_{j\ne i} A_{ij}^\text{pre} u_i u_j
= \frac{1}{n^2 q_n^\text{pre}} \sum_{i=1}^n \sum_{j\ne i} 
E\left[ A_{ij}^\text{pre} u_i u_j \right]  
+ O_{\mathbb{P}}\left(\frac{1}{\sqrt{n}}\right).
\label{eq:Aijuiuj}
\end{align}
For the second term of \eqref{eq:var_Aijuiuj}, 
\begin{align*}
& \frac{1}{n} \sum_{i=1}^n \sum_{j\ne i} A_{ij}^\text{pre} u_i ( \hat{\beta}^{\textsc{iv}} - \beta)^\top X_j \\
~=~& ( \hat{\beta}^{\textsc{iv}}_0 - \beta_0) \frac{1}{n} \sum_{i=1}^n \sum_{j\ne i} A_{ij}^\text{pre} u_i 
+ ( \hat{\beta}^{\textsc{iv}}_1 - \beta_1) \frac{1}{n} \sum_{i=1}^n \sum_{j\ne i} u_i A_{ij}^\text{pre} T_j
+ ( \hat{\beta}^{\textsc{iv}}_2 - \beta_2) \frac{1}{n} \sum_{i=1}^n \sum_{j\ne i} u_i A_{ij}^\text{pre} M_j.
\end{align*}
It suffices to show that $\frac{1}{n} \sum_{i=1}^n \sum_{j\ne i} A_{ij}^\text{pre} u_i X_j = O_{\mathbb{P}}(nq_n^\text{pre})$.
We bound these terms in $L_2$ norm. First,
\begin{align*}
& E\left[ \left( \frac{1}{n} \sum_{i=1}^n \sum_{j\ne i} A_{ij}^\text{pre} u_i  \right)^2 \right]
= \frac{1}{n^2} \sum_{i=1}^n E\left[ \left( \sum_{j\ne i} A_{ij}^\text{pre} u_i  \right)^2 \right]
+ \frac{1}{n^2} \sum_{(i,j)} 
E\left[ \left( \sum_{k\ne i} A_{ik}^\text{pre} u_i  \right)\left( \sum_{k\ne j} A_{jk}^\text{pre} u_j  \right) \right] \\
&= \frac{1}{n^2} \sum_{i=1}^n E\left[ \left( \sum_{j\ne i} A_{ij}^\text{pre} u_i^2 + \sum_{(j_1,j_2)} A_{ij_1}^\text{pre} A_{ij_2}^\text{pre} u_i^2 \right) \right]
+ \frac{1}{n^2} \sum_{(i,j)} 
E\left[ \sum_{k_1\ne i} \sum_{k_2\ne j} A_{ik_1}^\text{pre} u_i  A_{jk_2}^\text{pre} u_j \right] \\
& \le C_1 q_n^\text{pre} 
+ C_2 n (q_n^\text{pre})^2  
+ C_3 (n q_n^\text{pre})^2,
\end{align*}
which implies that 
\begin{align*}
\frac{1}{n} \sum_{i=1}^n \sum_{j\ne i} A_{ij}^\text{pre} u_i
= O_{\mathbb{P}}\left( nq_n^\text{pre} \right).
\end{align*}
Second, 
\begin{align*}
& E\left[ \left( \frac{1}{n} \sum_{i=1}^n \sum_{j\ne i} A_{ij}^\text{pre} u_i M_j \right)^2 \right]
= \frac{1}{n^2} \sum_{i=1}^n E\left[ \left( \sum_{j\ne i} A_{ij}^\text{pre} u_i M_j \right)^2 \right]
+ \frac{1}{n^2} \sum_{(i,j)} 
E\left[ \left( \sum_{k\ne i} A_{ik}^\text{pre} u_i M_k  \right)\left( \sum_{k\ne j} A_{jk}^\text{pre} u_j M_k \right) \right] \\
&= \frac{1}{n^2} \sum_{i=1}^n E\left[ \left( \sum_{j\ne i} A_{ij}^\text{pre} u_i^2 M_j^2 + \sum_{(j_1,j_2)} A_{ij_1}^\text{pre} A_{ij_2}^\text{pre} u_i^2 M_{j_1} M_{j_2} \right) \right]
+ \frac{1}{n^2} \sum_{(i,j)} 
E\left[ \sum_{k_1\ne i} \sum_{k_2\ne j} A_{ik_1}^\text{pre} u_i M_{k_1} A_{jk_2}^\text{pre} u_j M_{k_2} \right] \\
& \le \frac{1}{n^2} \sum_{i=1}^n E\left[ \left( \sum_{j\ne i} A_{ij}^\text{pre} u_i^2 + \sum_{(j_1,j_2)} A_{ij_1}^\text{pre} A_{ij_2}^\text{pre} u_i^2 \right) \right]
+ \frac{1}{2} \frac{1}{n^2} \sum_{(i,j)} 
E\left[ \sum_{k_1\ne i} \sum_{k_2\ne j} A_{ik_1}^\text{pre} M_{k_1} A_{jk_2}^\text{pre} M_{k_2} (u_i^2 + u_j^2) \right] \\
& \le C_1 q_n^\text{pre} + C_2 (n q_n^\text{pre}) q_n^\text{pre} + C_3 (n q_n^\text{pre})^2,
\end{align*}
which implies that 
\begin{align*}
\frac{1}{n} \sum_{i=1}^n \sum_{j\ne i} A_{ij}^\text{pre} u_i M_j
= O_{\mathbb{P}}\left( nq_n^\text{pre} \right).
\end{align*}
By analogous argument, we can show that 
\begin{align*}
\frac{1}{n} \sum_{i=1}^n \sum_{j\ne i} A_{ij}^\text{pre} u_i T_j
= O_{\mathbb{P}}\left( nq_n^\text{pre} \right).
\end{align*}
In particular, this implies that with consistent estimates, we have
\begin{align}
& \frac{1}{n^2 q_n^\text{pre}} \sum_{i=1}^n \sum_{j\ne i} A_{ij}^\text{pre} u_i ( \hat{\beta}^{\textsc{iv}} - \beta)^\top X_j 
= o_{\mathbb{P}}(1).
\label{eq:AijuiXj}
\end{align}
For the last term of \eqref{eq:Aijuiuj}, it is bounded above by 
\begin{align}
& \frac{1}{n} \sum_{i=1}^n \sum_{j\ne i} A_{ij}^\text{pre} \left( ( \hat{\beta}^{\textsc{iv}} - \beta)^\top X_i \right) \left( ( \hat{\beta}^{\textsc{iv}} - \beta)^\top X_j \right) \notag \\
\le & C \frac{1}{n} \sum_{i=1}^n \sum_{j\ne i} A_{ij}^\text{pre} \left[ ( \hat{\beta}^{\textsc{iv}}_0 - \beta_0)^2 + ( \hat{\beta}^{\textsc{iv}}_1 - \beta_1)^2T_i + ( \hat{\beta}^{\textsc{iv}}_2 - \beta_2)^2 M_i^2 \right]
\label{eq:AijXiXj}
\end{align}
by symmetric.
We bound each term in $L_1$ norm:
\begin{align*}
E\left[ \frac{1}{n} \sum_{i=1}^n \sum_{j\ne i} A_{ij}^\text{pre}  \right] = O\left( n q_n^\text{pre} \right)
\text{ and }
E\left[ \frac{1}{n} \sum_{i=1}^n \sum_{j\ne i} A_{ij}^\text{pre} T_i \right] = O\left( n q_n^\text{pre} \right).
\end{align*}
Also, by the boundedness of $M_i$, we have
\begin{align*}
E\left[ \frac{1}{n} \sum_{i=1}^n \sum_{j\ne i} A_{ij}^\text{pre} M_i^2 \right] 
\le E\left[ \frac{1}{n} \sum_{i=1}^n \sum_{j\ne i} A_{ij}^\text{pre} \right] 
= O\left( n q_n^\text{pre} \right).
\end{align*}
Therefore,
\begin{equation}
\frac{1}{n} \sum_{i=1}^n \sum_{j\ne i} A_{ij}^\text{pre} \left( ( \hat{\beta}^{\textsc{iv}} - \beta)^\top X_i \right) \left( ( \hat{\beta}^{\textsc{iv}} - \beta)^\top X_j \right)
= O_{\mathbb{P}}\left( n q_n^\text{pre} \right). 
\label{eq:AijXiXj}
\end{equation}
Thus, \eqref{eq:Aijuiuj}, \eqref{eq:AijuiXj} and \eqref{eq:AijXiXj} together imply that with consistent estimators, we have 
\begin{align*}
\frac{1}{n^2 q_n^\text{pre}} \sum_{i=1}^n \sum_{j\ne i} A_{ij}^\text{pre} \hat{u}_i^{\textsc{iv}} \hat{u}_j^{\textsc{iv}}
= \frac{1}{n^2 q_n^\text{pre}} \sum_{i=1}^n \sum_{j\ne i} E \left[ A_{ij}^\text{pre} u_i u_j \right] + o_{\mathbb{P}}(1).
\end{align*}

\paragraph{Step 3}
Steps 1--2 and CMT together imply the desired result:
\begin{align*}
& \left(\hat{V}^{\textsc{iv}}_{\text{num}} \right)^{-1/2}
\left( \sum_{i=1}^n {Z}_i u_i \right)
\overset{d}{\to} \mathcal{N}(0, I_3).
\end{align*}
By Theorem \ref{thm:ZX-1}, we have shown the probability limit of ${Z}^\top X$. By CMT again, we can show 
\[
\left( \hat{V}^{\textsc{iv}} \right)^{-1/2} 
 \left( \hat{\beta}^{\textsc{iv}}-\beta \right) 
\stackrel{\textup{d}}{\rightarrow} \mathcal{N}(0, I_3).
\]

\end{proof}

\section{Proof of results in Section \ref{sec:IV_PC}} 
\label{app:PC_IV}

\subsection{Useful lemmas}
By eigenvalue decomposition, we have $q_n^\text{pre} G_{n}^\text{pre} = \sum_{k=1}^n \lambda_k^* \psi_{k}^{*\top} \psi_{k}^*$.
In Assumption \ref{ass:low_rank}, we assume that the graphon is approximately low rank, i.e., can be approximated by the $r$ leading terms $q_n^\text{pre} \tilde{G}_{n}^\text{pre} = \sum_{k=1}^r \lambda_k^* \psi_{k}^{*\top} \psi_{k}^*$. 
We show that eigenvectors of $A^\text{pre}$, $(\hat{\psi}_{k})_{k=1}^n $, is close to 
$( {\psi}_{k}^*)_{k=1}^n $.

\begin{lemma} %
\label{lemma:25}
Suppose $q_n \succ {\frac{\log n}{\log \log n}} / n$.
Under Assumptions \ref{asu:network} and \ref{ass:low_rank}, we have
\begin{enumerate}[(1)]
\item 
$\left\|A^{\text{pre}} - q_n^{\text{pre}} G_{n}^\text{pre}\right\|_{\textup{op}}
= O_{\mathbb{P}}\left( \sqrt{n q_n^{\text{pre}}}  \left( {\frac{\log n}{\log \log n}} \right)^{1/4} \right)$;
\item 
$\left\| A^{\text{pre}} - q_n^\text{pre} \tilde{G}_{n}^\text{pre} \right\|_{\textup{op}}
= O_{\mathbb{P}}\left( \sqrt{n q_n^{\text{pre}}}  \left( {\frac{\log n}{\log \log n}} \right)^{1/4} \right)$.
\end{enumerate}
\end{lemma}

\begin{proof}[Proof of Lemma \ref{lemma:25}]
We first show result (1).
This is analogous to Lemma 25 in \cite{LiWager2022}.
Suppose $q_n^\text{pre} \succcurlyeq \log(n)/n$. 
By Theorem 5.2 in \cite{LeiRinaldo2015} with $d = n q_n^\text{pre}$, there exists some constant $C$ such that 
$\|A^{\text{pre}} - q_n^\text{pre} {G}_{n}^\text{pre}\|_{\textup{op}} \leq C \sqrt{n q_n^\text{pre}}$ with probability approaching to 1.
Suppose instead ${\frac{\log n}{\log \log n}} / n \prec q_n^\text{pre} \prec \log(n)/n$.
By Corollary 3.3 in \cite{Benaych-GeorgesBordenaveKnowles2020}, we can show that by setting their 
\[
\varepsilon^2 = \sqrt{\frac{\log n}{\log \log n}} / (n q_n^{\text{pre}}),
\]
we have that with probability approaching to 1, 
\[
\left\|A^{\text{pre}} - q_n^{\text{pre}} G_{n}^\text{pre}\right\|_{\textup{op}} 
\leq k \sqrt{n q_n^{\text{pre}}}  \left( {\frac{\log n}{\log \log n}} \right)^{1/4},
\]
where $k$ is a universal constant.
For result (2), by triangle inequality, we have
\begin{align*}
& \left\| A^{\text{pre}} - q_n^\text{pre} \tilde{G}_{n}^\text{pre}  \right\|_{\textup{op}}
\le \left\|A^{\text{pre}} - q_n^\text{pre} G_{n}^\text{pre} \right\|_{\textup{op}}
+ q_n^\text{pre} \left\| G_{n}^\text{pre} - \tilde{G}_{n}^\text{pre} \right\|_{\textup{op}} \\
& = \left\|A^{\text{pre}} - q_n^\text{pre} G_{n}^\text{pre} \right\|_{\textup{op}}
+ \left\| \sum_{k=r+1}^n \lambda_k^* \psi_{k}^{*\top} \psi_{k}^* \right\|_{\textup{op}} 
= O_{\mathbb{P}}\left( \sqrt{n q_n^{\text{pre}}}  \left( {\frac{\log n}{\log \log n}} \right)^{1/4}  \right).
\end{align*}
\end{proof}

\begin{lemma}\label{lemma:27}
Suppose $q_n^\text{pre} \succ {\frac{\log n}{\log \log n}} / n$.
Under Assumptions \ref{asu:network} and \ref{ass:low_rank}, for $k \in \{1,\cdots,r\}$, 
then 
\begin{enumerate}[(1)]
\item $\left| \hat{\lambda}_k - {\lambda}_k^* \right| = O_{\mathbb{P}}\left( \sqrt{n q_n^{\text{pre}}}  \left( {\frac{\log n}{\log \log n}} \right)^{1/4} \right)$;
\item $ \hat{\lambda}_k = \Omega_{\mathbb{P}}\left( n q_n^\text{pre} \right)$.
\end{enumerate}
\end{lemma}

\begin{proof}[Proof of Lemma \ref{lemma:27}]
This is analogous to Lemma 27 in \cite{LiWager2022}.
For (1), by Weyl's Inequality,
\[
\left| \hat{\lambda}_k - {\lambda}_k^* \right|
\le \left\|A^{\text{pre}} - q_n^\text{pre} G_{n}^\text{pre} \right\|_{\textup{op}} 
= O_{\mathbb{P}}\left( \sqrt{n q_n^{\text{pre}}}  \left( {\frac{\log n}{\log \log n}} \right)^{1/4} \right) 
\]
by Lemma \ref{lemma:25}. The result (2) follows from (1) and Assumption \ref{ass:low_rank}(a).
\end{proof}

Below we recall the statement of the Davis--Kahan theorem as given in \cite{YuWangSamworth2015}.
\begin{lemma}[Theorem 2 in \cite{YuWangSamworth2015}] \label{thm:Davis-Kahan}
Let $\Sigma $ and $\hat{\Sigma} \in \mathbb{R}^{p \times p}$ be symmetric, with eigenvalues $\lambda_1 \geq \ldots \geq \lambda_p$ and $\hat{\lambda}_1 \geq \ldots \geq \hat{\lambda}_p$, respectively. Fix $1 \leq t \leq s \leq p$ and assume that
$\min \left(\lambda_{t-1}-\lambda_t, \lambda_s-\lambda_{s+1}\right)>0$ where we define $\lambda_0=\infty$ and $\lambda_{p+1}=-\infty$. Let $d=s-t+1$, and let $V=\left(v_t, v_{t+1}, \ldots, v_s\right) \in \mathbb{R}^{p \times d}$ and $\hat{V}=\left(\hat{v}_t, \hat{v}_{t+1}, \ldots, \hat{v}_s\right) \in \mathbb{R}^{p \times d}$ have orthonormal columns satisfying $\Sigma v_j=\lambda_j v_j$ and $\hat{\Sigma} \hat{v}_j=\hat{\lambda}_j \hat{v}_j$ for $j = t, t+1, \ldots, \mathrm{s}$. Then there exists an orthogonal matrix $\hat{O} \in \mathbb{R}^{d \times d}$ such that
\begin{align}
\|\hat{V} \hat{O} - V\|_{\textsc{f}} \leq \frac{2^{3 / 2} \min \left(d^{1 / 2}\|\hat{\Sigma}-\Sigma\|_{\textup{op}},
\|\hat{\Sigma} - \Sigma\|_{\textsc{f}}\right)}{\min \left(\lambda_{t-1} - \lambda_t, \lambda_s - \lambda_{s+1}\right)}.
\label{eq:VO-V}
\end{align}    
\end{lemma}
Specifically, let $\hat{V}^\top V=O_1 D O_2^\top$ be the singular value decomposition of $\hat{V}^\top V$, then $\hat{O}$ is constructed by taking $\hat{O}=O_1 O_2^\top$.

\begin{lemma} %
\label{lemma:8}
Suppose $q_n^\text{pre} \succ {\frac{\log n}{\log \log n}} / n$.
Let $\psi_{k}^*$ denote the vector of $\psi_{k}^*(w_i)$. There exists an $r \times r$ orthogonal matrix $\hat{R}$, where $r$ is defined in Assumption \ref{ass:low_rank}, such that if we write $\hat{\Psi}^R=\hat{\Psi} \hat{R}$, and let $\hat{\psi}_{k}^R$ be the $k$-th column of $\hat{\Psi}^R$ for $k\le r$, then under Assumptions \ref{asu:network} and \ref{ass:low_rank}, 
we have
\begin{align}
\left( {\psi}_{l}^{*}\right)^\top
\left(\hat{\psi}_{k}^R-{\psi}_{k}^{*}\right)
=& O_{\mathbb{P}}\left( \sqrt{ \frac{\log n}{\log \log n}} / (n q_n^{\text{pre}}) \right), \text{ for } l = 1,\cdots, n
\label{eq:phi_hat_psi_l} 
\end{align}
\end{lemma}

\begin{proof}[Proof of Lemma \ref{lemma:8}]
We largely follow the proof of Lemma 8 in \cite{LiWager2022}.
By applying \eqref{eq:VO-V} to $\hat{\Psi}$ and $\Psi^*$, and together with Lemma \ref{lemma:25} and Lemma \ref{thm:Davis-Kahan}, we get that there exists an $r \times r$ orthogonal matrix $\hat{R}$ such that 
\begin{align}
\left\|\hat{\Psi} \hat{R} - \Psi^* \right\|_{\textsc{f}} 
&
= O_{\mathbb{P}}\left( \frac{1}{\sqrt{n q_n^{\text{pre}}}} \left( {\frac{\log n}{\log \log n}} \right)^{1/4} \right).
\label{eq:hat_star}
\end{align}
Define $\hat{\Psi}^R = \hat{\Psi} \hat{R}$. 
Let $\hat{\psi}_{k}^R$ be the $k$-th column of $\hat{\Psi}^R$ and $ \psi_{k}^{*}$ be the $k$-th column of $\Psi^{*}$. 
Then we have
\begin{align}
\left\|\hat{\psi}_{k}^R - \psi_{k}^{*}\right\|
= O_{\mathbb{P}}\left( \frac{1}{\sqrt{n q_n^{\text{pre}}}} \left( {\frac{\log n}{\log \log n}} \right)^{1/4} \right).
\label{eq:68}
\end{align}
For $l \le r$, we will show that $\left\|\left({\Psi}^{*}\right)^\top\left(\hat{\Psi}^R-{\Psi}^{*}\right)\right\|_{\textsc{f}}$ is small. 
Note that the construction of $\hat{O}$ in \eqref{eq:VO-V} ensures that $(\hat{V} \hat{O})^\top {V}$ is symmetric, i.e., 
$(\hat{V} \hat{O})^\top V = O_2 O_1^\top O_1 D O_2^\top=O_2 D O_2^\top$. This implies that $\left(\hat{\Psi}^R\right)^\top {\Psi}^{*}$ is symmetric.
By definition,
\[
\|\hat{V} \hat{O}-V\|_{\textsc{f}}^2 
= \operatorname{tr}\left( \hat{O}^\top \hat{V}^\top \hat{V} \hat{O} - (\hat{V} \hat{O})^\top V - V^\top \hat{V} \hat{O} + V^\top V  \right)
= 2 \operatorname{tr}(I_r-D).
\]
If we write $\left(\hat{\Psi}^R\right)^\top {\Psi}^{*} = O_{\Psi} D_{\Psi} O_{\Psi}^\top$, then by \eqref{eq:hat_star}, we have
\[
\left\|\hat{\Psi} \hat{R} - \Psi^* \right\|_{\textsc{f}}^2 
= \operatorname{tr}\left(I_r-D_{\Psi}\right)
= O_{\mathbb{P}}\left( \frac{1}{n q_n^{\text{pre}}} \sqrt{{\frac{\log n}{\log \log n}}}  \right).
\]
Note that 
\[
\left({\Psi}^{*}\right)^\top\left(\hat{\Psi}^R-{\Psi}^{*}\right)
= \left({\Psi}^{*}\right)^\top \hat{\Psi}^R - I_r
= \left(O_{\Psi} D_{\Psi} O_{\Psi}^\top-I_r\right).
\]
Therefore,
$$
\left\|\left({\Psi}^{*}\right)^\top
\left(\hat{\Psi}^R-{\Psi}^{*}\right)\right\|_{\textsc{f}}^2
= \operatorname{tr}\left(\left(O_{\Psi} D_{\Psi} O_{\Psi}^\top-I_r \right)^2\right)
= \operatorname{tr}\left(\left(D_{\Psi}-I_r \right)^2\right)
= O_{\mathbb{P}}\left( \left( \sqrt{\frac{\log n}{\log \log n}} / (n q_n^\text{pre}) \right)^{2} \right).
$$
It implies that for any $k, l \in\{1, \ldots, r\}$, we have
\begin{align*}
\left( {\psi}_{l}^{*}\right)^\top
\left(\hat{\psi}_{k}^R-{\psi}_{k}^{*}\right)
= O_{\mathbb{P}}\left( \frac{1}{n q_n^{\text{pre}}} \sqrt{{\frac{\log n}{\log \log n}}} \right).
\end{align*}
For $l > r$, we have 
$\left( {\psi}_{l}^{*}\right)^\top \left(\hat{\Psi}^{R}-{\Psi}^{*}\right) 
= \left({\psi}_{l}^{*}\right)^\top \hat{\Psi}^R$ by orthogonality. Note that 
\begin{align*}
\left\|\left( {\psi}_{l}^{*}\right)^\top \hat{\Psi}_r^R\right\|
= \left\|\left( {\psi}_{l}^{*}\right)^\top \hat{\Psi}_r \hat{R} \right\|
= \left\|\left( {\psi}_{l}^{*}\right)^\top \hat{\Psi}_r \hat{\Lambda}_r \hat{R} \hat{R}^\top \hat{\Lambda}_r^{-1} \hat{R}\right\|
\end{align*}
where $\hat{\Lambda}_r$ is the $r \times r$ diagonal matrix with $\hat{\lambda}_1, \ldots, \hat{\lambda}_r$ on its diagonal.
We complete the proof by bounding $\left\|\left( {\psi}_{l}^{*}\right)^\top \hat{\Psi}_r \hat{\Lambda}_r \hat{R} \right\|$. 
Note that
\begin{align*}
\left( {\psi}_{l}^{*}\right)^\top \hat{\Psi}_r \hat{\Lambda}_r \hat{R}
=& \left({\psi}_{l}^{*}\right)^\top {A}^{\text{pre}} \hat{\Psi} \hat{R}
= \left({\psi}_{l}^{*}\right)^\top {A}^{\text{pre}} \hat{\Psi}^R_r \\
=& \left({\psi}_{l}^{*}\right)^\top \left( {A}^{\text{pre}} - q_n^\text{pre} \tilde{G}_{n}^\text{pre} \right) \hat{\Psi}^R_r + \left({\psi}_{l}^{*}\right)^\top q_n^\text{pre} \tilde{G}_{n}^\text{pre} \hat{\Psi}^R_r \\
=& \left({\psi}_{l}^{*}\right)^\top \left( {A}^{\text{pre}} - q_n^\text{pre} \tilde{G}_{n}^\text{pre} \right) (\hat{\Psi}_r^R - {\Psi}_r^{*})
+ \left({\psi}_{l}^{*}\right)^\top \left( {A}^{\text{pre}} - q_n^\text{pre} \tilde{G}_{n}^\text{pre} \right) {\Psi}_r^{*}.
\end{align*}
The first term can be easily bounded by
\begin{align}
\left\|\left( {\psi}_{l}^{*}\right)^\top
\left( {A}^{\text{pre}} - q_n^\text{pre} \tilde{G}_{n}^\text{pre} \right)
\left(\hat{\Psi}_r^{R} - {\Psi}_r^{*}\right)\right\| 
&\leq 
\left\| {\psi}_{l}^{*}\right\| \left\| {A}^{\text{pre}} - q_n^\text{pre} \tilde{G}_{n}^\text{pre} \right\|_{\textup{op}}
\left\|\hat{\Psi}_r^R-{\Psi}_r^{*}\right\| \notag \\ 
&= O_{\mathbb{P}}\left( \sqrt{ \frac{\log n}{\log \log n}} \right),
\label{eq:lem8_term1}
\end{align}
where the last inequality follows from Lemma \ref{lemma:25} and \eqref{eq:hat_star}. 
For the second term, we rewrite it in the summation form: 
\begin{align*}
 \left( {\psi}_{l}^{*}\right)^\top 
 \left( {A}^{\text{pre}} - q_n^{\text{pre}} \tilde{G}_{n}^{\text{pre}} \right) {\psi}_{k}^{*}
=& \left( {\psi}_{l}^{*}\right)^\top 
\left( {A}^{\text{pre}} - q_n^\text{pre} {G}_{n}^\text{pre} \right) {\psi}_{k}^{*} 
+ \left( {\psi}_{l}^{*}\right)^\top \left( \sum_{h=r+1}^n \lambda_h^* \psi_{h}^{*\top} \psi_{h}^* \right) {\psi}_{k}^{*} \\
=& \sum_{(i, j), i \neq j} {\psi}_{l i}^{*} {\psi}_{k j}^{*} \left( {A}^{\text{pre}}_{i j} - q_n^\text{pre} g^{\text{pre}}(w_i,w_j)\right). 
\end{align*}
As $A^{\text{pre}}_{ij}$ 's are independent given $w$, we have
\begin{align}
& E\left[\left(\sum_{(i, j), i \neq j} {\psi}_{l i}^{*} {\psi}_{k j}^{*} \left( {A}^{\text{pre}}_{i j} - q_n^\text{pre} {g}_0(i,j)\right) \right)^2\right] \notag \\
= & 2 \sum_{(i, j), i \neq j} 
\mathbb{E}\left[\left( {\psi}_{l i}^{*}\right)^2
\left( {\psi}_{k j}^{*} \right)^2
\left(A^{\text{pre}}_{i j} - q_n^\text{pre} g_0(i,j)\right)^2 \right]
\le 2 \sum_{(i, j), i \neq j} \mathbb{E}\left[\left( {\psi}_{l i}^{*}\right)^2\left( {\psi}_{k j}^{*} \right)^2 q_n^\text{pre} g_0(i,j) \right] \notag \\
\leq & 2 q_n^\text{pre} \mathbb{E}\left[ \sum_{(i, j), i \neq j}  \left( {\psi}_{l i}^{*}\right)^2\left( {\psi}_{k j}^{*} \right)^2\right] 
\leq C q_n^\text{pre}.
\label{eq:lem8_term2}
\end{align}
Combining \eqref{eq:lem8_term1} and \eqref{eq:lem8_term2}, we can show that
$\left\|\left( {\psi}_{l}^{*} \right)^\top \hat{\Psi}_r \hat{\Lambda}_r \hat{R}\right\|
= O_{\mathbb{P}}\left( \sqrt{ \frac{\log n}{\log \log n}} \right)$.
Lemma \ref{lemma:27} shows that 
$\hat{\lambda}_k = \Omega_{\mathbb{P}}\left( n q_n^\text{pre} \right)$ for $k \leq r$. 
Thus
\begin{align*}
\left\|\left( {\psi}_{l}^{*}\right)^\top \hat{\Psi}_r^R\right\|
= \left\|\left( {\psi}_{l}^{*}\right)^\top \hat{\Psi}_r \hat{R} \right\|
= \left\|\left( {\psi}_{l}^{*}\right)^\top \hat{\Psi}_r \hat{\Lambda}_r \hat{R} \hat{R}^\top \hat{\Lambda}_r^{-1} \hat{R}\right\|
= O_{\mathbb{P}}\left( \sqrt{ \frac{\log n}{\log \log n}} / (n q_n^{\text{pre}}) \right).
\end{align*}
It implies that 
for any $k \in\{1, \ldots, r\}$ and $l > r$, we have
\begin{align}
\left( {\psi}_{l}^{*}\right)^\top
\left(\hat{\psi}_{k}^R-{\psi}_{k}^{*}\right)
= O_{\mathbb{P}}\left( \sqrt{ \frac{\log n}{\log \log n}} / (n q_n^{\text{pre}}) \right).
\end{align}

\end{proof}

\begin{lemma}\label{lemma:28}
Suppose $q_n^\text{pre} \succ {\frac{\log n}{\log \log n}} / n$.
Define $\hat{\psi}_k^R$ as in Lemma \ref{lemma:8}.
Under Assumptions \ref{asu:network} and \ref{ass:low_rank},
\begin{align*}
\left\|A^{\text{pre}} \left(\hat{\psi}_{k}^R-\psi_{k}^* \right)\right\|
= O_{\mathbb{P}}\left( \sqrt{ \frac{\log n}{\log \log n}} \right).   
\end{align*}
\end{lemma}

\begin{proof}[Proof of Lemma \ref{lemma:28}]
By decomposing the target expression:
\begin{align*}
\left\| A^{\text{pre}} \left(\hat{\psi}_{k}^R - \psi_{k}^*\right)\right\| 
& \leq \left\|q_n^\text{pre} \tilde{G}_{n}^\text{pre} \left(\hat{\psi}_{k}^R-\psi_{k}^*\right)\right\|
+ \left\|\left( A^{\text{pre}} - q_n^\text{pre} \tilde{G}_{n}^\text{pre} \right)\left(\hat{\psi}_{k}^R-\psi_{k}^* \right)\right\| \\
& \leq 
\left\|\sum_{l=1}^r {\lambda}_l^* {\psi}_{l}^* {\psi}_{l}^{*\top}
\left(\hat{\psi}_{k}^R-\psi_{k}^*\right) \right\|
+ \left\|A^{\text{pre}} - q_n^\text{pre} \tilde{G}_{n}^\text{pre}\right\|_{\textup{op}}
\left\|\hat{\psi}_{k}^R-\psi_{k}^* \right\| \\
&= O_{\mathbb{P}}\left( \sqrt{ \frac{\log n}{\log \log n}} \right) 
\end{align*}
by Lemma \ref{lemma:8}, Lemma \ref{lemma:25} and \eqref{eq:68}.    
\end{proof}

\begin{lemma}\label{lemma:29}
Suppose $q_n^\text{pre} \succ {\frac{\log n}{\log \log n}} / n$.
Define $\hat{R}$ and $\hat{\Psi}^R$ as in Lemma \ref{lemma:8} and $\hat{\gamma}^R=\hat{R}^\top \hat{\gamma}$. 
Under Assumptions \ref{asu:network} and \ref{ass:low_rank}, we have
\begin{equation*}
\sum_{k=1}^r \hat{\gamma}_k \hat{\psi}_{k i}
= \sum_{k=1}^r \hat{\gamma}_k^R \hat{\psi}_{k i}^R 
\text { and } 
\hat{\gamma}_k^R
= O_{\mathbb{P}}\left( n q_n^\text{pre} \right).
\end{equation*}
\end{lemma}

\begin{proof}[Proof of Lemma \ref{lemma:29}]
We follow the proof of Lemma 29 in \citet{LiWager2022}.
The first result holds by construction.
Recall the closed form of $\hat{\gamma}_k$:
$$
\hat{\gamma}_k
= - \frac{ \sum_{i=1}^n \hat{\psi}_{k i} Z_i^{\textsc{ssiv}} }{ \sum_{i=1}^n \hat{\psi}_{k i}^2 } 
= - \sum_{i=1}^n \hat{\psi}_{k i} Z_i^{\textsc{ssiv}}.
$$
Multiply $\hat{R}$ on both hand side, we can obtain
$$
\hat{\gamma}_k^R
= - \sum_{i=1}^n \hat{\psi}_{k i}^R Z_i^{\textsc{ssiv}}
= - \sum_{i=1}^n \psi_{k}^*(w_i)Z_i^{\textsc{ssiv}}
- \sum_{i=1}^n \left(\hat{\psi}_{k i}^R-\psi_{k}^*(w_i)\right) Z_i^{\textsc{ssiv}}.
$$
For the first term of $\hat{\gamma}_k^R$,
we bound it in $L_2$ norm:
\begin{align*}
& \operatorname{E}\left[ \left( \sum_{i=1}^n \psi_{k}^*(w_i) Z_i^{\textsc{ssiv}} \right)^2 \right]
= \operatorname{E}\left[ \left( \sum_{(i, j), i \neq j} \psi_{k}^*(w_i) A_{i j}^{\text{pre}} \left(T_j - \pi\right) \right)^2 \right] \\
&= \operatorname{E}\left[ \sum_{j=1}^n \left(T_j - \pi\right)^2 
\left( \sum_{i=1}^n  \psi_{k}^*(w_i)^2  A_{i j}^{\text{pre}} + \sum_{(i_1,i_2)}  \psi_{k}^*(w_{i_1}) \psi_{k}^*(w_{i_2}) A_{i_1 j}^{\text{pre}} A_{i_2 j}^{\text{pre}} \right) \right] \\
&= \operatorname{E}\left[ \sum_{j=1}^n \left(T_j - \pi\right)^2 
\left( \sum_{i=1}^n  \psi_{k}^*(w_i)^2  q_n^\text{pre} g_0(i,j) + \sum_{(i_1,i_2)}  \psi_{k}^*(w_{i_1}) \psi_{k}^*(w_{i_2}) q_n^\text{pre} g_0({i_1},j) q_n^\text{pre} g_0({i_2},j) \right) \right] \\
&= \operatorname{E}\left[ \sum_{j=1}^n \left(T_j - \pi\right)^2 
\left( \sum_{i=1}^n  \psi_{k}^*(w_i)^2  q_n^\text{pre} g_0(i,j) + \sum_{(i_1,i_2)}  \psi_{k}^*(w_{i_1}) \psi_{k}^*(w_{i_2}) \lambda_k^* \psi_{k}^*(w_{i_1}) \psi_{k}^*(w_{j}) \lambda_k^* \psi_{k}^*(w_{i_2}) \psi_{k}^*(w_{j}) \right) \right] \\
&= \pi(1-\pi) \sum_{(i, j), i \neq j}
\mathbb{E}\left[\psi_{k}^*(w_i)^2 q_n^\text{pre} g_0(i,j) \right] 
+ \pi(1-\pi) \left( \lambda_k^* \right)^2 \\
& \leq C_1 n q_n^\text{pre} + C_2  (n q_n^\text{pre})^2.
\end{align*}
This implies that
\begin{equation}
\sum_{i=1}^n \psi_{k}^*(w_i) Z_i^{\textsc{ssiv}} 
= O_{\mathbb{P}}\left(n q_n^\text{pre} \right).  
\label{eq:psi_Zi}
\end{equation}
For the second term of $\hat{\gamma}_k$, we bound it in $L_2$ norm:
\begin{align*}
& E\left[ \left( \sum_{j=1}^n (T_j - \pi) \sum_{i\ne j} A_{ij}^{\text{pre}} \left(\hat{\psi}_{k i}^R-\psi_{k}^*(w_i)\right)  \right)^2 \right] \\
=& \pi (1-\pi) E\left[ \sum_{j=1}^n  \left(  \sum_{i\ne j} A_{ij}^{\text{pre}} \left(\hat{\psi}_{k i}^R-\psi_{k}^*(w_i)\right)  \right)^2 \right] 
= \pi (1-\pi) E\left[ \left\| A^{\text{pre}} \left(\hat{\psi}_{k}^R-\psi_{k}^*(w)\right)  \right\|^2 \right] \\
=& O_{\mathbb{P}}\left( \frac{\log n}{\log \log n} \right)
\end{align*}
by Lemma \ref{lemma:28}. 
Therefore,
\begin{equation}
\sum_{i=1}^n \left(\hat{\psi}_{k i}^R-\psi_{k}^*(w_i)\right) Z_i^{\textsc{ssiv}} 
= O_{\mathbb{P}}\left( \sqrt{ \frac{\log n}{\log \log n}} \right).
\label{eq:psi_hat_Zi}
\end{equation}
Combining \eqref{eq:psi_Zi} and \eqref{eq:psi_hat_Zi}, we get $\hat{\gamma}_k^R
= O_{\mathbb{P}}\left( n q_n^\text{pre} \right)$.
\end{proof}

Define $\phi_i$ as any i.i.d. random variable with constant variance.
Define $\mu_{n,k}^{\phi} = \sum_{i=1}^n \phi_i \psi_{k}^*(w_i)$ and $\eta_i^{\phi} = \phi_i - \sum_{k=1}^r \mu_{n,k}^{\phi} \psi_{k}^*(w_i)$. By definition, we have 
\begin{align*}
\sum_{i=1}^n \eta_i^{\phi} \psi_{k}^*(w_i)
= \sum_{i=1}^n \left( \phi_i - \sum_{l=1}^r \mu_{n,l}^{\phi} \psi_{l}^*(w_i) \right) \psi_{k}^*(w_i)
= 0.
\end{align*}
Also, $\mu^{\phi}_{n,l} = O_{\mathbb{P}}(\sqrt{n})$ which follows by 
\begin{align*}
E\left( (\mu^{\phi}_{n,l})^2 \right)
=& E\left( \left(\sum_{i=1}^n \phi_i \psi_{li}^* \right)^2 \right)
= \sum_{i=1}^n E\left( \phi_i^2 (\psi_{li}^*)^2 \right)
+ \sum_{(i,j)} E\left( \phi_i \psi_{li}^* \phi_j \psi_{lj}^* \right) 
= O(n).
\end{align*}

\begin{lemma} \label{lemma:Ziphii_PC}
Suppose $q_n^{\text{pre}} \succ \frac{\log n}{\log \log n} / n$.
Under Assumptions \ref{asu:network} and \ref{ass:low_rank}, we have
\begin{align}
\frac{1}{n}\sum_{i=1}^n (Z_i^{\textsc{ssiv}} + \hat{\delta}_i) \phi_i
=& \frac{1}{n}\sum_{i=1}^n Z_i^{\textsc{ssiv}} \eta_i^{\phi} 
+ o_{\mathbb{P}}\left( \sqrt{q_n^{\text{pre}}} \right).
\label{eq:Ziphii_PC} 
\end{align}

\end{lemma}

\begin{proof}[Proof of Lemma \ref{lemma:Ziphii_PC}]
We rewrite it as 
\begin{align*}
& \frac{1}{n}\sum_{i=1}^n (Z_i^{\textsc{ssiv}} + \hat{\delta}_i) \phi_i
= \frac{1}{n}\sum_{i=1}^n \left( Z_i^{\textsc{ssiv}} + \sum_{k=1}^r \hat{\gamma}_k \hat{\psi}_{ki} \right) \phi_i \\
=& \frac{1}{n}\sum_{i=1}^n 
\left( Z_i^{\textsc{ssiv}} + \sum_{k=1}^r \hat{\gamma}_k \hat{\psi}_{ki} \right) 
\left( \sum_{l=1}^r \mu_{n,l}^\phi \psi_{l}^*(w_i) + \eta_i^\phi \right) \\
=& \frac{1}{n}\sum_{i=1}^n \left( Z_i^{\textsc{ssiv}} + \sum_{k=1}^r \hat{\gamma}_k^R \hat{\psi}_{ki}^R \right) \eta_i^\phi
+ \frac{1}{n} \sum_{l=1}^r \mu_{n,l}^\phi \sum_{i=1}^n 
\left( Z_i^{\textsc{ssiv}} + \sum_{k=1}^r \hat{\gamma}_k^R \hat{\psi}_{ki}^R \right) \left( \psi_{l}^*(w_i) - \hat{\psi}_{li}^R  \right)  \\
=& S_1 + S_2,
\end{align*}
where the last to second inequality holds by the construction of $\hat{\gamma}_k$. 
For the first part of $S_1$, $S_{11} = \frac{1}{n} \sum_{i=1}^n Z_i^{\textsc{ssiv}} \eta_i^\phi$, the expectation is zero, and the variance is
\begin{align*}
& \V(S_{11})
= \V\left( \frac{1}{n}\sum_{i=1}^n \sum_{j\ne i}A_{ij}^{\text{pre}} (T_j - \pi) \eta_i^\phi \right)  \\
=& \frac{1}{n^2} \sum_{(i,j),i\ne j} 
\V\left( A_{ij}^{\text{pre}} (T_j - \pi) \eta_i^\phi \right)
+ \frac{1}{n^2} \sum_{(i,j),i\ne j}
\textup{Cov}\left(A_{ij}^{\text{pre}} (T_j - \pi) \eta_i^\phi, A_{ji}^{\text{pre}} (T_i - \pi) \eta_j^\phi \right)  \\
+& \frac{1}{n^2} \sum_{\substack{(i,j,k) \\ \text{all distinct}}}
\textup{Cov}\left(A_{ij}^{\text{pre}} (T_j - \pi) \eta_i^\phi, A_{ik}^{\text{pre}} (T_k - p) \eta_i^\phi \right) 
+ \frac{1}{n^2} \sum_{\substack{(i,j,k) \\ \text{all distinct}}}
\textup{Cov}\left(A_{ij}^{\text{pre}} (T_j - \pi) \eta_i^\phi, A_{kj}^{\text{pre}} (T_j - \pi) \eta_k^\phi \right)  \\
+& \frac{1}{n^2} \sum_{\substack{(i,j,k,l) \\ \text{all distinct}}} 
\textup{Cov}\left(A_{ij}^{\text{pre}} (T_j - \pi) \eta_i^\phi, A_{kl}^{\text{pre}} (T_l - p) \eta_k^\phi \right) \\
=&  \frac{1}{n^2} \sum_{(i,j),i\ne j} 
\V\left( A_{ij}^{\text{pre}} (T_j - \pi) \eta_i^\phi \right)  
= O\left( q_n^\text{pre} \right),
\end{align*}
where by construction, we have
\begin{align*}
& \frac{1}{n^2} \sum_{\substack{(i,j,k) \\ \text{all distinct}}}
\textup{Cov}\left(A_{ij}^{\text{pre}} (T_j - \pi) \eta_i^\phi, A_{kj}^{\text{pre}} (T_j - \pi) \eta_k^\phi \right) \\
=& \pi (1-\pi) \frac{1}{n^2} \sum_{j=1}^n
(q_n^\text{pre})^2 \sum_{l_1=1}^n \sum_{l_2=1}^n \lambda_{l_1}^{\phi} \lambda_{l_2}^{\phi} 
E\left(  \psi_{l_1}^*(w_j) \psi_{l_2}^*(w_j) \sum_{i=1}^n \psi_{l_1}^*(w_i) \sum_{k=1}^n \eta_i^\phi \psi_{l_2}^*(w_k) \eta_k^\phi \right)
= 0.
\end{align*}
Therefore, $S_{11} = O_{\mathbb{P}}\left( \sqrt{q_n^\text{pre}} \right)$.

For $S_{12} = \frac{1}{n}\sum_{i=1}^n \left( \sum_{k=1}^r \hat{\gamma}_k^R \hat{\psi}_{ki}^R \right) \eta_i^\phi$, we span $\eta_i^\phi$ onto the subspace of $\{\psi_k \}_{k=1}^n$ as $\eta_i^\phi = \sum_{l=1}^n \alpha_l {\psi}_{l}^*(w_i)$ with $\sum_{l=1}^n \alpha_l^2 = n$. 
By the property of $\sum_{i=1}^n \eta_i^\phi {\psi}_{k}^*(w_i) = 0$ for $k\le r$, we have $\alpha_k = \mu^{\phi}_{n,k}$ for $k\le r$. Thus, $\eta_i^\phi = \sum_{l=r+1 }^n \alpha_l {\psi}_{l}^*(w_i)$. 
We bound the term in $L_2$ norm:
\begin{align*}
& E\left[ \left( \sum_{i=1}^n \eta_i^\phi (\hat{\psi}_{ki}^R - {\psi}_{k}^*(w_i)) \right)^2 \right] 
= E\left[ \left( \sum_{l=r+1 }^n \alpha_l {\psi}_{l}^* (\hat{\psi}_{k}^R - {\psi}_{k}^*) \right)^2 \right] \\
=& E\left[ \sum_{l=r+1 }^n \alpha_l^2 \left( {\psi}_{l}^* (\hat{\psi}_{k}^R - {\psi}_{k}^*) \right)^2 \right]
+ E\left[ \sum_{(l_1,l_2)>r } \alpha_{l_1} \alpha_{l_2} \left( {\psi}_{l_1}^* (\hat{\psi}_{k}^R - {\psi}_{k}^*) \right) \left( {\psi}_{l_2}^* (\hat{\psi}_{k}^R - {\psi}_{k}^*) \right) \right] \\
=& O_{\mathbb{P}}\left( n \left( \sqrt{ \frac{\log n}{\log \log n}} / (n q_n^{\text{pre}}) \right)^{2} \right)
\end{align*}
and thus 
\begin{align*}
\sum_{i=1}^n \eta_i^\phi (\hat{\psi}_{ki}^R - {\psi}_{k}^*(w_i)) 
= O_{\mathbb{P}}\left( \frac{1}{\sqrt{n} q_n^{\text{pre}}} \sqrt{ \frac{\log n}{\log \log n}}  \right).
\end{align*}
Therefore, 
\begin{align*}
S_{12}
=& \sum_{k=1}^r \frac{\hat{\gamma}_k}{n} \sum_{i=1}^n \eta_i^\phi \hat{\psi}_{ki} 
= \sum_{k=1}^r  \frac{\hat{\gamma}_k^R}{n} \sum_{i=1}^n \eta_i^\phi \hat{\psi}_{ki}^R 
= \sum_{k=1}^r  \frac{\hat{\gamma}_k^R}{n} \left[ \sum_{i=1}^n \eta_i^\phi (\hat{\psi}_{ki}^R - {\psi}_{k}^*(w_i)) \right] 
= O_{\mathbb{P}}\left( \sqrt{ \frac{\log n}{\log \log n} / n} \right)
\end{align*}
by Lemma \ref{lemma:8} and Lemma \ref{lemma:29}.
For $S_2$, we have
\begin{align*}
S_2 
= \frac{1}{n} \sum_{l=1}^r \mu_{n,l}^\phi \sum_{i=1}^n 
\left( \psi_{l}^*(w_i) - \hat{\psi}_{li}^R  \right)
\left( Z_i^{\textsc{ssiv}} + \sum_{k=1}^r \hat{\gamma}_k^R \hat{\psi}_{ki}^R \right) 
= S_{21} + S_{22}.
\end{align*}
For $S_{21}$, we have
\begin{align*}
S_{21} 
=& \frac{1}{n} \sum_{l=1}^r \mu_{n,l}^\phi \sum_{i=1}^n 
\left( \psi_{l}^*(w_i) - \hat{\psi}_{li}^R  \right)
Z_i^{\textsc{ssiv}} 
= \frac{1}{n} \sum_{l=1}^r \mu_{n,l}^\phi
\sum_{j=1}^n (T_j - \pi) 
\sum_{i=1}^n 
\left( \psi_{l}^*(w_i) - \hat{\psi}_{li}^R  \right)
 A_{ij}^{\text{pre}}  \\
=& O_{\mathbb{P}}\left( \sqrt{ \frac{\log n}{\log \log n} / n} \right).
\end{align*}
For $S_{22}$, we have
\begin{align*}
S_{22}
=& \frac{1}{n} \sum_{l=1}^r \mu_{n,l}^\phi \sum_{i=1}^n 
\left( \psi_{l}^*(w_i) - \hat{\psi}_{li}^R  \right)
\left( \sum_{k=1}^r \hat{\gamma}_k^R \hat{\psi}_{ki}^R \right)  = \frac{1}{n} \sum_{l=1}^r \sum_{k=1}^r 
\mu_{n,l}^\phi \hat{\gamma}_k^R
\sum_{i=1}^n 
\left( \psi_{l}^*(w_i) - \hat{\psi}_{li}^R  \right)
\hat{\psi}_{ki}^R.
\end{align*}
For each fixed pair of $(k,l)$, we have 
\begin{align*}
& \left| \sum_{i=1}^n 
\left( \psi_{l}^*(w_i) - \hat{\psi}_{li}^R  \right)
\hat{\psi}_{ki}^R \right|
= \left| \left( \psi_{l}^* - \hat{\psi}_{l}^R  \right)^\top 
\hat{\psi}_{k}^R \right| 
= \left| \left( \psi_{l}^* - \hat{\psi}_{l}^R  \right)^\top 
\left( \hat{\psi}_{k}^R - {\psi}_{k}^* \right) 
+ \left( \psi_{l}^* - \hat{\psi}_{l}^R  \right)^\top 
{\psi}_{k}^* \right| \\
\le & \left\| \psi_{l}^* - \hat{\psi}_{l}^R \right\|
\left\| \psi_{k}^* - \hat{\psi}_{k}^R \right\|
+ \left| \left( \psi_{l}^* - \hat{\psi}_{l}^R  \right)^\top 
{\psi}_{k}^* \right|
= O_{\mathbb{P}}\left( \sqrt{ \frac{\log n}{\log \log n}} / (n q_n^{\text{pre}}) \right)
\end{align*}
by \eqref{eq:68} and Lemma \ref{lemma:8}.
Therefore,
\[
S_{22} = 
O_{\mathbb{P}}\left( \frac{\sqrt{n}}{n} n q_n^{\text{pre}} \right)
O_{\mathbb{P}}\left(  \sqrt{ \frac{\log n}{\log \log n}} / (n q_n^{\text{pre}}) \right)
= O_{\mathbb{P}}\left( \sqrt{ \frac{\log n}{\log \log n} / n} \right).
\]
Thus, $S_{11}$ is the leading term 
$O_{\mathbb{P}}(\sqrt{q_n^{\text{pre}}})$, 
and 
\begin{align*}
\frac{1}{n}\sum_{i=1}^n (Z_i^{\textsc{ssiv}} + \hat{\delta}_i) \phi_i
= \frac{1}{n}\sum_{i=1}^n Z_i^{\textsc{ssiv}} \eta_i^\phi + o_{\mathbb{P}}\left( \sqrt{q_n^{\text{pre}}} \right). 
\end{align*}
\end{proof}

Define $\mu_{n,k}^{r} 
= \sum_{i=1}^n (r_{0,i} + r_{1,i} - \mu_{r_1,i}) \psi_{k}^*(w_i)$ and $\eta_i^{r} = (r_{0,i} + r_{1,i} - \mu_{r_1,i}) - \sum_{k=1}^r \mu_{n,k}^{r} \psi_{k}^*(w_i)$.

\begin{lemma}\label{lemma:ZiMi_PC}
Suppose $q_n^\text{pre} \succ {\frac{\log n}{\log \log n}} / n$.
Under Assumptions \ref{asu:network} and \ref{ass:low_rank},
then 
\begin{align}
\frac{1}{n}\sum_{i=1}^n (Z_i^{\textsc{ssiv}} + \hat{\delta}_i) M_i
=& \frac{1}{n}\sum_{i=1}^n Z_i^{\textsc{ssiv}} \eta_i^{r} 
+ O_{\mathbb{P}}\left( \sqrt{q_n^{\text{pre}}} \right). 
\label{eq:ZiMi_PC}
\end{align}   
\end{lemma}

\begin{proof}[Proof of Lemma \ref{lemma:ZiMi_PC}]

By expansion, 
\begin{align*}
\frac{1}{n}\sum_{i=1}^n (Z_i^{\textsc{ssiv}} + \hat{\delta}_i)M_i
= \frac{1}{n}\sum_{i=1}^n (Z_i^{\textsc{ssiv}} + \hat{\delta}_i)(\xi_i + \mu_{r_1,i})
+ \frac{1}{n}\sum_{i=1}^n (Z_i^{\textsc{ssiv}} + \hat{\delta}_i) (r_{0,i} + r_{1,i} - \mu_{r_1,i}).
\end{align*}
\underline{First}, 
by Lemma \ref{lemma:Ziphii_PC}, we can show 
\begin{align*}
\frac{1}{n}\sum_{i=1}^n (Z_i^{\textsc{ssiv}} + \hat{\delta}_i)(\xi_i + \mu_{r_1,i})
= \frac{1}{n}\sum_{i=1}^n Z_i^{\textsc{ssiv}} \eta_i^{(\xi_i + \mu_{r_1,i})} + O_{\mathbb{P}}\left( q_n^{\text{pre}} \right)
= O_{\mathbb{P}}\left( \sqrt{q_n^{\text{pre}}} \right).
\end{align*}
\underline{Second}, we have
\begin{align*}
& \frac{1}{n}\sum_{i=1}^n (Z_i^{\textsc{ssiv}} + \hat{\delta}_i) (r_{0,i} + r_{1,i} - \mu_{r_1,i})
= \frac{1}{n}\sum_{i=1}^n \left( Z_i^{\textsc{ssiv}} + \sum_{k=1}^r \hat{\gamma}_k \hat{\psi}_{ki} \right) (r_{0,i} + r_{1,i} - \mu_{r_1,i}) \\
=& \frac{1}{n}\sum_{i=1}^n 
\left( Z_i^{\textsc{ssiv}} + \sum_{k=1}^r \hat{\gamma}_k \hat{\psi}_{ki} \right) 
\left( \sum_{l=1}^r \mu_{n,l}^r \psi_{l}^*(w_i) + \eta_i^r \right) \\
=& \frac{1}{n}\sum_{i=1}^n \left( Z_i^{\textsc{ssiv}} + \sum_{k=1}^r \hat{\gamma}_k \hat{\psi}_{ki} \right) \eta_i^r
+ \frac{1}{n} \sum_{l=1}^r \mu_{n,l}^r \sum_{i=1}^n 
\left( Z_i^{\textsc{ssiv}} + \sum_{k=1}^r \hat{\gamma}_k \hat{\psi}_{ki} \right) \left( \psi_{l}^*(w_i) - \hat{\psi}_{li}^R  \right)  \\
=& S_1 + S_2.
\end{align*}
For $S_1$, we have $S_1 = S_{11} + S_{12}$.
For $S_{11}$, we have 
\begin{align*}
E(S_{11})
= E\left( \frac{1}{n}\sum_{i=1}^n \sum_{j\ne i}A_{ij}^{\text{pre}} (T_j - \pi) \eta_i^r \right) 
\asymp \frac{\min\{q_n^{\text{pre}}, q_n^\text{post}\}}{q_n^\text{post}}
\end{align*}
and 
\begin{align*}
& \V(S_{11})
= \V\left( \frac{1}{n}\sum_{i=1}^n \sum_{j\ne i}A_{ij}^{\text{pre}} (T_j - \pi) \eta_i^r \right)  \\
=& \frac{1}{n^2} \sum_{(i,j),i\ne j} 
\V\left( A_{ij}^{\text{pre}} (T_j - \pi) \eta_i^r \right)
+ \frac{1}{n^2} \sum_{(i,j),i\ne j}
\textup{Cov}\left(A_{ij}^{\text{pre}} (T_j - \pi) \eta_i^r, A_{ji}^{\text{pre}} (T_i - \pi) \eta_j^r \right)  \\
+& \frac{1}{n^2} \sum_{\substack{(i,j,k) \\ \text{all distinct}}}
\textup{Cov}\left(A_{ij}^{\text{pre}} (T_j - \pi) \eta_i^r, A_{ik}^{\text{pre}} (T_k - p) \eta_i^r \right) 
+ \frac{1}{n^2} \sum_{\substack{(i,j,k) \\ \text{all distinct}}}
\textup{Cov}\left(A_{ij}^{\text{pre}} (T_j - \pi) \eta_i^r, A_{kj}^{\text{pre}} (T_j - \pi) \eta_k^r \right)  \\
+& \frac{1}{n^2} \sum_{\substack{(i,j,k,l) \\ \text{all distinct}}} 
\textup{Cov}\left(A_{ij}^{\text{pre}} (T_j - \pi) \eta_i^r, A_{kl}^{\text{pre}} (T_l - p) \eta_k^r \right) \\
=&  \frac{1}{n^2} \sum_{(i,j),i\ne j} 
\V\left( A_{ij}^{\text{pre}} (T_j - \pi) \eta_i^r \right)  
= O\left( \frac{q_n^{\text{pre}}}{n q_n^\text{post}} \right).
\end{align*}
Therefore, 
\[
S_{11} = O_{\mathbb{P}}\left( \frac{\min\{q_n^{\text{pre}}, q_n^\text{post}\}}{q_n^\text{post}} \right) 
+ O_{\mathbb{P}}\left( \frac{ \sqrt{q_n^{\text{pre}}} }{ \sqrt{n q_n^{\text{post}}} }  \right).
\]
For $S_{12}$, we can show that
\begin{align*}
\frac{1}{n}\sum_{i=1}^n \left( \sum_{k=1}^r \hat{\gamma}_k \hat{\psi}_{ki} \right) \eta_i^r
=& \sum_{k=1}^r \frac{\hat{\gamma}_k}{n} \sum_{i=1}^n \eta_i^r \hat{\psi}_{ki} 
= \sum_{k=1}^r \frac{\hat{\gamma}_k}{n} \left( \sum_{i=1}^n \eta_i^r (\hat{\psi}_{ki} - \psi_{k}^*(w_i))  \right) \\
=& O_{\mathbb{P}}\left( \frac{1}{n\sqrt{q_n^\text{post}}} \sqrt{ \frac{\log n}{\log \log n} / n} \right)
\end{align*}
by Lemma \ref{lemma:8} and Lemma \ref{lemma:29}.
Therefore,
\begin{align*}
S_1
= \frac{1}{n}\sum_{i=1}^n \left( Z_i^{\textsc{ssiv}} + \sum_{k=1}^r \hat{\gamma}_k \hat{\psi}_{ki} \right) \eta_i^r
= O_{\mathbb{P}}\left( \frac{\min\{q_n^\text{pre}, q_n^\text{post}\}}{q_n^\text{post}} \right)
+ O_{\mathbb{P}}\left( \frac{ \sqrt{q_n^{\text{pre}}} }{ \sqrt{n q_n^{\text{post}}} }  \right).
\end{align*}
For $S_2$, we apply the same proof in Lemma \ref{lemma:Ziphii_PC}. By definition and Lemma \ref{lemma:airi_ri^2}, 
\begin{equation*}
\mu_{n,k}^r 
= \sum_{i=1}^n (r_{0,i} + r_{1,i} - \mu_{r_1,i}) \psi_{k}^*(w_i)
= O_{\mathbb{P}}\left( \frac{1}{\sqrt{n q_n^\text{post}}} \right),
\end{equation*}
and this implies that $S_2 = O_{\mathbb{P}}\left( \frac{1}{n \sqrt{ q_n^\text{post}}}  \sqrt{ \frac{\log n}{\log \log n} / n} \right) $.
To conclude,
\begin{align*}
\frac{1}{n}\sum_{i=1}^n (Z_i^{\textsc{ssiv}} + \hat{\delta}_i) (r_{0,i} + r_{1,i} - \mu_{r_1,i})
= E\left( \frac{1}{n}\sum_{i=1}^n Z_i^{\textsc{ssiv}} \eta_i^r \right)
+ O_{\mathbb{P}}\left( \frac{ \sqrt{q_n^{\text{pre}}} }{ \sqrt{n q_n^{\text{post}}} }  \right).
\end{align*}
\end{proof}

\subsection{Consistency}

Define 
\begin{align*}
D_n^{\textsc{de}} = 
\begin{pmatrix}
1 & 0 & 0 \\
0 & 1 & 0 \\
0 & 0 & \frac{1}{\sqrt{nq_n^\text{pre}}} 
\end{pmatrix}.
\end{align*}
We first show the probability limit of $\left( D_n^{\textsc{de}} (\tilde{Z}^{\textsc{de}})^\top X \right)^{-1}$ whose closed form is 
\begin{align}
& \left( D_n^{\textsc{de}} (\tilde{Z}^{\textsc{de}})^\top X \right)^{-1} 
= \frac{ 1 }{ \det( D_n^{\textsc{de}} (\tilde{Z}^{\textsc{de}})^\top X) } \begin{pmatrix}
c_{11} & c_{12} & c_{13} \\
c_{21} & c_{22} & c_{23} \\
c_{31} & c_{32} & c_{33}
\end{pmatrix}
\label{eq:PC_ZX-1}
\end{align}  
where 
\begin{align*}
c_{11}
~=~& \left( \frac{1}{n}\sum_{i=1}^n T_i \right) \left( \frac{1}{n \sqrt{n q_n^\text{pre}} }\sum_{i=1}^n M_i (Z_i^{\textsc{ssiv}} + \hat{\delta}_i) \right) - \left( \frac{1}{n }\sum_{i=1}^n T_i M_i \right) \left( \frac{1}{n \sqrt{n q_n^\text{pre}}  }\sum_{i=1}^n T_i (Z_i^{\textsc{ssiv}} + \hat{\delta}_i) \right) \\
c_{12}
~=~& - \left( \frac{1}{n}\sum_{i=1}^n T_i \right) \left( \frac{1}{n \sqrt{n q_n^\text{pre}}  }\sum_{i=1}^n M_i (Z_i^{\textsc{ssiv}} + \hat{\delta}_i) \right) + \left( \frac{1}{n \sqrt{n q_n^\text{pre}} }\sum_{i=1}^n T_i (Z_i^{\textsc{ssiv}} + \hat{\delta}_i) \right) \left( \frac{1}{n}\sum_{i=1}^n M_i \right)  \\
c_{13}
~=~& a_{13} \\
c_{21}
~=~& - \left( \frac{1}{n}\sum_{i=1}^n T_i \right) \left( \frac{1}{n \sqrt{n q_n^\text{pre}}  }\sum_{i=1}^n M_i (Z_i^{\textsc{ssiv}} + \hat{\delta}_i) \right) 
+ \left( \frac{1}{n \sqrt{n q_n^\text{pre}} }\sum_{i=1}^n (Z_i^{\textsc{ssiv}} + \hat{\delta}_i) \right) 
\left( \frac{1}{n}\sum_{i=1}^n T_i M_i \right) \\
c_{22}
~=~& \left( \frac{1}{n \sqrt{n q_n^\text{pre}} }\sum_{i=1}^n M_i (Z_i^{\textsc{ssiv}} + \hat{\delta}_i) \right) - \left( \frac{1}{n}\sum_{i=1}^n M_i \right) \left( \frac{1}{n \sqrt{n q_n^\text{pre}} }\sum_{i=1}^n (Z_i^{\textsc{ssiv}} + \hat{\delta}_i) \right) \\
c_{23}
~=~& a_{23} \\
c_{31}
~=~& 
- \left( \frac{1}{n}\sum_{i=1}^n T_i \right) \left( \frac{1}{n \sqrt{n q_n^\text{pre}} }\sum_{i=1}^n (1-T_i) (Z_i^{\textsc{ssiv}} + \hat{\delta}_i) \right) \\
c_{32}
~=~& - \left( \frac{1}{n \sqrt{n q_n^\text{pre}} }\sum_{i=1}^n T_i (Z_i^{\textsc{ssiv}} + \hat{\delta}_i) \right) + \left( \frac{1}{n \sqrt{n q_n^\text{pre}} }\sum_{i=1}^n (Z_i^{\textsc{ssiv}} + \hat{\delta}_i) \right) \left( \frac{1}{n}\sum_{i=1}^n T_i \right) \\
c_{33} 
~=~& a_{33}
\end{align*}
and
\begin{align*}
\det( D_n^{\textsc{de}} (\tilde{Z}^{\textsc{de}})^\top X) = c_{33} c_{22} - c_{32} c_{23}.
\end{align*}
Define $c_{22}^* = \frac{E\left[ Z_i^{\textsc{ssiv}}(r_{0,i} + r_{1,i} ) \right]}{\sqrt{n q_n^{\text{pre}}}} $ 
where $c_{22}^* \asymp \frac{\min\{q_n^{\text{pre}}, q_n^{\text{post}}\}}{q_n^{\text{post}} \sqrt{n q_n^{\text{pre}}}}$.

\begin{theorem}\label{thm:ZX-1_PC_limit}
Suppose $q_n^\text{pre} \succ \frac{\log(n)}{\log(\log(n))}/n$.
Under Assumptions \ref{asu:network} and \ref{ass:low_rank}, 
then
\paragraph{Case (a): }
\begin{align*}
 \left( D_n^{\textsc{de}} (\tilde{Z}^{\textsc{de}})^\top X \right)^{-1} 
=& \frac{ \begin{pmatrix}
\pi c_{22}^* & -\pi c_{22}^* & a_{13}^* \\
- \pi c_{22}^*  & c_{22}^* & 
- \textup{Cov}(T_i, \xi_i + \mu_{r_1,i}) \\
0 & 0 & \pi (1-\pi)
\end{pmatrix} + O_{\mathbb{P}}\left( \frac{1}{\sqrt{n}} \right) }{ \pi (1-\pi) c_{22}^* + O_{\mathbb{P}}\left( \frac{1}{\sqrt{n}} \right) }.
\end{align*}

\paragraph{Case (b):}

\begin{align*}
& \left( D_n^{\textsc{de}} (\tilde{Z}^{\textsc{de}})^\top X \right)^{-1}  
= \frac{ \begin{pmatrix}
\pi c_{22}^* + O_{\mathbb{P}}\left( \frac{1}{\sqrt{n}} \right) & - \pi c_{22}^* + O_{\mathbb{P}}\left( \frac{1}{\sqrt{n}} \right) & a_{13}^* + O_{\mathbb{P}}\left( \frac{1}{\sqrt{n}} \right) \\
- \pi c_{22}^* + O_{\mathbb{P}}\left( \frac{1}{\sqrt{n}} \right) & c_{22}^* + O_{\mathbb{P}}\left( \frac{ 1 }{ n \sqrt{ q_n^{\text{post}}} }  \right) & 
O_{\mathbb{P}}\left( \frac{1}{n\sqrt{q_n^\text{post}}} \right) \\
O_{\mathbb{P}}\left( \frac{1}{\sqrt{n}} \right) & O_{\mathbb{P}}\left( \frac{1}{\sqrt{n}} \right) & \pi (1-\pi) + O_{\mathbb{P}}\left( \frac{1}{\sqrt{n}} \right)
\end{pmatrix}  }{ \pi (1-\pi) c_{22}^* + O_{\mathbb{P}}\left( \frac{ 1 }{ n \sqrt{ q_n^{\text{post}}} }  \right) }. 
\end{align*}

\end{theorem}
\begin{proof}[Proof of Theorem \ref{thm:ZX-1_PC_limit}]
The proof is analogous to Theorem \ref{thm:ZX-1} by applying Lemma \ref{lemma:Ziphii_PC}.

\end{proof}

\begin{proof}[Proof of Theorem \ref{thm:consistency_ratio_PC_IV}]
We show the consistency by Theorem \ref{thm:ZX-1_PC_limit}.
By Lemma \ref{lemma:Ziphii_PC}, we can show that 
\begin{align*}
\frac{1}{n\sqrt{n q_n^\text{pre}}}\sum_{i=1}^n (Z_i^{\textsc{ssiv}} + \hat{\delta}_i)
=& \frac{1}{n\sqrt{n q_n^\text{pre}}}\sum_{i=1}^n Z_i^{\textsc{ssiv}} \eta_i^1 
+ O_{\mathbb{P}}\left( \frac{\sqrt{q_n^\text{pre}}}{\sqrt{n}} \right)
= O_{\mathbb{P}}\left( \frac{1}{\sqrt{n}} \right)
\end{align*}
and thus 
\begin{align*}
\frac{1}{n \sqrt{n q_n^\text{pre}}}\sum_{i=1}^n \left((Z_i^{\textsc{ssiv}} + \hat{\delta}_i) - \frac{1}{n}\sum_{i=1}^n (Z_i^{\textsc{ssiv}} + \hat{\delta}_i) \right) u_i
= O_{\mathbb{P}}\left( \frac{1}{\sqrt{n}} \right).
\end{align*}
\paragraph{Case (a):}   
By combining the results from Theorem \ref{thm:ZX-1_PC_limit}, we have
\begin{align*}
\hat{\beta}^{\textsc{de}}_1 - {\beta}_1
~=~ & \frac{c_{22} \frac{1}{n}\sum_{i=1}^n (T_i - \bar{T}) u_i + c_{23} \frac{1}{n \sqrt{n q_n^\text{pre}}}\sum_{i=1}^n \left((Z_i^{\textsc{ssiv}} + \hat{\delta}_i) - \frac{1}{n}\sum_{i=1}^n (Z_i^{\textsc{ssiv}} + \hat{\delta}_i) \right) u_i }{ c_{33} c_{22} - c_{32} c_{23} }  \\
~=~& \frac{ \left( c_{22}^* 
+ O_{\mathbb{P}}\left( \frac{1}{\sqrt{n}} \right) \right) O_{\mathbb{P}}\left( \frac{1}{\sqrt{n}} \right) 
+ \left( - \textup{Cov}(T_i, \xi_i + \mu_{r_1,i}) 
+ O_{\mathbb{P}}\left( \frac{1}{ \sqrt{n} } \right) \right) O_{\mathbb{P}}\left( \frac{1}{\sqrt{n}} \right)}{ \pi (1-\pi) c_{22}^* + O_{\mathbb{P}}\left( \frac{1}{\sqrt{n}} \right) } \\
~=~& \frac{  O_{\mathbb{P}}\left( \frac{1}{\sqrt{n}} \right) }{ \pi (1-\pi) c_{22}^* + O_{\mathbb{P}}\left( \frac{1}{\sqrt{n}} \right) }, 
\end{align*}
and
\begin{align*}
\hat{\beta}^{\textsc{de}}_2 - {\beta}_2
~=~& \frac{ c_{33} \frac{1}{n \sqrt{n q_n^\text{pre}}}\sum_{i=1}^n \left((Z_i^{\textsc{ssiv}} + \hat{\delta}_i) - \frac{1}{n}\sum_{i=1}^n (Z_i^{\textsc{ssiv}} + \hat{\delta}_i) \right) u_i + c_{32} \frac{1}{n}\sum_{i=1}^n (T_i - \bar{T}) u_i }{ c_{33} c_{22} - c_{32} c_{23} } \\
~=~& \frac{ \left( \pi(1-\pi) + O_{\mathbb{P}}\left( \frac{1}{\sqrt{n}} \right) \right) O_{\mathbb{P}}\left( \frac{1}{\sqrt{n}} \right) + O_{\mathbb{P}}\left( \frac{1}{\sqrt{n}} \right) O_{\mathbb{P}}\left( \frac{1}{\sqrt{n}} \right) }{ \pi (1-\pi) c_{22}^* + O_{\mathbb{P}}\left( \frac{1}{\sqrt{n}} \right)  } \\
~=~& \frac{ O_{\mathbb{P}}\left( \frac{1}{\sqrt{n}} \right) }{ \pi (1-\pi) c_{22}^* + O_{\mathbb{P}}\left( \frac{1}{\sqrt{n}} \right)  }.
\end{align*}
Therefore, $\hat{\beta}^{\textsc{de}}_1$ and $\hat{\beta}^{\textsc{de}}_2$ are consistent when $ \max\{q_n^\text{pre}, q_n^\text{post}\} \prec \sqrt{q_n^\text{pre}} $
with 
\begin{align*}
\hat{\beta}^{\textsc{de}}_1 - \beta_1 
= O_{\mathbb{P}}\left( \frac{ \max\{ q_n^\text{pre}, q_n^\text{post}\}  }{  \sqrt{q_n^\text{pre}} } \right) 
\text{ and }
\hat{\beta}^{\textsc{de}}_2 - \beta_2 
= O_{\mathbb{P}}\left( \frac{ \max\{ q_n^\text{pre}, q_n^\text{post}\}  }{  \sqrt{q_n^\text{pre}} } \right).
\end{align*}
It follows that $\hat{\beta}^{\textsc{de}}_0$ is also consistent 
when $ \max\{q_n^\text{pre}, q_n^\text{post}\} \prec \sqrt{q_n^\text{pre}} $
with 
\[
\hat{\beta}^{\textsc{de}}_0 - {\beta}_0
~=~ O_{\mathbb{P}}\left( \frac{ \max\{ q_n^\text{pre}, q_n^\text{post}\}  }{  \sqrt{q_n^\text{pre}} } \right).
\]

\paragraph{Case (b)}   
By combining the results from Theorem \ref{thm:ZX-1_PC_limit}, we have
\begin{align*}
\hat{\beta}^{\textsc{de}}_1 - {\beta}_1
~=~& \frac{c_{22} \frac{1}{n}\sum_{i=1}^n (T_i - \bar{T}) u_i + c_{23} \frac{1}{n \sqrt{n q_n^\text{pre}}}\sum_{i=1}^n \left((Z_i^{\textsc{ssiv}} + \hat{\delta}_i) - \frac{1}{n}\sum_{i=1}^n (Z_i^{\textsc{ssiv}} + \hat{\delta}_i) \right) u_i }{ c_{33} c_{22} - c_{32} c_{23} }  \\
~=~& \frac{ \left( c_{22}^* + O\left( \frac{1}{n \sqrt{q_n^\text{post}} } \right) \right) O_{\mathbb{P}}\left( \frac{1}{\sqrt{n}} \right) + O_{\mathbb{P}}\left( \frac{1}{n\sqrt{q_n^\text{post}}} \right) O_{\mathbb{P}}\left( \frac{1}{\sqrt{n}} \right)}{ \pi (1-\pi) c_{22}^* + O_{\mathbb{P}}\left( \frac{ 1 }{ n \sqrt{ q_n^{\text{post}}} }  \right) }.
\end{align*}
The consistency of $\hat{\beta}^{\textsc{de}}_1$ always holds with 
\begin{align*}
\hat{\beta}^{\textsc{de}}_1 - \beta_1 
= O_{\mathbb{P}}\left( \frac{1}{ \sqrt{n} } \right).
\end{align*}
Then, for $\hat{\beta}^{\textsc{de}}_2$, we have
\begin{align*}
\hat{\beta}^{\textsc{de}}_2 - {\beta}_2
~=~& \frac{ c_{33} \frac{1}{n \sqrt{n q_n^\text{pre}}}\sum_{i=1}^n \left((Z_i^{\textsc{ssiv}} + \hat{\delta}_i) - \frac{1}{n}\sum_{i=1}^n (Z_i^{\textsc{ssiv}} + \hat{\delta}_i) \right) u_i + c_{32} \frac{1}{n}\sum_{i=1}^n (T_i - \bar{T}) u_i }{ c_{33} c_{22} - c_{32} c_{23} } \\
~=~& \frac{ \left( \pi(1-\pi) + O_{\mathbb{P}}\left( \frac{1}{\sqrt{n}} \right) \right) O_{\mathbb{P}}\left( \frac{1}{\sqrt{n}} \right) + O_{\mathbb{P}}\left( \frac{1}{\sqrt{n}} \right) O_{\mathbb{P}}\left( \frac{1}{\sqrt{n}} \right) }{ \pi (1-\pi) c_{22}^* + O\left( \frac{1}{n \sqrt{q_n^\text{post}} } \right)  }.
\end{align*}
We can conclude that with $ \max\{q_n^\text{pre}, q_n^\text{post}\} \prec \sqrt{q_n^\text{pre}} $, then 
\begin{align*}
\hat{\beta}^{\textsc{de}}_2 - {\beta}_2
= O_{\mathbb{P}}\left( \frac{ \max\{ q_n^\text{pre}, q_n^\text{post}\}  }{  \sqrt{q_n^\text{pre}} } \right).
\end{align*}
Furthermore, with $ \max\{q_n^\text{pre}, q_n^\text{post}\} \prec \sqrt{q_n^\text{pre}} $, then
\[
\hat{\beta}^{\textsc{de}}_0 - {\beta}_0
~=~ O_{\mathbb{P}}\left( \frac{ \max\{ q_n^\text{pre}, q_n^\text{post}\}  }{  \sqrt{q_n^\text{pre}} } \right).
\]
 
\end{proof}

\begin{proof}[Proof of Corollary \ref{cor:IV_PC_consistency}]
Corollary \ref{cor:IV_PC_consistency} is a direct result of Theorem \ref{thm:consistency_ratio_PC_IV} under $q_n^\text{pre} \preccurlyeq q_n^\text{post}$.    
\end{proof}

\subsection{Asymptotic normality}
\begin{proof}[Proof of Theorem \ref{thm:asymnormal_ratio_PC_IV}]
We complete the proof in three steps.
\paragraph{Step 1: asymptotic normality.}
We prove the asymptotic normality of the numerator.
By Lemma \ref{lemma:Ziphii_PC},
we have
\begin{align}
\frac{1}{n \sqrt{n q_n^\text{pre}}}\sum_{i=1}^n (Z_i^{\textsc{ssiv}} + \hat{\delta}_i) u_i
= \frac{1}{n \sqrt{n q_n^\text{pre}}}\sum\limits_{i=1}^n \sum_{j\ne i} A_{ij}^{\text{pre}} (T_j - \pi) \eta_i^u 
+ o_{\mathbb{P}}\left( \frac{1}{\sqrt{n}} \right).
\label{eq:Zi_PC_ui}
\end{align}
Define $H_{ij} = A_{ij}^{\text{pre}} (T_j - \pi) \eta_i^u$ and $\sigma_n^2 = \V\left( \sum_{(i,j)} H_{ij} \right)$.
By Theorem \ref{thm:clt}, we have
\begin{align}
\Delta\left( \frac{1}{\sigma_n} 
 \sum_{(i,j)} H_{ij}, \mathcal{N}(0, 1)  \right)
& \le \frac{1}{\sigma_n^3} \sum_{i,j}  
\mathbb{E} \left[ \left|H_{ij} \left( \sum_{(k,l) \in \mathcal{S}_{(i,j)}} H_{kl} \right)^2 \right| \right] 
+ \frac{\sqrt{2}}{\sqrt{\pi} \sigma_n^2} 
\sqrt{\text{Var} \left( \sum_{(i,j)} H_{ij} \sum_{ (k,l) \in \mathcal{S}_{(i,j)}}  H_{kl} \right)}.
\label{eq:HijNormal}
\end{align}
By definition of $\mathcal{S}_{(i,j)}$, we have
\begin{align}
\sum_{(k,l) \in \mathcal{S}_{(i,j)}} H_{(k,l)}
= H_{ij} + H_{ji} + \sum_{k\ne i,j} H_{ik} 
+ \sum_{k\ne i,j} H_{ki}
+ \sum_{k\ne i,j} H_{jk}
+ \sum_{k\ne i,j} H_{kj}.
\label{eq:Hkl}
\end{align}
For the first term in \eqref{eq:HijNormal}, 
\begin{align*}
\mathbb{E} \left[ \left|H_{ij} \left( \sum_{(k,l) \in \mathcal{S}_{(i,j)}} H_{kl} \right)^2 \right| \right]
\le C \cdot 
E \left[ A_{ij}^{\text{pre}} \left( \sum_{(k,l) \in \mathcal{S}_{(i,j)}} H_{kl} \right)^2 \right].
\end{align*}
Then we consider the cases in \eqref{eq:Hkl}, respectively. \\
\underline{(1): $H_{ij} + H_{ji}$.}
By definition, 
\begin{align*}
E \left[ A_{ij}^{\text{pre}} \left( H_{ij} + H_{ji} \right)^2 \right]
= E \left[ A_{ij}^{\text{pre}} \left( (T_j - \pi) \eta_i^u + (T_i - \pi) \eta_j^u \right)^2 \right]
\le C q_n^\text{pre}.
\end{align*}
\underline{(2): $\sum_{k\ne i,j} H_{ik}$.}
It can be bounded above by 
\begin{align*}
& E \left[ A_{ij}^{\text{pre}} \left( \sum_{k\ne i,j} H_{ik} \right)^2 \right]
= E \left[ A_{ij}^{\text{pre}} \left( \sum_{k\ne i,j} H_{ik}^2 + \sum_{(k_1,k_2)} H_{ik_1} H_{ik_2} \right) \right] \\
=& E \left[ A_{ij}^{\text{pre}} \left( \sum_{k\ne i,j} A_{ik}^{\text{pre}} (T_k - \pi)^2 (\eta_i^u)^2 + \sum_{(k_1,k_2)} A_{ik_1}^{\text{pre}} (T_{k_1} - \pi) \eta_i^u A_{ik_2}^{\text{pre}} (T_{k_2} - \pi) \eta_i^u \right) \right] \\
=& E \left[ A_{ij}^{\text{pre}} \sum_{k\ne i,j} A_{ik}^{\text{pre}} (T_k - \pi)^2 (\eta_i^u)^2 \right]
\le C n (q_n^\text{pre})^2
\end{align*}
by independence between $(T_{k_1} - \pi)$ and $(T_{k_2} - \pi)$.
By symmetry, we have 
\begin{align*}
& E \left[ A_{ij}^{\text{pre}} \left( \sum_{k\ne i,j} H_{jk} \right)^2 \right]
\le C n (q_n^\text{pre})^2.
\end{align*}
\underline{(3): $\sum_{k\ne i,j} H_{ki}$.}
It can be bounded above by 
\begin{align*}
& E \left[ A_{ij}^{\text{pre}} \left( \sum_{k\ne i,j} H_{ki} \right)^2 \right]
= E \left[ A_{ij}^{\text{pre}} \left( \sum_{k\ne i,j} H_{ki}^2 + \sum_{(k_1,k_2)} H_{k_1i} H_{k_2i} \right) \right] \\
=& E \left[ A_{ij}^{\text{pre}} \left( \sum_{k\ne i,j} A_{ki}^{\text{pre}} (T_i - \pi)^2 (\eta_k^u)^2 
+ (T_i - \pi)^2 \sum_{(k_1,k_2)} A_{k_1i}^{\text{pre}} \eta_{k_1}^u A_{k_2i}^{\text{pre}}  \eta_{k_2}^u \right) \right] \\
=& E \left[ A_{ij}^{\text{pre}} (T_i - \pi)^2 \sum_{k\ne i,j} A_{ki}^{\text{pre}}  (\eta_k^u)^2 \right]
\le C n (q_n^\text{pre})^2,
\end{align*}
where the last line follows from 
\begin{align*}
& E \left[ A_{ij}^{\text{pre}} (T_i - \pi)^2 \sum_{(k_1,k_2)} A_{k_1i}^{\text{pre}} \eta_{k_1}^u A_{k_2i}^{\text{pre}}  \eta_{k_2}^u  \right] \\
=& \pi(1-\pi) (q_n^\text{pre})^3 E \left[ f_0(i,j) \sum_{k_1} \sum_{l_1=1}^n \lambda_{l_1}^* \psi_{l_1}^*(w_{k_1}) \psi_{l_1}^*(w_{i}) \eta_{k_1}^u \sum_{k_2} \sum_{l_2=1}^n \lambda_{l_2}^* \psi_{l_2}^*(w_{k_2}) \psi_{l_2}^*(w_{i}) \eta_{k_2}^u  \right] 
= 0.
\end{align*}
By combining these three cases, we have 
\begin{align*}
\frac{1}{\sigma_n^3} \sum_{i,j}  
\mathbb{E} \left[ \left|H_{ij} \left( \sum_{(k,l) \in \mathcal{S}_{(i,j)}} H_{kl} \right)^2 \right| \right]
\le C \frac{n^2 n (q_n^\text{pre})^2}{ n^3 (q_n^\text{pre})^{3/2} } 
= C \sqrt{q_n^\text{pre}}.
\end{align*}
For the second term in \eqref{eq:HijNormal}, we decompose it into several pieces in terms of \eqref{eq:Hkl}, and consider these terms one by one. \\
\underline{(1): $H_{ij} + H_{ji}$.}
By definition,
\begin{align*}
& \text{Var} \left( \sum_{(i,j)} H_{ij} \left( H_{ij} + H_{ji} \right) \right)
= \text{Var} \left( \sum_{(i,j)} A_{ij}^{\text{pre}} (T_j - \pi) \eta_i^u \left(  (T_j - \pi) \eta_i^u + (T_i - \pi) \eta_j^u \right) \right) \\
=& \sum_{(i,j)} \text{Var} \left(  A_{ij}^{\text{pre}} (T_j - \pi) \eta_i^u \left(  (T_j - \pi) \eta_i^u + (T_i - \pi) \eta_j^u \right) \right) \\
&+ \sum_{\substack{(i,j,k) \\ \text{all distinct}}} \text{Cov} \left( A_{ij}^{\text{pre}} (T_j - \pi) \eta_i^u \left(  (T_j - \pi) \eta_i^u + (T_i - \pi) \eta_j^u \right), A_{ik}^{\text{pre}} (T_k - \pi) \eta_i^u \left(  (T_k - \pi) \eta_i^u + (T_i - \pi) \eta_k^u \right) \right) \\
\le & C n^3 (q_n^{\text{pre}})^2.
\end{align*}
\underline{(2): $\sum_{k\ne i,j} H_{ik}$.}
By definition,
\begin{align*}
& \text{Var} \left( \sum_{(i,j)} H_{ij} \sum_{k\ne i,j} H_{ik} \right)
\le \text{Var} \left( \sum_{\substack{(i,j,k) \\ \text{all distinct}}} A_{ij}^{\text{pre}} (T_j - \pi) A_{ik}^{\text{pre}} (T_k - \pi) \left( \eta_i^u \right)^2 \right) \\
&= \sum_{\substack{(i,j,k) \\ \text{all distinct}}}  \text{Var} \left( (T_j - \pi) (T_k - \pi)  A_{ij}^{\text{pre}} A_{ik}^{\text{pre}} \left( \eta_i^u \right)^2 \right) \\
&+ \sum_{\substack{(i_1,i_2,j,k) \\ \text{all distinct}}} E \left( (T_j - \pi)^2 (T_k - \pi)^2  A_{i_1j}^{\text{pre}} A_{i_1k}^{\text{pre}} \left( \eta_{i_1}^u \right)^2 A_{i_2j}^{\text{pre}} A_{i_2k}^{\text{pre}} \left( \eta_{i_2}^u \right)^2 \right) \\
& \le  C_1 n^3 (q_n^{\text{pre}})^2 + C_2 n^4 ( q_n^{\text{pre}})^4.
\end{align*}
\underline{(3): $\sum_{k\ne i,j} H_{jk}$.}
By analogous argument, we have
\[
\text{Var} \left( \sum_{(i,j)} H_{ij} \sum_{k\ne i,j} H_{jk} \right)
\le C_1 n^3 (q_n^{\text{pre}})^2 
+ C_2 n^4 (q_n^{\text{pre}})^3.
\]
\underline{(4): $\sum_{k\ne i,j} H_{ki}$.}
By definition,
\begin{align*}
& \text{Var} \left( \sum_{(i,j)} H_{ij} \sum_{k\ne i,j} H_{ki} \right)
= \text{Var} \left( \sum_{\substack{ (i,j,k) \\ \text{all distinct}}} A_{ij}^{\text{pre}} (T_j - \pi) \eta_i^u A_{ki}^{\text{pre}} (T_i - \pi) \eta_k^u \right) \\
&= \sum_{\substack{ (i,j,k) \\ \text{all distinct}}} \text{Var} \left( A_{ij}^{\text{pre}} (T_j - \pi) \eta_i^u A_{ki}^{\text{pre}} (T_i - \pi) \eta_k^u \right) \\
&+ \sum_{ \substack{ (i,j,k_1, k_2) \\ \text{all distinct}}} E\left( (T_j - \pi)^2 (T_i - \pi)^2 A_{ij}^{\text{pre}}  A_{k_1i}^{\text{pre}}  \eta_{k_1}^u A_{ij}^{\text{pre}} (\eta_i^u)^2 A_{k_2i}^{\text{pre}}  \eta_{k_2}^u \right) 
\le C_1 n^3 (q_n^{\text{pre}})^2 .
\end{align*}
\underline{(5): $\sum_{k\ne i,j} H_{kj}$.}
Recall that $\sum_{i=1} \eta_i \psi_k^*(w_i) = 0$.
\begin{align*}
& \text{Var} \left( \sum_{(i,j)} H_{ij} \sum_{k\ne i,j} H_{kj} \right)
\le E\left[ \left( \sum_{\substack{ (i,j,k) \\ \text{all distinct}}} A_{ij}^{\text{pre}} (T_j - \pi)^2 \eta_i^u A_{kj}^{\text{pre}} \eta_k^u \right)^2 \right] \\
=& \sum_{\substack{(i_1,j_1,k_1) \\ (i_2,j_2,k_2)}} 
E\left[ (T_{j_1} - \pi)^2 (T_{j_2} - \pi)^2 
A_{i_1 j_1}^{\text{pre}} \eta_{i_1}^u A_{k_1 j_1}^{\text{pre}} \eta_{k_1}^u 
A_{i_2j}^{\text{pre}} \eta_{i_2}^u A_{k_2j_2}^{\text{pre}} \eta_{k_2}^u \right] \\
=& \pi^2 (1-\pi)^2
\sum_{ \substack{ (i,j_1,j_2,k) \\ \text{all distinct}}} 
E\left[ 
A_{i j_1}^{\text{pre}}  A_{k j_1}^{\text{pre}}  
 A_{k j_2}^{\text{pre}} 
 (\eta_{i}^u)^2 (\eta_{k}^u)^2
 \right]
+  \pi^2 (1-\pi)^2
\sum_{ \substack{ (i,j_1,j_2) \\ \text{all distinct}}} 
E\left[ 
A_{i j_1}^{\text{pre}}   A_{i j_2}^{\text{pre}} 
 (\eta_{i}^u)^4 
 \right] \\
\le & C_1 n^4 (q_n^\text{pre})^3 + C_2 n^3 (q_n^\text{pre})^2.
\end{align*}
By combining these results, we can show that 
\begin{align*}
\frac{\sqrt{2}}{\sqrt{\pi} \sigma_n^2} \sqrt{\text{Var} \left( \sum_{i,j} H_{ij} \sum_{ (k,l) \in \mathcal{S}_{(i,j)}}  H_{kl} \right)}
\le C \frac{1}{n^2 q_n^\text{pre}} \sqrt{n^4 (q_n^\text{pre})^3}
= O\left( \sqrt{q_n^\text{pre}} \right).
\end{align*}
Then by \eqref{eq:HijNormal}, we have that 
\begin{equation*}
\frac{1}{\sigma_n}\sum\limits_{i=1}^n \sum_{j\ne i} A_{ij}^{\text{pre}} (T_j - \pi) \eta_i^u
\overset{d}{\to}
\mathcal{N}\left(0, 1 \right).
\end{equation*}
Together with Cramér--Wold Theorem and \eqref{eq:Zi_PC_ui}, we can show that 
\begin{equation*}
( V_\text{num}^{\textsc{de}} )^{-1/2}
\left( \sum_{i=1}^n \tilde{{Z}}^{\textsc{de}}_i u_i \right)
\overset{d}{\to}
\mathcal{N}\left(0, I_3 \right).
\end{equation*}

\paragraph{Step 2: consistent variance estimator.}
Recall that $\hat{u}_i^{\textsc{de}} - u_i = (\hat{\beta}^{\textsc{de}} - \beta)^\top X_i$.
By definition and Lemma \ref{lemma:8}, we have
\begin{align*}
& \hat{\mu}_k^R - {\mu}_k
= \sum_{i=1}^n (\hat{u}_i^{\textsc{de}} - u_i) (\hat{\psi}_{ki}^R - {\psi}_{k}^*(w_i) ) 
+ \sum_{i=1}^n (\hat{u}_i^{\textsc{de}} - u_i) {\psi}_{k}^*(w_i) 
+ \sum_{i=1}^n u_i (\hat{\psi}_{ki}^R - {\psi}_{k}^*(w_i) ) 
= o_{\mathbb{P}}\left( \sqrt{n} \right),
\end{align*}
where the last line follows from
\begin{align*}
& \sum_{i=1}^n (\hat{u}_i^{\textsc{de}} - u_i) (\hat{\psi}_{ki}^R - {\psi}_{k}^*(w_i) ) \\
=& (\hat{\beta}^{\textsc{de}}_0 - \beta_0) \sum_{i=1}^n  (\hat{\psi}_{ki}^R - {\psi}_{k}^*(w_i) ) 
+ (\hat{\beta}^{\textsc{de}}_1 - \beta_1) \sum_{i=1}^n T_i (\hat{\psi}_{ki}^R - {\psi}_{k}^*(w_i) ) 
+ (\hat{\beta}^{\textsc{de}}_2 - \beta_2) \sum_{i=1}^n M_i (\hat{\psi}_{ki}^R - {\psi}_{k}^*(w_i) )  \\
=& O_{\mathbb{P}}\left( \sqrt{ \frac{\log n}{\log \log n}} / (\sqrt{n} q_n^{\text{pre}}) \right)
= o_{\mathbb{P}}(\sqrt{n}), \\
& \sum_{i=1}^n (\hat{u}_i^{\textsc{de}} - u_i) {\psi}_{k}^*(w_i) 
= (\hat{\beta}^{\textsc{de}}_0 - \beta_0) \sum_{i=1}^n  {\psi}_{k}^*(w_i) 
+ (\hat{\beta}^{\textsc{de}}_1 - \beta_1) \sum_{i=1}^n T_i {\psi}_{k}^*(w_i) 
+ (\hat{\beta}^{\textsc{de}}_2 - \beta_2) \sum_{i=1}^n M_i {\psi}_{k}^*(w_i) 
= o_{\mathbb{P}}\left(\sqrt{n} \right)
\end{align*}
and 
\begin{align*}
& \sum_{i=1}^n u_i (\hat{\psi}_{ki}^R - {\psi}_{k}^*(w_i) ) 
\le \sqrt{ \sum_{i=1}^n u_i^2 \sum_{i=1}^n (\hat{\psi}_{ki}^R - {\psi}_{k}^*(w_i) )^2 } 
= O_{\mathbb{P}}\left( \frac{1}{\sqrt{q_n^{\text{pre}}}} \left( {\frac{\log n}{\log \log n}} \right)^{1/4} \right)
= o_{\mathbb{P}}\left(\sqrt{n} \right).
\end{align*}
Now we show the consistency of the variance estimators one by one. \\
\underline{(1): The consistency of the $(1,1)$ element of $\hat{V}_\text{num}^{\textsc{de}}$.} 
The proof is analogous to that of Theorem \ref{thm:asymnormal_ratio_IV} and follows from $\hat{\beta}^{\textsc{de}} - \beta = o_{\mathbb{P}}(1)$. \\
\underline{(2): The consistency of the $(3,3)$ element of $\hat{V}_\text{num}^{\textsc{de}}$.} 
By definition,
\begin{align*}
\hat{\eta}_i^u
= \hat{u}_i^{\textsc{de}} - \sum_{k=1}^r \hat{\mu}_k \hat{\psi}_{ki} 
=& \hat{u}_i^{\textsc{de}} - u_i + u_i - \sum_{k=1}^r {\mu}_k {\psi}_{k}^*(w_i)
- \sum_{k=1}^r \hat{\mu}_k \hat{\psi}_{ki}
+ \sum_{k=1}^r {\mu}_k {\psi}_{k}^*(w_i) \\
=& \eta_i^u + \hat{u}_i^{\textsc{de}} - u_i - \Delta_i   
\end{align*}
where $\Delta_i = \sum_{k=1}^r \left( \hat{\mu}_k^R \hat{\psi}_{ki}^R - {\mu}_k {\psi}_{k}^*(w_i) \right)$.
The variance estimator can be decomposed as
\begin{align}
& \frac{1}{n^2 q_n^\text{pre}} \sum_{i=1}^n \sum_{j=1}^n A_{ij}^\text{pre} (\hat{\eta}_i^u)^2
= \frac{1}{n^2 q_n^\text{pre}} \sum_{i=1}^n \sum_{j=1}^n A_{ij}^\text{pre} (\eta_i^u + \hat{u}_i^{\textsc{de}} - u_i - \Delta_i)^2 \notag \\
=& \frac{1}{n^2 q_n^\text{pre}} \sum_{i=1}^n \sum_{j=1}^n A_{ij}^\text{pre}({\eta}_i^u)^2 + \frac{1}{n^2 q_n^\text{pre}} \sum_{i=1}^n \sum_{j=1}^n A_{ij}^\text{pre} (\hat{u}_i^{\textsc{de}} - u_i - \Delta_i)^2 
+ \frac{2}{n^2 q_n^\text{pre}} \sum_{i=1}^n \sum_{j=1}^n A_{ij}^\text{pre}{\eta}_i^u (\hat{u}_i^{\textsc{de}} - u_i - \Delta_i) \notag \\
\le & \frac{1}{n^2 q_n^\text{pre}} \sum_{i=1}^n \sum_{j=1}^n A_{ij}^\text{pre}({\eta}_i^u)^2 
+ \frac{2}{n^2 q_n^\text{pre}} \sum_{i=1}^n \sum_{j=1}^n A_{ij}^\text{pre}(\hat{u}_i^{\textsc{de}} - u_i )^2 
+ \frac{2}{n^2 q_n^\text{pre}} \sum_{i=1}^n \sum_{j=1}^n A_{ij}^\text{pre} \Delta_i^2 \notag \\
&+ \frac{2}{n^2 q_n^\text{pre}} \sum_{i=1}^n \sum_{j=1}^n A_{ij}^\text{pre}{\eta}_i^u (\hat{u}_i^{\textsc{de}} - u_i - \Delta_i) \notag \\
=& S_1 + S_2 + S_3 + S_4.
\label{eq:var_Aij_etai}
\end{align}
For $\underline{\underline{S_1}}$, the expectation is
\begin{align*}
E\left[ \frac{1}{n^2 q_n^\text{pre}} \sum_{i=1}^n \sum_{j=1}^n A_{ij}^\text{pre}({\eta}_i^u)^2 \right]
= O(1).
\end{align*}
The variance is 
\begin{align*}
& \V\left( \frac{1}{n^2 q_n^\text{pre}} \sum_{i=1}^n ({\eta}_i^u)^2 \sum_{j=1}^n A_{ij}^\text{pre} \right) 
= \frac{1}{n^4 (q_n^\text{pre})^2} 
\left[ 
\sum_{i=1}^n \V\left( ({\eta}_i^u)^2 \sum_{j=1}^n A_{ij}^\text{pre} \right)
+ \sum_{(i,k)} 
\textup{Cov}\left( ({\eta}_i^u)^2 \sum_{j=1}^n A_{ij}^\text{pre}, ({\eta}_k^u)^2 \sum_{j=1}^n A_{kj}^\text{pre} \right) 
\right] \\
&= \frac{1}{n^4 (q_n^\text{pre})^2} \sum_{i=1}^n \sum_{j=1}^n \V\left( ({\eta}_i^u)^2 A_{ij}^\text{pre} \right)
+ \frac{1}{n^4 (q_n^\text{pre})^2} \sum_{i=1}^n \sum_{(j_1,j_2)} 
\textup{Cov}\left( ({\eta}_i^u)^2 A_{ij_1}^\text{pre}, ({\eta}_i^u)^2 A_{ij_2}^\text{pre} \right) \\
&+ \frac{1}{n^4 (q_n^\text{pre})^2} \sum_{(i,k)} \sum_{j=1}^n \textup{Cov}\left( ({\eta}_i^u)^2  A_{ij}^\text{pre}, ({\eta}_k^u)^2 A_{kj}^\text{pre} \right) 
+ \frac{1}{n^4 (q_n^\text{pre})^2} \sum_{(i,k)} \sum_{(j_1, j_2)} 
\textup{Cov}\left(({\eta}_i^u)^2 A_{ij_1}^\text{pre}, ({\eta}_k^u)^2 A_{kj_2}^\text{pre} \right) \\
&= \frac{1}{n^4 (q_n^\text{pre})^2} 
\left[ \sum_{i=1}^n \sum_{j=1}^n \V\left( ({\eta}_i^u)^2 A_{ij}^\text{pre} \right)
+ \sum_{i=1}^n \sum_{(j_1,j_2)} \textup{Cov}\left( ({\eta}_i^u)^2 A_{ij_1}^\text{pre}, ({\eta}_i^u)^2 A_{ij_2}^\text{pre} \right) 
+ \sum_{(i,k)} \sum_{j=1}^n 
\textup{Cov}\left( ({\eta}_i^u)^2  A_{ij}^\text{pre}, ({\eta}_k^u)^2 A_{kj}^\text{pre} \right) \right]  \\
&= O\left( \frac{1}{n} \right).
\end{align*}
Therefore, 
\begin{align}
\frac{1}{n^2 q_n^\text{pre}} \sum_{i=1}^n \sum_{j=1}^n A_{ij}^\text{pre}({\eta}_i^u)^2
= E\left[ \frac{1}{n^2 q_n^{\text{pre}}} \sum_{i=1}^n \sum_{j=1}^n A_{ij}^\text{pre}({\eta}_i^u)^2 \right]
+ O_{\mathbb{P}}\left( \frac{ 1 }{ \sqrt{n} } \right).
\end{align}
For $\underline{\underline{S_2}}$, with $\hat{\beta}^{\textsc{de}} - \beta = o_{\mathbb{P}}(1)$, we can show that
\begin{align*}
& \frac{1}{n^2 q_n^\text{pre}} \sum_{i=1}^n \sum_{j=1}^n A_{ij}^\text{pre}(\hat{u}_i^{\textsc{de}} - u_i )^2 
= \frac{1}{n^2 q_n^\text{pre}} \sum_{i=1}^n \sum_{j=1}^n A_{ij}^\text{pre}
\left( (\hat{\beta}^{\textsc{de}} - \beta)^\top X_i \right)^2 \\
\le & 3 \frac{1}{n^2 q_n^\text{pre}} \left[ (\hat{\beta}^{\textsc{de}}_0 - \beta_0)^2 
\sum_{i=1}^n \sum_{j=1}^n A_{ij}^\text{pre} 
+ (\hat{\beta}^{\textsc{de}}_1 - \beta_1)^2 
\sum_{i=1}^n \sum_{j=1}^n A_{ij}^\text{pre} T_i + (\hat{\beta}^{\textsc{de}}_2 - \beta_2)^2 
\sum_{i=1}^n \sum_{j=1}^n A_{ij}^\text{pre} M_i^2 \right] \\
=& o_{\mathbb{P}}(1).
\end{align*}
For $\underline{\underline{S_3}}$, by definition, we have
\begin{align*}
& \frac{1}{n^2 q_n^\text{pre}} \sum_{i=1}^n \sum_{j=1}^n A_{ij}^\text{pre} \Delta_i^2
= \frac{1}{n^2 q_n^\text{pre}} \sum_{i=1}^n \left( \sum_{k=1}^r \left( \hat{\mu}_k^R \hat{\psi}_{ki}^R - {\mu}_k {\psi}_{k}^*(w_i) \right) \right)^2 \sum_{j=1}^n A_{ij}^\text{pre} \\
&= \frac{1}{n^2 q_n^\text{pre}} \sum_{k=1}^r \sum_{i=1}^n \left( \hat{\mu}_k^R \hat{\psi}_{ki}^R - {\mu}_k {\psi}_{k}^*(w_i) \right)^2 \sum_{j=1}^n A_{ij}^\text{pre} \\
&+ \frac{1}{n^2 q_n^\text{pre}} \sum_{(k,l)} \sum_{i=1}^n \left( \hat{\mu}_k^R \hat{\psi}_{ki}^R - {\mu}_k {\psi}_{k}^*(w_i) \right) \left( \hat{\mu}_l^R \hat{\psi}_{li}^R - {\mu}_l {\psi}_{l}^*(w_i) \right) \sum_{j=1}^n A_{ij}^\text{pre}. 
\end{align*}
By definition,
\begin{align*}
\hat{\mu}_k^R \hat{\psi}_{ki}^R - {\mu}_k {\psi}_{k}^*(w_i)
=& \hat{\mu}_k^R \hat{\psi}_{ki}^R - {\mu}_k \hat{\psi}_{ki}^R + {\mu}_k \hat{\psi}_{ki}^R - {\mu}_k {\psi}_{k}^*(w_i) \\
=& (\hat{\mu}_k^R - {\mu}_k) (\hat{\psi}_{ki}^R - {\psi}_{k}^*(w_i) ) 
+ (\hat{\mu}_k^R - {\mu}_k) {\psi}_{k}^*(w_i)
+ {\mu}_k (\hat{\psi}_{ki}^R - {\psi}_{k}^*(w_i) ).
\end{align*}
For the first term of $S_3$, we have
\begin{align*}
& \frac{1}{3} \frac{1}{n^2 q_n^\text{pre}} \sum_{k=1}^r \sum_{i=1}^n \left( \hat{\mu}_k^R \hat{\psi}_{ki}^R - {\mu}_k {\psi}_{k}^*(w_i) \right)^2 \sum_{j=1}^n A_{ij}^\text{pre} \\
\le & \frac{1}{n^2 q_n^\text{pre}} \sum_{k=1}^r (\hat{\mu}_k^R - {\mu}_k)^2 \sum_{i=1}^n  (\hat{\psi}_{ki}^R - {\psi}_{k}^*(w_i) ) ^2 \sum_{j=1}^n A_{ij}^\text{pre}
+ \frac{1}{n^2 q_n^\text{pre}} \sum_{k=1}^r (\hat{\mu}_k^R - {\mu}_k)^2
\sum_{i=1}^n  {\psi}_{k}^*(w_i)^2 \sum_{j=1}^n A_{ij}^\text{pre} \\
&+ \frac{1}{n^2 q_n^\text{pre}} \sum_{k=1}^r {\mu}_k^2 \sum_{i=1}^n (\hat{\psi}_{ki}^R - {\psi}_{k}^*(w_i) )^2 \sum_{j=1}^n A_{ij}^\text{pre} \\
=& o_{\mathbb{P}}(1)
\end{align*}
where the last equality follows from 
\begin{align*}
E\left[ \sum_{i=1}^n  {\psi}_{k}^*(w_i)^2 \sum_{j=1}^n A_{ij}^\text{pre} \right]
= E\left[ \sum_{i=1}^n  {\psi}_{k}^*(w_i)^2 \sum_{j=1}^n q_n^\text{pre} g_0(i,j) \right]
\le C n q_n^\text{pre}
\end{align*}
and 
\begin{align*}
\sum_{i=1}^n (\hat{\psi}_{ki}^R - {\psi}_{k}^*(w_i) )^2 \sum_{j=1}^n A_{ij}^\text{pre}
\le \max_i \sum_{j=1}^n A_{ij}^\text{pre}
\sum_{i=1}^n (\hat{\psi}_{ki}^R - {\psi}_{k}^*(w_i) )^2
= O_{\mathbb{P}}\left( \sqrt{ {\frac{\log n}{\log \log n}}}  \right).
\end{align*}
The second term follows from an analogous argument, and thus $S_3 = o_{\mathbb{P}}(1)$. \\
For $\underline{\underline{S_4}}$:
\begin{align*}
\frac{1}{n^2 q_n^\text{pre}} \sum_{i=1}^n \sum_{j=1}^n A_{ij}^\text{pre}{\eta}_i^u (\hat{u}_i^{\textsc{de}} - u_i - \Delta_i)
= \frac{1}{n^2 q_n^\text{pre}} \sum_{i=1}^n \sum_{j=1}^n A_{ij}^\text{pre} {\eta}_i^u (\hat{u}_i^{\textsc{de}} - u_i )
- \frac{1}{n^2 q_n^\text{pre}} \sum_{i=1}^n \sum_{j=1}^n A_{ij}^\text{pre}{\eta}_i^u \Delta_i.
\end{align*}
For the first term,
\begin{align*}
& \frac{1}{n^2 q_n^\text{pre}} \sum_{i=1}^n \sum_{j=1}^n A_{ij}^\text{pre}{\eta}_i (\hat{u}_i^{\textsc{de}} - u_i ) \\
=& (\hat{\beta}_0^{\textsc{de}} - \beta_0) \frac{1}{n^2 q_n^\text{pre}} \sum_{i=1}^n \sum_{j=1}^n A_{ij}^\text{pre}{\eta}_i^u
+ (\hat{\beta}_1^{\textsc{de}} - \beta_1) \frac{1}{n^2 q_n^\text{pre}} \sum_{i=1}^n \sum_{j=1}^n A_{ij}^\text{pre}{\eta}_i^u T_i
+ (\hat{\beta}_2^{\textsc{de}} - \beta_2) \frac{1}{n^2 q_n^\text{pre}} \sum_{i=1}^n \sum_{j=1}^n A_{ij}^\text{pre}{\eta}_i^u M_i.
\end{align*}
By bounding the term in $L_2$ norm:
\begin{align*}
& E\left[ \left( \sum_{i=1}^n \sum_{j=1}^n A_{ij}^\text{pre}{\eta}_i M_i \right)^2 \right] \\
=& E\left[ \sum_{i=1}^n \left( \sum_{j=1}^n A_{ij}^\text{pre}{\eta}_i^2 M_i^2 + \sum_{(j_1,j_2)}^n A_{ij_1}^\text{pre} A_{ij_2}^\text{pre} {\eta}_i^2 M_i^2 \right) + \sum_{(i_1, i_2)} {\eta}_{i_1} M_{i_1} {\eta}_{i_2} M_{i_2}  \left( \sum_{j=1}^n A_{i_1j}^\text{pre} \right) \left( \sum_{j=1}^n A_{i_2j}^\text{pre}\right) \right] \\
\le & E\left[ \sum_{i=1}^n \left( \sum_{j=1}^n A_{ij}^\text{pre}{\eta}_i^2 M_i^2 + \sum_{(j_1,j_2)}^n A_{ij_1}^\text{pre} A_{ij_2}^\text{pre} {\eta}_i^2 M_i^2 \right) + \sum_{(i_1, i_2)} \left( {\eta}_{i_1}^2 M_{i_1}^2 + {\eta}_{i_2}^2 M_{i_2}^2 \right)  \left( \sum_{j=1}^n A_{i_1j}^\text{pre} \right) \left( \sum_{j=1}^n A_{i_2j}^\text{pre}\right) \right] \\
\le & C_1 n^2 q_n^\text{pre} + C_2 n^3 (q_n^\text{pre})^2 + C_3 n^4 (q_n^\text{pre})^2,
\end{align*}
which implies that 
\begin{align*}
\sum_{i=1}^n \sum_{j=1}^n A_{ij}^\text{pre}{\eta}_i^u M_i  
= O_{\mathbb{P}}\left( n^2 q_n^\text{pre} \right).
\end{align*}
With analogous arguments and $\hat{\beta}^{\textsc{de}} - \beta = o_{\mathbb{P}}(1)$, we can show that 
\begin{align*}
\frac{1}{n^2 q_n^\text{pre}} \sum_{i=1}^n \sum_{j=1}^n A_{ij}^\text{pre} {\eta}_i^u (\hat{u}_i^{\textsc{de}} - u_i )
= o_{\mathbb{P}}(1).
\end{align*}
For the second term,
\begin{align*}
& \frac{1}{n^2 q_n^\text{pre}} \sum_{i=1}^n \sum_{j=1}^n A_{ij}^\text{pre}{\eta}_i^u \Delta_i
= \frac{1}{n^2 q_n^\text{pre}} \sum_{i=1}^n \sum_{j=1}^n A_{ij}^\text{pre}{\eta}_i^u \sum_{k=1}^r \left( \hat{\mu}_k^R \hat{\psi}_{ki}^R - {\mu}_k {\psi}_{k}^*(w_i) \right) \\
=& \frac{1}{n^2 q_n^\text{pre}} \sum_{k=1}^r (\hat{\mu}_k^R - {\mu}_k)  \sum_{i=1}^n {\eta}_i^u 
(\hat{\psi}_{ki}^R - {\psi}_{k}^*(w_i) ) 
\sum_{j=1}^n A_{ij}^\text{pre}
+ \frac{1}{n^2 q_n^\text{pre}} \sum_{k=1}^r (\hat{\mu}_k^R - {\mu}_k) \sum_{i=1}^n {\eta}_i^u  {\psi}_{k}^*(w_i) 
\sum_{j=1}^n A_{ij}^\text{pre} \\
&+ \frac{1}{n^2 q_n^\text{pre}} \sum_{k=1}^r {\mu}_k \sum_{i=1}^n  \sum_{j=1}^n 
A_{ij}^\text{pre} {\eta}_i^u
 (\hat{\psi}_{ki}^R - {\psi}_{k}^*(w_i) ) 
= o_{\mathbb{P}}(1)
\end{align*}
where the last equality follows from 
\begin{align*}
 \sum_{i=1}^n \sum_{j=1}^n A_{ij}^\text{pre}{\eta}_i^u (\hat{\psi}_{ki}^R - {\psi}_{k}^*(w_i) )
\le \left\| {\eta}^u \right\|_2 
\left\| A^\text{pre} (\hat{\psi}_{k}^R - {\psi}_{k}^* )\right\|_2 
= O_{\mathbb{P}}\left(\sqrt{n} \right)
O_{\mathbb{P}}\left( \sqrt{ \frac{\log n}{\log \log n}} \right)
= O_{\mathbb{P}}(\sqrt{n q_n^\text{pre}})
\end{align*}
and 
\begin{align*}
& E\left[ \left( \sum_{i=1}^n \sum_{j=1}^n A_{ij}^\text{pre}{\eta}_i^u  {\psi}_{k}^*(w_i) \right)^2 \right]  \\ 
=& E\left[ \sum_{i=1}^n ({\eta}_i^2)^2  {\psi}_{k}^*(w_i)^2 \left(  \sum_{j=1}^n A_{ij}^\text{pre} \right)^2 \right] 
+ E\left[ \sum_{(i_1,i_2)} {\eta}_{i_1}^u {\psi}_{k}^*(w_{i_1}) {\eta}_{i_2}^u  {\psi}_{k}^*(w_{i_2})\left(  \sum_{j=1}^n A_{i_1j}^\text{pre} \right)  
\left(  \sum_{j=1}^n A_{i_2j}^\text{pre} \right) \right] 
\le n^3 (q_n^\text{pre})^2.
\end{align*}
Thus, $S_4 = o_{\mathbb{P}}(1)$. 
Therefore, we show that
\[
\frac{1}{n^2 q_n^\text{pre}} \sum_{i=1}^n \sum_{j=1}^n A_{ij}^\text{pre} (\hat{\eta}_i^u)^2 = E\left[ \frac{1}{n^2 q_n^\text{pre}} \sum_{i=1}^n \sum_{j=1}^n A_{ij}^\text{pre}({\eta}_i^u)^2 \right] + o_{\mathbb{P}}(1).
\]

\paragraph{Step 3}
Steps 1--2 and CMT together imply the desired result:
\begin{align*}
& \left(\hat{V}_{\text{num}}^{\textsc{de}}  \right)^{-1/2} 
\left(  (\tilde{Z}^{\textsc{de}})^\top u \right)
\overset{d}{\to} \mathcal{N}(0, I_3).
\end{align*}
By Theorem \ref{thm:ZX-1_PC_limit}, we have shown the probability limit of $({\tilde{Z}^{\textsc{de}}})^\top X$. By combining with CMT, we can show 
\[
\left( \hat{V}^{\textsc{de}} \right)^{-1/2} 
\left(\hat{\beta}^{\textsc{de}}-\beta \right) 
\stackrel{\textup{d}}{\rightarrow} \mathcal{N}(0, I_3).
\]

\end{proof}

\section{Normalized SSIV}
\label{app:IV_alt}

\cite{BorusyakHullJaravel2022} suggest the following normalized SSIV for $M_i$:
\begin{equation}
Z_i^{\text{alt}} = \frac{\sum_{j= 1}^n A_{ij}^\text{pre} T_j}{\sum_{j= 1}^n A_{ij}^\text{pre}}. 
\label{eq:SSIV_ratio}  
\end{equation}
The exclusion restriction holds by the randomness of the shocks and the normalization of the shares:
\begin{align*}
E\left[ Z_i^{\text{alt}} u_i \right]
= E\left[ \frac{ \sum_{j= 1}^n A_{ij}^{\text{pre}} T_j }{ \sum_{j= 1}^n A_{ij}^{\text{pre}} } u_i \right]
= \pi E\left[ \frac{ \sum_{j= 1}^n A_{ij}^{\text{pre}} }{ \sum_{j= 1}^n A_{ij}^{\text{pre}} } u_i \right]
= 0.
\end{align*}
\cite{BorusyakHullJaravel2022} state that the relevance condition holds 
when individual units are mostly exposed to only a small number of shocks. 
In this section, we quantify the regime where the normalized SSIV in \eqref{eq:SSIV_ratio} yields consistent estimators. 

\subsection{Consistency}

For estimation, we consider the IV estimation with the IV vector:
\[
\Tilde{Z}_i^{\text{alt}} = (1,T_i, Z_i^{\text{alt}}).
\]
Let $\hat{\beta}^{\text{alt}}$ denote the vector of coefficients obtained from the above IV regression.
Define $\tilde{Z}^{\text{alt}}$ as the $n\times 3$ matrix obtained by stacking $\Tilde{Z}_i^{\text{alt}}$.

\begin{theorem}\label{thm:IV_alt_consistency}
Under Assumptions \ref{asu:network} and \ref{asu:linearY}, 
\begin{enumerate}[(a)]
\item if $\V(\xi_i)>0$ with $\max\{q_n^{\text{pre}}, q_n^{\text{post}}\} \prec \frac{1}{\sqrt{n}}$,
then
\begin{align*}
\hat{\beta}^{\text{alt}}- {\beta} = 
O_{\mathbb{P}}\left( \sqrt{n } \max\{q_n^{\text{pre}}, q_n^{\text{post}}\}  \right);
\end{align*}
\item if $\V(\xi_i)=0$ with
$\max\{q_n^{\text{pre}}, q_n^{\text{post}}\}
\prec \sqrt{n} q_n^{\text{post}}$, then
\begin{align*}
\hat{\beta}_1^{\text{alt}}- {\beta}_1 
= O_{\mathbb{P}}\left( \frac{1}{\sqrt{n}} \max\left\{ \frac{ \max\{q_n^\text{pre}, q_n^\text{post} \} }{  \sqrt{ q_n^\text{post}}  }, 1 \right\}  \right).
\end{align*}
Moreover, with $\max\{q_n^{\text{pre}}, q_n^{\text{post}}\} \prec \frac{1}{\sqrt{n}}$, then 
\begin{align*}
\hat{\beta}_0^{\text{alt}} - {\beta}_0 = 
O_{\mathbb{P}}\left( \sqrt{n } \max\{q_n^{\text{pre}}, q_n^{\text{post}}\}  \right)
\text{ and }
\hat{\beta}_2^{\text{alt}} - {\beta}_2 = 
O_{\mathbb{P}}\left( \sqrt{n } \max\{q_n^{\text{pre}}, q_n^{\text{post}}\}  \right).
\end{align*}

\end{enumerate}
\end{theorem}
Theorem \ref{thm:IV_alt_consistency} establishes the consistency of $\hat{\beta}^{\text{alt}}$ under both Cases (a) and (b). The consistency regime and convergence rates of $\hat{\beta}^{\text{alt}}_0$ and $\hat{\beta}^{\text{alt}}_2$ remain unchanged across both cases. 
However, in Case (b), where $A^{\text{post}}$ is conditionally mean-independent of the treatment, $\hat{\beta}^{\text{alt}}_1$ exhibits a faster convergence rate compared to Case (a). 
Additionally, $\hat{\beta}^{\text{alt}}_1$ is consistent under a less restrictive condition in Case (b). 

\begin{remark}
The consistency regime for the IV estimatos using the normalized SSIV in \eqref{eq:SSIV_ratio} is the same as that for that of SSIV $Z_i^{\textsc{ssiv}}$ in Theorem \ref{thm:IV_consistency}. 
This similarity arises because the primary component of both IVs is $\sum_{j= 1}^n A_{ij}^{\text{post}} T_j$, which is endogenous due to its correlation with the error term $u_i$. 
Either normalization or centering around the expected number of treated friends helps to control this endogenous component.
\end{remark}

Corollary \ref{cor:alt_IV_consistency} simplifies the results in Theorem \ref{thm:IV_alt_consistency} to the special case where $q_n^{\text{pre}} \preccurlyeq q_n^{\text{post}}$.
\begin{corollary}
\label{cor:alt_IV_consistency}
Suppose $q_n^{\text{pre}} \preccurlyeq q_n^{\text{post}}$.
Under Assumptions \ref{asu:network} and \ref{asu:linearY}, 
\begin{enumerate}[(a)]
\item if $\V(\xi_i)>0$ with $q_n^{\text{post}} \prec \frac{1}{\sqrt{n}}$, then $\hat{\beta}^{\text{alt}}- {\beta} = O_{\mathbb{P}}(\sqrt{n} q_n^{\text{post}})$;
\item if $\V(\xi_i)=0$ with $q_n^{\text{post}} \prec \frac{1}{\sqrt{n}}$, then 
$\hat{\beta}_0^{\text{alt}}- {\beta}_0 = O_{\mathbb{P}}(\sqrt{n} q_n^{\text{post}})$ and $\hat{\beta}_2^{\text{alt}}- {\beta}_2 = O_{\mathbb{P}}(\sqrt{n} q_n^{\text{post}})$.
Moreover, it always holds that $\hat{\beta}^{\text{alt}}_1 - {\beta}_1 = O_{\mathbb{P}}\left( \frac{ 1 }{ \sqrt{n} } \right)$.
\end{enumerate}
\end{corollary}
Corollary \ref{cor:alt_IV_consistency} establishes the consistency of $\hat{\beta}^{\text{alt}}$ under both Cases (a) and (b). The consistency regime and convergence rates of $\hat{\beta}^{\text{alt}}_0$ and $\hat{\beta}^{\text{alt}}_2$ remain unchanged across both cases. 
However, in Case (b), where $A^{\text{post}}$ is conditionally mean-independent of the treatment, $\hat{\beta}^{\text{alt}}_1$ converges faster than in Case (a) and maintains the standard rate $\sqrt{n}$.
Additionally, $\hat{\beta}^{\text{alt}}_1$ is consistent under a less restrictive condition in Case (b). 

\subsection{Proof of Theorem \ref{thm:IV_alt_consistency}}

We start by providing some useful lemmas.

\begin{lemma}[Lemma 15 in \cite{LiWager2022}]\label{lemma:15}
Consider an RCT under network interference satisfying Assumption \ref{asu:network}, with treatment assigned independently as $T_i \sim \text{Bernoulli}(\pi)$ for some $0 < \pi < 1$.

\begin{enumerate}
\item Suppose furthermore that if we define $g(w) = \int_0^1 \min(1, g(w, t))dF(t)$, then the function
$g$ is bounded away from 0, i.e.,
\begin{equation}
g(w_i) \geq c_l \text{ for any } w_i. 
\end{equation}
Then for any $k \in \mathbb{N}, k \geq 1$, there exists some constant $C_k$, depending on $k$, s.t.
\begin{enumerate}[(a)]
\item $E\left[ \left( \frac{1}{N_i} \right)^k {1}{\{N_i>0\}} \mid w_i \right] \leq \frac{C_k}{(n q_n c_l)^k}$,
\item $E\left[ \left( \frac{1}{N_i} \right)^k {1}{\{N_i>0\}} \right] \leq \frac{C_k}{(n q_n c_l)^k}$, 
\end{enumerate}

\item Assume the graphon has a finite $K^{th}$ moment, i.e.
\begin{equation*}
E [g(w_1, w_2)^k] \leq c^k_u, \text{ for } k = 1, 2, \ldots, K.
\end{equation*}
Then 
\[
E [(N_i - nq_n g(i))^{2k}] \leq C_k (n q_n c_u)^k
\]
for $k=1,\cdots, K$, where $C_k$ is some constant depending on $k$.
\end{enumerate}
\end{lemma}

\begin{lemma}\label{lemma:Zi_alt_phii}
Define $\phi_i$ as an i.i.d. random variable with constant variance.
Define the conditional expectation $Q_j^{\phi} = E\left[ \frac{ A_{ij}^{\text{pre}} \phi_i }{ q_n^{\text{pre}} g_0(i) } \mid w_j \right]$. 
Under Assumptions \ref{asu:network} and \ref{asu:linearY}, 
then
\begin{align}
\frac{1}{n}\sum_{i=1}^n (Z_i^{\text{alt}} - \pi)\phi_i
= \frac{1}{n}\sum_{i=1}^n (T_i - \pi) Q_i^{\phi}
+ O_{\mathbb{P}}\left( \frac{1}{\sqrt{n} \sqrt{n q_n^{\text{pre}}}} \right). 
\label{eq:Ziphi}
\end{align}
\end{lemma}
\begin{proof}[Proof of Lemma \ref{lemma:Zi_alt_phii}]
We make use of the proof of \citet[Theorem 4]{LiWager2022}.
By reordering the index, we have
\begin{align*}
\frac{1}{n}\sum_{i=1}^n (Z_i^{\text{alt}} - \pi) \phi_i
=& \frac{1}{n}\sum_{i=1}^n \frac{ \sum_{j=1}^n A_{ij}^{\text{pre}} (T_j - \pi) }{ \sum_{j=1}^n A_{ij}^{\text{pre}} } \phi_i 
= \frac{1}{n}\sum_{j=1}^n (T_j - \pi) 
\sum_{i\ne j}
\frac{ A_{ij}^{\text{pre}} \phi_i }{ \sum_{k=1}^n A_{ik}^{\text{pre}} }  
\end{align*}   
where 
\begin{align*}
\sum_{i\ne j}
\frac{ A_{ij}^{\text{pre}} \phi_i }{ \sum_{k=1}^n A_{ik}^{\text{pre}} }  
=& \sum_{i\ne j}
\frac{ A_{ij}^{\text{pre}} \phi_i }{ N_i } 
= \sum_{i\ne j}
\frac{ A_{ij}^{\text{pre}} \phi_i }{ (n-1) q_n^{\text{pre}} g_0(i) } 
- \sum_{i\ne j}
\frac{ A_{ij}^{\text{pre}} \phi_i (N_i - (n-1)q_n^{\text{pre}} g_0(i)) }{ (n-1) q_n^{\text{pre}} g_0(i) N_i }. 
\end{align*}
For the first term, for fixed $j$, given $w_j$, $\frac{ A_{ij}^{\text{pre}} \phi_i }{ q_n^{\text{pre}} g_0(i) } $ are i.i.d.. 
Then we have 
\begin{align}
E\left[ \left( \frac{1}{n-1} \sum_{i\ne j}
\frac{ A_{ij}^{\text{pre}} \phi_i }{ q_n^{\text{pre}} g_0(i) } - Q_j^{\phi} \right)^2 \right]
&= \frac{1}{n-1} 
E\left[ \left( 
\frac{ A_{ij}^{\text{pre}} \phi_i }{ q_n^{\text{pre}} g_0(i) } - Q_j^{\phi} \right)^2 \right] \notag \\
&\le \frac{1}{n-1} 
E\left[ \left( 
\frac{ A_{ij}^{\text{pre}} \phi_i^2 }{ (q_n^{\text{pre}} g_0(i))^2 } \right) \right]
\le \frac{C}{(n-1) q_n^{\text{pre}}}. 
\label{eq:term1}
\end{align}
This implies that $\frac{1}{n-1} \sum_{i\ne j}
\frac{ A_{ij}^{\text{pre}} \phi_i }{ q_n^{\text{pre}} g_0(i) }$ can be well approximated by $Q_j^{\phi}$ with a small error of $O_{\mathbb{P}}\left( \frac{1}{nq_n^{\text{pre}}} \right)$.
For the second term, we start by noting that 
\begin{align}
\frac{ A_{ij}^{\text{pre}} \phi_i (N_i - (n-1)q_n^{\text{pre}} g_0(i)) }{ q_n^{\text{pre}} g_0(i) N_i }
= \frac{ A_{ij}^{\text{pre}} \phi_i ((N_i - A_{ij}^{\text{pre}} + 1) - (n-1) q_n^{\text{pre}} g_0(i))  }{ q_n^{\text{pre}} g_0(i) (N_i - A_{ij}^{\text{pre}} + 1) }.   
\label{eq:AijNi}
\end{align}
Conditional on $w$, $(N_i - A_{ij}^{\text{pre}})$ is distributed as a $\text{Binomial}(n-1, q_n^{\text{pre}} g_0(i))$. 
By the property of binomial distribution, for a random variable $X\sim \text{Binomial}(n,p)$, we have 
\begin{align*}
E\left( \frac{1}{X+1} \right) 
=& \frac{1-(1-p)^{n+1}}{(n+1)p}.
\end{align*}
Thus, we have
\begin{align}
& E\left[ \frac{ (N_i - A_{ij}^{\text{pre}} + 1) - (n-1) q_n^{\text{pre}} g_0(i)  }{ (n-1)q_n^{\text{pre}} g_0(i)(N_i - A_{ij}^{\text{pre}} + 1) } \mid w \right]
= \frac{ 1 }{ (n-1) q_n^{\text{pre}} g_0(i)} - E\left[ \frac{ 1  }{ N_i - A_{ij}^{\text{pre}} + 1 } \mid w \right] \notag \\
~=~& \frac{1}{(n-1) q_n^{\text{pre}} g_0(i)} - \frac{1-(1-q_n^{\text{pre}} g_0(i))^{n}}{n q_n^{\text{pre}} g_0(i)} 
= \frac{ 1 }{ n(n-1) q_n^{\text{pre}} g_0(i)} 
+ \frac{1}{n} \frac{ (1-q_n^{\text{pre}} g_0(i))^{n} }{ q_n^{\text{pre}} g_0(i) } \notag \\
~\le~ & \frac{ 1 }{ n(n-1) q_n^{\text{pre}} g_0(i)} 
+ \frac{ e^{-C n q_n^{\text{pre}}} }{ n q_n^{\text{pre}} g_0(i) }
\le C \frac{ 1 }{ \left(n q_n^{\text{pre}}\right)^2 }.
\label{eq:inverBij}
\end{align}
By Cauchy--Schwarz inequality and Lemma \ref{lemma:15},
\begin{align}
& E\left[ \left( \frac{ (N_i - A_{ij}^{\text{pre}} + 1) - (n-1) q_n^{\text{pre}} g_0(i)  }{ (n-1) q_n^{\text{pre}} g_0(i) (N_i - A_{ij}^{\text{pre}} + 1) } \right)^2 \mid w \right] \notag \\
\le & \frac{1}{(n-1)^2 (q_n^{\text{pre}} g_0(i))^2} 
\sqrt{E\left( \left( (N_i - A_{ij}^{\text{pre}} + 1) - (n-1) q_n^{\text{pre}} g_0(i)\right)^4 \right)}
\sqrt{E\left( \frac{1}{(N_i - A_{ij}^{\text{pre}} + 1)^4}  \right)} \notag \\
\le & C \frac{1}{\left((n-1) q_n^{\text{pre}}\right)^2} \frac{n q_n^{\text{pre}}}{\left(n q_n^{\text{pre}}\right)^2}
\le C \frac{1}{\left(n q_n^{\text{pre}}\right)^3}. \label{eq:inverBij^2}
\end{align}
Let $B_{ij} = \frac{ A_{ij}^{\text{pre}} \phi_i (N_i - (n-1)q_n^{\text{pre}} g_0(i)) }{ (n-1) q_n^{\text{pre}} g_0(i) N_i }$. 
Then we bound the second moment:
\begin{align*}
E\left[ \left( \sum_{i\ne j} B_{ij} \right)^2 \right]
= \sum_{i\ne j} E\left[ B_{ij}^2 \right]
+ \sum_{(i,k)} E\left[ B_{ij} B_{kj} \right].
\end{align*}
For the diagonal term, we have 
\begin{align*}
& E\left[ B_{ij}^2 \right]
= E\left[ A_{ij}^{\text{pre}} \phi_i^2 \left( \frac{ (N_i - A_{ij}^{\text{pre}} + 1) - (n-1) q_n^{\text{pre}} g_0(i)  }{ (n-1) q_n^{\text{pre}} g_0(i) (N_i - A_{ij}^{\text{pre}} + 1) } \right)^2 \right] 
\le C \frac{1}{ n (n q_n^{\text{pre}})^2 }
\end{align*}
by \eqref{eq:inverBij^2}. 
For the off-diagonal terms, conditional on $w$, $A_{ij}^{\text{pre}}$, $A_{jk}^{\text{pre}}$, $\frac{N_j - (n-1)q_n^{\text{pre}} g_0(w_j)}{N_j}$ and $\frac{N_k - (n-1)q_n^{\text{pre}} g_0(w_k)}{N_k}$ are all independent. Thus
\begin{align*}
 E [B_{ij} B_{ik}] 
=& E \left[ (q_n^{\text{pre}})^2 g_0(i, j) g_0(i, k) \phi_j \phi_k 
E \left[ \frac{ N_j - (n - 1) q_n^{\text{pre}} g_0(w_j) }{ (n-1) q_n^{\text{pre}} g_0(w_j) N_j} \mid w \right] 
E \left[ \frac{ N_k - (n - 1) q_n^{\text{pre}} g_0(w_k) }{ (n-1) q_n^{\text{pre}} g_0(w_k) N_k} \mid w \right] \right] \\
\le & C \frac{1}{n^2 (n q_n^{\text{pre}})^2},
\end{align*}
where the last inequality follows from \eqref{eq:inverBij}. Combing the diagonal and off-diagonal terms, we have
\begin{equation}
E\left[ \left( \sum_{i\ne j}
\frac{ A_{ij}^{\text{pre}} \phi_i (N_i - (n-1)q_n^{\text{pre}} g_0(i)) }{ (n-1) q_n^{\text{pre}} g_0(i) N_i } \right)^2 \right]
\leq \frac{C }{(n q_n^{\text{pre}})^2}.
\label{eq:term2}
\end{equation}
Combining \eqref{eq:term1} and \eqref{eq:term2}, we get
\[
E \left[ \left( \sum_{i\ne j}
\frac{ A_{ij}^{\text{pre}}  }{ N_i } \phi_i - Q_{j}^{\phi} \right)^2 \right] \leq \frac{C}{n q_n^{\text{pre}}}.
\]
This implies that 
\begin{align*}
& \frac{1}{n}\sum_{j=1}^n (T_j - \pi) \left( \sum_{i\ne j}
\frac{ A_{ij}^{\text{pre}} \phi_i }{ \sum\limits_{k\ne i}A_{ik}^{\text{pre}} } - Q_j^{\phi} \right) 
= O_{\mathbb{P}}\left( \frac{1}{\sqrt{n} \sqrt{nq_n^{\text{pre}}}} \right) 
\end{align*} 
and thus 
\begin{align*}
\frac{1}{n}\sum_{i=1}^n (Z_i^{\text{alt}} - \pi)\phi_i
= \frac{1}{n}\sum_{i=1}^n (T_i - \pi) Q_i^{\phi}
+ O_{\mathbb{P}}\left( \frac{1}{\sqrt{n} \sqrt{nq_n^{\text{pre}}}} \right).
\end{align*}
\end{proof}

Define
\[
S_1^0
= \frac{1}{n}\sum_{j=1}^n (T_j - \pi) \sum_{i \ne j} \frac{ A_{ij}^{\text{pre}} r_{0,i} }{ (n-1)q_n^{\text{pre}} g_0(i) }
\]
with  $E(S_1^0) \asymp \frac{ 1 }{ n \max\{q_n^{\text{pre}}, q_n^{\text{post}}\} }$.
In Lemma \ref{lemma:ZiMi_alt} below, we analyze the order of $\frac{1}{n} \sum_{i=1}^n (Z_i^{\text{alt}} - \pi) M_i$.

\begin{lemma}\label{lemma:ZiMi_alt} 
Under Assumptions \ref{asu:network} and \ref{asu:linearY}, we have
\begin{align}
\frac{1}{n} \sum_{i=1}^n (Z_i^{\text{alt}} - \pi) M_i
=& 
E(S_1^0)
+ o_{\mathbb{P}}\left( \frac{ 1 }{ n \max\{q_n^{\text{pre}}, q_n^{\text{post}}\} } \right)
+ O_{\mathbb{P}}\left( \frac{1}{ \sqrt{n}  } \right).
\label{eq:ZiMi_alt}
\end{align}
\end{lemma}
\begin{proof}
See proof in Section \ref{proof:ZiMi_alt}.    
\end{proof}

Now we are ready to prove Theorem \ref{thm:IV_alt_consistency}.

\begin{proof}[Proof of Theorem \ref{thm:IV_alt_consistency}]
We start with deriving the probability limit of $\left( \tilde{Z}^{\text{alt} \top} X \right)^{-1}$:
\begin{align*}
& \left( \tilde{Z}^{\text{alt} \top} X \right)^{-1} 
= \frac{ 1 }{ \det( \tilde{Z}^{\text{alt} \top}X) } \begin{pmatrix}
* & * & * \\
* & d_{22} & d_{23} \\
* & d_{32} & d_{33}
\end{pmatrix}
\end{align*}  
where 
\begin{align*}
d_{22}
~=~& \left( \frac{1}{n}\sum_{i=1}^n M_i Z_i^{\text{alt}}  \right) - \left( \frac{1}{n}\sum_{i=1}^n M_i \right) \left( \frac{1}{n }\sum_{i=1}^n Z_i^{\text{alt}} \right) \\
d_{23}
~=~& a_{23} \\
d_{32}
~=~& - \left( \frac{1}{n}\sum_{i=1}^n T_i Z_i^{\text{alt}} \right) 
+ \left( \frac{1}{n}\sum_{i=1}^n T_i \right)\left( \frac{1}{n }\sum_{i=1}^n Z_i^{\text{alt}} \right)  \\
d_{33} 
~=~& a_{33}
\end{align*}
and
\begin{align*}
\det( \tilde{Z}^{\text{alt} \top}X) = d_{33} d_{22} - d_{32} d_{23}.
\end{align*}
Let $*$ denote the terms that are not relevant to show Theorem \ref{thm:IV_alt_consistency}.
We analyze the terms one by one. \\
\underline{(1): $d_{32}$.}
By Lemma \ref{lemma:Zi_alt_phii}, we have
\begin{align*}
d_{32}
=& - \left( \frac{1}{n}\sum_{i=1}^n (T_i - \pi) (Z_i^{\text{alt}} - \pi) \right) + \left( \frac{1}{n }\sum_{i=1}^n (Z_i^{\text{alt}} - \pi) \right) \left( \frac{1}{n}\sum_{i=1}^n (T_i - \pi) \right) \\
=& - \left( \frac{1}{n}\sum_{i=1}^n (T_i - \pi) E\left[ \frac{ A_{ij}^{\text{pre}} (T_j - \pi) }{ q_n^{\text{pre}} g_0(j) } \mid w_i \right] + O_{\mathbb{P}}\left( \frac{1}{\sqrt{n}\sqrt{n q_n^{\text{pre}}}} \right) \right)  \\
&+ \left( \frac{1}{n}\sum_{i=1}^n (T_i - \pi) E\left[ \frac{ A_{ij}^{\text{pre}} }{ q_n^{\text{pre}} g_0(j) } \mid w_i \right] + O_{\mathbb{P}}\left( \frac{1}{\sqrt{n}\sqrt{n q_n^{\text{pre}}}} \right) \right) O_{\mathbb{P}}\left( \frac{1}{\sqrt{n} } \right) \\
=& O_{\mathbb{P}}\left( \frac{1}{ n \sqrt{q_n^{\text{pre}}}} \right).
\end{align*}
\underline{(2): $\frac{1}{n}\sum_{i=1}^n (Z_i^{\text{alt}} - \bar{Z}) u_i$.}
Similarly, by Lemma \ref{lemma:Zi_alt_phii}, we have
\begin{align*}
& \frac{1}{n}\sum_{i=1}^n (Z_i^{\text{alt}} - \bar{Z}^{\text{alt}}) u_i
= \frac{1}{n}\sum_{i=1}^n (Z_i^{\text{alt}} - \pi) u_i - (\bar{Z}^{\text{alt}} - \pi) \frac{1}{n}\sum_{i=1}^n u_i
= O_{\mathbb{P}}\left( \frac{1}{\sqrt{n}} \right).
\end{align*}
We consider the remaining terms under Cases (a) and (b), respectively.
\paragraph{Case (a)}
\underline{(1): $d_{22}$.} 
By Lemma \ref{lemma:ZiMi_alt}, we have
\begin{align*}
d_{22}
=& \left( \frac{1}{n}\sum_{i=1}^n M_i (Z_i^{\text{alt}} - \pi)  \right) - \left( \frac{1}{n}\sum_{i=1}^n M_i \right) \left( \frac{1}{n }\sum_{i=1}^n (Z_i^{\text{alt}} - \pi)\right) \\
=& \left( \frac{1}{n}\sum_{i=1}^n (\xi_i + \mu_{r_1,i}) (Z_i^{\text{alt}} - \pi)  \right) - \left( \frac{1}{n}\sum_{i=1}^n (\xi_i + \mu_{r_1,i}) \right) \left( \frac{1}{n }\sum_{i=1}^n (Z_i^{\text{alt}} - \pi) \right) \\
&+ \left( \frac{1}{n}\sum_{i=1}^n (r_{0,i} + r_{1,i} - \mu_{r_1,i}) (Z_i^{\text{alt}} - \pi)  \right) - \left( \frac{1}{n}\sum_{i=1}^n (r_{0,i} + r_{1,i} - \mu_{r_1,i}) \right) \left( \frac{1}{n }\sum_{i=1}^n (Z_i^{\text{alt}} - \pi) \right) \\
=& \left( \frac{1}{n}\sum_{i=1}^n (T_i - \pi) E\left[ \frac{ A_{ij}^{\text{pre}} (\xi_j + \mu_{r_1,j}) }{ q_n^{\text{pre}} g_0(j) } \mid w_i \right] + O_{\mathbb{P}}\left( \frac{ 1 }{\sqrt{n} \sqrt{n q_n^{\text{pre}}} } \right) \right) \\
&- \left( E(\xi_i + \mu_{r_1,i}) + O_{\mathbb{P}}\left( \frac{ 1 }{\sqrt{n} } \right) \right) 
\left( \frac{1}{n}\sum_{i=1}^n (T_i - \pi) E\left[ \frac{ A_{ij}^{\text{pre}} }{ q_n^{\text{pre}} g_0(j) } \mid w_i \right] + O_{\mathbb{P}}\left( \frac{ 1 }{\sqrt{n} \sqrt{n q_n^{\text{pre}}} } \right) \right) \\
&+ \left( \frac{1}{n}\sum_{i=1}^n (r_{0,i} + r_{1,i} - \mu_{r_1,i}) (Z_i^{\text{alt}} - \pi)  \right) 
- O_{\mathbb{P}}\left( \frac{ 1 }{ n\sqrt{q_n^{\text{post}}} } \right)
\left( \frac{1}{n }\sum_{i=1}^n (Z_i^{\text{alt}} - \pi) \right) \\
=& E(S_1^0)
+ O_{\mathbb{P}}\left( \frac{ 1 }{\sqrt{n} } \right) 
+ o_{\mathbb{P}}\left( \frac{ 1 }{ n \max\{q_n^{\text{pre}}, q_n^{\text{post}}\} } \right).
\end{align*}
\underline{(2): $d_{23}$.}
By \eqref{eq:a_23_limit}, we can show $d_{23}
= - \textup{Cov}(T_i, \xi_i + \mu_{r_1,i})
+ O_{\mathbb{P}}\left( \frac{1}{ \sqrt{n} } \right)$. \\
\underline{(3): $\det( \tilde{Z}^{\text{alt} \top}X)$.}
By combining the results above, we have 
\begin{align*}
& \det( \tilde{Z}^{\text{alt} \top}X)
= d_{33} d_{22} - d_{32} d_{23} \\
=& \left( \pi(1-\pi) + O_{\mathbb{P}}\left( \frac{ 1 }{\sqrt{n} } \right) 
\right)
\left( E(S_1^0)
+ O_{\mathbb{P}}\left( \frac{ 1 }{\sqrt{n} } \right) 
+ o_{\mathbb{P}}\left( \frac{ 1 }{ n \max\{q_n^{\text{pre}}, q_n^{\text{post}}\} } \right) \right) \\
&- \left( - \textup{Cov}(T_i, \xi_i + \mu_{r_1,i})
+ O_{\mathbb{P}}\left( \frac{1}{ \sqrt{n} } \right) \right)
O_{\mathbb{P}}\left( \frac{1}{ n \sqrt{q_n^{\text{pre}}}} \right) \\
=& \pi(1-\pi) E(S_1^0) + O_{\mathbb{P}}\left( \frac{ 1 }{\sqrt{n} } \right)
+ o_{\mathbb{P}}\left( \frac{ 1 }{ n \max\{q_n^{\text{pre}}, q_n^{\text{post}}\} } \right).
\end{align*}
Therefore, by the closed form of IV estimates, we can show that
\begin{align*}
\hat{\beta}^\text{alt}_1 - {\beta}_1
~=~& \frac{d_{22} \frac{1}{n}\sum_{i=1}^n (T_i - \bar{T}) u_i + d_{23} \frac{1}{n}\sum_{i=1}^n (Z_i^{\text{alt}} - \bar{Z}^{\text{alt}}) u_i }{ d_{33} d_{22} - d_{32} d_{23} }  \\
~=~& \frac{ \left( E(S_1^0) + O_{\mathbb{P}}\left( \frac{ 1 }{\sqrt{n} } \right) \right) O_{\mathbb{P}}\left( \frac{ 1 }{\sqrt{n} } \right) + \left( - \textup{Cov}(T_i, \xi_i + \mu_{r_1,i})
+ O_{\mathbb{P}}\left( \frac{1}{ \sqrt{n} } \right) \right) O_{\mathbb{P}}\left( \frac{ 1 }{\sqrt{n} } \right) }{ \pi(1-\pi) E(S_1^0) + O_{\mathbb{P}}\left( \frac{ 1 }{\sqrt{n} } \right) } \\
~=~& 
\frac{ O_{\mathbb{P}}\left( \frac{1}{\sqrt{n} } \right) }{ \pi(1-\pi) E(S_1^0) + O_{\mathbb{P}}\left( \frac{1}{\sqrt{n} } \right) }, 
\end{align*}
and
\begin{align*}
\hat{\beta}^\text{alt}_2 - {\beta}_2
~=~& \frac{ d_{33} \frac{1}{n}\sum_{i=1}^n (Z_i^{\text{alt}} - \bar{Z}^{\text{alt}}) u_i + d_{32} \frac{1}{n}\sum_{i=1}^n (T_i - \bar{T}) u_i }{ d_{33} d_{22} - d_{32} d_{23} } \\
~=~& \frac{ \left( \pi(1-\pi) + O_{\mathbb{P}}\left( \frac{ 1 }{\sqrt{n} } \right) \right) O_{\mathbb{P}}\left( \frac{ 1 }{\sqrt{n} } \right) + O_{\mathbb{P}}\left( \frac{ 1 }{ n \sqrt{q_n^{\text{pre}}} } \right) O_{\mathbb{P}}\left( \frac{ 1 }{\sqrt{n} } \right) }{ \pi(1-\pi) E(S_1^0) + O_{\mathbb{P}}\left( \frac{ 1 }{\sqrt{n} } \right) } \\
~=~& \frac{ O_{\mathbb{P}}\left( \frac{1}{\sqrt{n} } \right) }{ \pi(1-\pi) E(S_1^0) + O_{\mathbb{P}}\left( \frac{1}{\sqrt{n} } \right)  }.
\end{align*}
Recall that $E(S_1^0) \asymp \frac{1}{ n \max\{q_n^{\text{pre}}, q_n^{\text{post}}\} } $.
We can conclude that with $\max\{q_n^{\text{pre}}, q_n^{\text{post}}\} \prec \frac{1}{\sqrt{n}}$,
then
\[
\hat{\beta}^{\text{alt}} - {\beta} =
O_{\mathbb{P}}\left( \sqrt{n} \max\{q_n^{\text{pre}}, q_n^{\text{post}}\} \right).
\]

\paragraph{Case (b)}
\underline{(1) $d_{22}$.}
By Lemma \ref{lemma:ZiMi_alt}, we have
\begin{align*}
& d_{22}
= \left( \frac{1}{n}\sum_{i=1}^n M_i Z_i^{\text{alt}}  \right) - \left( \frac{1}{n}\sum_{i=1}^n M_i \right) \left( \frac{1}{n }\sum_{i=1}^n Z_i^{\text{alt}} \right) 
= \left( \frac{1}{n}\sum_{i=1}^n M_i (Z_i^{\text{alt}} - \pi)  \right) - \left( \frac{1}{n}\sum_{i=1}^n M_i \right) \left( \frac{1}{n }\sum_{i=1}^n (Z_i^{\text{alt}} - \pi)\right) \\
&= \left( \frac{1}{n}\sum_{i=1}^n \mu_{r_1,i} (Z_i^{\text{alt}} - \pi)  \right) - \left( \frac{1}{n}\sum_{i=1}^n \mu_{r_1,i} \right) \left( \frac{1}{n }\sum_{i=1}^n (Z_i^{\text{alt}} - \pi) \right) \\
&+ \left( \frac{1}{n}\sum_{i=1}^n (r_{0,i} + r_{1,i} - \mu_{r_1,i}) (Z_i^{\text{alt}} - \pi)  \right) - \left( \frac{1}{n}\sum_{i=1}^n (r_{0,i} + r_{1,i} - \mu_{r_1,i}) \right) \left( \frac{1}{n }\sum_{i=1}^n (Z_i^{\text{alt}} - \pi) \right) \\
&= \left( \frac{1}{n}\sum_{i=1}^n (T_i - \pi) E\left[ \frac{ A_{ij}^{\text{pre}} \mu_{r_1,j} }{ q_n^{\text{pre}} g_0(j) } \mid w_i \right] + O_{\mathbb{P}}\left( \frac{ 1 }{\sqrt{n} \sqrt{n q_n^{\text{pre}}} } \right) \right) 
- \left( E(\mu_{r_1,i}) + O_{\mathbb{P}}\left( \frac{ 1 }{\sqrt{n} n q_n^{\text{post}} } \right) \right) 
\left( \frac{1}{n }\sum_{i=1}^n (Z_i^{\text{alt}} - \pi) \right) \\
&+ \left( \frac{1}{n}\sum_{i=1}^n (r_{0,i} + r_{1,i} - \mu_{r_1,i}) (Z_i^{\text{alt}} - \pi)  \right) 
- O_{\mathbb{P}}\left( \frac{ 1 }{ n\sqrt{q_n^{\text{post}}} } \right)
\left( \frac{1}{n }\sum_{i=1}^n (Z_i^{\text{alt}} - \pi) \right) \\
&= E(S_1^0)
+ O_{\mathbb{P}}\left( \frac{ 1 }{\sqrt{n} n q_n^{\text{post}} } \right)
+ o_{\mathbb{P}}\left( \frac{ 1 }{ n \max\{q_n^{\text{pre}}, q_n^{\text{post}}\} } \right).
\end{align*}
\underline{(2) $d_{23}$.}
By \eqref{eq:a_23_limit}, we can show $d_{23} = O_{\mathbb{P}}\left( \frac{1}{n\sqrt{q_n^{\text{post}}}} \right)$. \\
\underline{(3) $\det( \tilde{Z}^{\text{alt} \top}X)$.}
By combining the results above, we have 
\begin{align*}
& \det( \tilde{Z}^{\text{alt} \top}X)
= d_{33} d_{22} - d_{32} d_{23} \\
&= \left( \pi(1-\pi) + O_{\mathbb{P}}\left( \frac{ 1 }{\sqrt{n} } \right) 
\right)
\left( E(S_1^0)
+ O_{\mathbb{P}}\left( \frac{ 1 }{\sqrt{n} n q_n^{\text{post}} } \right) 
+ o_{\mathbb{P}}\left( \frac{ 1 }{ n \max\{q_n^{\text{pre}}, q_n^{\text{post}}\} } \right) \right) \\
&- O_{\mathbb{P}}\left( \frac{1}{ n \sqrt{q_n^{\text{post}}} } \right)
O_{\mathbb{P}}\left( \frac{1}{ n \sqrt{q_n^{\text{pre}}}} \right) \\
&= \pi(1-\pi) E(S_1^0) 
+ O_{\mathbb{P}}\left( \frac{ 1 }{\sqrt{n} n q_n^{\text{post}}} \right)
+ o_{\mathbb{P}}\left( \frac{ 1 }{ n \max\{q_n^{\text{pre}}, q_n^{\text{post}}\} } \right).
\end{align*}
Therefore, we have 
\begin{align*}
& \hat{\beta}^{\text{alt}}_1 - {\beta}_1
~=~ \frac{d_{22} \frac{1}{n}\sum_{i=1}^n (T_i - \bar{T}) u_i + d_{23} \frac{1}{n}\sum_{i=1}^n (Z_i^{\text{alt}} - \bar{Z}^{\text{alt}}) u_i }{ d_{33} d_{22} - d_{32} d_{23} }  \\
~=~& \frac{ \left( E(S_1^0)
+ O_{\mathbb{P}}\left( \frac{ 1 }{\sqrt{n} n q_n^{\text{post}} } \right)
+ o_{\mathbb{P}}\left( \frac{ 1 }{ n \max\{q_n^{\text{pre}}, q_n^{\text{post}}\} } \right) \right) O_{\mathbb{P}}\left( \frac{1}{\sqrt{n}} \right) + O_{\mathbb{P}}\left( \frac{1}{n\sqrt{q_n^{\text{post}}}} \right) O_{\mathbb{P}}\left( \frac{1}{\sqrt{n} } \right) }{ \pi(1-\pi) E(S_1^0) 
+ O_{\mathbb{P}}\left( \frac{ 1 }{\sqrt{n} n q_n^{\text{post}}} \right)
+ o_{\mathbb{P}}\left( \frac{ 1 }{ n \max\{q_n^{\text{pre}}, q_n^{\text{post}}\} } \right) } \\
~=~& \frac{ O_{\mathbb{P}}\left( \frac{1}{\sqrt{n} n q_n^{\text{post}}} \right)   }{ \pi(1-\pi) E(S_1^0) +  O_{\mathbb{P}}\left( \frac{ 1 }{\sqrt{n} n q_n^{\text{post}} } \right) 
+ o_{\mathbb{P}}\left( \frac{ 1 }{ n \max\{q_n^{\text{pre}}, q_n^{\text{post}}\} } \right) }. 
\end{align*}
Therefore, with $ \max\{q_n^{\text{pre}}, q_n^{\text{post}} \} \prec \sqrt{n}  q_n^{\text{post}} $,
then
\begin{align*}
 \hat{\beta}^{\text{alt}}_1 - {\beta}_1
= O_{\mathbb{P}}\left( \frac{1}{\sqrt{n}} \max\left\{ \frac{ \max\{q_n^\text{pre}, q_n^\text{post} \} }{  \sqrt{ q_n^\text{post}}  }, 1 \right\}  \right).
\end{align*}
Then,
\begin{align*}
\hat{\beta}^{\text{alt}}_2 - {\beta}_2
~=~& \frac{ d_{33} \frac{1}{n}\sum_{i=1}^n (Z_i^{\text{alt}} - \bar{Z}^{\text{alt}} ) u_i + d_{32} \frac{1}{n}\sum_{i=1}^n (T_i - \bar{T}) u_i }{ d_{33} d_{22} - d_{32} d_{23} } \\
~=~& \frac{ \left( \pi(1-\pi) + O_{\mathbb{P}}\left( \frac{1}{\sqrt{n}} \right) \right) O_{\mathbb{P}}\left( \frac{1}{\sqrt{n}} \right) + O_{\mathbb{P}}\left( \frac{1}{ n \sqrt{q_n^{\text{pre}}}} \right) O_{\mathbb{P}}\left( \frac{1}{ \sqrt{n}} \right) }{ \pi(1-\pi) E(S_1^0) 
+ O_{\mathbb{P}}\left( \frac{ 1 }{\sqrt{n} n q_n^{\text{post}}} \right)
+ o_{\mathbb{P}}\left( \frac{ 1 }{ n \max\{q_n^{\text{pre}}, q_n^{\text{post}}\} } \right) }  \\
~=~& \frac{ O_{\mathbb{P}}\left( \frac{1}{\sqrt{n}} \right) }{ \pi(1-\pi) E(S_1^0) 
+ O_{\mathbb{P}}\left( \frac{ 1 }{\sqrt{n} n q_n^{\text{post}}} \right)
+ o_{\mathbb{P}}\left( \frac{ 1 }{ n \max\{q_n^{\text{pre}}, q_n^{\text{post}}\} } \right) }.
\end{align*}
Therefore, with $n \max\{q_n^{\text{pre}}, q_n^{\text{post}}\} \prec \sqrt{n} $,
then
\begin{align*}
\hat{\beta}^{\text{alt}}_2 - {\beta}_2
= O_{\mathbb{P}}\left( \sqrt{n} \max\{q_n^{\text{pre}}, q_n^{\text{post}}\} \right). 
\end{align*}
Therefore, $\hat{\beta}^{\text{alt}}_0$ has the same convergence rate as $\hat{\beta}^{\text{alt}}_2$, such that 
when 
$n \max\{q_n^{\text{pre}}, q_n^{\text{post}}\} \prec \sqrt{n} $, 
\begin{align*}
\hat{\beta}^{\text{alt}}_0 - {\beta}_0
= O_{\mathbb{P}}\left( \sqrt{n} \max\{q_n^{\text{pre}}, q_n^{\text{post}}\} \right). 
\end{align*}

\end{proof}

\begin{proof}[Proof of Corollary \ref{cor:alt_IV_consistency}]
Corollary \ref{cor:alt_IV_consistency} is a direct result of Theorem \ref{thm:IV_alt_consistency} under $q_n^{\text{pre}} \preccurlyeq q_n^{\text{post}}$.    
\end{proof}

\subsection{Simulation results}
In this subsection, we provide simulation results to support Theorem \ref{thm:IV_alt_consistency}.
The networks and outcomes are generated in the same fashion as in Section \ref{sec:MC_IV}.
Table \ref{tab:IV_alt} presents the results for sparsity levels $q_n^{\text{pre}} = q_n^{\text{post}} = q_n$, with $q_n \in \{n^{-2/3}, n^{-1/3}, n^{-1/5}\}$ and sample sizes $n \in \{200, 800\}$.
We report the average of the estimates under ``$\hat{\beta}^{\text{alt}}$,'' along with the standard deviation across the simulations under ``std($\hat{\beta}^{\text{alt}}$).''
It demonstrates that the estimates are consistent if the standard deviation decreases with larger sample sizes.
We observe that the estimates are consistent for all designs when 
$q_n = n^{-2/3}$. 
However, when $q_n$ exceeds $n^{-1/2}$, $\hat{\beta}^{\text{alt}}$ is no longer consistent for Designs 1 and 2. 
For Designs 3 and 4, $\hat{\beta}^{\text{alt}}_1$ remains consistent, while $\hat{\beta}^{\text{alt}}_2$ is not.

\begin{table}[ht]
\caption{Simulation results of normalized SSIV in \eqref{eq:SSIV_ratio}}
\centering
\footnotesize
\label{tab:IV_alt}
\begin{tabular}{cl|cc|cc|cc}
\hline \hline
& & \multicolumn{2}{c}{$q_n = n^{-2/3}$} & \multicolumn{2}{c}{$q_n = n^{-1/3}$}	 & \multicolumn{2}{c}{$q_n = n^{-1/5}$} \\			
\hline	
Design  & $\beta$ & $\hat{\beta}^{\text{alt}}$ &  std($\hat{\beta}^{\text{alt}}$) & $\hat{\beta}^{\text{alt}}$ &  std($\hat{\beta}^{\text{alt}}$) & $\hat{\beta}^{\text{alt}}$ &  std($\hat{\beta}^{\text{alt}}$) \\
\hline 
\multicolumn{8}{c}{$n = 200$} \\
\hline 
1	&	$\beta_1=1$	&	1.003	&	0.082	&	1.001	&	0.099	&	1.004	&	0.209	\\
	&	$\beta_2=0.5$	&	0.476	&	0.192	&	0.486	&	0.639	&	0.467	&	1.367	\\
\hline															
2	&	$\beta_1=1$	&	1.001	&	0.082	&	0.990	&	0.131	&	0.935	&	0.386	\\
	&	$\beta_2=0.5$	&	0.506	&	0.229	&	0.606	&	0.908	&	1.005	&	2.528	\\
\hline															
3	&	$\beta_1=1$	&	1.000	&	0.082	&	0.998	&	0.084	&	0.997	&	0.083	\\
	&	$\beta_2=0.5$	&	0.518	&	0.164	&	0.503	&	0.474	&	0.515	&	0.759	\\
\hline															
4	&	$\beta_1=1$	&	1.000	&	0.082	&	1.001	&	0.087	&	1.001	&	0.092	\\
	&	$\beta_2=0.5$	&	0.496	&	0.211	&	0.494	&	0.797	&	0.509	&	1.701	\\
\hline 
\multicolumn{8}{c}{$n = 800$} \\
\hline 
1	&	$\beta_1=1$	&	1.003	&	0.039	&	1.002	&	0.113	&	1.006	&	0.336	\\
	&	$\beta_2=0.5$	&	0.488	&	0.115	&	0.483	&	0.716	&	0.466	&	1.937	\\
\hline															
2	&	$\beta_1=1$	&	1.000	&	0.040	&	0.986	&	0.135	&	0.841	&	2.325	\\
	&	$\beta_2=0.5$	&	0.499	&	0.149	&	0.605	&	0.888	&	1.479	&	13.681	\\
\hline															
3	&	$\beta_1=1$	&	0.999	&	0.041	&	1.000	&	0.041	&	1.000	&	0.043	\\
	&	$\beta_2=0.5$	&	0.505	&	0.105	&	0.499	&	0.370	&	0.491	&	0.637	\\
\hline															
4	&	$\beta_1=1$	&	1.000	&	0.040	&	0.999	&	0.043	&	1.000	&	0.047	\\
	&	$\beta_2=0.5$	&	0.500	&	0.142	&	0.493	&	0.796	&	0.496	&	2.075	\\
\hline
\hline
\end{tabular}
\caption*{\small Note: Simulation results for the IV estimators using the normalized SSIV with $n \in \{200, 800\}$, $q_n \in \{n^{-2/3}, n^{-1/3}, n^{-1/5}\}$, and 5, 000 replications. }
\end{table}

\section{Proof of Lemmas}
\label{app:lemma}

\subsection{Auxiliary Lemmas}

\begin{lemma}\label{lemma:Eri^c}
Under Assumption \ref{asu:network}, we have for $c\in \{2,4\}$,
\begin{align}
E\left[ (r_{0,i} + r_{1,i} - \mu_{r_1,i})^c \right]
= O\left( \frac{1}{(n q_n^{\text{post}})^{c/2}} \right). 
\label{eq:Er^c}
\end{align}
\end{lemma}

\begin{proof}[Proof of Lemma \ref{lemma:Eri^c}]
We only show it for $c=4$ since the argument for $c=2$ is analogous and we omit it for brevity.
By $C_r$ inequality, we have
\begin{align}
E\left[ (r_{0,i} + r_{1,i} - \mu_{r_1,i})^4 \right]  
\le 3^{3} \cdot \left[  
E\left( r_{0,i}^4 \right) 
+ E\left( r_{1,i}^4 \right) 
+ E\left( \mu_{r_1,i}^4 \right) 
\right]. 
\label{eq:M_i^c}
\end{align}   
For the first term, by the Multinomial Theorem, we have 
\begin{align*}
E\left[ r_{0,i}^4 \mid T_i, w_i \right]
=& \frac{\xi_i^4}{(n-1)^4}  E\left[ \left( \sum_{j\ne i} R_{ij} \right)^4 \mid T_i, w_i \right] 
= \frac{\xi_i^4}{(n-1)^4} 
E\left[ \sum_{(j_1,j_2):j_1 \ne j_2} R_{ij_1}^{2} R_{i_2}^{2} \mid T_i, w_i \right].
\end{align*}
where the last equality follows from $E\left( R_{ij} \mid T_i, w_i \right) = 0$ such that any product involving $R_{ij}$ with a multiplicity of 1 has a mean of zero.
Recall that
\[
R_{i{j_k}}
= \frac{ A_{i{j_k}}^{\text{post}} T_{j_{j_k}} }{E(A_{i{j_k}}^{\text{post}} T_{j_k}\mid T_i, w_i)} - \frac{A_{i{j_k}}^{\text{post}}}{E(A_{i{j_k}} \mid T_i, w_i)}
= \frac{A_{i{j_k}}^{\text{post}} }{q_n^{\text{post}}}  \left( \frac{ T_{j_k} }{E(f_1(i,{j_k}) T_{j_k} \mid T_i, w_i)} - \frac{1}{f_1(i)} \right).
\]
Therefore, $E\left[ r_{0,i}^4 \right]
\le \frac{C}{(n q_n^{\text{post}})^2}.$
For the second term in \eqref{eq:M_i^c}, 
define $X_{ij} = U_{ij} W_{ij}$ and $V_{ij} = W_{ij}^2$.
By AM--GM inequality, we have
\begin{align}
& r_{1,i}^4
= \frac{1}{(n-1)^{8}} 
\left( - \sum_{j\ne i} \frac{X_{ij}}{\theta_i^*(y)^2}   
+ \frac{2\theta_i^*(x)}{\theta_i^*(y)^3} \sum_{j\ne i} V_{ij} 
- \sum_{\substack{(j,l): j\ne l}}  \frac{U_{ij} W_{ik} }{\theta_i^*(y)^2} 
+ \frac{2\theta_i^*(x)}{\theta_i^*(y)^3}   
\sum_{\substack{(j,l): j\ne l}} W_{ij} W_{ik} \right)^4 \notag \\
\le & \frac{ 4^{3} }{(n-1)^{8}} 
\left[ \left( \sum_{j\ne i} \frac{X_{ij}}{\theta_i^*(y)^2} 
 \right)^4
+ \left( \frac{2\theta_i^*(x)}{\theta_i^*(y)^3} \sum_{j\ne i} V_{ij} \right)^4
+ \left(  
\sum_{\substack{(j,l): j\ne l}} \frac{U_{ij} W_{il}}{\theta_i^*(y)^2} \right)^4
+ \left( \frac{2\theta_i^*(x)}{\theta_i^*(y)^3}   
\sum_{\substack{(j,l): j\ne l}} W_{ij} W_{il} \right)^4 \right]. \label{eq:upper}
\end{align}
For the first two terms in \eqref{eq:upper}, the order depends on the number of indexes. We take $ X_{ij}$ for example:
\begin{align*}
E\left[ \left( \sum_{j\ne i} X_{ij} \right)^4 \mid T_i, w_i \right] 
=  E\left[ \sum_{\substack{K\in\{1,\cdots,4\} \\ a_1+\cdots+a_K=4}} \sum_{\substack{(j_1,\cdots,j_K)\\ \text{all distinct}}} X_{ij_1}^{a_{1}} \cdots X_{ij_K}^{a_{K}} \mid T_i, w_i \right].
\end{align*}
Then for any $K$,
\begin{align*}
E\left[ \sum_{\substack{(j_1,\cdots,j_K)\\ \text{all distinct}}} X_{ij_1}^{a_1} \cdots X_{ij_K}^{a_K} \mid T_i, w_i \right]
= \sum_{\substack{(j_1,\cdots,j_K)\\ \text{all distinct}}}
E\left[ X_{ij_1}^{a_1} \mid T_i, w_i \right]
\cdots
E\left[ X_{ij_K}^{a_K} \mid T_i, w_i \right],
\end{align*}
where for any $k$,
\begin{align*}
& E\left[ X_{ij_k}^{a_k} \mid T_i, w_i \right] 
\le 3^{a_k -1} \cdot 
\left[ 
\begin{array}{c}
p_n E(f_1(i,j_k)T_{j_k} \mid T_i, w_i) \\
+ p_n^{a_k+1} f_1(i) \left\{E(f_1(i,j_k)T_{j_k} \mid T_i, w_i)
+ T_{j_k} f_1(i) \right\}^{a_k} \\
+ p_n^{2a_k} E(f_1(i,j_k)T_{j_k} \mid T_i, w_i)^{a_k} f_1(i)^{a_k}
\end{array}
\right].
\end{align*}
Taking expectation over $T_i$ and $w_i$, we have
\begin{align*}
E\left[ \sum_{\substack{(j_1,\cdots,j_K)\\ \text{all distinct}}} X_{ij_1}^{a_1} \cdots X_{ij_K}^{a_K} \right]
\le  C(n-1)^K p_n^K.
\end{align*}
Therefore, the upper bound depends on the largest $K$, and thus we have
\begin{align*}
& E\left[ \left(  \sum_{j\ne i} X_{ij} \right)^4 \right] 
\le  C n^4 (q_n^{\text{post}})^4.  
\end{align*}
For the third term in \eqref{eq:upper}, 
\begin{align*}
E\left[ \left( \sum_{(j,l): j\ne l} U_{ij} W_{il} \right)^c \mid T_i, w_i \right] 
=& E\left[ \sum_{\substack{(j_1,l_1)\\ \text{all distinct}}} \cdots \sum_{\substack{(j_c,l_c)\\ \text{all distinct}}}
U_{ij_1} W_{il_1} \cdots U_{ij_c}W_{il_c} \mid T_i, w_i \right] \\
=& E\left[ 
\sum_{\substack{ K\in\{2,\cdots, 2c\} \\ a_1+\cdots+a_{K}=c \\ b_1+\cdots+b_{K}=c}} 
\sum_{\substack{(j_1,\cdots, j_K)\\ \text{all distinct}}} 
(U_{ij_1}^{a_{1}} W_{ij_1}^{b_{1}}) 
\cdots (U_{ij_K}^{a_{K}}W_{ij_K}^{b_{K}}) \mid T_i, w_i \right].
\end{align*}
For any $K$,
\begin{align*}
E\left[ \sum_{\substack{(j_1,\cdots,j_K)\\ \text{all distinct}}} U_{ij_1}^{a_1} \cdots U_{ij_K}^{a_K}
W_{ij_1}^{b_1} \cdots W_{ij_K}^{b_K}\mid T_i, w_i \right]
= \sum_{\substack{(j_1,\cdots,j_K)\\ \text{all distinct}}}
E\left[ U_{ij_1}^{a_1} W_{ij_1}^{b_1} \mid T_i, w_i \right]
\cdots
E\left[ U_{ij_K}^{a_K} W_{ij_K}^{b_K} \mid T_i, w_i \right].
\end{align*}
If $a_k + b_k = 1$, then 
\[
E\left[ U_{ij_k}^{a_k} W_{ij_k}^{b_k} \mid T_i, w_i \right] = 0.
\]
If $a_k + b_k > 1$, then 
\begin{align*}
& E\left[ U_{ij_k}^{a_k} W_{ij_k}^{b_k} \mid T_i, w_i \right] 
\le  2^{a_k + b_k - 2} E\left[ 
\left( A_{ij}T_j - E(A_{ij}T_j | T_i, w_i)^{a_k} \right)^{1(a_k>0)} 
\left( A_{ij} - E(A_{ij} | T_i, w_i)^{b_k} \right)^{1(b_k>0)} | T_i, w_i \right].
\end{align*}
The order depends on the largest $K$ such that for any $k\le K$ we have $a_k + b_k >1$.
Therefore, we have
\begin{align*}
& E\left[ \left(  \sum_{j\ne k} U_{ij} W_{ik} \right)^4 \right] 
\le C (n q_n^{\text{post}})^4.
\end{align*}
The argument of the last term is analogous to that of the third term, so we omit it here. 
By combining these results, we have $E\left[r_{1,i}^4 \right]
\le \frac{C}{(n q_n^{\text{post}})^4}$.
Also, $E\left[ \mu_{r_1,i}^4 \right] = O\left( \frac{1}{(n q_n^{\text{post}})^4} \right)$ by definition in \eqref{eq:mu_r1i}.
Therefore, we conclude that 
\begin{align*}
E\left[ (r_{0,i} + r_{1,i} - \mu_{r_1,i})^4 \right]
= O\left( \frac{1}{ (n q_n^{\text{post}})^2 } \right).
\end{align*}

\end{proof}

\begin{lemma}\label{lemma:Cov_order}
Define $a_i$ and $b_i$ as any functions of $T_i$ and $w_i$ with constant variance. 
Under Assumption \ref{asu:network}, we have for any $i\ne j$,
\begin{align}
& \textup{Cov}\left( a_i (r_{0,i} + r_{1,i} - \mu_{r_1,i}), b_j (r_{0,j} + r_{1,j} - \mu_{r_1,j}) \right) 
= O\left( \frac{1}{n} \right), 
\label{eq:Cov_r_{0,i} + r_{1,i}_r_{0,j} + r_{1,j}} \\
& \textup{Cov}\left( a_i (r_{0,i} + r_{1,i} - \mu_{r_1,i})^2, b_j (r_{0,j} + r_{1,j} - \mu_{r_1,j})^2 \right) 
= O\left( \frac{1}{n^3 (q_n^{\text{post}})^2} \right)
+ O\left( \frac{1}{n^2} \right).
\label{eq:Cov_r_{0,i} + r_{1,i}^2r_{0,j} + r_{1,j}^2}
\end{align}    
\end{lemma}

We write down the closed form of $r_{0,i} r_{1,i}$, $r_{1,i}^2$ and $r_{1,i_1} r_{1,i_2}$ for reference:
\begin{align*}
r_{0,i}r_{1,i}
= \frac{1}{(n-1)^3} 
\left[
\begin{array}{c}
\xi_i \sum\limits_{k\ne i} R_{ik}
\sum\limits_{l\ne i} 
\left( 
- \frac{U_{il} W_{il} }{\theta_{i}^*(y)^2} 
+ \frac{2 \theta_{i}^*(x)}{\theta_{i}^*(y)^3} W_{il}^2
\right)  \\
+ \xi_i \sum\limits_{k\ne i} R_{ik}
\sum\limits_{(l,h):l\ne h} 
\left( 
- \frac{U_{il} W_{ih} }{\theta_{i}^*(y)^2} 
+ \frac{2 \theta_{i}^*(x)}{\theta_{i}^*(y)^3} U_{il} W_{ih}
\right) 
\end{array}
\right],
\end{align*}
where $r_{0,i}r_{1,i}$ has at most three indexes.
Next,
\begin{align}
r_{1,i}^2
= \frac{1}{(n-1)^4} 
\left[ 
\begin{array}{c}
\sum\limits_{k\ne i} 
\left( 
- \frac{U_{ik} W_{ik} }{\theta_i^*(y)^2} 
+ \frac{2 \theta_i^*(x)}{\theta_i^*(y)^3} W_{ik}^2
\right)^2 \\
+ \sum\limits_{(k,l):k\ne l}
\left( 
- \frac{U_{ik} W_{ik} }{\theta_i^*(y)^2} 
+ \frac{2 \theta_i^*(x)}{\theta_i^*(y)^3} W_{ik}^2
\right)
\left( 
- \frac{U_{il} W_{il} }{\theta_i^*(y)^2} 
+ \frac{2 \theta_i^*(x)}{\theta_i^*(y)^3} W_{il}^2
\right) \\
+ \sum\limits_{(k,l):k\ne l} 
\left(
- \frac{U_{ik} W_{il} }{\theta_i^*(y)^2} 
+ \frac{2 \theta_i^*(x)}{\theta_i^*(y)^3} W_{ik} W_{il}
\right)^2 \\
+ \sum\limits_{\substack{(k,l,h) \\ \text{all distinct}}} 
\left(
- \frac{U_{ik} W_{il} }{\theta_i^*(y)^2} 
+ \frac{2 \theta_i^*(x)}{\theta_i^*(y)^3} W_{ik} W_{il}
\right)
\left(
- \frac{U_{ik} W_{ih} }{\theta_i^*(y)^2} 
+ \frac{2 \theta_i^*(x)}{\theta_i^*(y)^3} W_{ik} W_{ih}
\right) \\
+ \sum\limits_{\substack{(k,l,h,m) \\ \text{all distinct}}} 
\left(
- \frac{U_{ik} W_{il} }{\theta_i^*(y)^2} 
+ \frac{2 \theta_i^*(x)}{\theta_i^*(y)^3} W_{ik} W_{il}
\right)
\left(
- \frac{U_{ih} W_{im} }{\theta_i^*(y)^2} 
+ \frac{2 \theta_i^*(x)}{\theta_i^*(y)^3} W_{ih} W_{im}
\right)\\
+ 2 \sum\limits_{k\ne i} 
\left( 
- \frac{U_{ik} W_{ik} }{\theta_i^*(y)^2} 
+ \frac{2 \theta_i^*(x)}{\theta_i^*(y)^3} W_{ik}^2
\right) 
\sum\limits_{(k,l):k\ne l} 
\left( 
- \frac{U_{ik} W_{il} }{\theta_i^*(y)^2} 
+ \frac{2 \theta_i^*(x)}{\theta_i^*(y)^3} 
W_{ik} W_{il}
\right)
\end{array}
\right],
\label{eq:r1i^2}
\end{align}
where $r_{1,i}^2$ has at most 4 different indexes.
Last,
\begin{align*}
r_{1,i_1} r_{1,i_2}
= \frac{1}{(n-1)^4} 
\left[
\begin{array}{c}
\sum\limits_{l\ne i_1} 
\left( 
- \frac{U_{i_1l} W_{i_1l} }{\theta_{i_1}^*(y)^2} 
+ \frac{2 \theta_{i_1}^*(x)}{\theta_{i_1}^*(y)^3} W_{i_1l}^2
\right) 
\sum\limits_{l\ne i_2} 
\left( 
- \frac{U_{i_2l} W_{i_2l} }{\theta_{i_2}^*(y)^2} 
+ \frac{2 \theta_{i_2}^*(x)}{\theta_{i_2}^*(y)^3} W_{i_2l}^2
\right) \\
+ \sum\limits_{l\ne i_1} 
\left( 
- \frac{U_{i_1l} W_{i_1l}}{\theta_{i_1}^*(y)^2}  
+ \frac{2 \theta_{i_1}^*(x)}{\theta_{i_1}^*(y)^3} W_{i_1l}^2
\right) 
\sum\limits_{(l,h):l\ne h} \left( 
- \frac{U_{i_2l} W_{i_2h} }{\theta_{i_2}^*(y)^2} 
+ \frac{2 \theta_{i_2}^*(x)}{\theta_{i_2}^*(y)^3} W_{i_2l} W_{i_2h}
\right) \\
+ \sum\limits_{(l,h):l\ne h} 
\left( 
- \frac{U_{i_1l} W_{i_1h} }{\theta_{i_1}^*(y)^2} 
+ \frac{2 \theta_{i_1}^*(x)}{\theta_{i_1}^*(y)^3} U_{il} W_{ih}
\right) 
\sum\limits_{l\ne i_2} 
\left( 
- \frac{U_{i_2l} W_{i_2l}}{\theta_{i_2}^*(y)^2}  
+ \frac{2 \theta_{i_2}^*(x)}{\theta_{i_2}^*(y)^3} W_{i_2l}^2
\right) \\
+ \sum\limits_{(l,h):l\ne h} 
\left( 
- \frac{U_{i_1l} W_{i_1h}}{\theta_{i_1}^*(y)^2}  
+ \frac{2 \theta_{i_1}^*(x)}{\theta_{i_1}^*(y)^3} 
W_{i_1l} W_{i_1h}
\right) 
\sum\limits_{(l,h):l\ne h} 
\left( 
- \frac{U_{i_2l} W_{i_2h} }{\theta_{i_2}^*(y)^2} 
+ \frac{2 \theta_{i_2}^*(x)}{\theta_{i_2}^*(y)^3} 
W_{i_2l} W_{i_2h}
\right)
\end{array}
\right],
\end{align*}
where $r_{1,i_1} r_{1,i_2}$ has at most 4 different indexes.

\begin{proof}[Proof of Lemma \ref{lemma:Cov_order}]
To show \eqref{eq:Cov_r_{0,i} + r_{1,i}_r_{0,j} + r_{1,j}}, we can rewrite it as 
\begin{align*}
& \textup{Cov}\left( a_i (r_{0,i} + r_{1,i} - \mu_{r_1,i}), b_j (r_{0,j} + r_{1,j} - \mu_{r_1,j}) \right) \\
=& \textup{Cov}\left( a_ir_{0,i}, b_j r_{0,j} \right) 
+ \textup{Cov}\left( a_i r_{0,i}, b_j ( r_{1,j} - \mu_{r_1,j}) \right)
+ \textup{Cov}\left( a_i (r_{1,i} - \mu_{r_1,i}), b_j r_{0,j} \right) \\
&+ \textup{Cov}\left( a_i ( r_{1,i} - \mu_{r_1,i}), b_j ( r_{1,j} - \mu_{r_1,j}) \right).
\end{align*}
We analyze each term one by one.
\begin{enumerate}[(a)]
\item The covariance is nonzero when $R_{ik} $ and $ R_{jl}$ share common second index:
\begin{align*}
\textup{Cov}\left( a_ir_{0,i}, b_j r_{0,j} \right)
= \frac{1}{(n-1)^2} 
\sum_{k\ne i,j} 
\textup{Cov}\left( a_i \xi_i R_{ik}, b_j \xi_j R_{jk} \right) 
\le C \frac{n (q_n^{\text{post}})^2}{n^2 (q_n^{\text{post}})^2} = O\left( \frac{1}{n} \right).
\end{align*}
\item The covariance is nonzero when both sides share common index:
\begin{align*}
& \textup{Cov}\left( a_i r_{0,i}, b_j ( r_{1,j} - \mu_{r_1,j}) \right) \\
=& \frac{1}{(n-1)^3}  
\left[ 
\begin{array}{c}
\sum\limits_{k\ne i,j}
\textup{Cov}\left( a_i \xi_i R_{ik}, b_j 
\left( - \frac{U_{jk} W_{jk}}{\theta_j^*(y)^2} + \frac{2\theta_j^*(x)}{\theta_j^*(y)^3}  W_{jk}^2 \right) 
\right) \\
+ \sum\limits_{k\ne i,j}
\textup{Cov}\left( a_i \xi_i R_{ik}, b_j 
\left( - \frac{U_{jk} W_{ji}}{\theta_j^*(y)^2} + \frac{2\theta_j^*(x)}{\theta_j^*(y)^3}  W_{jk} W_{ji} \right)
\right)
\end{array}
\right] 
- \frac{1}{n-1} 
\textup{Cov}\left( a_i \xi_i R_{ij}, b_j \mu_{r_1,j} \right) \\
\le & 
C_1 \frac{n (q_n^{\text{post}})^2}{(n q_n^{\text{post}})^3} + C_2 \frac{n (q_n^{\text{post}})^3}{(n q_n^{\text{post}})^3}
+ C_3 \frac{1}{n} \frac{1}{n q_n^{\text{post}}}
= O\left( \frac{1}{n^2 q_n^{\text{post}}} \right).
\end{align*}
\item The covariance is nonzero when both sides share a common index:
\begin{align}
& \textup{Cov}\left( a_i ( r_{1,i} - \mu_{r_1,i}), b_j ( r_{1,j} - \mu_{r_1,j}) \right) \notag \\
&= \textup{Cov}\left( a_i r_{1,i}, b_j r_{1,j} \right) 
- \textup{Cov}\left( a_i r_{1,i}, b_j \mu_{r_1,j} \right)
- \textup{Cov}\left( a_i \mu_{r_1,i}, b_j r_{1,j} \right)
+ \textup{Cov}\left( a_i \mu_{r_1,i}, b_j \mu_{r_1,j} \right) \notag \\
&= \frac{1}{(n-1)^4} 
\left[ 
\begin{array}{c}
\sum\limits_{k\ne i,j} 
\textup{Cov}\left( a_i 
\left( - \frac{U_{ik} W_{ik}}{\theta_i^*(y)^2} + \frac{2\theta_i^*(x)}{\theta_i^*(y)^3}  W_{ik}^2 \right),
b_j \left( - \frac{U_{jk} W_{jk}}{\theta_j^*(y)^2} + \frac{2\theta_j^*(x)}{\theta_j^*(y)^3}  W_{jk}^2 \right) \right) \\
+ \sum\limits_{k\ne i,j} 
\textup{Cov}\left( a_i
\left( - \frac{U_{ik} W_{ik}}{\theta_i^*(y)^2} + \frac{2\theta_i^*(x)}{\theta_i^*(y)^3}  W_{ik}^2 \right), 
b_j \left( - \frac{U_{jk} W_{ji}}{\theta_j^*(y)^2} + \frac{2\theta_j^*(x)}{\theta_j^*(y)^3}  W_{jk} W_{ji} \right) \right) \\
+ \sum\limits_{k\ne i,j} 
\textup{Cov}\left( a_i 
\left( - \frac{U_{ik} W_{ij}}{\theta_i^*(y)^2} + \frac{2\theta_i^*(x)}{\theta_i^*(y)^3}  W_{ik} W_{ij} \right),
b_j \left( - \frac{U_{jk} W_{jk}}{\theta_j^*(y)^2} + \frac{2\theta_j^*(x)}{\theta_j^*(y)^3}  W_{jk}^2 \right) \right) \\
+ \sum\limits_{(k,l):k\ne l}
\textup{Cov}\left( a_i 
\left( - \frac{U_{ik} W_{il}}{\theta_i^*(y)^2} + \frac{2\theta_i^*(x)}{\theta_i^*(y)^3}  W_{ik} W_{il} \right), b_j\left( - \frac{U_{jk} W_{jl}}{\theta_j^*(y)^2} + \frac{2\theta_j^*(x)}{\theta_j^*(y)^3}  W_{jk} W_{jl} \right) \right)
\end{array}
\right] \notag \\
&- \frac{1}{(n-1)^2}
\left[ \textup{Cov}\left( a_i  
\left( - \frac{U_{ij} W_{ij}}{\theta_i^*(y)^2} + \frac{2\theta_i^*(x)}{\theta_i^*(y)^3}  W_{ij}^2 \right), b_j \mu_{r_1,j} \right) 
+ \textup{Cov}\left( a_i \mu_{r_1,i}, 
b_j \left( - \frac{U_{ji} W_{ji}}{\theta_j^*(y)^2} + \frac{2\theta_j^*(x)}{\theta_j^*(y)^3}  W_{ji}^2 \right) \right)
\right] \notag \\
& \le C_1 \frac{n (q_n^{\text{post}})^2}{n^4 (q_n^{\text{post}})^4} 
+ C_2 \frac{n (q_n^{\text{post}})^3}{n^4 (q_n^{\text{post}})^4}
+ C_3 \frac{n^2 (q_n^{\text{post}})^4}{n^4 (q_n^{\text{post}})^4}
+ C_4 \frac{q_n^{\text{post}}}{n^2 (q_n^{\text{post}})^2 np_n}
= O\left( \frac{1}{n^2} \right)
+ O\left( \frac{1}{n^3 (q_n^{\text{post}})^2} \right).
\label{eq:Cov_r1ir1j}
\end{align}
\end{enumerate}
Therefore, by combining results in (a)-(c), we can conclude that
\begin{align*}
& \textup{Cov}\left( a_i (r_{0,i} + r_{1,i} - \mu_{r_1,i}), b_j (r_{0,j} + r_{1,j} - \mu_{r_1,j}) \right) 
= O\left( \frac{1}{n} \right).
\end{align*}
To show \eqref{eq:Cov_r_{0,i} + r_{1,i}^2r_{0,j} + r_{1,j}^2}, we decompose it as 
\begin{align*}
& \textup{Cov}\left( a_i (r_{0,i} + r_{1,i} - \mu_{r_1,i})^2, b_j (r_{0,j} + r_{1,j} - \mu_{r_1,j})^2 \right) \\
&= \textup{Cov}\left( a_i r_{0,i}^2, b_j r_{0,j}^2 \right) 
+ \textup{Cov}\left( a_i r_{0,i}^2, b_j (r_{1,j} - \mu_{r_1,j})^2 \right)
+ \textup{Cov}\left( a_i (r_{1,i} - \mu_{r_1,i})^2, b_j r_{0,j}^2 \right) \\
&+ 2 \textup{Cov}\left( a_i r_{0,i}^2, b_j r_{0,j}(r_{1,j} - \mu_{r_1,j}) \right)  
+ 2 \textup{Cov}\left( a_i r_{0,i}(r_{1,i} - \mu_{r_1,i}), b_j r_{0,j}^2 \right) \\
&+ 2 \textup{Cov}\left( a_i (r_{1,i} - \mu_{r_1,i})^2, b_j r_{0,j}(r_{1,j} - \mu_{r_1,j}) \right) 
+ 2 \textup{Cov}\left( a_i r_{0,i}(r_{1,i} - \mu_{r_1,i}), b_j (r_{1,j} - \mu_{r_1,j})^2 \right) \\
&+ 4 \textup{Cov}\left( a_i r_{0,i}(r_{1,i} - \mu_{r_1,i}), b_j r_{0,j}(r_{1,j} - \mu_{r_1,j}) \right)
+ \textup{Cov}\left( a_i (r_{1,i} - \mu_{r_1,i})^2, b_j (r_{1,j} - \mu_{r_1,j})^2 \right)
\end{align*}
We analyze each term one by one.
\begin{enumerate}[(a)]
\item 
By definition of $R_{ik}$, any index pair that appears once would have zero covariance:
\begin{align*}
& \textup{Cov}\left( a_i r_{0,i}^2, b_j r_{0,j} ^2 \right) 
= \frac{1}{(n-1)^4} \textup{Cov}\left( a_i \xi_i^2 \left(  \sum_{k\ne i} R_{ik}^2 + \sum_{(k,l):k\ne l} R_{ik}R_{il}  \right), b_j \xi_j^2 \left(  \sum_{k\ne j} R_{jk}^2 + \sum_{(k,l):k\ne l} R_{jk}R_{jl} \right) \right) \\
&= \frac{1}{(n-1)^4} \left[
\begin{array}{c}
\sum\limits_{k\ne i,j} \textup{Cov}\left( a_i \xi_i^2 R_{ik}^2, b_j \xi_j^2 R_{jk}^2 \right) 
+ \sum\limits_{k\ne i,j} \textup{Cov}\left( a_i \xi_i^2 R_{ik}^2, b_j \xi_j^2 R_{jk}R_{ji} \right) \\
+ \sum\limits_{k\ne i,j} \textup{Cov}\left( a_i \xi_i^2 R_{ik}R_{ij}, b_j \xi_j^2 R_{jk}^2 \right)
+ \sum\limits_{(k,l):k\ne l} \textup{Cov}\left( a_i \xi_i^2 R_{ik}R_{il}, b_j \xi_j^2 R_{jk} R_{jl} \right)
\end{array}
\right] \\
&\le \frac{C_1 n (q_n^{\text{post}})^2 + C_2 n (q_n^{\text{post}})^3 + C_3 n^2 (q_n^{\text{post}})^4}{n^4 (q_n^{\text{post}})^4}
= O\left( \frac{1}{n^3 (q_n^{\text{post}})^2} \right) + O\left( \frac{1}{n^2} \right).
\end{align*}

\item 
By decomposition, we have
\begin{align*}
\textup{Cov}\left( a_i r_{0,i}^2, b_j (r_{1,j} - \mu_{r_1,j})^2 \right) 
= \textup{Cov}\left( a_i r_{0,i}^2, b_j r_{1,j}^2 \right) 
- 2 \textup{Cov}\left( a_i r_{0,i}^2, b_j r_{1,j} \mu_{r_1,j} \right) 
+ \textup{Cov}\left( a_i r_{0,i}^2, b_j \mu_{r_1,j}^2 \right).
\end{align*}
For the first term, $r_{0,i}^2$ has at most two indexes and $r_{1,i}^2$ has at most four indexes. To have nonzero covariance, both sides should share a common index or for each side, the multiplicity of each pair index should be larger than 1.
Below we discuss each possibility: 
\begin{enumerate}[(1)]
\item if there is only one (second) index, the covariance is bounded above by $C \frac{n (q_n^{\text{post}})^2}{(n q_n^{\text{post}})^6 }$;
\item if there are two (second) indexes, the covariance is bounded above by $C \frac{n^2 (q_n^{\text{post}})^4}{(n q_n^{\text{post}})^6 }$;
\item if there are three (second) indexes, the covariance is bounded above by $C \frac{n^3 ( q_n^{\text{post}})^5}{(n q_n^{\text{post}})^6 }$.
\end{enumerate}
Therefore, $ \textup{Cov}\left( a_i r_{0,i}^2, b_j r_{1,j}^2 \right) = 
O\left( \frac{1}{n^5 (q_n^{\text{post}})^4 } \right) + O\left( \frac{1}{n^3 q_n^{\text{post}} } \right)$. For the second term, 
\begin{align*}
& \textup{Cov}\left( a_i r_{0,i}^2, b_j r_{1,j} \mu_{r_1,j} \right) \\
=& \frac{1}{(n-1)^4}
\left[ 
\begin{array}{c}
\sum\limits_{k\ne i,j}
\textup{Cov}\left( a_i \xi_i^2 R_{ik}^2, 
b_j \mu_{r_1,j} \left( - \frac{U_{jk} W_{jk}}{\theta_j^*(y)^2}   
+ \frac{2\theta_j^*(x)}{\theta_j^*(y)^3}  W_{jk}^2
\right) \right) \\
+ \sum\limits_{(k,l):k\ne l}
\textup{Cov}\left( a_i \xi_i^2 R_{ik} R_{il}, 
b_j \mu_{r_1,j} \left( - \frac{U_{jk} W_{jl}}{\theta_j^*(y)^2}   
+ \frac{2\theta_j^*(x)}{\theta_j^*(y)^3}  W_{jk} W_{jl}
\right) \right)
\end{array}
\right] \\
\le & C_1 \frac{ n (q_n^{\text{post}})^2 }{n^4 (q_n^{\text{post}})^4} \frac{1}{nq_n^{\text{post}}}
+ C_2 \frac{ n^2 (q_n^{\text{post}})^4}{n^4 (q_n^{\text{post}})^4} \frac{1}{nq_n^{\text{post}}}
= O\left( \frac{ 1 }{n^4 (q_n^{\text{post}})^3} \right)
+ O\left( \frac{ 1 }{n^3 q_n^{\text{post}}} \right).
\end{align*}
For the last term,
\begin{align*}
\textup{Cov}\left( a_i r_{0,i}^2, b_j \mu_{r_1,j}^2 \right)
= \frac{1}{(n-1)^2}
\textup{Cov}\left( a_i \xi_i^2 R_{ij}^2, b_j \mu_{r_1,j}^2 \right)
= O\left( \frac{ 1 }{n^4 (q_n^{\text{post}})^3 } \right).
\end{align*}
To conclude, 
\begin{align*}
\textup{Cov}\left( a_i r_{0,i}^2, b_j (r_{1,j} - \mu_{r_1,j})^2 \right) 
= O\left( \frac{ 1 }{n^4 (q_n^{\text{post}})^3 } \right) + O\left( \frac{ 1 }{n^3 q_n^{\text{post}}} \right).
\end{align*}

\item 
By definition of $R_{ik}$, any index pair that appears once would have zero covariance:
\begin{align*}
& \textup{Cov}\left( a_i r_{0,i}^2, b_j r_{0,j} (r_{1,j} - \mu_{r_1,j}) \right)  
= \textup{Cov}\left( a_i r_{0,i}^2, b_j r_{0,j} r_{1,j} \right) - \textup{Cov}\left( a_i r_{0,i}^2, b_j r_{0,j} \mu_{r_1,j} \right) \\
&= \frac{1}{(n-1)^5}
\left[
\begin{array}{c}
\sum\limits_{(k,l):k\ne l} 
\textup{Cov}\left( a_i \xi_i^2 R_{ik}^2, 
b_j \xi_j R_{jk} \left( - \frac{U_{jl} W_{jl}}{\theta_j^*(y)^2}   
+ \frac{2\theta_j^*(x)}{\theta_j^*(y)^3}  W_{jl}^2
\right) \right) \\
+ \sum\limits_{(k,l):k\ne l} 
\textup{Cov}\left( a_i \xi_i^2 R_{ik}^2, 
b_j \xi_j R_{jl} 
\left(  - \frac{U_{jl} W_{jk}}{\theta_j^*(y)^2} 
+ \frac{2\theta_j^*(x)}{\theta_j^*(y)^3} W_{jl} W_{jk}  \right)  \right) \\
+ \sum\limits_{(k,l):k\ne l} 
\textup{Cov}\left( a_i \xi_i^2 R_{ik} R_{il}, 
b_j \xi_j R_{jk} \left( - \frac{U_{jl} W_{jl}}{\theta_j^*(y)^2}   
+ \frac{2\theta_j^*(x)}{\theta_j^*(y)^3}  W_{jl}^2
\right) \right) \\
+ \sum\limits_{(k,l):k\ne l} 
\textup{Cov}\left( a_i \xi_i^2 R_{ik} R_{il}, 
b_j \xi_j R_{jk} 
\left(  - \frac{U_{jk} W_{jl}}{\theta_j^*(y)^2} 
+ \frac{2\theta_j^*(x)}{\theta_j^*(y)^3} W_{jk} W_{jl}  \right)  \right)
\end{array}
\right] \\
&- \frac{1}{(n-1)^3} 
\left[ \sum_{k\ne i,j}  \textup{Cov}\left( a_i \xi_i^2 R_{ik}^2, b_j \xi_j  \mu_{r_1,j} R_{jk} \right) 
+ \sum_{k\ne i,j} \textup{Cov}\left( a_i \xi_i^2 R_{ik}R_{ij}, b_j \xi_j  \mu_{r_1,j} R_{jk} \right) \right] \\
&\le \frac{C_1 n^2 (q_n^{\text{post}})^3 }{(n q_n^{\text{post}})^5} + C_2 \frac{n (q_n^{\text{post}})^2}{(n q_n^{\text{post}})^3}\frac{1}{nq_n^{\text{post}}}
= O\left( \frac{1}{n^3 (q_n^{\text{post}})^2} \right).
\end{align*}

\item 
By decomposition, we have:
\begin{align*}
& \textup{Cov}\left( a_i (r_{1,i} - \mu_{r_1,i})^2, b_j r_{0,j} (r_{1,j} - \mu_{r_1,j}) \right)  
= \textup{Cov}\left( a_i r_{1,i}^2, b_j r_{0,j} r_{1,j} \right) \\
& - \textup{Cov}\left( a_i r_{1,i}^2, b_j \mu_{r_1,j} r_{0,j} \right)
- 2 \textup{Cov}\left( a_i r_{1,i} \mu_{r_1,i}, b_j r_{0,j} (r_{1,j} - \mu_{r_1,j})  \right)
+ \textup{Cov}\left( a_i \mu_{r_1,i}^2, b_j r_{0,j} (r_{1,j} -\mu_{r_1,j} ) \right).
\end{align*}
For the first term $\textup{Cov}\left( a_i r_{1,i}^2, b_j r_{0,j} r_{1,j} \right)$, $r_{1,i}^2$ has at most 4 different indexes, and $r_{0,j} r_{1,j}$ has at most 3 different indexes. We discuss all possibilities:
\begin{enumerate}[(1)]
\item if there is only one (second) index, the covariance is bounded above by $C\frac{n (q_n^{\text{post}})^2}{(n q_n^{\text{post}})^7}$;
\item if there are only two (second) indexes, the covariance is bounded above by $C\frac{n^2 (q_n^{\text{post}})^3}{(n q_n^{\text{post}})^7}$;
\item if there are only three (second) indexes, the covariance is bounded above by $C\frac{n^3 (q_n^{\text{post}})^4}{(n q_n^{\text{post}})^7}$.
\end{enumerate}
Therefore, $\textup{Cov}\left( a_i r_{1,i}^2, b_j r_{0,j} r_{1,j} \right) = O\left( \frac{1}{n^4 (q_n^{\text{post}})^3} \right)$.
For the second term, $r_{0,j}$ has 1 index, and thus 
\begin{align*}
\textup{Cov}\left( a_i r_{1,i}^2, b_j \mu_{r_1,j} r_{0,j} \right)
\le C \frac{n^2 (q_n^{\text{post}})^3}{(n q_n^{\text{post}})^5} \frac{1}{n q_n^{\text{post}}}
= O\left( \frac{1}{n^4 (q_n^{\text{post}})^3} \right).
\end{align*}
For the third term, $r_{1,i}$ has at most 2 different indexes, and thus there are two possibilities for $\textup{Cov}\left( a_i \mu_{r_1,i} r_{1,i}, b_j r_{0,j} r_{1,j}  \right)$:
\begin{enumerate}[(1)]
\item if there is only one (second) index, the covariance is bounded above by $C\frac{n (q_n^{\text{post}})^2}{(n q_n^{\text{post}})^6}$;
\item if there are two (second) indexes, the covariance is bounded above by $C\frac{n^2 (q_n^{\text{post}})^3}{(n q_n^{\text{post}})^6}$.
\end{enumerate}
Therefore, $\textup{Cov}\left( a_i \mu_{r_1,i} r_{1,i}, b_j r_{0,j} r_{1,j}  \right) = O\left( \frac{1}{n^4 (q_n^{\text{post}})^3} \right)$. Moreover, we have 
\begin{align*}
& \textup{Cov}\left( a_i \mu_{r_1,i} r_{1,i}, b_j \mu_{r_1,j} r_{0,j}  \right) \\
=& \frac{1}{(n-1)^3} 
\left[
\begin{array}{c}
\sum\limits_{k\ne i,j} 
\textup{Cov}\left( a_i \mu_{r_1,i} 
\left( 
- \frac{U_{ik} W_{ik} }{\theta_{i}^*(y)^2} 
+ \frac{2 \theta_{i}^*(x)}{\theta_{i}^*(y)^3} W_{il}^2
\right), b_j \mu_{r_1,j} \xi_j R_{jk}  \right) \\
+ \sum\limits_{k\ne i,j} 
\textup{Cov}\left( a_i \mu_{r_1,i} 
\left( 
- \frac{U_{ik} W_{ij} }{\theta_{i}^*(y)^2} 
+ \frac{2 \theta_{i}^*(x)}{\theta_{i}^*(y)^3} U_{ik} W_{ij}
\right), b_j \mu_{r_1,j} \xi_j R_{jk} \right)
\end{array}
\right]
= O\left( \frac{1}{n^4 (q_n^{\text{post}})^3} \right).
\end{align*}
Therefore, we can conclude that 
\begin{align*}
\textup{Cov}\left( a_i (r_{1,i} - \mu_{r_1,i})^2, b_j r_{0,j} (r_{1,j} - \mu_{r_1,j}) \right) 
= O\left( \frac{1}{n^4 (q_n^{\text{post}})^3} \right).
\end{align*}

\item 
By decomposition, we have
\begin{align*}
& \textup{Cov}\left( a_i r_{0,i} (r_{1,i} - \mu_{r_1,i}), b_j r_{0,j} (r_{1,j} - \mu_{r_1,j}) \right) 
= \textup{Cov}\left( a_i r_{0,i} r_{1,i}, b_j r_{0,j} r_{1,j} \right) \\
&- \textup{Cov}\left( a_i r_{0,i} r_{1,i}, b_j r_{0,j} \mu_{r_1,j} \right) - \textup{Cov}\left( a_i r_{0,i} \mu_{r_1,i}, b_j r_{1,j} \right)
+ \textup{Cov}\left( a_i r_{0,i} \mu_{r_1,i}, b_j r_{0,j} \mu_{r_1,j} \right) \\
=& \frac{1}{(n-1)^6} 
\left[ 
\begin{array}{c}
\sum\limits_{(k,l,h)}
\textup{Cov}\left( a_i \xi_i R_{ik}
\left( 
- \frac{U_{il} W_{il} }{\theta_{i}^*(y)^2} 
+ \frac{2 \theta_{i}^*(x)}{\theta_{i}^*(y)^3} W_{il}^2
\right), 
b_j \xi_j R_{jk}
\left( 
- \frac{U_{jh} W_{jh} }{\theta_{j}^*(y)^2} 
+ \frac{2 \theta_{j}^*(x)}{\theta_{j}^*(y)^3} W_{jh}^2
\right)
\right) \\
+ \sum\limits_{(k,l,h)}
\textup{Cov}\left( 
a_i \xi_i R_{ik}
\left( 
- \frac{U_{il} W_{il} }{\theta_{i}^*(y)^2} 
+ \frac{2 \theta_{i}^*(x)}{\theta_{i}^*(y)^3} W_{il}^2
\right),
b_j \xi_j R_{jh}
\left( 
- \frac{U_{jk} W_{jh} }{\theta_{j}^*(y)^2} 
+ \frac{2 \theta_{j}^*(x)}{\theta_{j}^*(y)^3} U_{jk} W_{jh}
\right)
\right) \\
+ \sum\limits_{(k,l,h)}
\textup{Cov}\left( a_i \xi_i R_{ik}
\left( 
- \frac{U_{ik} W_{il} }{\theta_{i}^*(y)^2} 
+ \frac{2 \theta_{i}^*(x)}{\theta_{i}^*(y)^3} U_{ik} W_{il}
\right), 
b_j \xi_j R_{jl}
\left( 
- \frac{U_{jh} W_{jh} }{\theta_{j}^*(y)^2} 
+ \frac{2 \theta_{j}^*(x)}{\theta_{j}^*(y)^3} W_{jh}^2
\right)
\right) \\
+ \sum\limits_{(k,l,h)}
\textup{Cov}\left( 
a_i \xi_i R_{ik}
\left( 
- \frac{U_{il} W_{ih} }{\theta_{i}^*(y)^2} 
+ \frac{2 \theta_{i}^*(x)}{\theta_{i}^*(y)^3} U_{il} W_{ih}
\right),
\xi_j R_{jk}
\left( 
- \frac{U_{jl} W_{jh} }{\theta_{j}^*(y)^2} 
+ \frac{2 \theta_{j}^*(x)}{\theta_{j}^*(y)^3} U_{jl} W_{jh}
\right)
\right)
\end{array}
\right] \\
&- \frac{1}{(n-1)^4} 
\left[ \begin{array}{c}
\sum\limits_{(k,l)} 
\textup{Cov} \left(
a_i \xi_i R_{ik}
\left( 
- \frac{U_{il} W_{il} }{\theta_{i}^*(y)^2} 
+ \frac{2 \theta_{i}^*(x)}{\theta_{i}^*(y)^3} W_{il}^2
\right),
b_j \mu_{r_1,j} \xi_j R_{jl} \right) \\
+ \sum\limits_{(k,l)} 
\textup{Cov} \left( a_i \xi_i R_{il}
\left( 
- \frac{U_{ik} W_{il} }{\theta_{i}^*(y)^2} 
+ \frac{2 \theta_{i}^*(x)}{\theta_{i}^*(y)^3} U_{ik} W_{il}
\right),
b_j \mu_{r_1,j} \xi_j R_{jk}
\right)
\end{array} \right] \\
&- \frac{1}{(n-1)^4} 
\left[ 
\begin{array}{c}
\sum\limits_{(k,l)} 
\textup{Cov}\left( a_i \mu_{r_1,i} \xi_i R_{ik},
b_j \xi_j R_{jk}
\left( 
- \frac{U_{jl} W_{jl} }{\theta_{j}^*(y)^2} 
+ \frac{2 \theta_{j}^*(x)}{\theta_{j}^*(y)^3} W_{jl}^2
\right) \right) \\
+ \sum\limits_{(k,l)} 
\textup{Cov}\left( a_i \mu_{r_1,i} \xi_i R_{ik},
b_j \xi_j R_{jl}
\left( 
- \frac{U_{jk} W_{jl} }{\theta_{j}^*(y)^2} 
+ \frac{2 \theta_{j}^*(x)}{\theta_{j}^*(y)^3} U_{jk} W_{jl}
\right) \right)
\end{array} \right] \\
&+ \frac{1}{(n-1)^2} \sum\limits_{k\ne i,j}  \textup{Cov}\left( a_i \mu_{r_1,i} \xi_i R_{ik}, b_j \mu_{r_1,j} \xi_j R_{jk} \right) \\
=& O\left( \frac{n^3 (q_n^{\text{post}})^4}{n^6 p_n^6} \right)
+ O\left( \frac{n^2 (q_n^{\text{post}})^3}{n^4 (q_n^{\text{post}})^4} \frac{1}{np_n} \right)
+ O\left( \frac{n (q_n^{\text{post}})^2}{n^2 (q_n^{\text{post}})^2} \frac{1}{n^2 (q_n^{\text{post}})^2} \right)
= O\left( \frac{1}{n^3 (q_n^{\text{post}})^2} \right).
\end{align*}

\item 
By definition of $R_{ik}$, any index pair that appears once would have zero covariance.
\begin{align*}
& \textup{Cov}\left( a_i (r_{1,i} - \mu_{r_1,i})^2, b_j (r_{1,j} - \mu_{r_1,j})^2 \right)  \\
=& \textup{Cov}\left( a_i r_{1,i}^2, b_j r_{1,j}^2 \right)
- 2 \textup{Cov}\left( a_i r_{1,i}^2, b_j r_{1,j} \mu_{r_1,j} \right)
-2 \textup{Cov}\left( a_i r_{1,i} \mu_{r_1,i}, b_j r_{1,j}^2 \right) \\
&+ \textup{Cov}\left( a_i r_{1,i}^2, b_j \mu_{r_1,j}^2 \right) 
+ \textup{Cov}\left( a_i \mu_{r_1,i}^2, b_j r_{1,j}^2 \right)
+4 \textup{Cov}\left( a_i r_{1,i} \mu_{r_1,i}, b_j r_{1,j} \mu_{r_1,j} \right) \\
&-2 \textup{Cov}\left( a_i \mu_{r_1,i} r_{1,i}, b_j \mu_{r_1,j}^2 \right)
- 2 \textup{Cov}\left( a_i \mu_{r_1,i}^2, b_j \mu_{r_1,j} r_{1,j} \right).
\end{align*}
For $\textup{Cov}\left( a_i r_{1,i}^2, b_j r_{1,j}^2 \right)$, by the definition of in $r_{1,i}^2$ in \eqref{eq:r1i^2}, the order depends on how many indexes we have. We consider each possible case:
\begin{enumerate}[(1)]
\item for the case with 1 index, the covariance term is bounded above by $C\frac{1}{n(n q_n^{\text{post}})^6}$;
\item for the case with 2 indexes, the covariance term is bounded above by $C\frac{1}{n^6 (q_n^{\text{post}})^4}$;
\item for the case with 3 indexes, the covariance term is bounded above by $C\frac{1}{n^5 (q_n^{\text{post}})^2}$;
\item for the case with 4 indexes, the covariance term is bounded above by $C\frac{1}{n^4}$.
\end{enumerate}
For the second term $\textup{Cov}\left( a_i r_{1,i}^2, b_j r_{1,j} \mu_{r_1,j} \right)$, $r_{1,j}$ has at most 2 different indexes, and the order depends on how many indexes. We consider each possible case:
\begin{enumerate}[(1)]
\item for the case with 1 index, the covariance term is bounded above by $C\frac{n (q_n^{\text{post}})^2}{(n q_n^{\text{post}})^7}$;
\item for the case with 2 indexes, the covariance term is bounded above by $C\frac{n^2 (q_n^{\text{post}})^3}{(n q_n^{\text{post}})^7}$;
\item for the case with 3 indexes, the covariance term is bounded above by $C\frac{n^3 p_n^5}{(n q_n^{\text{post}})^7}$.
\end{enumerate}
Therefore $\textup{Cov}\left( a_i r_{1,i}^2, b_j r_{1,j}^2 \right) = O\left( \frac{1}{n^5 (q_n^{\text{post}})^4} \right) + O\left( \frac{1}{n^4 (q_n^{\text{post}})^2} \right)$.
For the third term $\textup{Cov}\left( a_i r_{1,i}^2, b_j \mu_{r_1,j}^2 \right)$, by the definition in \eqref{eq:r1i^2}, we have
\begin{align*}
& \textup{Cov}\left( a_i r_{1,i}^2, b_j \mu_{r_1,j}^2 \right) \\
=& \frac{1}{(n-1)^4} 
\left[ 
\begin{array}{c}
\textup{Cov}\left( a_i \left( 
- \frac{U_{ij} W_{ij} }{\theta_i^*(y)^2} 
+ \frac{2 \theta_i^*(x)}{\theta_i^*(y)^3} W_{ij}^2
\right)^2, b_j \mu_{r_1,j}^2 \right) \\
+ \sum\limits_{k\ne i,j}
\textup{Cov}\left( a_i 
\left( 
- \frac{U_{ik} W_{ik} }{\theta_i^*(y)^2} 
+ \frac{2 \theta_i^*(x)}{\theta_i^*(y)^3} W_{ik}^2
\right)
\left( 
- \frac{U_{ij} W_{ij} }{\theta_i^*(y)^2} 
+ \frac{2 \theta_i^*(x)}{\theta_i^*(y)^3} W_{ij}^2
\right), b_j \mu_{r_1,j}^2 \right) \\
+ \sum\limits_{k\ne i,j}
\textup{Cov}\left( a_i  
\left(
- \frac{U_{ik} W_{ij} }{\theta_i^*(y)^2} 
+ \frac{2 \theta_i^*(x)}{\theta_i^*(y)^3} W_{ik} W_{ij}
\right)^2, b_j \mu_{r_1,j}^2 \right) \\
+ 2 \sum\limits_{k\ne i,j} 
\textup{Cov}\left( a_i 
\left( 
- \frac{U_{ik} W_{ik} }{\theta_i^*(y)^2} 
+ \frac{2 \theta_i^*(x)}{\theta_i^*(y)^3} W_{ik}^2
\right) 
\left( 
- \frac{U_{ik} W_{ij} }{\theta_i^*(y)^2} 
+ \frac{2 \theta_i^*(x)}{\theta_i^*(y)^3} 
W_{ik} W_{ij}
\right), b_j \mu_{r_1,j}^2 \right)
\end{array}
\right] \\
=& O\left( \frac{1}{n^6 p_n^5} \right)
+ O\left( \frac{1}{n^5 (q_n^{\text{post}})^4} \right)
= O\left( \frac{1}{n^5 (q_n^{\text{post}})^4} \right).
\end{align*}

For the fourth term $\textup{Cov}\left( a_i r_{1,i} \mu_{r_1,i}, b_j r_{1,j} \mu_{r_1,j} \right)$, it is analogous to \eqref{eq:Cov_r1ir1j}, and we can show that 
\begin{align*}
\textup{Cov}\left( a_i r_{1,i} \mu_{r_1,i}, b_j r_{1,j} \mu_{r_1,j} \right)
= O\left( \frac{1}{n^5 (q_n^{\text{post}})^4 } \right)
+ O\left( \frac{1}{n^4 (q_n^{\text{post}})^2 } \right).
\end{align*}

For the fifth term, we have
\begin{align*}
\textup{Cov}\left( a_i \mu_{r_1,i} r_{1,i}, b_j \mu_{r_1,j}^2 \right)
= \frac{1}{(n-1)^2} 
\textup{Cov}\left(
a_i \mu_{r_1,i} 
\left( 
- \frac{U_{ij} W_{ij} }{\theta_{i}^*(y)^2} 
+ \frac{2 \theta_{i}^*(x)}{\theta_{i}^*(y)^3} W_{ij}^2
\right), b_j \mu_{r_1,j}^2  
\right)
= O\left( \frac{1}{n^5 (q_n^{\text{post}})^4 } \right).
\end{align*}
By combining the results, we can conclude that 
\begin{align*}
& \textup{Cov}\left( a_i (r_{1,i} - \mu_{r_1,i})^2, b_j (r_{1,j} - \mu_{r_1,j})^2 \right) 
= O\left( \frac{1}{n^5 (q_n^{\text{post}})^4 } \right)
+ O\left( \frac{1}{n^4 (q_n^{\text{post}})^2 } \right).
\end{align*}
\end{enumerate}
Therefore, we can conclude that 
\begin{align*}
\textup{Cov}\left( a_i  (r_{0,i} + r_{1,i} - \mu_{r_1,i})^2, b_j (r_{0,j} + r_{1,j} - \mu_{r_1,j})^2 \right) 
= O\left( \frac{1}{n^3 (q_n^{\text{post}})^2} \right) 
+ O\left( \frac{1}{n^2} \right). 
\end{align*}

\end{proof}

\subsection{Proof of Lemma \ref{lemma:consistencyM}}\label{proof:consistencyM}

We only show \eqref{eq:M_i^2_con} since the remaining terms can be shown with analogous arguments.
By Taylor expansion, we have
\begin{equation*}
M_i^2
= (\xi_i + \mu_{r_1,i})^2 
+ (r_{0,i} + r_{1,i} - \mu_{r_1,i})^2 
+ 2 (\xi_i + \mu_{r_1,i}) (r_{0,i} + r_{1,i} - \mu_{r_1,i}).  \label{eq:Mi^2}
\end{equation*}
The first term of $M_i^2$ is a function of $T_i$ and $w_i$, which is i.i.d. across $i$. 
The second and the third terms follow from Lemma \ref{lemma:airi_ri^2}. Therefore, we have
\begin{align*}
\frac{1}{n}\sum_{i=1}^n M_i^2
= \frac{1}{n}\sum_{i=1}^n  (\xi_i + \mu_{r_1,i})^2 
+ E\left((r_{0,i} + r_{1,i} - \mu_{r_1,i})^2 \right) 
+ O_{\mathbb{P}}\left( \frac{1}{n \sqrt{q_n^{\text{post}}}} \right).
\end{align*}

\subsection{Proof of Lemma \ref{lemma:airi_ri^2}}
\label{proof:airi_ri^2}
To show \eqref{eq:airi}, the expectation is zero and the variance is 
\begin{align*}
\V\left( \frac{1}{n}\sum_{i=1}^n a_i (r_{0,i} + r_{1,i} - \mu_{r_1,i}) \right) 
&= \frac{1}{n^2}\sum_{i=1}^n E\left( a_i^2 (r_{0,i} + r_{1,i} - \mu_{r_1,i} )^2 \right) \\
+ \frac{1}{n^2}\sum_{i=1}^n \sum_{j\ne i} & \textup{Cov}\left( a_i(r_{0,i} + r_{1,i} - \mu_{r_1,i}), a_j(r_{0,j} + r_{1,j} - \mu_{r_1,j}) \right) 
= O\left( \frac{1}{n^2 q_n^{\text{post}}} \right)
\end{align*}
by Lemma \ref{lemma:Eri^c}  and Lemma \ref{lemma:Cov_order}.
In particular, this implies 
\begin{align*}
\frac{1}{n}\sum_{i=1}^n a_i (r_{0,i} + r_{1,i} - \mu_{r_1,i})
= O_{\mathbb{P}}\left( \frac{1}{n \sqrt{q_n^{\text{post}}}} \right). 
\end{align*}
To show \eqref{eq:ri^2}, the expectation is
\begin{align*}
& E\left[ \frac{1}{n}\sum_{i=1}^n (r_{0,i} + r_{1,i} - \mu_{r_1,i})^2 \right] 
= O\left( \frac{1}{n q_n^{\text{post}}} \right)
\end{align*}
and the variance is 
\begin{align*}
\V\left( \frac{1}{n}\sum_{i=1}^n (r_{0,i} + r_{1,i} - \mu_{r_1,i})^2 \right) 
=& \frac{1}{n^2}\sum_{i=1}^n \V\left( (r_{0,i} + r_{1,i} - \mu_{r_1,i})^2 \right) \\
+ \frac{1}{n^2}\sum_{i=1}^n \sum_{j\ne i} & \textup{Cov}\left( (r_{0,i} + r_{1,i} - \mu_{r_1,i})^2, (r_{0,j} + r_{1,j} - \mu_{r_1,j})^2 \right) \\
=& O\left( \frac{1}{n^3 (q_n^{\text{post}})^2} \right) 
+ O\left( \frac{1}{n^2} \right)
\end{align*}
by Lemma \ref{lemma:Eri^c}  and Lemma \ref{lemma:Cov_order}.
Therefore, 
\begin{equation}
\frac{1}{n} \sum_{i=1}^n (r_{0,i} + r_{1,i} - \mu_{r_1,i})^2 
= E\left((r_{0,i} + r_{1,i} - \mu_{r_1,i})^2 \right) 
+ O_{\mathbb{P}}\left( \frac{1}{\sqrt{n} n q_n^{\text{post}}} \right)
+ O_{\mathbb{P}}\left( \frac{1}{n} \right).
\label{eq:ri^2}
\end{equation}

\subsection{Proof of Lemma \ref{lemma:Zixx}} \label{proof_lemma:Zixx}
To show \eqref{eq:MiZi}, we decompose it into
\begin{align*}
\frac{1}{n}\sum_{i=1}^n M_i Z_i^{\textsc{ssiv}}
= \frac{1}{n}\sum_{i=1}^n (\xi_i + \mu_{r_1,i}) Z_i^{\textsc{ssiv}}
+ \frac{1}{n}\sum_{i=1}^n (Z_i^{\textsc{ssiv}} r_{0,i} + Z_i^{\textsc{ssiv}} r_{1,i} - Z_i^{\textsc{ssiv}} \mu_{r_1,i} ).
\end{align*}
For the first term, we can show by Lemma \ref{lemma:Ziphii} and \eqref{eq:mu_r1i} that 
\begin{align*}
\frac{1}{n}\sum_{i=1}^n Z_i^{\textsc{ssiv}} \xi_i 
=& O_{\mathbb{P}}\left( \sqrt{n} q_n^{\text{pre}} \right)
\text{ and }
\frac{1}{n}\sum_{i=1}^n Z_i^{\textsc{ssiv}} \mu_{r_1,i} 
= O_{\mathbb{P}}\left( \frac{ q_n^{\text{pre}} }{\sqrt{n} q_n^{\text{post}}} \right).
\end{align*}
For the remaining terms, we show them one by one.
For $\frac{1}{n}\sum_{i=1}^n r_{0,i} Z_i^{\textsc{ssiv}}$, the expectation is 
\begin{align*}
E\left[ \frac{1}{n}\sum_{i=1}^n Z_i^{\textsc{ssiv}} r_{0,i} \right] 
~=~& \frac{1}{n(n-1)}\sum_{i=1}^n 
E\left[ \left( \sum_{k\ne i} A_{ik}^\text{pre} (T_k-\pi) \right) \left( \xi_i \sum_{k\ne i} R_{ik} \right) \right] \\
~=~& \frac{1}{n(n-1)}\sum_{i=1}^n 
E\left[ \xi_i  \sum_{k\ne i} A_{ik}^\text{pre} (T_k-\pi) R_{ik}  \right] 
~\le~ C \frac{\min\{q_n^{\text{pre}}, q_n^{\text{post}}\}}{q_n^{\text{post}}}.
\end{align*}
The variance is 
\begin{align*}
\V\left( \frac{1}{n}\sum_{i=1}^n Z_i^{\textsc{ssiv}} r_{0,i} \right)
&= \frac{1}{n^2}\sum_{i=1}^n \V\left( Z_i^{\textsc{ssiv}} r_{0,i} \right) 
+ \frac{1}{n^2}\sum_{i=1}^n \sum_{j\ne i} \textup{Cov}\left( Z_i^{\textsc{ssiv}}r_{0,i}, Z_j^{\textsc{ssiv}} r_{0,j} \right).
\end{align*}
For the diagonal term, we have 
\begin{align*}
& \V\left( Z_i^{\textsc{ssiv}} r_{0,i} \right)
\le E\left[ (Z_i^{\textsc{ssiv}})^2 r_{0,i}^2 \right] 
= E\left[ \left( \sum_{j\ne i} A_{ij}^{\text{pre}} (T_j -\pi)  \right)^2 \xi_i^2 \left( \frac{1}{n-1}\sum_{j\ne i} R_{ij} \right)^2  \right] \\
&= \frac{\xi_i^2}{(n-1)^2}  
E \left[
\left( \sum\limits_{j\ne i} A_{ij}^{\text{pre}} (T_j -\pi)^2 
+ \sum\limits_{(j,k)} A_{ij}^{\text{pre}} (T_j -\pi)   A_{ik}^{\text{pre}} (T_k -\pi) \right) 
\left( 
\sum\limits_{j\ne i} R_{ij}^2 
+ \sum\limits_{(j,k)} R_{ij} R_{ik}   
\right)
\mid T_i, w_i \right] \\
&= \frac{\xi_i^2}{ (n-1)^2}  
E\left[
\begin{array}{c}
\sum\limits_{j\ne i} A_{ij}^{\text{pre}} (T_j -\pi)^2 
\sum\limits_{j\ne i} R_{ij}^2 
+ \sum\limits_{(j,k)} A_{ij}^{\text{pre}} (T_j -\pi)   A_{ik}^{\text{pre}} (T_k -\pi)
\sum\limits_{j\ne i} R_{ij}^2 \\
+ \sum\limits_{j\ne i} A_{ij}^{\text{pre}} (T_j -\pi)^2 
\sum\limits_{(j,k)} R_{ij} R_{ik} 
+ \sum\limits_{(j,k)} A_{ij}^{\text{pre}} A_{ik}^{\text{pre}} (T_j -\pi) (T_k -\pi)
\sum\limits_{(j,k)} R_{ij} R_{ik}
\end{array}
\mid T_i, w_i \right] \\
&= \frac{\xi_i^2 }{ (n-1)^2} 
E\left[
\sum\limits_{(j,k)} A_{ij}^{\text{pre}} (T_j -\pi)^2 R_{ik}^2 
+ \sum\limits_{(j,k)} A_{ij}^{\text{pre}} A_{ik}^{\text{pre}} 
(T_j -\pi) (T_k -\pi)
R_{ij} R_{ik}
\mid T_i, w_i \right] \\
&\le \frac{ \xi_i^2 }{ (n-1)^2}  
\left(
C_1 \frac{n^2 q_n^{\text{pre}} q_n^{\text{post}} }{(q_n^{\text{post}})^2}  
+ C_2 \frac{n^2 \min\{q_n^{\text{pre}}, q_n^{\text{post}} \}^2 }{(q_n^{\text{post}})^2}
\right)
\le C \frac{ q_n^{\text{pre}} }{q_n^{\text{post}}}.
\end{align*}
For the cross terms, we have 
\begin{align*}
& \textup{Cov}\left( Z_i^{\textsc{ssiv}} r_{0,i}, Z_j^{\textsc{ssiv}} r_{0,j} \right) 
= \frac{1}{(n-1)^2} \textup{Cov} \left( \xi_i \sum\limits_{k\ne i} A_{ik}^\text{pre} (T_k-\pi)  \sum\limits_{k\ne i} R_{ik}, \xi_j \sum\limits_{k\ne j} A_{jk}^\text{pre} (T_k-\pi) \sum\limits_{k\ne j} R_{jk} \right) \\
&= \frac{1}{(n-1)^2}
\left[
\begin{array}{c}
\sum\limits_{k\ne i,j} 
\textup{Cov} \left( \xi_i A_{ik}^\text{pre} (T_k-\pi)  R_{ik}, 
\xi_j A_{jk}^\text{pre} (T_k-\pi) R_{jk} \right) \\
+ \sum\limits_{k \ne i,j} 
\textup{Cov} \left( \xi_i A_{ik}^\text{pre} (T_k-\pi) R_{ik}, 
\xi_j A_{jk}^\text{pre} (T_k-\pi) R_{ji} \right) \\
+ \sum\limits_{k\ne i,j}
\textup{Cov} \left( \xi_i A_{ik}^\text{pre} (T_k-\pi) R_{ij}, 
\xi_j A_{jk}^\text{pre} (T_k-\pi) R_{jk} \right) \\
+ \sum\limits_{(k,l)}
\textup{Cov} \left( \xi_i A_{ik}^\text{pre} (T_k-\pi) R_{il}, 
\xi_j A_{jk}^\text{pre} (T_k-\pi) R_{jl} \right)
\end{array}
\right] \\
\le & C_1\frac{n \min\{q_n^{\text{pre}}, q_n^{\text{post}}\}^2 }{n^2 (q_n^{\text{post}})^2}
+ C_2 \frac{ n \min\{q_n^{\text{pre}}, q_n^{\text{post}}\} q_n^{\text{pre}} q_n^{\text{post}} }{n^2 (q_n^{\text{post}})^2}
+ C_3 \frac{ n^2 (q_n^{\text{pre}})^2 (q_n^{\text{post}})^2 }{n^2 (q_n^{\text{post}})^2}
= O\left( \frac{ \min\{q_n^{\text{pre}},q_n^{\text{post}}\}^2  }{ n (q_n^{\text{post}})^2 } \right).
\end{align*}
Therefore, 
\begin{align*}
\V\left( \frac{1}{n}\sum_{i=1}^n Z_i^{\textsc{ssiv}} r_{0,i} \right)
&\le C_1 \frac{ q_n^{\text{pre}} }{ n q_n^{\text{post}} }   
+ C_2 \frac{ \min\{q_n^{\text{pre}}, q_n^{\text{post}}\}^2  }{ n (q_n^{\text{post}})^2 }
\end{align*}
and thus
\begin{equation}
\frac{1}{n}\sum_{i=1}^n Z_i^{\textsc{ssiv}} r_{0,i}
= O_{\mathbb{P}}\left( \frac{\min\{q_n^{\text{pre}}, q_n^{\text{post}}\}}{q_n^{\text{post}}} \right)
+ O_{\mathbb{P}}\left( \frac{ \sqrt{q_n^{\text{pre}}} }{ \sqrt{n q_n^{\text{post}}} }  \right).
\label{eq:Zir0i}
\end{equation}
For $\frac{1}{n}\sum_{i=1}^n Z_i^{\textsc{ssiv}} r_{1,i}$, the expectation is 
\begin{align*}
E\left[ \frac{1}{n}\sum_{i=1}^n Z_i^{\textsc{ssiv}} r_{1,i} \right] 
=& \frac{1}{n(n-1)^2}\sum_{i=1}^n 
E\left[ \left( \sum_{k\ne i} A_{ik}^\text{pre} (T_k-\pi) \right) \left( 
\begin{array}{c}
\sum\limits_{k\ne i} \left( - \frac{U_{ik} W_{ik}}{\theta_i^*(y)^2} + \frac{2\theta_i^*(x)}{\theta_i^*(y)^3}  W_{ik}^2 \right) \\
+ \sum\limits_{(k,l)} \left( - \frac{U_{ik} W_{il}}{\theta_i^*(y)^2} + \frac{2\theta_i^*(x)}{\theta_i^*(y)^3}  W_{ik} W_{il} \right)
\end{array}
\right) \right] \\
=& 
\frac{1}{n(n-1)^2}\sum_{i=1}^n 
E\left[ \sum_{k\ne i} A_{ik}^\text{pre} (T_k-\pi) \left( - \frac{U_{ik} W_{ik}}{\theta_i^*(y)^2} + \frac{2\theta_i^*(x)}{\theta_i^*(y)^3}  W_{ik}^2 \right) 
\right] \\
\le & 
C \frac{\min\{q_n^{\text{pre}}, q_n^{\text{post}}\}}{n (q_n^{\text{post}})^2}.
\end{align*}
The variance is 
\begin{align*}
\V\left( \frac{1}{n}\sum_{i=1}^n Z_i^{\textsc{ssiv}} r_{1,i} \right)
= \frac{1}{n^2}\sum_{i=1}^n \V\left( Z_i^{\textsc{ssiv}} r_{1,i} \right) 
+ \frac{1}{n^2}\sum_{i=1}^n \sum_{j\ne i} \textup{Cov}\left( Z_i^{\textsc{ssiv}} r_{1,i}, Z_j^{\textsc{ssiv}} r_{1,j} \right).
\end{align*}
For the diagonal term, we have
\begin{align*}
& \V\left( Z_i^{\textsc{ssiv}} r_{1,i} \right) 
\le E\left[(Z_i^{\textsc{ssiv}})^2 r_{1,i}^2 \right]
= E
\left[ 
\left( \sum_{j\ne i} A_{ij}^{\text{pre}}(T_j - \pi)^2 
+ \sum_{(j,k)} A_{ij}^{\text{pre}} A_{ik}^{\text{pre}} (T_j - \pi) (T_k - \pi) \right) r_{1,i}^2
 \right] \\
&= \frac{1}{(n-1)^4} 
E\left\{ 
\begin{array}{c}
\sum\limits_{j\ne i} 
A_{ij}^{\text{pre}}(T_j - \pi)^2 
\left[ 
\begin{array}{c}
\left( 
\sum\limits_{j\ne i} 
\left(
- \frac{U_{ij} W_{ij}}{\theta_i^*(y)^2} 
+ \frac{2 \theta_i^*(x)}{\theta_i^*(y)^3} 
W_{ij}^2
\right)
\right)^2 
+ \left( 
\sum\limits_{(j,k)} 
\left(
- \frac{U_{ij} W_{ik}}{\theta_i^*(y)^2} 
+ \frac{2 \theta_i^*(x)}{\theta_i^*(y)^3} 
W_{ij} W_{ik}
\right)
\right)^2 \\
+ 2 \sum\limits_{j\ne i}  
\left( 
- \frac{U_{ij} W_{ij}}{\theta_i^*(y)^2} 
+ \frac{2 \theta_i^*(x)}{\theta_i^*(y)^3} 
W_{ij}^2
\right) 
\sum\limits_{(j,k)} 
\left( 
- \frac{U_{ij} W_{ik}}{\theta_i^*(y)^2} 
+ \frac{2 \theta_i^*(x)}{\theta_i^*(y)^3} 
W_{ij} W_{ik}
\right)
\end{array}
\right]
\end{array}
\right\} \\
&+ \frac{1}{(n-1)^4} 
E\left\{ 
\begin{array}{c}
\sum\limits_{(j,k)} 
A_{ij}^{\text{pre}}(T_j - \pi)
A_{ik}^{\text{pre}}(T_k - \pi)
\left[ 
\begin{array}{c}
\left( 
\sum\limits_{j\ne i} 
\left(
- \frac{U_{ij} W_{ij}}{\theta_i^*(y)^2} 
+ \frac{2 \theta_i^*(x)}{\theta_i^*(y)^3} 
W_{ij}^2
\right)
\right)^2 \\
+ \left( 
\sum\limits_{(j,k)} 
\left(
- \frac{U_{ij} W_{ik}}{\theta_i^*(y)^2} 
+ \frac{2 \theta_i^*(x)}{\theta_i^*(y)^3} 
W_{ij} W_{ik}
\right)
\right)^2 \\
+ 2 \sum\limits_{j\ne i}  
\left( 
- \frac{U_{ij} W_{ij}}{\theta_i^*(y)^2} 
+ \frac{2 \theta_i^*(x)}{\theta_i^*(y)^3} 
W_{ij}^2
\right) 
\sum\limits_{(j,k)} 
\left( 
- \frac{U_{ij} W_{ik}}{\theta_i^*(y)^2} 
+ \frac{2 \theta_i^*(x)}{\theta_i^*(y)^3} 
W_{ij} W_{ik}
\right)
\end{array}
\right]
\end{array}
\right\} \\
&= \frac{1}{(n-1)^4} 
E\left\{ 
\begin{array}{c}
\sum\limits_{j\ne i} 
A_{ij}^{\text{pre}}(T_j - \pi)^2 
\left[ 
\begin{array}{c}
\sum\limits_{j\ne i} 
\left(
- \frac{U_{ij} W_{ij}}{\theta_i^*(y)^2} 
+ \frac{2 \theta_i^*(x)}{\theta_i^*(y)^3} W_{ij}^2
\right)^2 \\
+ \sum\limits_{(j,k)} 
\left(
- \frac{U_{ij} W_{ij}}{\theta_i^*(y)^2} 
+ \frac{2 \theta_i^*(x)}{\theta_i^*(y)^3} W_{ij}^2
\right)
\left(
- \frac{U_{ik} W_{ik}}{\theta_i^*(y)^2} 
+ \frac{2 \theta_i^*(x)}{\theta_i^*(y)^3} W_{ik}^2
\right)\\
+ \sum\limits_{(j,k)} 
\left(
- \frac{U_{ij} W_{ik}}{\theta_i^*(y)^2} 
+ \frac{2 \theta_i^*(x)}{\theta_i^*(y)^3} 
W_{ij} W_{ik}
\right)^2 \\
+ 2 \sum\limits_{k\ne j}  
\left( 
- \frac{U_{ik} W_{ik}}{\theta_i^*(y)^2} 
+ \frac{2 \theta_i^*(x)}{\theta_i^*(y)^3} W_{ik}^2
\right) 
\left( 
- \frac{U_{ij} W_{ik}}{\theta_i^*(y)^2} 
+ \frac{2 \theta_i^*(x)}{\theta_i^*(y)^3} 
W_{ij} W_{ik}
\right)
\end{array}
\right]
\end{array}
\right\} \\
&+ \frac{1}{(n-1)^4} 
E\left\{ 
\begin{array}{c}
\sum\limits_{(j,k)} 
A_{ij}^{\text{pre}}(T_j - \pi)
A_{ik}^{\text{pre}}(T_k - \pi)
\left[ 
\begin{array}{c}
\left(
- \frac{U_{ij} W_{ij}}{\theta_i^*(y)^2} 
+ \frac{2 \theta_i^*(x)}{\theta_i^*(y)^3} W_{ij}^2
\right)
\left(
- \frac{U_{ik} W_{ik}}{\theta_i^*(y)^2} 
+ \frac{2 \theta_i^*(x)}{\theta_i^*(y)^3} W_{ik}^2
\right)\\
+ \left(
- \frac{U_{ij} W_{ik}}{\theta_i^*(y)^2} 
+ \frac{2 \theta_i^*(x)}{\theta_i^*(y)^3} W_{ij} W_{ik}
\right)^2 \\
+ \sum\limits_{l \ne j,k} 
\left(
- \frac{U_{il} W_{ij}}{\theta_i^*(y)^2} 
+ \frac{2 \theta_i^*(x)}{\theta_i^*(y)^3} W_{il} W_{ij}
\right)
\left(
- \frac{U_{il} W_{ik}}{\theta_i^*(y)^2} 
+ \frac{2 \theta_i^*(x)}{\theta_i^*(y)^3} W_{il} W_{ik}
\right)\\
+ 2 \sum\limits_{l\ne i}  
\left( 
- \frac{U_{il} W_{il}}{\theta_i^*(y)^2} 
+ \frac{2 \theta_i^*(x)}{\theta_i^*(y)^3} W_{il}^2
\right) 
\left( 
- \frac{U_{ij} W_{ik}}{\theta_i^*(y)^2} 
+ \frac{2 \theta_i^*(x)}{\theta_i^*(y)^3} 
W_{ij} W_{ik}
\right)
\end{array}
\right]
\end{array}
\right\} \\
&\le C_1 \frac{n^3 q_n^\text{pre} (q_n^{\text{post}})^2}{n^4 (q_n^{\text{post}})^4} + C_2 \frac{n^3 \min\{q_n^{\text{post}}, q_n^\text{pre}\}^2 q_n^{\text{post}} }{n^4 (q_n^{\text{post}})^4} 
= O\left( \frac{ q_n^\text{pre} }{n (q_n^{\text{post}})^2} \right). 
\end{align*}
For the cross terms, we have 
\begin{align*}
& \textup{Cov}\left( Z_i^{\textsc{ssiv}} r_{1,i}, Z_j^{\textsc{ssiv}} r_{1,j} \right) \\
&= \frac{1}{(n-1)^4} \\
& \left[
\begin{array}{c}
\textup{Cov}\left( \sum\limits_{k\ne i} A_{ik}^\text{pre} (T_k-\pi)  \sum\limits_{k\ne i} \left( - \frac{U_{ik} W_{ik}}{\theta_i^*(y)^2} + \frac{2\theta_i^*(x)}{\theta_i^*(y)^3}  W_{ik}^2 \right), \sum\limits_{k\ne j} A_{jk}^\text{pre} (T_k-\pi) \sum\limits_{k\ne j} \left( - \frac{U_{jk} W_{jk}}{\theta_j^*(y)^2} + \frac{2\theta_j^*(x)}{\theta_j^*(y)^3}  W_{jk}^2 \right) \right) \\
+ \textup{Cov}\left( \sum\limits_{k\ne i} A_{ik}^\text{pre} (T_k-\pi)  \sum\limits_{k\ne i} \left( - \frac{U_{ik} W_{ik}}{\theta_i^*(y)^2} + \frac{2\theta_i^*(x)}{\theta_i^*(y)^3}  W_{ik}^2 \right), \sum\limits_{k\ne j} A_{jk}^\text{pre} (T_k-\pi) \sum\limits_{(k,l)} \left( - \frac{U_{jk} W_{jl}}{\theta_j^*(y)^2}  + \frac{2\theta_j^*(x)}{\theta_j^*(y)^3}  W_{jk} W_{jl} \right) \right) \\
+ \textup{Cov}\left( \sum\limits_{k\ne i} A_{ik}^\text{pre} (T_k-\pi)  \sum\limits_{(k,l)} \left( - \frac{U_{ik} W_{il}}{\theta_i^*(y)^2}  + \frac{2\theta_i^*(x)}{\theta_i^*(y)^3}  W_{ik} W_{il} \right), \sum\limits_{k\ne j} A_{jk}^\text{pre} (T_k-\pi) \sum\limits_{k\ne j} \left( - \frac{U_{jk} W_{jk}}{\theta_j^*(y)^2} + \frac{2\theta_j^*(x)}{\theta_j^*(y)^3}  W_{jk}^2 \right) \right) \\
+ \textup{Cov}\left( \sum\limits_{k\ne i} A_{ik}^\text{pre} (T_k-\pi)  \sum\limits_{(k,l)} \left( - \frac{U_{ik} W_{il}}{\theta_i^*(y)^2}  + \frac{2\theta_i^*(x)}{\theta_i^*(y)^3}  W_{ik} W_{il} \right), \sum\limits_{k\ne j} A_{jk}^\text{pre} (T_k-\pi) \sum\limits_{(k,l)} \left( - \frac{U_{jk} W_{jl}}{\theta_j^*(y)^2}  + \frac{2\theta_j^*(x)}{\theta_j^*(y)^3}  W_{jk} W_{jl} \right) \right)
\end{array}
\right] \\
&= \frac{1}{(n-1)^4} 
\left[
\begin{array}{c}
\sum\limits_{(k,l,h)} 
\textup{Cov}\left( A_{ik}^\text{pre} (T_k-\pi) \left( - \frac{U_{il} W_{il}}{\theta_i^*(y)^2} + \frac{2\theta_i^*(x)}{\theta_i^*(y)^3}  W_{ik}^2 \right), 
A_{jk}^\text{pre} (T_k-\pi) 
\left( - \frac{U_{jh} W_{jh}}{\theta_j^*(y)^2} + \frac{2\theta_j^*(x)}{\theta_j^*(y)^3}  W_{jh}^2 \right) \right) \\
+ \sum\limits_{(k,l,h)} 
\textup{Cov}\left( A_{ik}^\text{pre} (T_k-\pi) \left( - \frac{U_{il} W_{il}}{\theta_i^*(y)^2} + \frac{2\theta_i^*(x)}{\theta_i^*(y)^3}  W_{il}^2 \right), 
A_{jh}^\text{pre} (T_h-\pi) 
\left( - \frac{U_{jk} W_{jh}}{\theta_j^*(y)^2}  + \frac{2\theta_j^*(x)}{\theta_j^*(y)^3}  W_{jk} W_{jh} \right) \right) \\
+ \sum\limits_{(k,l,h)} 
\textup{Cov}\left( A_{ik}^\text{pre} (T_k-\pi) \left( - \frac{U_{ik} W_{ih}}{\theta_i^*(y)^2}  + \frac{2\theta_i^*(x)}{\theta_i^*(y)^3}  W_{ik} W_{ih} \right), 
A_{jh}^\text{pre} (T_h-\pi) \left( - \frac{U_{jl} W_{jl}}{\theta_j^*(y)^2} + \frac{2\theta_j^*(x)}{\theta_j^*(y)^3}  W_{jl}^2 \right) \right) \\
+ \sum\limits_{(k,l,h)} 
\textup{Cov}\left( A_{ik}^\text{pre} (T_k-\pi) \left( - \frac{U_{il} W_{ih}}{\theta_i^*(y)^2}  + \frac{2\theta_i^*(x)}{\theta_i^*(y)^3}  W_{il} W_{ih} \right), 
A_{jk}^\text{pre} (T_k-\pi) 
\left( - \frac{U_{jl} W_{jh}}{\theta_j^*(y)^2}  + \frac{2\theta_j^*(x)}{\theta_j^*(y)^3}  W_{jl} W_{jh} \right) \right)
\end{array}
\right] \\
&\le C_1 \frac{n^3 (q_n^{\text{pre}})^2 (q_n^{\text{post}})^2}{n^4 (q_n^{\text{post}})^4}
+ C_2 \frac{n^3 q_n^{\text{pre}} \min\{q_n^{\text{pre}},q_n^{\text{post}}\} (q_n^{\text{post}})^2 }{n^4 (q_n^{\text{post}})^4}
+ C_3 \frac{n^3 (q_n^{\text{pre}})^2 (q_n^{\text{post}})^4}{n^4 (q_n^{\text{post}})^4}
= O\left( \frac{ (q_n^{\text{pre}})^2 }{n (q_n^{\text{post}})^2} \right).
\end{align*}
Therefore,
\begin{equation*}
\V\left( \frac{1}{n}\sum_{i=1}^n Z_i^{\textsc{ssiv}} r_{1,i} \right)
\le C_1 \frac{ q_n^{\text{pre}} }{n^2 (q_n^{\text{post}})^2} + C_2 \frac{ (q_n^{\text{pre}})^2 }{n (q_n^{\text{post}})^2}
\end{equation*}
and thus 
\begin{equation}
\frac{1}{n}\sum_{i=1}^n Z_i^{\textsc{ssiv}} r_{1,i}
= O_{\mathbb{P}}\left( \frac{\min\{q_n^{\text{pre}}, q_n^{\text{post}}\}}{n (q_n^{\text{post}})^2} \right)
+ O_{\mathbb{P}}\left( \frac{ \sqrt{q_n^{\text{pre}} } }{n q_n^{\text{post}}} \right)
+ O_{\mathbb{P}}\left( \frac{ q_n^{\text{pre}} }{\sqrt{n} q_n^{\text{post}}} \right).
\label{eq:Zir1i}
\end{equation}
By combining with \eqref{eq:Zir0i} and \eqref{eq:Zir1i},
\begin{align}
\frac{1}{n}\sum_{i=1}^n Z_i^{\textsc{ssiv}}(r_{0,i} + r_{1,i} )
&= O_{\mathbb{P}}\left( \frac{\min\{q_n^{\text{pre}}, q_n^{\text{post}}\}}{q_n^{\text{post}}} \right)
+ O_{\mathbb{P}}\left( \frac{ q_n^{\text{pre}} }{\sqrt{n} q_n^{\text{post}}} \right)
+ O_{\mathbb{P}}\left( \frac{ \sqrt{q_n^{\text{pre}}} }{ \sqrt{n q_n^{\text{post}}} }  \right).
\label{eq:Ziri}
\end{align}
In particular, this implies that 
\begin{equation*}
\frac{1}{n}\sum_{i=1}^n M_i Z_i^{\textsc{ssiv}}
= E\left[ Z_i^{\textsc{ssiv}}(r_{0,i} + r_{1,i} - \mu_{r_1,i}) \right] + O_{\mathbb{P}}\left( \sqrt{n} q_n^{\text{pre}} \right). 
\end{equation*}

\subsection{Proof of Lemma \ref{lemma:ZiMi_alt}}\label{proof:ZiMi_alt}

\fontsize{11}{13.2}
By the decomposition, we have
\[
\frac{1}{n}\sum_{i=1}^n (Z_i^{\text{alt}} - \pi) M_i 
= \frac{1}{n}\sum_{i=1}^n (Z_i^{\text{alt}} - \pi) \xi_i
+ \frac{1}{n}\sum_{i=1}^n (Z_i^{\text{alt}} - \pi) r_{0,i}
+ \frac{1}{n}\sum_{i=1}^n (Z_i^{\text{alt}} - \pi) r_{1,i}.
\]
\underline{First}, by Lemma \ref{lemma:Zi_alt_phii}, we have 
\begin{align}
\frac{1}{n}\sum_{i=1}^n (Z_i^{\text{alt}} - \pi) \xi_i 
= \frac{1}{n}\sum_{i=1}^n (T_i - \pi) E\left[ \frac{ A_{ij}^{\text{pre}} \xi_j }{ q_n^{\text{pre}} g_0(j) } \mid w_i \right]
+ O_{\mathbb{P}}\left( \frac{1}{\sqrt{n} \sqrt{nq_n^{\text{pre}}}} \right)
= O_{\mathbb{P}}\left( \frac{1}{\sqrt{n}} \right).
\label{eq:Zigi}
\end{align}
\underline{Next}, for $\frac{1}{n} \sum_{i=1}^n (Z_i^{\text{alt}} - \pi) r_{0,i}$, we rewrite it as 
\begin{align*}
& \frac{1}{n}\sum_{j=1}^n (T_j - \pi) \sum_{i \ne j}^n \frac{ A_{ij}^{\text{pre}} r_{0,i} }{ N_i } 
= S_1^0 - S_2^0
\end{align*}
where 
\begin{align}
S_1^0
=& \frac{1}{n}\sum_{j=1}^n (T_j - \pi) \sum_{i \ne j}^n \frac{ A_{ij}^{\text{pre}} r_{0,i} }{ (n-1)q_n^{\text{pre}} g_0(i) }, 
\label{eq:S10} \\
S_2^0 
=& \frac{1}{n}\sum_{j=1}^n (T_j - \pi) \sum_{i \ne j}^n \frac{ A_{ij}^{\text{pre}} r_{0,i} \left(N_i - (n-1)q_n^{\text{pre}} g_0(i) \right) }{ (n-1)q_n^{\text{pre}} g_0(i) N_i }. 
\label{eq:S20}
\end{align}
For $S_1^0$, the expectation is 
\begin{align*}
& E\left[ S_1^0 \right]
= E\left[ \frac{1}{n}\sum_{j=1}^n (T_j - \pi) \sum_{i \ne j}^n \frac{ A_{ij}^{\text{pre}} r_{0,i} }{ (n-1)q_n^{\text{pre}} g_0(i) } \right]
= \frac{1}{n}\sum_{j=1}^n E\left[ (T_j - \pi) \sum_{i \ne j}^n \frac{ A_{ij}^{\text{pre}} \xi_i \frac{1}{n-1} R_{ij} }{ (n-1)q_n^{\text{pre}} g_0(i) } \right]
\asymp  \frac{ \min\{q_n^{\text{pre}}, q_n^{\text{post}} \} }{n q_n^{\text{pre}} q_n^{\text{post}}},
\end{align*}
and the variance is 
\begin{align*}
& \V\left[ S_1^0 \right]
= \V\left[ \frac{1}{n}\sum_{j=1}^n (T_j - \pi) \sum_{i \ne j}^n \frac{ A_{ij}^{\text{pre}} r_{0,i} }{ (n-1)q_n^{\text{pre}} g_0(i) } \right] \\
=& \frac{1}{(n-1)^2 (n q_n^{\text{pre}})^2  } 
\left[ 
\begin{array}{c}
\sum\limits_{j=1}^n \sum\limits_{i \ne j}^n 
\V\left( \frac{ (T_j - \pi)  A_{ij}^{\text{pre}} r_{0,i} }{  g_0(i) } \right)
+ \sum\limits_{j=1}^n \sum\limits_{(i_1,i_2)}  
\textup{Cov} \left(  \frac{ (T_j - \pi) A_{i_1j}^{\text{pre}} r_{0,i_1} }{ g_0(i_1) }, 
\frac{ (T_j - \pi) A_{i_2 j}^{\text{pre}} r_{0,i_2} }{ g_0(i_2) } \right) \\
+ \sum\limits_{(j_1,j_2)} 
\sum\limits_{i\ne j_1,j_2}  
\textup{Cov} \left(  \frac{ (T_{j_1} - \pi) A_{ij_1}^{\text{pre}} r_{0,i} }{ g_0(i) }, 
\frac{ (T_{j_2} - \pi) A_{ij_2}^{\text{pre}} r_{0,i} }{ g_0(i) } \right) \\
+ \sum\limits_{(j_1,j_2)} \sum\limits_{(i_1,i_2)}  
\textup{Cov} \left(  \frac{ (T_{j_1} - \pi) A_{i_1j_1}^{\text{pre}} r_{0,i_1} }{ g_0(i_1) }, \frac{ (T_{j_2} - \pi) A_{i_2j_2}^{\text{pre}} r_{0,i_2} }{ g_0(i_2) } \right)
\end{array}
\right]
\end{align*}
where we calculate each term:
\begin{align}
& \V\left( (T_j - \pi) \frac{ A_{ij}^{\text{pre}} r_{0,i} }{  g_0(i) } \right)
\le \frac{1}{(n-1)^2} E\left[ \frac{ (T_j - \pi)^2 }{ g_0(i)^2 } A_{ij}^{\text{pre}} \xi_i^2 \sum\limits_{l\ne i} R_{il}^2 \right]
\asymp 
\frac{ q_n^{\text{post}} }{ n q_n^{\text{post}} }, \notag \\
& \textup{Cov} \left(  \frac{ (T_j - \pi) A_{i_1j}^{\text{pre}} r_{0,i_1} }{ g_0(i_1) }, 
\frac{ (T_j - \pi) A_{i_2 j}^{\text{pre}} r_{0,i_2} }{ g_0(i_2) } \right) \notag \\
=& E\left[ (T_j - \pi)^2  \frac{ A_{i_1j}^{\text{pre}} r_{0,i_1} }{ g_0(i_1) } \frac{ A_{i_2j}^{\text{pre}} r_{0,i_2} }{ g_0(i_2) } \right]
- E\left[ \frac{ (T_j - \pi) A_{i_1j}^{\text{pre}} r_{0,i_1} }{ g_0(i_1) } \right]
E\left[ \frac{ (T_j - \pi) A_{i_2 j}^{\text{pre}} r_{0,i_2} }{ g_0(i_2) }  \right] \notag \\
=& \frac{1}{ (n-1)^2}  
\left( 
E\left[ \frac{(T_j - \pi)^2 }{ g_0(i_1) g_0(i_2)}   
\sum\limits_{l\ne i_1,i_2,j}
\xi_{i_1} \xi_{i_2} A_{i_1j}^\text{pre} A_{i_2j}^\text{pre} R_{i_1l}R_{i_2l}
\right]
- E\left[ \frac{ (T_j - \pi) A_{i_1j}^{\text{pre}} \xi_{i_1} R_{i_1 j} }{ g_0(i_1) }  \right]
E\left[ \frac{ (T_j - \pi) A_{i_2 j}^{\text{pre}} \xi_{i_2} R_{i_2 j} }{ g_0(i_2) }  \right]
\right) \notag \\
\asymp & \frac{(q_n^{\text{pre}})^2}{n} 
- \frac{\min\{q_n^{\text{pre}}, q_n^{\text{post}}\}^2}{n^2 (q_n^{\text{post}})^2}, \notag \\
& \textup{Cov} \left(  \frac{ (T_{j_1} - \pi) A_{ij_1}^{\text{pre}} r_{0,i} }{ g_0(i) }, 
\frac{ (T_{j_2} - \pi) A_{ij_2}^{\text{pre}} r_{0,i} }{ g_0(i) } \right) \notag \\
=& E\left[ (T_{j_1} - \pi) (T_{j_2} - \pi) 
\frac{ A_{ij_1}^{\text{pre}} A_{ij_2}^{\text{pre}} r_{0,i}^2 }{ g_0(i)^2 }  \right]
- E\left[ \frac{ (T_{j_1} - \pi) A_{ij_1}^{\text{pre}} r_{0,i} }{ g_0(i) } \right]
E\left[ \frac{ (T_{j_2} - \pi) A_{ij_2}^{\text{pre}} r_{0,i} }{ g_0(i) } \right] \notag \\
=& \frac{1}{(n-1)^2}
\left(
E\left[ (T_{j_1} - \pi) (T_{j_2} - \pi) 
\frac{ A_{ij_1}^{\text{pre}} A_{ij_2}^{\text{pre}} \xi_i^2 R_{ij_1} R_{ij_2} }{ g_0(i)^2 }  \right]
- E\left[ \frac{ (T_{j_1} - \pi) A_{i j_1}^{\text{pre}} \xi_i R_{ij_1} }{ g_0(i) } \right]
E\left[ \frac{ (T_{j_2} - \pi) A_{i j_2}^{\text{pre}} \xi_i R_{ij_2} }{ g_0(i) } \right]
\right) \notag \\
\asymp & \frac{\min\{q_n^{\text{pre}}, q_n^{\text{post}}\}^2}{n^2 (q_n^{\text{post}})^2}, \notag \\
& \textup{Cov} \left(  \frac{ (T_{j_1} - \pi) A_{i_1j_1}^{\text{pre}} r_{0,i_1} }{ g_0(i_1) }, \frac{ (T_{j_2} - \pi) A_{i_2j_2}^{\text{pre}} r_{0,i_2} }{ g_0(i_2) } \right) \notag \\
=& E\left[ (T_{j_1} - \pi) (T_{j_2} - \pi) 
\frac{ A_{i_1j_1}^{\text{pre}} r_{0,i_1} }{ g_0(i_1) } \frac{ A_{i_2j_2}^{\text{pre}} r_{0,i_2} }{ g_0(i_2) } \right]
- E\left[ \frac{ (T_{j_1} - \pi) A_{i_1j_1}^{\text{pre}} r_{0,i_1} }{ g_0(i_1) } \right]
E\left[ \frac{ (T_{j_2} - \pi) A_{i_2j_2}^{\text{pre}} r_{0,i_2} }{ g_0(i_2) } \right] \notag \\
=& \frac{1}{(n-1)^2}
E\left[ \frac{(T_{j_1} - \pi) (T_{j_2} - \pi) }{ g_0(i_1) g_0(i_2) } \xi_{i_1} \xi_{i_2} A_{i_1j_1}^{\text{pre}} R_{i_1 j_2} A_{i_2j_2}^{\text{pre}} R_{i_2 j_1} \right]
\asymp \frac{ (q_n^{\text{pre}})^2 }{n^2 }.
\label{eq:Aijr02}
\end{align}
By combining these terms in \eqref{eq:Aijr02}, we have
\begin{align}
S_1^0 
= E(S_1^0)
+ o_{\mathbb{P}}\left( \frac{ 1 }{n \max\{q_n^{\text{pre}}, q_n^{\text{post}} \} } \right).
\label{eq:Zir0_S1}
\end{align}
Recall \eqref{eq:AijNi} that
\begin{align*}
\frac{ A_{ij}^{\text{pre}} r_{0,i} (N_i - (n-1)q_n^{\text{pre}} g_0(i)) }{ q_n^{\text{pre}} g_0(i) N_i }
= \frac{ A_{ij}^{\text{pre}} r_{0,i} ((N_i - A_{ij}^{\text{pre}} + 1) - (n-1) q_n^{\text{pre}} g_0(i)) }{ q_n^{\text{pre}} g_0(i) (N_i - A_{ij}^{\text{pre}} + 1) }.
\end{align*}
Define 
\[
B_{ij}
= \frac{ A_{ij}^{\text{pre}} r_{0,i} \left(N_i - (n-1)q_n^{\text{pre}} g_0(i) \right) }{ (n-1) q_n^{\text{pre}} g_0(i) N_i }
\]
and thus 
\[
S_2^0 
= \frac{1}{n}\sum_{j=1}^n (T_j - \pi) \sum_{i \ne j}^n B_{ij}.
\]
We bound it in $L_2$ norm:
\begin{align*}
E\left[ (S_2^0)^2 \right]
= E\left[ \left( \frac{1}{n}\sum_{j=1}^n (T_j - \pi) \sum_{i \ne j}^n B_{ij} \right)^2 \right] 
&= \frac{1}{n^2} \sum_{j=1}^n E\left[ (T_j - \pi)^2 \left(  \sum_{i \ne j}^n B_{ij}^2 + \sum_{(i_1,i_2)} B_{i_1j} B_{i_2j} \right) \right] \\
+ \frac{1}{n^2} \sum_{(j,k)} & E\left[ (T_j - \pi)(T_k - \pi) \left(  \sum_{i \ne j,k}^n B_{ij} B_{ik} + \sum_{(i_1,i_2)}^n B_{i_1j} B_{i_2k} \right)  \right].
\end{align*}
We analyze these terms one by one.
Conditional on $w$, $A_{ij}^{\text{pre}}$ and $\frac{N_i - (n-1)q_n^{\text{pre}} g_0(i)}{N_i}$ are independent.
For $E\left( (T_j - \pi)^2 B_{ij}^2 \right)$, we have 
\begin{align*}
E\left((T_j - \pi)^2 B_{ij}^2\right)
=& E\left[  E\left[ (T_j - \pi)^2 A_{ij}^{\text{pre}} r_{0,i}^2  \mid w\right] 
E\left[ \left( \frac{ (N_i - A_{ij}^{\text{pre}} + 1) - (n-1) q_n^{\text{pre}} g_0(i) }{ (n-1) q_n^{\text{pre}} g_0(i) (N_i - A_{ij}^{\text{pre}} + 1) } \right)^2 \mid w \right] \right] \\
\le & C \frac{1}{(n q_n^{\text{pre}})^3} \frac{1}{(n-1)^2}
E\left[ (T_j - \pi)^2 A_{ij}^\text{pre}  \xi_i^2 \sum_{k\ne i} R_{ik}^2 \right]  
\asymp \frac{ 1 }{n^4 (q_n^{\text{pre}})^2 q_n^{\text{post}} }
\end{align*}
by \eqref{eq:inverBij^2} and analogous argument to \eqref{eq:Aijr02}.
Conditional on $w$, $A_{ij}^{\text{pre}}$, $A_{jk}^{\text{pre}}$, $\frac{N_j - (n-1)q_n^{\text{pre}} g_0(j)}{N_j}$ and $\frac{N_k - (n-1)q_n^{\text{pre}} g_0(k)}{N_k}$ are independent.
For $E( (T_j - \pi)^2 B_{i_1 j} B_{i_2 j})$, we have 
\begin{align*}
& E\left((T_j - \pi)^2 B_{i_1 j} B_{i_2 j} \right) \\
=& E\left[  
E\left[ (T_j - \pi)^2 A_{i_1j}^{\text{pre}} A_{i_2j}^{\text{pre}} r_{0,i_1} r_{0,i_2} \mid w \right]
E\left[ \frac{ N_{i_1} - (n-1) q_n^{\text{pre}} g_0(i_1) }{ (n-1) q_n^{\text{pre}} g_0(i_1) N_{i_1} } \mid w\right]
E\left[ \frac{ N_{i_2} - (n-1) q_n^{\text{pre}} g_0(i_2) }{ (n-1) q_n^{\text{pre}} g_0(i_2) N_{i_2} }\mid w\right]
\right] \\
\le & 
C \frac{\min\{q_n^{\text{pre}}, q_n^{\text{post}} \}^2}{ (n q_n^{\text{post}})^2}
\frac{ 1 }{(n q_n^{\text{pre}})^4 }
\asymp \frac{\min\{q_n^{\text{pre}}, q_n^{\text{post}}\}^2}{ n^6 (q_n^{\text{pre}})^4 (q_n^{\text{post}})^2}
\end{align*}
where by \eqref{eq:inverBij} and analogous argument to \eqref{eq:Aijr02}.
For $E( (T_j - \pi)(T_k-\pi) B_{i j} B_{i k})$, we have 
\begin{align*}
& E( (T_j - \pi)(T_k-\pi) B_{i j} B_{i k}) \\
=& E\left[ E\left[ (T_j - \pi)(T_k-\pi)  A_{ij}^{\text{pre}} A_{ik}^{\text{pre}} r_{0,i}^2 \mid w \right] 
E\left[ \left( \frac{ (N_i - (n-1) q_n^{\text{pre}} g_0(i)) }{ (n-1) q_n^{\text{pre}} g_0(i) N_i } \right)^2  \mid w \right]
\right] \\
\le & C \frac{\min\{q_n^{\text{pre}},q_n^{\text{post}}\}^2}{(n q_n^{\text{post}})^2 } \frac{1}{(n q_n^{\text{pre}})^3}
= C \frac{\min\{q_n^{\text{pre}}, q_n^{\text{post}}\}^2}{n^5 (q_n^{\text{pre}})^3 (q_n^{\text{post}})^2}
\end{align*}
by \eqref{eq:inverBij^2} and analogous argument to \eqref{eq:Aijr02}.
For $E( (T_j - \pi)(T_k-\pi) B_{i_1 j} B_{i_2 k})$, we have 
\begin{align*}
& E( (T_j - \pi)(T_k-\pi) B_{i_1 j} B_{i_2 k}) \\
=& E\left[  
E\left[ (T_j - \pi)(T_k-\pi) A_{i_1j}^{\text{pre}} A_{i_2k}^{\text{pre}} r_{0,i_1} r_{0,i_2}  \mid w \right] 
E\left[ \tfrac{   (N_{i_1} - (n-1) q_n^{\text{pre}} g_0(i_1)) }{ (n-1) q_n^{\text{pre}} g_0(i_1) N_{i_1} } \mid w \right]
E\left[ \tfrac{  (N_{i_2} - (n-1) q_n^{\text{pre}} g_0(i_2)) }{ (n-1) q_n^{\text{pre}} g_0(i_2) N_{i_2} } \mid w \right] \right] \\
=& E\left[ 
E\left[ \frac{(T_j - \pi)(T_k-\pi) }{(n-1)^2} \xi_{i_1} \xi_{i_2} A_{i_1j}^{\text{pre}} A_{i_2k}^{\text{pre}} R_{i_1j} R_{i_2k} \mid T,w \right] 
E\left[ \tfrac{   (N_{i_1} - (n-1) q_n^{\text{pre}} g_0(i_1)) }{ (n-1) q_n^{\text{pre}} g_0(i_1) N_{i_1} } \mid w \right]
E\left[ \tfrac{  (N_{i_2} - (n-1) q_n^{\text{pre}} g_0(i_2)) }{ (n-1) q_n^{\text{pre}} g_0(i_2) N_{i_2} } \mid w \right] \right] \\
\le & C \frac{1}{(n q_n^{\text{pre}})^4} \frac{\min\{q_n^{\text{pre}}, q_n^{\text{post}}\}^2}{n^2 (q_n^{\text{post}})^2}
= C  \frac{\min\{q_n^{\text{pre}}, q_n^{\text{post}}\}^2}{n^6 (q_n^{\text{pre}})^4 (q_n^{\text{post}})^2 }
\end{align*}
by \eqref{eq:inverBij} and analogous argument to \eqref{eq:Aijr02}.
By combining these results, we have 
\begin{align*}
E\left[(S_2^0)^2\right]
\le & 
C_1 \frac{ 1 }{n^4 (q_n^{\text{pre}})^2 q_n^{\text{post}}}
+ C_2 \frac{n \min\{q_n^{\text{pre}}, q_n^{\text{post}}\}^2}{ n^6 (q_n^{\text{pre}})^4 (q_n^{\text{post}})^2}
+ C_3 \frac{n \min\{q_n^{\text{pre}}, q_n^{\text{post}}\}^2}{n^5 (q_n^{\text{pre}})^3 (q_n^{\text{post}})^2} 
+ C_4 \frac{n^2 \min\{q_n^{\text{pre}}, q_n^{\text{post}}\}^2}{n^6 (q_n^{\text{pre}})^4 (q_n^{\text{post}})^2} \\
=& 
O\left( \frac{1}{n^4 (q_n^{\text{pre}})^2 q_n^{\text{post}}} \right) 
+ O\left( \frac{\min\{q_n^{\text{pre}}, q_n^{\text{post}}\}^2}{(n q_n^{\text{pre}})^4 (q_n^{\text{post}})^2} \right)
\end{align*}
and thus
\begin{align}
S_2^0
= O_{\mathbb{P}}\left( \frac{ 1 }{ n^2 q_n^{\text{pre}} \sqrt{q_n^{\text{post}}} } \right) 
+ O_{\mathbb{P}}\left( \frac{\min\{q_n^{\text{pre}}, q_n^{\text{post}}\} }{ n^2 (q_n^{\text{pre}})^2 q_n^{\text{post}}} \right).
\label{eq:Zir0_S2}
\end{align}
Therefore, by combining \eqref{eq:Zir0_S1} and \eqref{eq:Zir0_S2}, we have
\begin{align}
\frac{1}{n} \sum_{i=1}^n (Z_i^{\text{alt}} - \pi) r_{0,i}
= E(S_1^0)
+ o_{\mathbb{P}}\left( \frac{ 1 }{n \max\{q_n^{\text{pre}}, q_n^{\text{post}}\} } \right)
= O_{\mathbb{P}}\left( \frac{ 1 }{n \max\{q_n^{\text{pre}}, q_n^{\text{post}}\} } \right).
\label{eq:Ziroi_alt}
\end{align}
\underline{Next}, for $\frac{1}{n} \sum_{i=1}^n (Z_i^{\text{alt}} - \pi) r_{1,i}$, we have
\begin{align*}
\frac{1}{n} \sum_{i=1}^n (Z_i^{\text{alt}} - \pi) r_{1,i}
= \frac{1}{n}\sum_{j=1}^n (T_j - \pi) \sum_{i \ne j}^n \frac{ A_{ij}^{\text{pre}} r_{1,i} }{ N_i } 
= S_1^1 - S_2^1
\end{align*}
where 
\begin{align*}
S_1^1
=& \frac{1}{n }\sum_{j=1}^n (T_j - \pi) \sum_{i \ne j}^n \frac{ A_{ij}^{\text{pre}} r_{1,i} }{ (n-1) q_n^{\text{pre}} g_0(i) }, \\
S_2^1 
=& \frac{1}{n }\sum_{j=1}^n (T_j - \pi) \sum_{i \ne j}^n \frac{ A_{ij}^{\text{pre}} r_{1,i} (N_i - (n-1)q_n^{\text{pre}} g_0(i)) }{ (n-1) q_n^{\text{pre}} g_0(i) N_i }. 
\end{align*}
Again, for $S_1^1$, the expectation is 
\begin{align*}
E\left[ S_1^1 \right]
=& E\left[ \frac{1}{n }\sum_{j=1}^n (T_j - \pi) \sum_{i \ne j}^n \frac{ A_{ij}^{\text{pre}} r_{1,i} }{ (n-1) q_n^{\text{pre}} g_0(i) } \right] \\
=& \frac{1}{n (n-1)^3 q_n^{\text{pre}}} \sum_{j=1}^n 
E\left[ (T_j - \pi) \sum_{i \ne j}^n \frac{ A_{ij}^{\text{pre}}  }{   g_0(i) } \left( - \frac{U_{ij} W_{ij}}{\theta_i^*(y)^2} + \frac{2\theta_i^*(x)}{\theta_i^*(y)^3}  W_{ij}^2 \right) \right] 
\asymp \frac{ \min\{q_n^{\text{pre}}, q_n^{\text{post}}\} }{n^2 q_n^{\text{pre}} (q_n^{\text{post}})^2},
\end{align*}
and the variance is 
\begin{align*}
& \V(S_1^1)
= \V\left( \frac{1}{n }\sum_{j=1}^n (T_j - \pi) \sum_{i \ne j}^n \frac{ A_{ij}^{\text{pre}} r_{1,i} }{ (n-1) q_n^{\text{pre}} g_0(i) } \right) \\
=& \frac{1}{ (n-1)^2 (n q_n^{\text{pre}})^2 }
\left[ 
\begin{array}{c}
\sum\limits_{j=1}^n \sum\limits_{i \ne j}^n  
\V\left( \frac{ (T_j - \pi) 
 A_{ij}^{\text{pre}} r_{1,i} }{  g_0(i) } \right)
+ \sum\limits_{j=1}^n \sum\limits_{(i_1,i_2)}
\textup{Cov}\left( \frac{ (T_j - \pi) 
 A_{i_1j}^{\text{pre}} r_{1,i_1} }{  g_0(i_1) }, \frac{ (T_j - \pi) 
 A_{i_2j}^{\text{pre}} r_{1,i_2} }{  g_0(i_2) } \right) \\
 + \sum\limits_{(j_1,j_2)} \sum\limits_{i\ne j_1,j_2}
\textup{Cov}\left( \frac{ (T_{j_1} - \pi) 
 A_{ij_1}^{\text{pre}} r_{1,i} }{  g_0(i) }, \frac{ (T_{j_2} - \pi) 
 A_{ij_2}^{\text{pre}} r_{1,i} }{  g_0(i) } \right) \\
 + \sum\limits_{(j_1,j_2)} \sum\limits_{(i_1,i_2)}
\textup{Cov}\left( \frac{ (T_{j_1} - \pi) 
 A_{i_1j_1}^{\text{pre}} r_{1,i_1} }{  g_0(i_1) }, \frac{ (T_{j_2} - \pi) 
 A_{i_2j_2}^{\text{pre}} r_{1,i_2} }{  g_0(i_2) } \right)
\end{array}
\right] \\
=& \frac{1}{ (n-1)^2 (n q_n^{\text{pre}})^2 }
\left( S_{11}^1 + S_{12}^1 + S_{13}^1 + S_{14}^1 \right).
\end{align*}
For $S_{11}^1$, we have 
\begin{align}
& \V\left( \frac{ (T_j - \pi) 
 A_{ij}^{\text{pre}} r_{1,i} }{  g_0(i) } \right)
\le E\left[ \frac{ (T_j - \pi)^2 }{g_0(i)^2} A_{ij}^\text{pre} r_{1,i}^2 \right] \notag \\
&= \frac{1}{(n-1)^4} 
E\left[ 
\begin{array}{c}
\frac{ (T_j - \pi)^2 }{g_0(i)^2}
\left(
\begin{array}{c}
A_{ij}^\text{pre}  
\sum\limits_{k\ne i} 
\left( 
- \frac{U_{ik} W_{ik} }{\theta_i^*(y)^2} 
+ \frac{2 \theta_i^*(x)}{\theta_i^*(y)^3} W_{ik}^2
\right)^2 \\
+ A_{ij}^\text{pre}  
\sum\limits_{(k,l)}
\left( 
- \frac{U_{ik} W_{ik} }{\theta_i^*(y)^2} 
+ \frac{2 \theta_i^*(x)}{\theta_i^*(y)^3} W_{ik}^2
\right)
\left( 
- \frac{U_{il} W_{il} }{\theta_i^*(y)^2} 
+ \frac{2 \theta_i^*(x)}{\theta_i^*(y)^3} W_{il}^2
\right) \\
+ A_{ij}^\text{pre} 
\sum\limits_{(k,l)} 
\left(
- \frac{U_{ik} W_{il} }{\theta_i^*(y)^2} 
+ \frac{2 \theta_i^*(x)}{\theta_i^*(y)^3} W_{ik} W_{il}
\right)^2 \\
+ 2 A_{ij}^\text{pre} 
\sum\limits_{k\ne i,j} 
\left( 
- \frac{U_{ik} W_{ik} }{\theta_i^*(y)^2} 
+ \frac{2 \theta_i^*(x)}{\theta_i^*(y)^3} W_{ik}^2
\right) 
\left( 
- \frac{U_{ik} W_{ij} }{\theta_i^*(y)^2} 
+ \frac{2 \theta_i^*(x)}{\theta_i^*(y)^3} 
W_{ik} W_{ij}
\right)
\end{array}
\right)
\end{array}
\right] \notag \\
& \asymp \frac{q_n^{\text{pre}}}{(n q_n^{\text{post}})^2}.
\label{eq:A_ijr1}
\end{align}
For $S_{12}^1$, we have 
\begin{align}
& \textup{Cov}\left( \frac{ (T_j - \pi) 
A_{i_1j}^{\text{pre}} r_{1,i_1} }{  g_0(i_1) }, \frac{ (T_j - \pi) 
A_{i_2j}^{\text{pre}} r_{1,i_2} }{  g_0(i_2) } \right) \notag \\
=& E\left[ \frac{(T_j-\pi)^2}{g_0(i_1) g_0(i_2)} A_{i_1j}^{\text{pre}} A_{i_2 j}^{\text{pre}} r_{1,i_1} r_{1,i_2}
\right]
- E\left[ \frac{ (T_j - \pi) 
A_{i_1j}^{\text{pre}} r_{1,i_1} }{  g_0(i_1) } \right]
E\left[ \frac{ (T_j - \pi) 
A_{i_2j}^{\text{pre}} r_{1,i_2} }{  g_0(i_2) } \right] \notag \\
=& \frac{1}{(n-1)^4} 
E\left[ 
\begin{array}{c}
\frac{(T_j-\pi)^2}{g_0(i_1) g_0(i_2)}
\left(
\begin{array}{c}
A_{i_1j}^\text{pre} A_{i_2j}^\text{pre}
\sum\limits_{(l,h)} 
\left( 
- \frac{U_{i_1l} W_{i_1l} }{\theta_{i_1}^*(y)^2} 
+ \frac{2 \theta_{i_1}^*(x)}{\theta_{i_1}^*(y)^3} W_{i_1 l}^2
\right) 
\left( 
- \frac{U_{i_2h} W_{i_2h} }{\theta_{i_2}^*(y)^2} 
+ \frac{2 \theta_{i_2}^*(x)}{\theta_{i_2}^*(y)^3} W_{i_2h}^2
\right) \\
+ A_{i_1j}^\text{pre} A_{i_2j}^\text{pre}
\sum\limits_{(l,h)} 
\left( 
- \frac{U_{i_1l} W_{il} }{\theta_{i_1}^*(y)^2} 
+ \frac{2 \theta_{i_1}^*(x)}{\theta_{i_1}^*(y)^3} W_{i_2 l}^2
\right) 
\left( 
- \frac{U_{{i_2}h} W_{{i_2}j} }{\theta_{i_2}^*(y)^2} 
+ \frac{2 \theta_{i_2}^*(x)}{\theta_{i_2}^*(y)^3} W_{{i_2}h} W_{{i_2}j}
\right) \\
+ A_{i_1j}^\text{pre} A_{{i_2}j}^\text{pre} 
\sum\limits_{l\ne i_1,{i_2}} 
\left( 
- \frac{U_{i_1l} W_{i_1l}}{\theta_{i_1}^*(y)^2}  
+ \frac{2 \theta_{i_1}^*(x)}{\theta_{i_1}^*(y)^3} W_{i_1l}^2
\right) 
\left( 
- \frac{U_{{i_2}l} W_{{i_2}j} }{\theta_{i_2}^*(y)^2} 
+ \frac{2 \theta_{i_2}^*(x)}{\theta_{i_2}^*(y)^3} W_{{i_2}l} W_{{i_2}j}
\right) \\
+ A_{ij}^\text{pre} A_{{i_2}j}^\text{pre}
\sum\limits_{l\ne i_1,{i_2}} 
\left( 
- \frac{U_{i_1l} W_{i_1j} }{\theta_{i_1}^*(y)^2} 
+ \frac{2 \theta_{i_1}^*(x)}{\theta_{i_1}^*(y)^3} U_{i_1l} W_{i_1j}
\right) 
\left( 
- \frac{U_{{i_2}l} W_{{i_2}l}}{\theta_{i_2}^*(y)^2}  
+ \frac{2 \theta_{i_2}^*(x)}{\theta_{i_2}^*(y)^3} W_{{i_2}l}^2
\right) \\
+ A_{i_1j}^\text{pre} A_{{i_2}j}^\text{pre}
\sum\limits_{(l,h)} 
\left( 
- \frac{U_{i_1l} W_{i_1h}}{\theta_{i_1}^*(y)^2}  
+ \frac{2 \theta_{i_1}^*(x)}{\theta_{i_1}^*(y)^3} 
W_{il} W_{ih}
\right) 
\left( 
- \frac{U_{{i_2}l} W_{{i_2}h} }{\theta_{i_2}^*(y)^2} 
+ \frac{2 \theta_{i_2}^*(x)}{\theta_{i_2}^*(y)^3} 
W_{{i_2}l} W_{{i_2}h}
\right) \\
+ A_{i_1j}^\text{pre} A_{{i_2}j}^\text{pre}
\sum\limits_{l\ne i_1,j} 
\left( 
- \frac{U_{i_1l} W_{i_1j}}{\theta_{i_1}^*(y)^2}  
+ \frac{2 \theta_{i_1}^*(x)}{\theta_{i_1}^*(y)^3} 
W_{i_1l} W_{i_1j}
\right) 
\left( 
- \frac{U_{{i_2}l} W_{{i_2}j} }{\theta_{i_2}^*(y)^2} 
+ \frac{2 \theta_{i_2}^*(x)}{\theta_{i_2}*(y)^3} 
W_{{i_2}l} W_{{i_2}j}
\right)
\end{array}
\right)
\end{array}
\right] \notag \\
& - E\left[ \frac{ (T_j - \pi) 
A_{i_1j}^{\text{pre}} r_{1,i_1} }{  g_0(i_1) } \right]
E\left[ \frac{ (T_j - \pi) 
A_{i_2j}^{\text{pre}} r_{1,i_2} }{  g_0(i_2) } \right] \notag \\
&= O\left( \frac{ \min\{q_n^{\text{pre}},q_n^{\text{post}}\}^2 }{n^3 (q_n^{\text{post}})^2} \right) + O\left( \frac{ (q_n^{\text{pre}})^2 }{(n q_n^{\text{post}})^2} \right)
- O\left( \frac{ \min\{q_n^{\text{pre}}, q_n^{\text{post}}\}^2 }{n^4 (q_n^{\text{post}})^4} \right)
= O\left( \frac{ (q_n^{\text{pre}})^2 }{(n q_n^{\text{post}})^2} \right).
\label{eq:AijAkjrirk}
\end{align}
For $S_{13}^1$, we have 
\begin{align}
& \textup{Cov}\left( \frac{ (T_{j_1} - \pi) 
 A_{ij_1}^{\text{pre}} r_{1,i} }{  g_0(i) }, \frac{ (T_{j_2} - \pi) 
 A_{ij_2}^{\text{pre}} r_{1,i} }{  g_0(i) } \right) \notag \\
 =& E\left[ \frac{(T_{j_1} - \pi) (T_{j_2} - \pi) }{ g_0(i)^2 } A_{i{j_1}}^{\text{pre}} A_{i{j_2}}^{\text{pre}} r_{1,i}^2 \right] 
 - E\left[ \frac{ (T_{j_1} - \pi) 
 A_{ij_1}^{\text{pre}} r_{1,i} }{  g_0(i) } \right]
 E\left[ \frac{ (T_{j_2} - \pi) 
 A_{ij_2}^{\text{pre}} r_{1,i} }{  g_0(i) } \right]
 \notag \\
=& \frac{1}{(n-1)^4} 
E\left[ 
\begin{array}{c}
\frac{(T_{j_1} - \pi) (T_{j_2} - \pi)}{ g_0(i)^2 } 
A_{i{j_1}}^{\text{pre}} A_{i{j_2}}^{\text{pre}} \left( 
\begin{array}{c}
\left( - \frac{U_{i{j_1}} W_{i{j_1}}}{\theta_i^*(y)^2} + \frac{2\theta_i^*(x)}{\theta_i^*(y)^3}  W_{i{j_1}}^2 \right)
\left( - \frac{U_{i{j_2}} W_{i{j_2}}}{\theta_i^*(y)^2} + \frac{2\theta_i^*(x)}{\theta_i^*(y)^3}  W_{i{j_2}}^2 \right) \\
+ 2 \sum\limits_{l\ne i} 
\left( - \frac{U_{i{j_1}} W_{il}}{\theta_i^*(y)^2} + \frac{2\theta_i^*(x)}{\theta_i^*(y)^3}  W_{il}^2 \right) 
\left( - \frac{U_{i{j_1}} W_{i{j_2}}}{\theta_i^*(y)^2} + \frac{2\theta_i^*(x)}{\theta_i^*(y)^3}  W_{i{j_1}} W_{i{j_2}} \right) \\
+ \left( - \frac{U_{i{j_1}} W_{i{j_2}}}{\theta_i^*(y)^2} + \frac{2\theta_i^*(x)}{\theta_i^*(y)^3}  W_{i{j_1}} W_{i{j_2}} \right)^2 \\
+ \sum\limits_{l\ne {j_1},{j_2}} 
\left( - \frac{U_{il} W_{i{j_1}}}{\theta_i^*(y)^2} + \frac{2\theta_i^*(x)}{\theta_i^*(y)^3}  W_{il} W_{i{j_1}} \right) 
\left( - \frac{U_{il} W_{i{j_2}}}{\theta_i^*(y)^2} + \frac{2\theta_i^*(x)}{\theta_i^*(y)^3}  W_{il} W_{i{j_2}} \right) 
\end{array}
\right) 
\end{array}
\right] \notag  \\
& - E\left[ \frac{ (T_{j_1} - \pi) 
 A_{ij_1}^{\text{pre}} r_{1,i} }{  g_0(i) } \right]
 E\left[ \frac{ (T_{j_2} - \pi) 
 A_{ij_2}^{\text{pre}} r_{1,i} }{  g_0(i) } \right]
 \notag \\
\asymp 
& \frac{ \min\{q_n^{\text{pre}}, q_n^{\text{post}}\}^2 }{( n q_n^{\text{post}})^3}. 
\label{eq:AijAikr1i}
\end{align}
For $S_{14}^1$, we have
\begin{align}
& \textup{Cov}\left( \frac{ (T_{j} - \pi) 
 A_{i_1j}^{\text{pre}} r_{1,i_1} }{  g_0(i_1) }, \frac{ (T_{k} - \pi) 
 A_{i_2 k}^{\text{pre}} r_{1,i_2} }{  g_0(i_2) } \right)
 \notag \\
&= E\left[ (T_j - \pi) (T_k - \pi) \frac{ A_{i_1j}^{\text{pre}} A_{i_2k}^{\text{pre}} r_{1,i_1} r_{1,i_2} }{  g_0(i_1) g_0(i_2) } \right] 
- E\left[ \frac{ (T_{j} - \pi) 
 A_{i_1j}^{\text{pre}} r_{1,i_1} }{  g_0(i_1) } \right]
 E\left[ \frac{ (T_{k} - \pi) 
 A_{i_2 k}^{\text{pre}} r_{1,i_2} }{  g_0(i_2) } \right]
 \notag \\
&= \frac{1}{(n-1)^4} 
E\left[ 
\begin{array}{c}
\frac{(T_j - \pi) (T_k - \pi)}{ g_0(i_1) g_0(i_2) } A_{i_1j}^{\text{pre}} A_{i_2k}^{\text{pre}} 
\left( 
\begin{array}{c}
\sum\limits_{l\ne i_1, i_2}
\left( - \frac{U_{i_1 j} W_{i_1l}}{\theta_{i_1}^*(y)^2} + \frac{2\theta_{i_1}^*(x)}{\theta_{i_1}^*(y)^3}  W_{i_1 j} W_{i_1l} \right)
\left( - \frac{U_{i_2k} W_{i_2l}}{\theta_{i_2}^*(y)^2} + \frac{2\theta_{i_2}^*(x)}{\theta_{i_2}^*(y)^3}  W_{i_2 k} W_{i_2l} \right) \\
+ \sum\limits_{k_1\ne i_1} 
\left( - \frac{U_{i_1k_1} W_{ik_1}}{\theta_{i_1}^*(y)^2} + \frac{2\theta_{i_1}^*(x)}{\theta_{i_1}^*(y)^3}  W_{i_1k_1}^2 \right)
\left( - \frac{U_{i_2 j} W_{i_2 k}}{\theta_{i_2}^*(y)^2} + \frac{2\theta_{i_2}^*(x)}{\theta_{i_2}^*(y)^3}  W_{i_2 j} W_{i_2 k} \right) \\
+ \left( - \frac{U_{i_1 j} W_{i_1 k}}{\theta_{i_1}^*(y)^2} + \frac{2\theta_{i_1}^*(x)}{\theta_{i_1}^*(y)^3}  W_{i_1 j} W_{i_1 k} \right)
\sum\limits_{k_2\ne i_2} 
\left( - \frac{U_{i_2k_2} W_{i_2k_2}}{\theta_{i_2}^*(y)^2} + \frac{2\theta_{i_2}^*(x)}{\theta_{i_2}^*(y)^3}  W_{i_2k_2}^2 \right) 
\end{array}
\right) 
\end{array}
\right] \notag \\
& \asymp 
\frac{\min\{q_n^{\text{pre}}, q_n^{\text{post}}\}^2 n (q_n^{\text{post}})^2 }{n^4 (q_n^{\text{post}})^4}
+ \frac{\min\{q_n^{\text{pre}}, q_n^{\text{post}}\} n q_n^{\text{pre}} (q_n^{\text{post}})^2 }{n^4 (q_n^{\text{post}})^4} \notag \\
& \asymp
\frac{\min\{q_n^{\text{pre}}, q_n^{\text{post}}\}  q_n^{\text{pre}}  }{n^3 (q_n^{\text{post}})^2}.
\label{eq:AijAikr1r1}
\end{align}
By combining these results, we have 
\begin{align*}
V(S_1^1)
\asymp &
\frac{1}{ (n-1)^2 (n q_n^{\text{pre}})^2 }
\left( \frac{q_n^{\text{pre}}}{(q_n^{\text{post}})^2} 
+  \frac{ n (q_n^{\text{pre}})^2 }{(q_n^{\text{post}})^2} 
+ \frac{ \min\{q_n^{\text{pre}}, q_n^{\text{post}}\}^2 }{( q_n^{\text{post}})^3}
+ \frac{\min\{q_n^{\text{pre}}, q_n^{\text{post}}\} n q_n^{\text{pre}}  }{ ( q_n^{\text{post}})^2 }
\right)  \\
\asymp & 
\frac{ 1 }{ n^3  ( q_n^{\text{post}})^2 }.
\end{align*}
In particular, this implies that 
\begin{equation}
 S_1^1 = O_{\mathbb{P}}\left( \frac{ 1 }{\sqrt{n} n q_n^{\text{post}}}  \right)
+ O_{\mathbb{P}}\left( \frac{\min\{q_n^{\text{pre}}, q_n^{\text{post}} \} }{ n^2 q_n^{\text{pre}} (q_n^{\text{post}})^2 } \right).   
\label{eq:Zir1_S1}
\end{equation}
Next, for $S_2$,
define $B_{ij}
= \frac{ A_{ij}^{\text{pre}} r_{1,i} (N_i - (n-1)q_n^{\text{pre}} g_0(i)) }{ (n-1) q_n^{\text{pre}} g_0(i) N_i }$.
Then we bound $S_2$ in $L_2$ norm:
\begin{align*}
E[(S_2^1)^2]
=& \frac{1}{n^2  }\sum_{j=1}^n E\left[ (T_j - \pi)^2 \left( \sum_{i \ne j}^n B_{ij}^2 + \sum_{(i_1, i_2)} B_{i_1j} B_{i_2j} \right) \right] \\
&+ \frac{1}{n^2 } \sum_{(j,k)} E\left[ (T_j - \pi)(T_k - \pi) 
\left(  \sum_{i \ne j}^n B_{ij} B_{ik} + \sum_{(i_1,i_2)} B_{i_1j} B_{i_2k} \right) \right].
\end{align*}
For the first term, by \eqref{eq:inverBij^2} and analogous argument to \eqref{eq:A_ijr1}, we have
\begin{align*}
E\left[ (T_j - \pi)^2 B_{ij}^2  \right]
=& E\left[ E\left[ (T_j - \pi)^2 A_{ij}^{\text{pre}} r_{1,i}^2 \mid w\right] E\left[ \left( \frac{ (N_i - (n-1)q_n^{\text{pre}} g_0(i)) }{ (n-1) q_n^{\text{pre}} N_i } \right)^2 \mid w\right] \right] \\
\le & \frac{1}{(n q_n^{\text{pre}})^3}  E\left[ (T_j - \pi)^2 A_{ij}^{\text{pre}} r_{1,i}^2 \right] 
\le C \frac{ q_n^{\text{pre}} }{ n^2 (q_n^{\text{post}})^2} \frac{1}{(n q_n^{\text{pre}})3} 
= C \frac{ 1 }{ n^5 (q_n^{\text{post}})^2 (q_n^{\text{pre}})^2}.
\end{align*}
For the second term, by \eqref{eq:inverBij} and analogous argument to \eqref{eq:AijAkjrirk}, we have
\begin{align*}
& E\left[ (T_j - \pi)^2 B_{i_1j} B_{i_2j} \right] \\ 
=& E\left[   
E\left[ (T_j - \pi)^2 A_{i_1j}^{\text{pre}} A_{i_2j}^{\text{pre}} r_{1,i_1} r_{1,i_2} \mid w \right] E\left[ \frac{ (N_{i_1} - (n-1)q_n^{\text{pre}} g_0(i_1)) }{ (n-1) q_n^{\text{pre}} g_0(i_1) N_{i_1} } \mid w \right] 
E\left[\frac{  (N_{i_2} - (n-1)q_n^{\text{pre}} g_0(i_2)) }{ (n-1) q_n^{\text{pre}} g_0(i_2) N_{i_2} } \mid w \right]  \right] \\
\le & \frac{1}{(n q_n^{\text{pre}})^4} \left( C_1 \frac{\min\{q_n^{\text{pre}}, q_n^{\text{post}}\}^2}{n^3 (q_n^{\text{post}})^2} + C_2 \frac{(q_n^{\text{pre}})^2}{n^2 (q_n^{\text{post}})^2} \right)
\le C \frac{1}{n^6 (q_n^{\text{post}})^2 (q_n^{\text{pre}})^2}.
\end{align*}
For the third term, by \eqref{eq:inverBij^2} and analogous argument to \eqref{eq:AijAikr1i}, we have
\begin{align*}
& E\left[ (T_j - \pi)(T_k-\pi) B_{ij} B_{ik} \right] 
= E\left[ 
E\left[ (T_j - \pi)(T_k-\pi) A_{ij}^{\text{pre}} A_{ik}^{\text{pre}} r_{1,i}^2 \mid w \right]
E\left[ \left( \frac{ N_{i} - (n-1)q_n^{\text{pre}} g_0(i) }{ (n-1) q_n^\text{pre} g_0(i) N_{i} } \right)^2 \mid w \right] 
\right] \\
\le & C \frac{\min\{q_n^{\text{pre}}, q_n^{\text{post}}\}^2  }{(n q_n^{\text{post}})^3} 
\frac{1}{n^3 (q_n^{\text{pre}})^3}
= C \frac{\min\{q_n^{\text{pre}}, q_n^{\text{post}}\}^2  }{n^6 (q_n^{\text{pre}})^3 (q_n^{\text{post}})^3}.
\end{align*}
For the fourth term, by \eqref{eq:inverBij} and analogous argument to \eqref{eq:AijAikr1r1}, we have
\begin{align*}
& E\left[ (T_j - \pi) (T_k - \pi) B_{i_1j} B_{i_2k} \right] \\
=& E\left[ 
E\left[ (T_j - \pi) (T_k - \pi) A_{i_1j}^{\text{pre}} A_{i_2k}^{\text{pre}} r_{1,i_1} r_{1,i_2} \mid w \right] 
E\left[ \frac{ N_{i_1} - (n-1)q_n^{\text{pre}} g_0(i_1) }{ (n-1)q_n^{\text{pre}} g_0(i_1) N_{i_1} } \mid w \right] 
E\left[ \frac{ N_{i_2} - (n-1)q_n^{\text{pre}} g_0(i_2) }{ (n-1)q_n^{\text{pre}} g_0(i_2) N_{i_2} } \mid w \right] \right] \\
\le & \left( C_1 \frac{ \min\{q_n^{\text{pre}}, q_n^{\text{post}}\}^2 }{ n^4 (q_n^{\text{post}})^4 } 
+ C_2 \frac{ \min\{q_n^{\text{pre}}, q_n^{\text{post}}\} n q_n^{\text{pre}} (q_n^{\text{post}})^2 }{ n^4 (q_n^{\text{post}})^4 } \right) \frac{1}{(n q_n^{\text{pre}})^4} \\
=& C_1 \frac{ \min\{q_n^{\text{pre}}, q_n^{\text{post}}\}^2 }{ n^8 q_n^{\text{post}} (q_n^{\text{pre}})^4} 
+ C_2 \frac{ \min\{q_n^{\text{pre}}, q_n^{\text{post}}\} }{ n^7 (q_n^{\text{post}})^2 (q_n^{\text{pre}})^3}.
\end{align*}
By combining these results, we have 
\begin{align*}
E\left[(S_2^1)^2\right]
=& O\left( \frac{1 }{n^5 (q_n^{\text{post}})^2 (q_n^{\text{pre}})^2} \right)
+ O\left( \frac{ \min\{q_n^{\text{pre}}, q_n^{\text{post}}\}^2 }{ n^6 q_n^{\text{post}} (q_n^{\text{pre}})^4 } \right)
\end{align*}
and thus 
\begin{align}
S_2^1
= O_{\mathbb{P}}\left( \frac{ 1 }{ \sqrt{n} n^2 q_n^{\text{pre}} q_n^{\text{post}}} \right)
+ O_{\mathbb{P}}\left( \frac{ \min\{q_n^{\text{pre}}, q_n^{\text{post}}\} }{ n^3 (q_n^{\text{post}})^2 (q_n^{\text{pre}})^2 }  \right).
\label{eq:Zir1_S2}
\end{align}
In particular, combining \eqref{eq:Zir1_S1} and \eqref{eq:Zir1_S2} implies that 
\begin{align}
\frac{1}{n} \sum_{i=1}^n (Z_i^{\text{alt}} - \pi) r_{1,i} 
=& O_{\mathbb{P}}\left( \frac{ 1 }{\sqrt{n} n q_n^{\text{post}}}  \right)
+ O_{\mathbb{P}}\left( \frac{\min\{q_n^{\text{pre}}, q_n^{\text{post}}\} }{ n^2 q_n^{\text{pre}} (q_n^{\text{post}})^2 } \right).
\label{eq:Zir1_alt}
\end{align}
By combining \eqref{eq:Ziroi_alt} and \eqref{eq:Zir1_alt}, we have 
\begin{align}
\frac{1}{n} \sum_{i=1}^n (Z_i^{\text{alt}} - \pi) (r_{0,i} + r_{1,i}) 
=& O_{\mathbb{P}}\left( \frac{ 1 }{n \max\{q_n^{\text{pre}}, q_n^{\text{post}}\} } \right)
+ O_{\mathbb{P}}\left( \frac{ 1 }{\sqrt{n} n q_n^{\text{post}}}  \right).
\label{eq:eq:Zir0ir1i}
\end{align}
Combining \eqref{eq:Zigi} and \eqref{eq:eq:Zir0ir1i} together completes the proof.

Also, we can show that by Lemma \ref{lemma:Zi_alt_phii} and \eqref{eq:mu_r1i}, 
\begin{align*}
& \frac{1}{n} \sum_{i=1}^n (Z_i^{\text{alt}} - \pi) \mu_{r_1,i} 
= \frac{1}{n} \sum_{i=1}^n (T_i - \pi) E\left[ \frac{A_{ij}^{\text{pre}} \mu_{r_1,j} }{q_n^{\text{pre}} g_0(j)} \mid w_i\right]
+ O_{\mathbb{P}}\left( \frac{ 1 }{\sqrt{n} \sqrt{nq_n^{\text{pre}}} n q_n^{\text{post}}}  \right) 
\end{align*}
and thus 
\begin{align*}
\frac{1}{n} \sum_{i=1}^n (Z_i^{\text{alt}} - \pi) (r_{0,i} + r_{1,i} - \mu_{r_1,i})
= O_{\mathbb{P}}\left( \frac{ 1 }{n \max\{q_n^{\text{pre}}, q_n^{\text{post}}\} } \right)
+ O_{\mathbb{P}}\left( \frac{ 1 }{\sqrt{n} n q_n^{\text{post}}}  \right).
\end{align*}

\end{appendix}

\end{document}